\documentclass[10pt]{phdthesis}

\usepackage{amsthm}
\usepackage{changepage}
\usepackage{fancyhdr}
\usepackage{amssymb,latexsym,pifont}
\usepackage{graphics}
\usepackage{float}
\usepackage{captionhack}
\usepackage{hyperref}
\usepackage{bm,upgreek}
\usepackage{enumitem}
\usepackage{parskip}
\usepackage{titlesec}
\usepackage{multirow}
\usepackage[toc,page]{appendix}
\usepackage[caption=false,labelformat=simple]{subfig}
\usepackage{emptypage}
\usepackage{mathtools}
\usepackage{algorithm,algpseudocode}
\usepackage{tikz}
\usepackage{setspace}
\usepackage{cite}
\usepackage{microtype}
\usepackage{pgfplots}
\usepackage{pgfplotstable}
\usepgfplotslibrary{statistics}
\usepackage{tikz}
\usetikzlibrary{spy}
\usepackage{etoolbox}
\usepackage[mathscr]{euscript}
\usepackage{amsbsy}
\usepackage{euscript}
\usepackage[most]{tcolorbox}
\usepackage{arydshln}
\usepackage{booktabs}

\hypersetup{colorlinks=true,citecolor=blue,linkcolor=black}

\titlespacing*{\section}      {0em}{.75em}{0.25em}
\titlespacing*{\subsection}   {0em}{1.00em}{.0em}
\titlespacing*{\subsubsection}{0em}{.50em}{.25em}
\titlespacing*{\paragraph}    {0em}{.25em}{.25em}

\tcbset{colback=blue!5,boxrule=0pt,left=5pt,right=5pt,top=3pt,bottom=3pt,
title filled, oversize,breakable,enhanced,before upper={\parindent15pt}, before
skip=15pt plus 2pt,after skip=15pt plus 2pt}

\newcommand\ph[1]{\mathscr{#1}}

\newcommand\xqed[1]{%
\leavevmode\unskip\penalty9999 \hbox{}\nobreak\hfill
\quad\hbox{#1}}
\newcommand\demo{\xqed{$\triangle$}}

\newtheorem{theorem}{Theorem}[section]
\newtheorem{example}[theorem]{Example}
\newtheorem{remark}{Remark}

\makeatletter
\def\myline{\pgfutil@ifnextchar[{\my@line}{\my@line[]}}%
\def\my@line[#1](#2)(#3){%
\tikz[overlay] \draw[#1]  (#2)--(#3);}%

\makeatletter
  \renewcommand*\env@matrix[1][*\c@MaxMatrixCols c]{%
    \hskip -\arraycolsep
    \let\@ifnextchar\new@ifnextchar
  \array{#1}}
\makeatother

\newcommand\chap[1]{%
  \chapter*{#1}%
  \addcontentsline{toc}{chapter}{#1}}
  
\makeatletter
\newcommand{\tocfill}{\cleaders\hbox{$\m@th \mkern\@dotsep mu . \mkern\@dotsep mu$}\hfill}
\makeatother

\newenvironment{abbreviations}{\begin{list}{}{%
\setlength{\labelwidth}{3cm}\setlength{\leftmargin}{\labelwidth+\labelsep}%
\setlength{\itemsep}{-3pt}}}{\end{list}}

\pgfplotsset{compat=1.5}
\pgfplotsset{grid style={dotted,gray}}
\pgfplotsset{legend image with text/.style={legend image code/.code={%
\node[anchor=west, align=right] at (0.0cm,0cm) {#1};}},}

\pgfplotsset{
  box plot/.style={/pgfplots/.cd,black,only marks,mark=-,
  mark size=\pgfkeysvalueof{/pgfplots/box plot width},
  /pgfplots/error bars/y dir=plus,/pgfplots/error bars/y explicit,
  /pgfplots/table/x index=\pgfkeysvalueof{/pgfplots/box plot x index},},
  box plot box/.style={
  /pgfplots/error bars/draw error bar/.code 2 args={%
  \draw  ##1 -- ++(\pgfkeysvalueof{/pgfplots/box plot width},0pt) 
  |- ##2 -- ++(-\pgfkeysvalueof{/pgfplots/box plot width},0pt) |- ##1 -- cycle;},
  /pgfplots/table/.cd,
  y index=\pgfkeysvalueof{/pgfplots/box plot box top index},
  y error expr={
  \thisrowno{\pgfkeysvalueof{/pgfplots/box plot box bottom index}}
  - \thisrowno{\pgfkeysvalueof{/pgfplots/box plot box top index}}},
  /pgfplots/box plot},
  box plot top whisker/.style={
  /pgfplots/error bars/draw error bar/.code 2 args={%
  \pgfkeysgetvalue{/pgfplots/error bars/error mark}%
  {\pgfplotserrorbarsmark}%
  \pgfkeysgetvalue{/pgfplots/error bars/error mark options}%
  {\pgfplotserrorbarsmarkopts}%
  \path ##1 -- ##2;},/pgfplots/table/.cd,
  y index=\pgfkeysvalueof{/pgfplots/box plot whisker top index},
  y error expr={
  \thisrowno{\pgfkeysvalueof{/pgfplots/box plot box top index}}
  - \thisrowno{\pgfkeysvalueof{/pgfplots/box plot whisker top index}}},
  /pgfplots/box plot},
  box plot bottom whisker/.style={
  /pgfplots/error bars/draw error bar/.code 2 args={%
  \pgfkeysgetvalue{/pgfplots/error bars/error mark}%
  {\pgfplotserrorbarsmark}%
  \pgfkeysgetvalue{/pgfplots/error bars/error mark options}%
  {\pgfplotserrorbarsmarkopts}%
  \path ##1 -- ##2;},
  /pgfplots/table/.cd,
  y index=\pgfkeysvalueof{/pgfplots/box plot whisker bottom index},
  y error expr={
  \thisrowno{\pgfkeysvalueof{/pgfplots/box plot box bottom index}}
  - \thisrowno{\pgfkeysvalueof{/pgfplots/box plot whisker bottom index}}},
  /pgfplots/box plot},
  box plot median/.style={/pgfplots/box plot,
  /pgfplots/table/y index=\pgfkeysvalueof{/pgfplots/box plot median index},
  semithick,black },
  box plot width/.initial=1em,
  box plot x index/.initial=0,
  box plot median index/.initial=1,
  box plot box top index/.initial=2,
  box plot box bottom index/.initial=3,
  box plot whisker top index/.initial=4,
  box plot whisker bottom index/.initial=5,}

\newcommand{\boxplot}[2][]{
    \addplot [box plot median,#1] table {#2};
    \addplot [forget plot, box plot box,#1] table {#2};
    \addplot [forget plot, box plot top whisker,#1] table {#2};
    \addplot [forget plot, box plot bottom whisker,#1] table {#2};}

\begin{document}

\pagestyle{fancyplain}
\pagenumbering{roman}
\allowdisplaybreaks

\newcommand{\publ}{}
\renewcommand{\sectionmark}[1]{\markright{\it \thesection.\ #1}}
\renewcommand{\chaptermark}[1]{\markboth{\it \thechapter.\ #1}{}}
\lhead[\thepage]{\fancyplain{\publ}{\rightmark}}
\chead[\fancyplain{}{}]{\fancyplain{}{}}
\rhead[\fancyplain{}{\leftmark}]{\fancyplain{}{\thepage}}
\lfoot[]{}
\cfoot[]{}
\rfoot[]{}
\pagenumbering{arabic}

\thispagestyle{empty}

\begin{adjustwidth*}{-7.5mm}{-7.5mm}
\begin{flushleft}
{\LARGE \textsc{\textbf{Design and Analysis of \\ Distributed State Estimation Algorithms \\ Based on Belief Propagation and\\ \vskip 2.1mm Applications in Smart Grids}}}\\ \vskip 16mm
{\large \textbf{Mirsad \'{C}osovi\'{c}}}
\end{flushleft}
\end{adjustwidth*}

\pagebreak
\thispagestyle{empty}
\vspace*{9cm}


\thispagestyle{empty}

\begin{centering}

{\Large Design and Analysis of \\ Distributed State Estimation Algorithms \\ Based on Belief Propagation and\\ \vskip 2.1mm Applications in Smart Grids}


{\large by\\
Mirsad \'Cosovi\'c\\
}

\vskip 0.7cm

Mr.-Ing. Power Electrical Engineering, University of Sarajevo, \\Bosnia and Herzegovina, 2013.\\
Dipl.-Ing. Power Electrical Engineering, University of Sarajevo, \\Bosnia and Herzegovina, 2009.

\vskip 0.7cm

for the degree of \\ \vspace{1.0cm}
{\Large Doctor of Technical Sciences}

\vskip 0.7cm
{\large A dissertation submitted to the\\
\vskip 0.3cm 
Department of Power, Electronics\\
and Communication Engineering,\\ 
Faculty of Technical Sciences,\\
University of Novi Sad,\\
Serbia.\\}

\vspace{1.0cm}

2018.\\

\end{centering}

\vskip 4cm
\pagebreak
\thispagestyle{empty}

{\Large Advisor:}\\
\\
\hspace*{1.0cm} {\Large Dr Dejan Vukobratovi\'c}, {Associate Professor}\\  
\hspace*{1.0cm} Department of Power, Electronics and Communication Engineering,\\
\hspace*{1.0cm} University of Novi Sad, Serbia.

\vspace*{1cm}
{\Large Thesis Committee Members:} \\
\\
\hspace*{1.0cm} {\Large Dr Andrija Sari\'c}, {Full Professor}\\ 
\hspace*{1.0cm} Department of Power, Electronics and Communication Engineering,\\
\hspace*{1.0cm} University of Novi Sad, Serbia.\\
\\
\hspace*{1.0cm} {\Large Dr Petar Popovski}, {Full Professor}\\ 
\hspace*{1.0cm} Department of Electronic Systems,\\
\hspace*{1.0cm} Aalborg University, Denmark.\\
\\
\hspace*{1.0cm} {\Large Dr \v{C}edomir Stefanovi\'c}, {Associate Professor}\\ 
\hspace*{1.0cm} Department of Electronic Systems,\\
\hspace*{1.0cm} Aalborg University Copenhagen, Denmark.\\
\\
\hspace*{1.0cm} {\Large Dr Izudin D\v{z}afi\'c}, {Full Professor},\\ 
\hspace*{1.0cm} Department of Electrical Engineering,\\
\hspace*{1.0cm} International University of Sarajevo, Bosnia and Herzegovina.\\       
\\
\hspace*{1.0cm} {\Large Dr Du\v{s}an Jakoveti\'c}, {Assistant Professor}\\ 
\hspace*{1.0cm} Department of Mathematics and Informatics,\\
\hspace*{1.0cm} University of Novi Sad, Serbia.\\  

\vspace*{1cm}
This research has received funding from the EU 7th Framework Programme for research, technological development and demonstration under grant agreement no. 607774.
\thispagestyle{empty}

\vspace*{1cm}

{\flushright \emph{Engineering: where the semi-skilled laborers \\execute the vision of those who think and dream. \\ Hello, Oompa Loompas of science.}\\Dr. Sheldon Lee Cooper, B.S., M.S., M.A., Ph.D., Sc.D.\\}

\tableofcontents
\chap{Acknowledgments} 
\pagestyle{empty}

This thesis is a final result of an incredible journey that has lasted for the last four years. It is impossible to list all the incredible people that I met during this period, who left an indelible mark in my life. I would like to point out that the value of this thesis is not in mathematical equations, its main value is that it made me a better person. 

First and foremost, I would like to dedicate this thesis to my family, they have always been with me. I want to express my deep gratitude for their support, encouraging and love. 

Further, I would like to express my deepest gratitude to my friend and advisor Prof. Dejan Vukobratovic, without his support, nothing would be possible. There are simply no words to describe his influence on this work and on my life. You can meet such a person and scientist once in your life if you have the very best of luck, few people had so much positive influence on my life as he had. 

I owe very much to Dragana Bajovic, Dusan Jakovetic and Gorana Mijatovic, they were always there for me, without them these four years would be significantly different. Also, I would like to thank all members of the Communications and Signal Processing Group at Faculty of Technical Sciences, University of Novi Sad, 

I would like to give special thanks to ADVANTAGE and SENSIBLE projects, and special thanks to Hazel Cox, Prof. John Thompson, Prof. Cedomir Stefanovic, Charalampos Kalalas, Marko Angjelichinoski, Achilleas Tsitsimelis and Alexandros Kleidaras, for all wonderful moments that we shared through ADVANTAGE project. 

Finally, I owe my thanks to the great people who hosted me during my secondment time, Prof. Carles Anton-Haro, CTTC, Barcelona; Prof. Vladimir Stankovic, University of Strathclyde, Glasgow; Prof. Juraj Machaj, University of Zilina, Zilina; and Prof. Chao Wang, Tongji University, Shanghai.

\chap{List of Publications and Awards} 
\pagestyle{empty}

{\Large Journal Publications:}

\noindent
M. Cosovic and D. Vukobratovic, “Distributed Gauss-Newton Method for State
Estimation Using Belief Propagation,” in IEEE Transactions on Power Systems, 2018 (early access).

\noindent
M. Cosovic, A. Tsitsimelis, D. Vukobratovic, J. Matamoros and C. Anton-Haro, "5G Mobile Cellular Networks: Enabling Distributed State Estimation for Smart Grids," in IEEE Communications Magazine, vol. 55, no. 10, pp. 62-69, October 2017.

\vspace{0.8cm}
\noindent
{\Large Conference Publications:}

\noindent
M. Cosovic, D. Vukobratovic and V. Stankovic, "Linear state estimation via 5G C-RAN cellular networks using Gaussian belief propagation," 2018 IEEE Wireless Communications and Networking Conference (WCNC), Barcelona, 2018, pp. 1-6.

\noindent
M. Cosovic and D. Vukobratovic, "Fast real-time DC state estimation in electric power systems using belief propagation," 2017 IEEE International Conference on Smart Grid Communications (SmartGridComm), Dresden, 2017, pp. 207-212.

\noindent
A. Kleidaras, M. Cosovic, D. Vukobratovic and A. E. Kiprakis, "Demand response for thermostatically controlled loads using belief propagation," 2017 IEEE PES Innovative Smart Grid Technologies Conference Europe (ISGT-Europe), Torino, 2017, pp. 1-6.

\noindent
M. Cosovic and D. Vukobratovic, "Distributed Gauss-Newton method for AC state estimation: A belief propagation approach," 2016 IEEE International Conference on Smart Grid Communications (SmartGridComm), Sydney, NSW, 2016, pp. 643-649.

\noindent
M. Cosovic and D. Vukobratovic, "State estimation in electric power systems using belief propagation: An extended DC model," 2016 IEEE 17th International Workshop on Signal Processing Advances in Wireless Communications (SPAWC), Edinburgh, 2016, pp. 1-5.

\vspace{0.8cm}
\noindent
{\Large Book Chapter:}

\noindent
M. Angjelichinoski, M. Cosovic, C. Kalalas, R. Lliuyacc, M. Zeinali, J. Alonso-Zarate, J. M. Mauricio, P. Popovski, C. Stefanovic, J. S. Thompson and D. Vukobratovic, “Overview of research in the ADVANTAGE project,” Chapter 12, “Book title: Smarter Energy: from Smart Metering to the Smart Grid,” Editors: H. Sun, N.D. Hatziargyriou, H.V. Poor, L. Carpanini, and M. Forni{\'e}, ser. Energy Engineering,
Institution of Engineering and Technology, 2016.

\vspace{0.8cm}
\noindent
{\Large Awards:}

\noindent
Early Career Research Award in NSF US-Serbia \& West Balkan Data Science Workshop, Belgrade, Serbia, 2018, for the poster ``Distributed Power System State Estimation Algorithms Based on the Belief Propagation".

\noindent
Best Student Paper Award in IEEE International Conference on Smart Grid Communications (SmartGridComm), Dresden, Germany, 2017, for the paper ``Fast real-time DC state estimation in electric power systems using belief propagation”.

\phantomsection
\cleardoublepage
\addcontentsline{toc}{chapter}{\listfigurename}
\listoffigures

\cleardoublepage
\addcontentsline{toc}{chapter}{\listtablename}
\phantomsection
\listoftables

\chap{Abstract} 
\pagestyle{empty}

We present a detailed study on application of factor graphs and the belief propagation (BP) algorithm to the power system state estimation (SE) problem. We start from the BP solution for the linear DC model, for which we provide a detailed convergence analysis. Using BP-based DC model we propose a fast real-time state estimator for the power system SE. The proposed estimator is easy to distribute and parallelize, thus alleviating computational limitations and allowing for processing measurements in real time. The presented algorithm may run as a continuous process, with each new measurement being seamlessly processed by the distributed state estimator. In contrast to the matrix-based SE methods, the BP approach is robust to ill-conditioned scenarios caused by significant differences between measurement variances, thus resulting in a solution that eliminates observability analysis. Using the DC model, we numerically demonstrate the performance of the state estimator in a realistic real-time system model with asynchronous measurements. We note that the extension to the non-linear SE is possible within the same framework.

Using insights from the DC model, we use two different approaches to derive the BP algorithm for the non-linear model. The first method directly applies BP methodology, however, providing only approximate BP solution for the non-linear model. In the second approach, we make a key further step by providing the solution in which the BP is applied sequentially over the non-linear model, akin to what is done by the Gauss-Newton method. The resulting iterative Gauss-Newton belief propagation (GN-BP) algorithm can be interpreted as a distributed Gauss-Newton method with the same accuracy as the centralized SE, however, introducing a number of advantages of the BP framework. The thesis provides extensive numerical study of the GN-BP algorithm, provides details on its convergence behavior, and gives a number of useful insights for its implementation.

Finally, we define the bad data test based on the BP algorithm for the non-linear model. The presented model establishes local criteria to detect and identify bad data measurements. We numerically demonstrate that the BP-based bad data test significantly improves the bad data detection over the largest normalized residual test.
\chap{Abbreviations} 
\pagestyle{empty}


\begin{abbreviations}
\item[AC-BP]	Native Belief Propagation Approximate Solution for the Non-Linear State 					Estimation Model
\item[BP]		Belief Propagation
\item[BP-BDT]	Belief Propagation based Bad Data Test
\item[CDF]		Cumulative Density Function
\item[DC]		Direct Current
\item[DC-BP]	Belief Propagation based DC State Estimation Algorithm
\item[EMS]		Energy Management System
\item[GN-BP]	Gauss-Newton Belief Propagation based Algorithm
\item[LNRT]		Largest Normalized Residual Test
\item[MAD]		Mean Absolute Difference
\item[MAP]		Maximum a Posteriori
\item[PMU]		Phasor Measurement Unit
\item[SCADA]	Supervisory Control and Data Acquisition
\item[SE]		State Estimation
\item[WAMS]		Wide Area Measurement System
\item[WLS]		Weighted Least-Squares
\item[WRSS]		Weighted Residual Sum of Squares
\item[5G]		Fifth-Generation
\end{abbreviations}

\pagestyle{fancyplain}
\chapter{Introduction}\label{ch:introduction}
The major topic of the thesis is to provide novel distributed state estimation (SE) algorithms applicable to electric power systems. In essence, we provide algorithms that solve systems of linear and non-linear equations with real coefficients and variables. Consequently, the implications of our results go far beyond SE in electric power systems and can be applied in different areas, such as for demand response \cite{kleidaras} or water distribution systems \cite{nalini}.

Proposed SE algorithms are suitable to cope with near-real-time and asynchronous operation requirements, bypassing established routines (e.g., system observability). They are flexible and easy to distribute across local processors that are located at different physical locations, and/or in parallel fashion, where local processors run in parallel at the same physical place. Novel algorithms do not involve direct matrix inversion, which makes them attractive from the point of computational complexity and in some special conditions are numerically more stable.

In this chapter, we present the formulation of the problems that we intend to solve and introduce the basic terms, giving the reader a clearer picture of the problems. We clearly state assumptions and limitations that we use throughout this thesis and present main advantages over the current state-of-the-art SE models in electric power systems. Finally, we note that results presented in the thesis are based on our previous publications with additional clarifications, and enriched with many useful examples.

\section{Power System State Estimation}  
Electric power systems consist of generation, transmission and consumption spread over wide geographical areas and operated from the control centers by the system operators. Maintaining normal operation conditions is of the central importance for the power system operators \cite[Ch.~1]{abur}. Control centers are traditionally operated in centralized and independent fashion. However, increase in the system size and complexity, as well as external socio-economic factors, lead to deregulation of power systems, resulting in decentralized structure with distributed control centers. Cooperation in control and monitoring across distributed control centers is critical for efficient system operation. Consequently, existing centralized algorithms have to be redefined based on a new requirements for distributed operation, scalability and computational efficiency \cite{wu}.

The system monitoring is an essential part of the control centers, providing control and optimization functionality whose efficiency relies on accurate SE. The centralized SE assumes that the measurements collected across the system are available at the control center, where the centralized SE algorithm provides the system state estimate. Precisely, the centralized SE algorithm typically uses the Gauss-Newton method to solve the non-linear weighted least-squares (WLS) problem \cite{monticelli}, \cite{schweppe}. In contrast, decentralized SE distributes communication and computational effort across multiple control centers to provide the system state estimate. There are two main approaches to distributed SE: i) algorithms which require a global control center to exchange data with local control centers, and ii) algorithms with local control centers only \cite{gupta}. Distributed SE algorithms target the same state estimate accuracy as achievable using the centralized SE algorithms. 

Input data for the SE arrive from supervisory control and data acquisition (SCADA) technology. SCADA provides communication infrastructure to collect legacy measurements (voltage and line current magnitude, power flow and injection measurements) from measurement devices and transfer them to a central computational unit for processing and storage. In the last decades, phasor measurement units (PMUs) were developed that measure voltage and line current phasors and provide highly accurate measurements with high sampling rates. PMUs were instrumental to the development of the wide area measurement systems (WAMSs) that should provide real-time monitoring and control of electric power systems \cite{zhu, anna, bose2010smart}. The WAMS requires significant investments in deployment of a large number of PMUs across the system, which is why SCADA systems will remain important technology, particularly at medium and low voltage levels. However, with the evolution and adoption of PMU technology and, consequently, with decline in price of PMUs, it is realistic to assume that future power systems will be fully observable by PMUs \cite{gol}.  Exploiting PMU inputs by robust, decentralized and real-time SE solution calls for novel distributed algorithms and communication infrastructure that would support future WAMS and aims to detect and counteract power grid disturbances in real-time \cite{commag, terzija}. 

Monitoring and control capability of the system, besides the SE accuracy, strongly depends on the periodicity of evaluation of state estimates. Ideally, in the presence of both legacy and phasor measurements, SE should run at the scanning rate (seconds), but due to the computational limitations, practical SE algorithms run every few minutes or when a significant change occurs \cite{monticelli}. 
 
\subsection{Distributed SE Algorithms}
The mainstream distributed SE algorithms exploit matrix decomposition techniques applied over the Gauss-Newton method. These algorithms usually achieve the same accuracy as the centralized SE algorithm and work either with global control center \cite{korres, jiang, aburali, contaxis} or without it \cite{minot, marelli, tai, reza}. Furthermore, SE algorithms based on distributed optimization \cite{conejo}, and in particular, the alternating direction method of multipliers \cite{boyddistributed} became very popular\cite{giannakis, kekatos, matamoros}. Authors in \cite{anna} present the robust decentralized Gauss-Newton algorithm which provides flexible communication model, but suffers from slight performance degradation compared to the centralized SE. The work in \cite{poor} proposed a fully distributed SE algorithm for wide-area monitoring which provably converges to the centralized SE. The paper \cite{chakrabarti} proposed a new multi-area SE approach with the central coordinator, where is no requirement to share the topology information among the sub-areas and from sub-areas to the central coordinator. Recently, in \cite{guo}, a new hierarchical multi-area SE method is proposed, where the algorithm converges close to the centralized SE solution with improved convergence speed. We refer the reader to \cite{gomeztax} for a detailed survey of the distributed multi-area SE. In addition, we note that most of the distributed SE papers implicitly consider wide-area monitoring and transmission grid scenario, which is the approach we follow in this thesis.

\section{Belief Propagation Approach}
In this thesis, we solve the SE problem using probabilistic graphical models \cite{pearl}, a powerful tool for modeling the independence/dependence relationships among the systems of random variables \cite[Ch.~4]{barber}. Graphical models are useful since they provide a framework for studying a wide class of probabilistic models and associated algorithms. Factor graph represents a graphical model which allows a graph-based representation of probability density functions using variable and factor nodes connected by edges. In contrast to directed and undirected graphical models, factor graphs provide the details of the factorization in more explicit way \cite[Ch.~8]{bishop}. 

We represent the SE problem using factor graphs and solve it using the belief propagation (BP) algorithm. Applying the BP algorithm on probabilistic graphical models without loops, one obtains exact marginal distributions or a mode of the joint distribution of the system of random variables \cite{pearl}, \cite{bishop}. The BP algorithm can be also applied to graphical models with loops (loopy BP)\cite{loop}, although in that case, the solution may not converge to the correct marginals/modes of the joint distribution. BP is a fully distributed algorithm suitable for accommodation of distributed power sources and time-varying loads. Moreover, placing the SE into the probabilistic graphical modelling framework enables not only efficient inference, but also, a rich collection of tools for learning parameters or structure of the graphical model from observed data \cite{KollerFreidmanBook, bajovic}.

In the standard setup, the goal of the BP algorithm is to efficiently evaluate the marginals of a system of random variables $\mathbf y = [y_1,\dots,y_n]^{\mathrm{T}}$ described via the joint probability density function $g(\mathbf{y})$\footnote{With a slight abuse of notation, here we use $\mathbf{y}$ to define a general system of random variables, hereinafter we use different symbols to describe those. However, throughout the thesis, we use $\mathcal{V}$ to describe the set of nodes.}. Assuming that the function $g(\mathbf{y})$ can be factorized proportionally ($\propto$) to a product of local functions:
		\begin{equation}
        \begin{aligned}
        g(\mathbf{y}) \propto \prod_{i=1}^k \psi_i(\mathcal{V}_i),
        \end{aligned}
		\label{FG_factorize}
		\end{equation}
where $\mathcal{V}_i \subseteq \{y_1,\dots,y_n\}$, the marginalization problem can be efficiently solved using BP algorithm. The first step is forming a factor graph, which is a bipartite graph that describes the structure of the factorization \eqref{FG_factorize}. The factor graph structure comprises the set of factor nodes $\mathcal{F}=\{f_1,\dots,f_k\}$, where each factor node  $f_i$ represents local function $\psi_i(\mathcal{V}_i)$, and the set of variable nodes $\mathcal{V}=\{y_1,\dots,y_n\}$. The factor node $f_i$ connects to the variable node $y_s$ if and only if $y_s \in \mathcal{V}_i$ \cite{kschischang}.

The BP algorithm on factor graphs proceeds by passing two types of messages along the edges of the factor graph: 
\begin{enumerate}[label=(\roman*)]
\item a variable node $y_s \in \mathcal{V}$ to a factor node $f_i \in \mathcal{F}$ message $\mu_{y_s \to f_i}(y_s)$, and  
\item a factor node $f_i \in \mathcal{F}$ to a variable node $y_s \in \mathcal{V} $ message $\mu_{f_i \to y_s}(y_s)$.
\end{enumerate} 
Both variable and factor nodes in a factor graph process the incoming messages and calculate outgoing messages, where an output message on any edge depends on incoming messages from all other edges. BP messages represent "beliefs" about variable nodes, thus a message that arrives or departs a certain variable node is a function (distribution) of the random variable corresponding to the variable node. 

We are employing a loopy BP since the corresponding factor graph usually contains cycles. Loopy BP is an iterative algorithm, and requires a message-passing schedule. Typically, the scheduling where messages from variable to factor nodes, and messages from factor nodes to variable nodes, are updated in parallel in respective half-iterations, is known as synchronous scheduling. Synchronous scheduling updates all messages in a given iteration using the output of the previous iteration as an input \cite{elidan}.

\subsection{Belief Propagation SE Algorithms} 
The work in \cite{kavcic, kavcicconf} provides the first demonstration of BP applied to the SE problem. Although this work is elaborate in terms of using, e.g., environmental correlation via historical data, it applies BP to a linear approximation of the non-linear functions \cite{cain}. The non-linear model is recently addressed in \cite{ilic}, where tree-reweighted BP is applied using preprocessed weights obtained by randomly sampling the space of spanning trees. The work in \cite{fu} investigates Gaussian BP convergence for the DC model. Although the above results provide initial insights on using BP for distributed SE, the BP-based solution for non-linear SE model and the corresponding performance and convergence analysis is still missing. This thesis intends to fill this gap.

\subsection{Belief Propagation Based DC SE Algorithm} 
In general, the DC SE model is obtained by linearisation of the non-linear model, and the model ignores the reactive powers and transmission losses and takes into account only the active powers. Our methodology is to start with the simplest linear DC SE model and use insights obtained therein to derive the BP solution for the non-linear SE model; we refer to the corresponding method as the DC-BP. As a side-goal of this part, we aimed at thorough and detailed presentation of applying BP on the simple DC SE problem in order to make the powerful BP algorithm more accessible and more popular within the power-engineering community \cite{cosovicexdc}. 

Using the DC-BP algorithm, we demonstrate capability of the BP algorithm. More precisely, we propose a fast real-time state estimator based on the BP algorithm. In other words, unlike the usual scenario where measurements are transmitted directly to the control center, in the BP framework, measurements are locally collected and processed by local modules (at substations, generators or load units) that exchange BP messages with neighboring local modules. Furthermore, even in the scenario where measurements are transmitted to the centralized control entity, the BP solution is advantageous over the classical centralized solutions in that it can be easily distributed and parallelized for high performance. We note that the extension to the non-linear SE is possible within the same framework.

Finally, this thesis provides a novel and detailed convergence analysis of the BP-DC algorithm and points to extension of this analysis for the proposed BP-based non-linear SE algorithm, and an improved algorithm that applies synchronous scheduling with randomized damping.

\subsection{Belief Propagation Based Non-Linear SE Algorithms}
The non-linear SE model is defined using the measurement functions that precisely follow the physical laws that connect the measured variables and the state variables. In the process of deriving non-liner algorithms, we provide a step-by-step guide for application of BP algorithm to the SE problem, giving this part of the thesis strong tutorial flavor. 

Using insights from the linear BP-based DC SE model, we derive the native BP solution for the non-linear SE model. Unfortunately, as closed-form expressions for certain classes of BP messages cannot be obtained, that lead us to propose the AC-BP algorithm as an approximate BP solution for the non-linear SE model. However, we include the resulting AC-BP method for methodological reasons, although it is outperformed by the subsequent Gauss-Newton BP (GN-BP) method. 

Finally, as a main contribution, we make a key further step where we change the perspective of our BP approach and, instead of applying the BP directly onto the non-linear SE model, we present the solution where the BP is applied sequentially over the non-linear model, akin to what is done by the Gauss-Newton method. The resulting GN-BP represents a BP counterpart of the Gauss-Newton method achieving the same accuracy, however, preserving a number of advantages brought in by the BP framework.

\section{Contributions}
Some of the contributions have already been mentioned throughout previous discussion, however, as the main contribution, we adopt different methodology to derive efficient BP-based non-linear SE method, and propose the GN-BP algorithm. The GN-BP is the first BP-based solution for the non-linear SE model achieving exactly the same accuracy as the centralized SE via Gauss-Newton method.

In general, solving the SE problems using factor graphs and BP algorithm introduce a number of advantages over the current state-of-the-art in power systems SE algorithms:
\begin{itemize}[leftmargin=*]
\item In comparison with the distributed SE algorithms that exploit matrix decomposition, the BP-based SE algorithms are robust to ill-conditioned scenarios caused by significant differences between measurement variances, thus allowing inclusion of arbitrary number of pseudo-measurements without impact to the solution within the observable islands.
\item Due to the sparsity of the underlying factor graph, the algorithms has optimal computational complexity (linear per iteration), making it particularly suitable for solving large-scale systems.
\item BP-based algorithms can be easily designed to provide \emph{asynchronous} operation and integrated as part of the \emph{real-time} systems where newly arriving measurements are processed as soon as they are received. 
\item Algorithms can easily integrate new measurements: the arrival of a measurement at the control center will define a new factor node which will be seamlessly integrated in the graph as part of the time continuous process. 
\item In the multi-area scenario, BP-based algorithms can be implemented over the non-overlapping multi-area SE scenario without the central coordinator, where algorithms neither requires exchanging measurements nor local network topology among the neighboring areas.
\item BP-based algorithms are flexible and easy to \emph{distribute} and \emph{parallelize}. Thus, even if implemented in the framework of centralized SE, it can be flexibly matched to distributed computation resources (e.g., parallel processing on graphical-processing units).
\item The BP approach allows to define the novel bad data test that significantly improves the bad data detection.
\end{itemize}

Finally, even if electric power systems observable only by PMUs \cite{gol, xupmu, aburpmu} are beyond the thesis scope, we note that the BP can be applied to this problem. Then, in the multi-area scenario, areas exchange only ``beliefs'' about specific state variables, where algorithm ensures data privacy in the distributed architecture. Furthermore, the BP framework allows integration of legacy and phasor measurements in fifth-generation (5G) communication infrastructure, as we demonstrate in \cite{commag, stankovic}.

\section{Assumptions}
In this thesis, we provide BP-based algorithms using following assumptions: 
\begin{itemize}
\item the network topology and parameters are known without errors,
\item power system operates in the steady state under balanced condition,
\item phase shifting transformers are neglected,
\item the SE model is described with an overdetermined system of equations,
\item legacy and phasor measurements are uncorrelated,
\item measurement errors follow a zero-mean Gaussian distribution,
\item complex bus voltages are observed state variables.
\end{itemize}

\section{Summary}
In the thesis, we solve power system SE problems using factor graphs and BP algorithm. We proposed three BP-based algorithms: 
\begin{enumerate}[label=(\roman*)]
\item DC-BP to solve linear DC SE model, 
\item AC-BP that provides an approximate solution of the non-linear SE model, 
\item GN-BP that is the first BP-based solution for the non-linear SE model achieving exactly the same solution as the Gauss-Newton method.
\end{enumerate}
Presented architectures directly exploit system sparsity, can be flexibly paralellized (in the extreme case, the algorithm can be implemented as a fully distributed) and results in substantially lower computational complexity compared to traditional SE solutions.
\chapter{Power System State Estimation}	\label{ch:se_power}
\addcontentsline{lof}{chapter}{2 Power System State Estimation}
In this chapter, we review the state-of-the-art SE models in electric power systems. The power system represents a dynamic system, where power generation and power demand is changing values depending on various factors. The SE is used for describing the present state of the power system, unlike the power flow analysis which is used for defining load profiles, generator capabilities, voltage specification, contingency analysis, and planning. 
	\begin{figure}[ht]
	\centering
	\includegraphics[width=9.5cm]{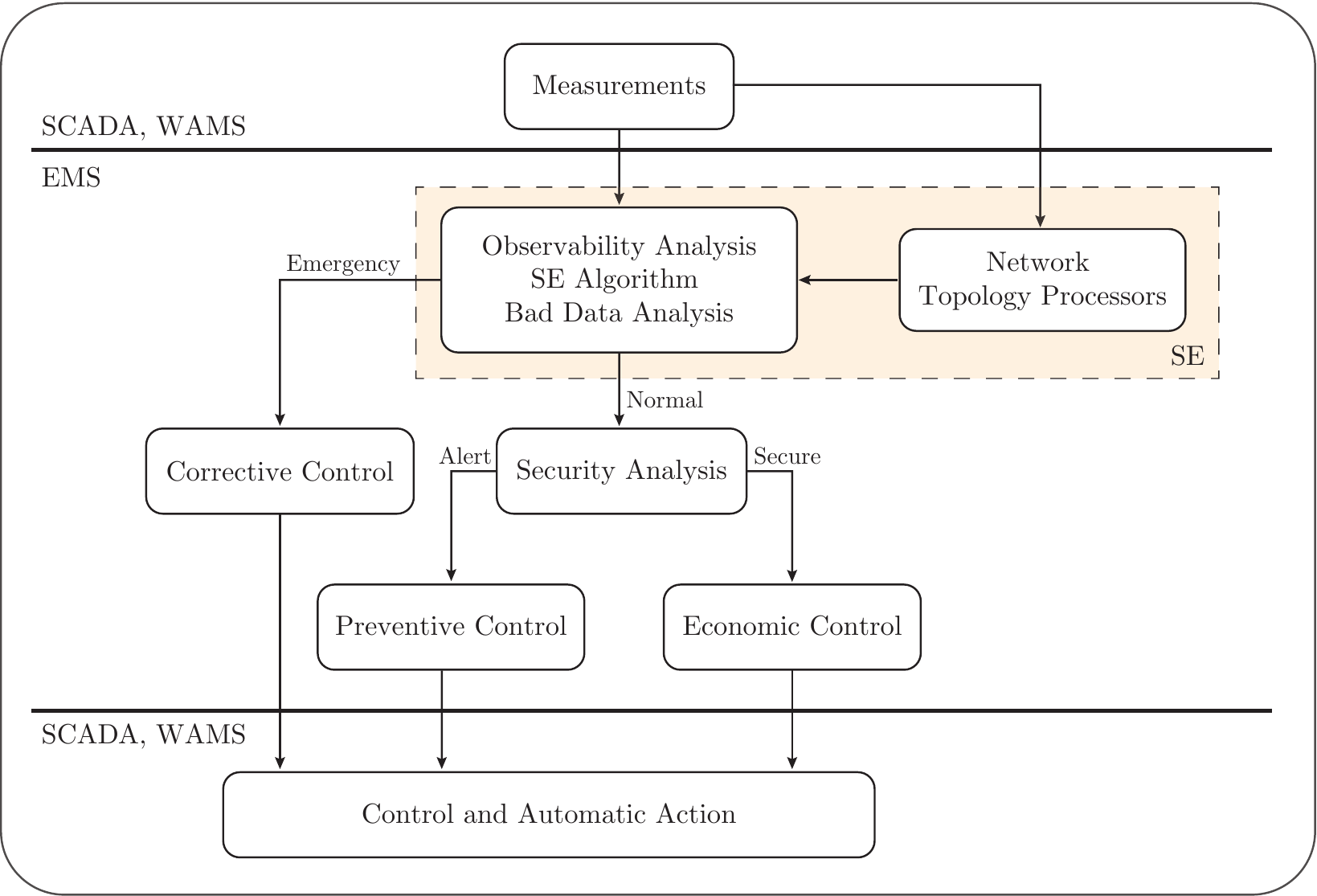}
	\caption{EMS configuration and SE routines.}
	\label{fig1_1}
	\end{figure}

The SE is a part of the energy management systems (EMS) and typically includes network topology processors, observability analysis, SE algorithm and bad data analysis, as shown in \autoref{fig1_1}. Data for the SE arrives from SCADA and WAMS technology. SCADA provides legacy measurements with low sampling rates insufficient to capture system dynamics in real-time and provides a snapshot SE with order of seconds and minutes latency. In contrast, WAMS provides data from PMUs with high sampling rates $(10\,\mbox{ms} - 20\,\mbox{ms})$ enabling the real-time system monitoring.

In a usual scenario, the SE model is described with the system of non-linear equations, where bus voltage magnitudes and bus voltage angles, transformer magnitudes of turns ratio and transformer angles of turns ratio as state variables $\mathbf{x}$. The core of the SE is the SE algorithm that provides an estimate of the system state $\mathbf{x}$ based on the network topology and available measurements. SE is performed on a bus/branch model and used to reconstruct the state of the system. Conventional SE algorithms are centralized and typically use the Gauss-Newton method to solve the non-linear WLS problem \cite{monticelli, schweppe}. Besides the non-linear SE model, the DC model is obtained by linearization of the non-linear model, and it provides an approximate solution. The DC state estimate is obtained through non-iterative procedure by solving the linear WLS problem.

\section{Measurement Model}	\label{sec:Measurement_Model}
The SE algorithm estimates the values of the state variables based on the knowledge of network topology and parameters, and measured values obtained from measurement devices spread across the power system. The knowledge of the network topology and parameters is provided by the network topology processor in the form of the bus/branch model, where branches of the grid are usually described using the two-port $\pi$-model \cite[Ch.~1,2]{abur}. The bus/branch model can be represented using a graph $\mathcal{G} =$ $(\mathcal{H},\mathcal{E})$, where the set of nodes $\mathcal{H} =$ $\{1,\dots,N  \}$ represents the set of buses, while the set of edges $\mathcal{E} \subseteq \mathcal{H} \times \mathcal{H}$ represents the set of branches of the power network. 

As an input, the SE requires a set of measurements $\mathcal{M}$ of different electrical quantities spread across the power network. Using the bus/branch model and available measurements, the observability analysis defines observable and unobservable parts of the network, subsequently defining the additional set of pseudo-measurements needed to determine the solution \cite[Ch.~4]{abur}.
Finally, the measurement model can be described as the system of equations \cite{schweppe}:  
		\begin{equation}
        \begin{aligned}
        \mathbf{z}=\mathbf{h}(\mathbf{x})+\mathbf{u},
        \end{aligned}
		\label{SE_model}
		\end{equation}
where $\mathbf {x}=[x_1,\dots,x_{n}]^{\mathrm{T}}$ is the vector of the state variables, $\mathbf{h}(\mathbf{x})=$ $[h_1(\mathbf{x})$, $\dots$, $h_k(\mathbf{x})]^{\mathrm{T}}$ is the vector of measurement functions,  $\mathbf{z} = [z_1,\dots,z_k]^{\mathrm{T}}$ is the vector of measurement values, and $\mathbf{u} = [u_1,\dots,u_k]^{\mathrm{T}}$ is the vector of uncorrelated measurement errors. The SE problem in transmission grids is commonly an overdetermined system of equations $(k>n)$ \cite[Sec.~2.1]{monticelliBook}.

Each measurement $M_i \in \mathcal{M}$ is associated with measured value $z_i$, measurement error  $u_i$, and measurement function $h_i(\mathbf{x})$. Under the assumption that measurement errors $u_i$ follow a zero-mean Gaussian distribution, the probability density function associated with the \textit{i}-th measurement is proportional to:
		\begin{equation}
        \begin{gathered}
        \mathcal{N}(z_i|\mathbf{x},v_i) \propto
        \exp\Bigg\{\cfrac{[z_i-h_i(\mathbf{x})]^2}{2v_i}\Bigg\},
        \end{gathered}
		\label{SE_Gauss_mth}
		\end{equation}
where $v_i$ is the measurement variance defined by the measurement error  $u_i$, and the measurement function $h_i(\mathbf{x})$ connects the vector of state variables $\mathbf{x}$ to the value of the \textit{i}-th measurement.

The SE in electric power systems deals with the problem of determining state variables $\mathbf{x}$ according to the noisy observed data $\mathbf{z}$ and a prior knowledge: 
		\begin{equation}
        \begin{gathered}
 		p(\mathbf{x}|\mathbf{z})=
		\cfrac{p(\mathbf{z}|\mathbf{x})p(\mathbf{x})}{p(\mathbf{z})}.
        \end{gathered}
		\label{SE_problem}
		\end{equation}
Assuming that the prior probability distribution $p(\mathbf{x})$ is uniform, and given that $p(\mathbf{z})$ does not depend on $\mathbf{x}$, the maximum a posteriori (MAP) solution of \eqref{SE_problem} reduces to the maximum likelihood solution, as given below \cite{barber}:
		\begin{equation}
        \begin{gathered}
		\hat{\mathbf{x}}=
		\mathrm{arg}\max_{\mathbf{x}}p(\mathbf{x}|\mathbf{z})=
		\mathrm{arg}\max_{\mathbf{x}}p(\mathbf{z}|\mathbf{x})=
		\mathrm{arg}\max_{\mathbf{x}}\mathcal{L}(\mathbf{z}|\mathbf{x}).
        \end{gathered}
		\label{SE_MAP_MLE}
		\end{equation}

\begin{tcolorbox}[title=Maximum Likelihood Estimator]
One can find the solution \eqref{SE_MAP_MLE} via maximization of the likelihood function $\mathcal{L}(\mathbf{z}|\mathbf{x})$, which is defined via likelihoods of $k$ independent measurements:  
		\begin{equation}
        \begin{gathered}
		\hat{\mathbf x}=
		\mathrm{arg} \max_{\mathbf{x}}\mathcal{L}(\mathbf{z}|\mathbf{x})=
		\mathrm{arg} \max_{\mathbf{x}}  
		\prod_{i=1}^k \mathcal{N}(z_i|\mathbf{x},v_i).
        \end{gathered}
		\label{SE_likelihood}
		\end{equation}
\end{tcolorbox}

It can be shown that the solution of the MAP problem can be obtained by solving the following optimization problem, known as the WLS problem \cite[Sec.~9.3]{wood}:
		\begin{equation}
        \begin{gathered}
		\hat{\mathbf x} =
		\mathrm{arg}\min_{\mathbf{x}} \sum_{i=1}^k 
		\cfrac{[z_i-h_i(\mathbf x)]^2}{v_i}.
        \end{gathered}
		\label{SE_WLS_problem}
		\end{equation}
The state estimate $\hat{\mathbf x}$ representing the solution of the optimization problem \eqref{SE_WLS_problem} is known as the WLS estimator, the maximum likelihood and WLS estimator are equivalent to the maximum a posteriori (MAP) solution \cite[Sec.~8.6]{barber}.

\subsection{Measurement Set}
The typical set of measurements $\mathcal{M}$ is defined according to type of measurement devices and includes:

\begin{enumerate}[label=(\roman*)]
\item Legacy measurements that contain active and reactive power flow and line current magnitude $\{M_{P_{ij}},$ $M_{Q_{ij}},$ $M_{I_{ij}} \}$, $(i,j) \in \mathcal{E}$, respectively; active and reactive power injection and bus voltage magnitude $\{M_{P_{i}},$ $M_{Q_{i}},M_{V_{i}} \}$, $i \in \mathcal{H}$, respectively.
\item Phasor measurements provide by PMUs contain line current $\mathcal{M}_{\ph{I}_{ij}}$, $(i,j) \in \mathcal{E}$ and bus voltage $\mathcal{M}_{\ph{V}_{i}}$, $i \in \mathcal{H}$ phasors, where each phasor measurement can be represented by a pair of measurements, for example, the bus voltage phasor measurement can be represented over the bus voltage magnitude and angle measurements $\mathcal{M}_{\ph{V}_{i}} =$ $\{{M}_{{V}_{i}}, {M}_{{\theta}_{i}}\}$, $i \in \mathcal{H}$.
\end{enumerate}

Each legacy measurement is described by non-linear measurement function $h_i(\mathbf{x})$, where the state vector $\mathbf{x}$ is given in polar coordinates. In contrast, phasor measurements can be described with both non-linear and linear measurement functions $h_i(\mathbf{x})$, where the state vector $\mathbf{x}$ can be given in polar or rectangular coordinates. Phasor measurements integration into the SE defines different models for solving the SE problem.

\subsection{The Equivalent Branch Model} \label{subsec:pi_model}
To solve SE problem, it is necessary to establish expressions of measurement functions  $\mathbf h (\mathbf x)$ related to measurements in the set $\mathcal{M}$. The equivalent $\pi$-model for a branch, shown in \autoref{fig1_2}, is sufficient to describe all measurement functions using currents, voltages and apparent powers. For simplicity, we assume that the model does not contain phase-shifting transformers.

The series admittance is $y_{ij}$ and shunt admittances of the branch are denoted as $y_{\mathrm{s}i} =$ $g_{\mathrm{s}i} + $ $\mathrm{j} b_{\mathrm{s}i}$ and $y_{\mathrm{s}j} =$ $g_{\mathrm{s}j} + $ $\mathrm{j} b_{\mathrm{s}j}$.
	\begin{figure}[ht]
	\centering
	\includegraphics[width=65mm]{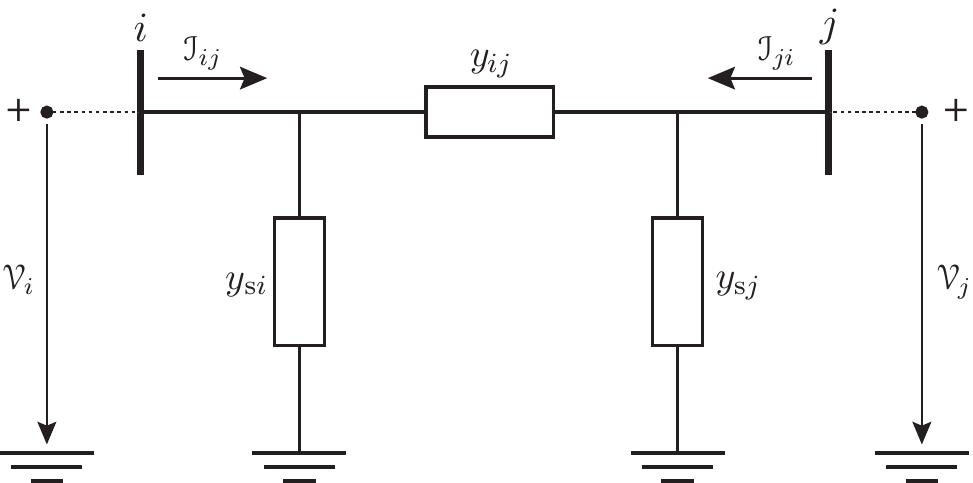}
	\caption{The equivalent branch $\pi$-model.
	\label{fig1_2}}
	\end{figure}
The branch series admittance $y_{ij}$ is inversely proportional to the branch series impedance $z_{ij}$:  
	\begin{equation}
   	\begin{aligned}
    y_{ij} = \frac{1}{z_{ij}} = 
    \frac{1}{{r_{ij}}+\mathrm{j}x_{ij}} = 
    \frac{r_{ij}}{r_{ij}^2+x_{ij}^2}-
    \mathrm{j}\frac{x_{ij}}{r_{ij}^2+x_{ij}^2}=g_{ij}+\mathrm{j}b_{ij},
   	\end{aligned}
   	\label{br_admittance}
	\end{equation}
where $r_{ij}$ is a resistance, $x_{ij}$ is a reactance, $g_{ij}$ is a conductance and $b_{ij}$ is a susceptance of the branch. In \autoref{fig1_2}, $\{i,j\} \in \mathcal{H}$ denotes buses, where, in power networks the bus represents elements such as a generator, load, substation, etc.

Using Kirchoff's laws, the complex current at buses $i$ and $j$ are:
	\begin{subequations}
   	\begin{align}
    \ph{I}_{ij} &= 
    {y}_{ij}(\ph{V}_{i}-\ph{V}_{j})+{y}_{\mathrm{s}i}
    \ph{V}_{i}=
    ({y}_{ij}+{y}_{\mathrm{s}i})\ph{V}_{i}-{y}_{ij}\ph{V}_{j}
    \label{current_ij}\\
    \ph{I}_{ji} &=
    - {y}_{ij}(\ph{V}_{i}-\ph{V}_{j})+{y}_{\mathrm{s}j}
    \ph{V}_{j}=
    ({y}_{ij}+{y}_{\mathrm{s}j})\ph{V}_{j}-{y}_{ij}\ph{V}_{i}.
    \label{current_ji}
   	\end{align}
   	\label{current_ij_ji}%
	\end{subequations}		
From \eqref{current_ij_ji} the complex currents at the bus are proportional to admittances incident to the bus (i.e. the sum of admittances) and the admittance between buses. These equations refer to the Node-Voltage method, and we apply \eqref{current_ij} to derive SE models (i.e., measurement functions). Further, complex bus voltages can be written:
	\begin{subequations}
   	\begin{align}
    \ph{V}_{i}&=V_{i}\mathrm{e}^{\mathrm{j}\theta_{i}} 
    = V_i\cos\theta_i + \mathrm{j}V_i\sin\theta_i \label{vol_polar1} \\ 
    \ph{V}_{j}&=V_{j}\mathrm{e}^{\mathrm{j}\theta_{j}}  
    = V_j\cos\theta_j + \mathrm{j}V_j\sin\theta_j, 
   	\end{align}
   	\label{vol_polar}%
	\end{subequations}
where $V_i$ and $V_j$ are bus voltage magnitudes, and $\theta_i$ and $\theta_j$ are bus voltage angles at buses $i$ and $j$. The apparent power $\ph S_{ij}$ from bus $i$ to bus $j$ is equal to:
	\begin{equation}
   	\begin{aligned}
	\ph {S}_{ij}&=\ph{V}_i \ph{I}_{ij}^* = P_{ij}+\mathrm{j}Q_{ij},
	\end{aligned}
   	\label{apparent_ij}
	\end{equation}
where $P_{ij}$ and $Q_{ij}$ represent active and reactive power flow from bus $i$ to bus $j$. 

Further, the injection complex current into the bus $i \in \mathcal{H}$ can be obtained by observing a set of buses $\mathcal{H}_i \setminus i = \{k, \dots, K\} \subset \mathcal{H}$ connected to the bus $i$, illustrated in \autoref{fig1_3}. 
	\begin{figure}[ht]
	\centering
	\includegraphics[width=48mm]{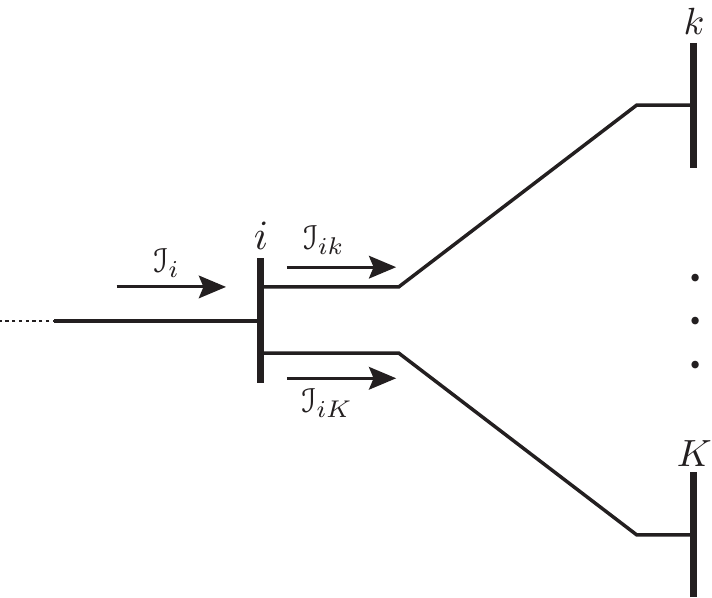}
	\caption{The set of buses $\mathcal{H}_i \setminus i = \{k, \dots, K\} 
	\subset \mathcal{H}$ connected to the bus $i$.} 
	\label{fig1_3}
	\end{figure}	
Using Kirchoff's law and \eqref{current_ij}, the injection complex current $\ph{I}_{i}$ into the bus $i$ is defined:
	\begin{equation}
   	\begin{aligned}
    \ph{I}_{i} = 
    \ph{I}_{ik}+...+\ph{I}_{iK}= 
    \sum\limits_{j \in \mathcal{H}_i \setminus i} 
    \Big[({y}_{ij}+{y}_{\mathrm{s} i})
    \ph{V}_{i}-{y}_{ij}\ph{V}_{j} \Big].  
   	\end{aligned}
   	\label{injcurr_sum}
	\end{equation}

In general, for the electric power system with $\mathcal{H} =$ $\{1,\dots,N  \}$ buses, the injection complex current $\ph{I}_{i}$ for each bus $i \in \mathcal{H}$ can be computed using:
	\begin{equation}
   	\begin{aligned}
    \ph{I}_{i} =\sum\limits_{j=1}^N \Big[({y}_{ij}+{y}_{\mathrm{s}i})
    \ph{V}_{i}-{y}_{ij}\ph{V}_{j} \Big].
   	\end{aligned}
   	\label{injcurr_vol_all}
	\end{equation}
Further, the expanded form is:
	\begin{equation}
   	\begin{aligned}
    Y_{11}&\ph{V}_{1}+Y_{12}\ph{V}_{2}+Y_{13}\ph{V}_{3} + 
    \cdots+ Y_{1N}\ph{V}_{N}= \ph{I}_{1} \\
    Y_{21}&\ph{V}_{1}+Y_{22}\ph{V}_{2}+Y_{23}\ph{V}_{3} + 
    \cdots+ Y_{2N}\ph{V}_{N}= \ph{I}_{2} \\
    \;\vdots  & \\
    Y_{N1}&\ph{V}_{1}+Y_{N2}\ph{V}_{2}+Y_{N3}\ph{V}_{3} + 
    \cdots+ Y_{NN}\ph{V}_{N}= \ph{I}_{N}.
   	\end{aligned}
   	\label{injcurr_expa}
	\end{equation}
Above system of equations can be written in the the matrix form:
	\begin{equation}
   	\begin{aligned}
    \mathbf{Y}\pmb{\mathscr{V}}=\pmb{\mathscr{I}},
   	\end{aligned}
   	\label{injcurr_mat}
	\end{equation}
where the elements of the bus or nodal admittance matrix $\mathbf{Y}$, when the bus is incident to the branch, can be formed as:
	\begin{equation}
   	\begin{aligned}
  	Y_{ij}= G_{ij} + \mathrm{j}B_{ij} = 
  	\begin{cases}
   	\sum\limits_{j\in \mathcal{H}_i \setminus i} 
   	({y}_{ij}+{y}_{\mathrm{s}i}), & \text{if} \;\; i=j 
   	 \;\;(\mathrm{diagonal\;element})\\
   	-{y}_{ij}, & \text{if} \;\; i \not = j
   	 \;\;(\mathrm{non-diagonal\;element).}
	\end{cases}
	\end{aligned}
   	\label{adm_mat_ele}
	\end{equation}
When the branch is not incident (or adjacent) to the bus the corresponding element in the nodal admittance matrix $\textbf{Y}$ is equal to zero. The nodal admittance matrix $\mathbf{Y}$ is a sparse matrix (i.e., a small number of elements are non-zeros) for a real power systems. Note that, if bus $i$ contains shunt element (capacitor or reactor), positive or negative susceptance value will be added to the diagonal element $i = j$ of the matrix $\mathbf{Y}$. Although it is often assumed that the matrix $\mathbf{Y}$ is symmetrical, it is not a general case, for example, in the presence of phase shifting transformers the matrix $\mathbf{Y}$ is not symmetrical \cite[Sec.~9.6]{stevenson}.

The apparent power injection $\ph S_i$ into the bus $i$ is a function of the complex voltage $\ph{V}_i$ at the bus and the conjugate value of the injection complex current $\ph{I}_i$ into the bus $i$:
	\begin{equation}
   	\begin{aligned}
  	\ph {S}_{i} =\ph{V}_{i}\ph{I}_{i}^* = P_i + \mathrm{j}Q_i,
	\end{aligned}
   	\label{apparent_inj}
	\end{equation}
where $P_i$ and $Q_i$ represent active power and reactive power injection into bus $i$. According to \eqref{injcurr_expa}, \eqref{adm_mat_ele} and \eqref{apparent_inj} apparent injection power $S_i$ into the bus $i$ is:
	\begin{equation}
   	\begin{aligned}
   	\ph {S}_{i} =\ph{V}_{i}\sum\limits_{j \in \mathcal{H}_i} {Y}_{ij}^*
   	\ph{V}_{j}^*,                   
	\end{aligned}
   	\label{apparent_inj_vol}
	\end{equation}
where $\mathcal{H}_i$ is the set of buses adjacent to the bus $i$, including the bus $i$. Using \eqref{vol_polar}, apparent injection power $\ph S_i$ is defined:
	\begin{equation}
   	\begin{aligned}
   	\ph {S}_{i} ={V}_{i}\sum\limits_{j \in \mathcal{H}_i} {V}_{j}
   	\mathrm{e}^{\mathrm{j}\theta_{ij}}(G_{ij}-\mathrm{j}B_{ij}).   
	\end{aligned}
   	\label{apparent_inj_state}
	\end{equation}

\subsection{State Variables}
In typical scenario, the SE model takes complex bus voltages and transformer turns ratio as state variables $\mathbf{x}$. Without loss of generality, in the rest of the thesis, for the SE model we observe complex bus voltages $\ph V_i$, $i \in \mathcal{H}$ as state variables: 
	\begin{equation}
   	\begin{aligned}
    \ph V_i = V_{i}\mathrm{e}^{\mathrm{j}\theta_{i}} = 
    \Re {(\ph V_i)} + \mathrm{j} \Im{(\ph V_i)},
   	\end{aligned}
   	\label{state_vol}
	\end{equation} 
where $\Re {(\ph V_i)}$ and $\Im{(\ph V_i)}$ represent the real and imaginary components of the complex bus voltage $\ph V_i$, respectively.   

Thus, the vector of state variables $\mathbf{x}$ can be given in polar coordinates $\mathbf x \equiv[\bm \uptheta,\mathbf V]^{\mathrm{T}}$, where we observe bus voltage angles and magnitudes as state variables respectively: 
	\begin{equation}
   	\begin{aligned}
    \bm \uptheta&=[\theta_1,\dots,\theta_N]\\
    \mathbf V&=[V_1,\dots, V_N].
   	\end{aligned}
   	\label{polar_coord}
	\end{equation} 
One voltage angle from the vector $\bm \uptheta$ corresponds to the slack or reference bus where the voltage angle has a given value. Consequently, the SE operates with  $n=2N - 1$ state variables\footnote{For convenience, BP-based SE algorithms take state variables defined with \eqref{polar_coord} as probabilistic variable nodes, where each state variable defines a variable node (i.e., the number of state variables is $n=2N$).}. The conventional SE model in the presence of legacy measurements usually implies above approach. 

Furthermore, the vector of state variables $\mathbf{x}$ can be given in rectangular coordinates $\mathbf x \equiv[\mathbf{V}_\mathrm{re},\mathbf{V}_\mathrm{im}]^{\mathrm{T}}$, where we can observe real and imaginary components of bus voltages as state variables:   
	\begin{equation}
   	\begin{aligned}
    \mathbf{V}_\mathrm{re}&=\big[\Re(\ph{V}_1),\dots,\Re(\ph{V}_N)\big]\\
	\mathbf{V}_\mathrm{im}&=\big[\Im(\ph{V}_1),\dots,\Im(\ph{V}_N)\big].     
   	\end{aligned}
   	\label{rect_coord}
	\end{equation} 
One of the elements from the vector $\mathbf{V}_\mathrm{im}$ corresponds to the slack bus. This way of assignment is frequently used for phasor measurements, whereupon measurement functions $h_i(\mathbf{x})$ become linear. However, same as before, the number of state variables is $n=2N - 1$.

\section{State Estimation Models} \label{subsec:se_models}
Power system SE models can be defined in several ways by using different criteria, such as type of measurements or according to state variables and measurements representation, as well as whether the system is linear or non-linear and how to interpret the obtained state estimator. 

\autoref{fig1_4} shows SE models described with measurement functions that precisely follow the physical laws. In general, the model where only legacy measurements exist is described with non-linear measurement functions, where state variables are given in the polar coordinate system $\mathbf x \equiv[\bm \uptheta,\mathbf V]^{\mathrm{T}}$, and it defines the conventional SE model, described in \autoref{sub:leg_meas}.
    \begin{figure}[ht]
    \centering
    \includegraphics[width=11.5cm]{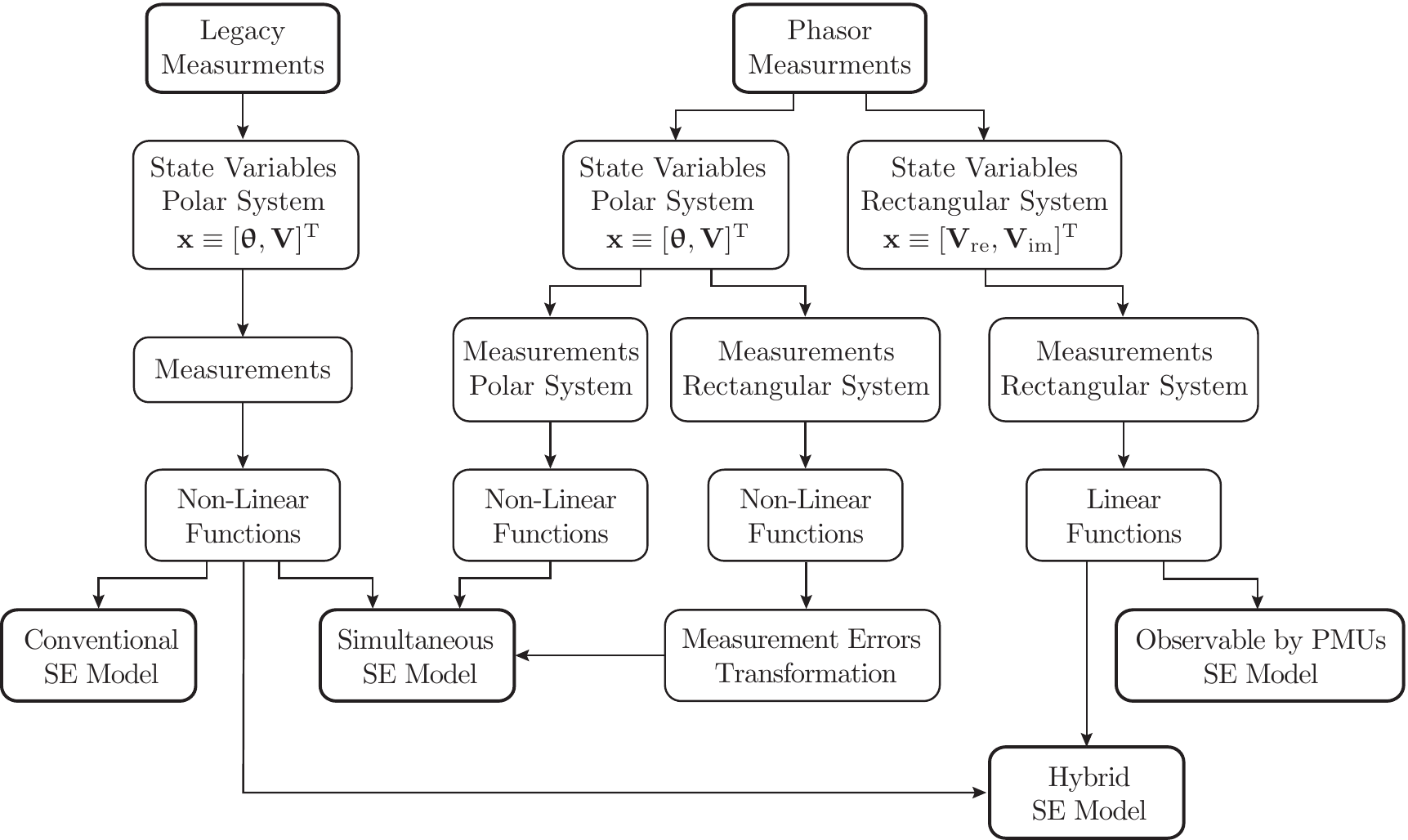}
    \caption{SE models described with measurement functions 
    that precisely follow the physical laws.}
    \label{fig1_4}
    \end{figure}

The simultaneous SE model represents the expansion of the conventional SE model with phasor measurements. State variables are given in the polar coordinate system $\mathbf x \equiv[\bm \uptheta,\mathbf V]^{\mathrm{T}}$, while phasor measurements can be given in the polar or rectangular coordinates. Phasor measurements in polar coordinate system enable straightforward inclusion in the conventional SE model (see \autoref{sub:pmu_polar}), whereas it is necessary to convert measurement variances for the case of phasor measurements in the rectangular coordinate system (see \autoref{sub:pmu_rect}) \cite{catalina}.

Hybrid SE models \cite{yang, kashyap, ghosh, dzafic, phadke} use advantages of linear functions related to phasor measurements, where state variables are given in the rectangular coordinate system $\mathbf x \equiv[\mathbf{V}_\mathrm{re},\mathbf{V}_\mathrm{im}]^{\mathrm{T}}$. Finally, to provide a state estimator only with PMUs, the system needs to be observable by PMUs only, which is currently difficult to achieve. However, with the evolution and adoption of PMU technology and, consequently, decline in the price of PMUs, it is realistic to assume that future power systems will be fully observable by PMUs, where the SE model becomes linear \cite{gol}, as will be described in \autoref{sub:pmu_linear}.

\autoref{fig1_5} shows SE models related to the SE accuracy and solving methods. In the presence of legacy measurements where measurement functions follow the physical laws, the SE model represents the non-convex problem and the Gauss-Newon provides a solution, described in \autoref{subsec:gauss-newton}.  
    \begin{figure}[ht]
    \centering
    \includegraphics[width=5.9cm]{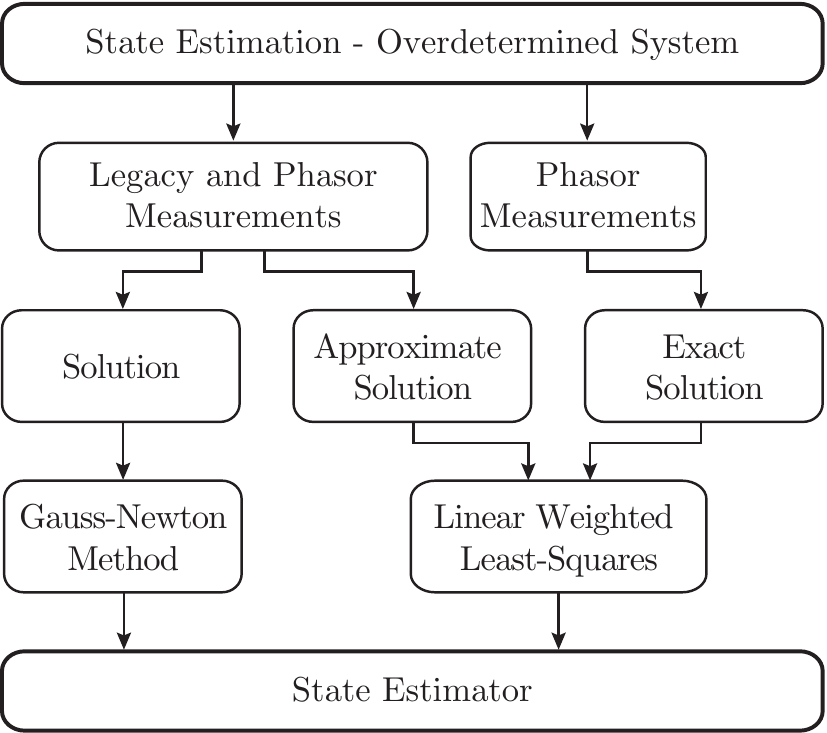}
    \caption{Different SE models related to the SE accuracy and solving methods.}
    \label{fig1_5}
    \end{figure}

The approximate solution is related with the DC SE model, obtained by linearisation of the non-linear model. The DC model ignores the reactive powers and transmission losses and takes into account only the active powers. Therefore, the DC SE takes only bus voltage angles as state variables, see \autoref{sec:dc_se_model}.     
 
Electric power systems observable by PMUs results with linear SE model where measurement functions follow the physical laws, thus obtained solution represent the optimal or exact state estimator. This model is beyond the thesis scope.  

\section{The Gauss-Newton Method} \label{subsec:gauss-newton}
In the presence of both, legacy and phasor measurements, the system \eqref{SE_model} in general represents the system of non-linear equations. The Gauss-Newton method is typically used to solve the non-linear SE model defined using measurement functions $\mathbf {h(x)}$ that precisely follow the physical laws that connect the measured variables and the state variables.  

\begin{tcolorbox}[title=Gauss-Newton Method]
Based on the available set of measurements $\mathcal{M}$, the WLS estimator $\hat{\mathbf x}$, i.e., the solution of the WLS problem \eqref{SE_WLS_problem}, can be found using the Gauss-Newton method:	
		\begin{subequations}
        \begin{gather}  
		\Big[\mathbf J (\mathbf x^{(\nu)})^\mathrm{T} \mathbf R^{-1} 
		\mathbf J (\mathbf x^{(\nu)})\Big] \Delta \mathbf x^{(\nu)} =
		\mathbf J (\mathbf x^{(\nu)})^\mathrm{T}
		\mathbf R^{-1} \mathbf r (\mathbf x^{(\nu)})
		\label{AC_GN_increment}\\
		\mathbf x^{(\nu+1)} = 
		\mathbf x^{(\nu)} + \Delta \mathbf x^{(\nu)}, \label{AC_GN_update}
        \end{gather}
        \label{AC_GN}%
		\end{subequations}	
where $\nu = \{0,1,2,\dots\}$ is the iteration index, $\Delta \mathbf x \in \mathbb {R}^{n}$ is the vector of increments of the state variables, $\mathbf J (\mathbf x)\in \mathbb {R}^{k\mathrm{x}n}$ is the Jacobian matrix of measurement functions $\mathbf h (\mathbf x)$ at $\mathbf x=\mathbf x^{(\nu)}$, $\mathbf{R}\in \mathbb {R}^{k\mathrm{x}k}$ is a measurement error covariance matrix, and $\mathbf r (\mathbf x) = \mathbf{z} - \mathbf h (\mathbf x)$ is the vector of residuals \cite[Ch.~10]{monticelliBook}. Note that, assumption that measurement errors are uncorrelated leads to the diagonal covariance matrix $\mathbf {R}$ that corresponds to measurement variances.
\end{tcolorbox}	

The non-linear SE represents non-convex problem arising from the non-linear measurement functions \cite{weng}. Due the fact that the values of state variables usually fluctuate in narrow boundaries, the non-linear SE model represents the mildly non-linear problem, where solutions are in a reasonable-sized neighborhood which enables the use of the Gauss-Newton method. The Gauss-Newton method can produce different rates of convergence, which can be anywhere from linear to quadratic \cite[Sec.~9.2]{hansen}. The convergence rate in regards to power system SE depends of the topology and measurements, and if parameters are consistent (e.g., free bad data measurement set), the method shows near quadratic convergence rate \cite[Sec.~11.2]{monticelliBook}.

\section{Legacy Measurments} \label{sub:leg_meas}
In the following, we provide expressions for measurement functions $\mathbf h (\mathbf x)$ and corresponding Jacobian elements of the matrix $\mathbf J (\mathbf x)$ related to legacy measurements, where state variables (i.e., unknown variables) are given in polar coordinates $\mathbf x \equiv[\bm \uptheta,\mathbf V]^{\mathrm{T}}$. To recall, legacy measurements contain active and reactive power flow and line current magnitude $\{M_{P_{ij}},$ $M_{Q_{ij}},$ $M_{I_{ij}}\}$, $(i,j) \in \mathcal{E}$, respectively; active and reactive power injection and bus voltage magnitude $\{M_{P_{i}},$ $M_{Q_{i}},M_{V_{i}} \}$, $i \in \mathcal{H}$, respectively.  

The active and reactive power flow at the branch $(i,j) \in \mathcal{E}$ that connects buses $i$ and $j$ can be obtained using \eqref{current_ij}, \eqref{vol_polar} and \eqref{apparent_ij}. It is easy to show that the apparent power $\ph S_{ij}$ equals:
	\begin{equation}
   	\begin{aligned}
	\ph {S}_{ij}&= 
	V_i^2 (g_{ij} + g_{\mathrm{s}i}) - 
	\mathrm{j}V_i^2(b_{ij} + b_{\mathrm{s}i}) -
	V_iV_j\mathrm{e}^{\mathrm{j}\theta_{ij}}(g_{ij} - \mathrm{j}b_{ij}),
	\end{aligned}
   	\label{apparent_state}
	\end{equation}
where $\theta_{ij} = \theta_i - \theta_j$ is the bus voltage angle difference between bus voltage angles at buses $i$ and $j$. The apparent power $\ph S_{ij}$ consists of the active $P_{ij}$ and reactive $Q_{ij}$ power flow \eqref{apparent_ij}. Hence, the real and imaginary components of the complex expression \eqref{apparent_state} define the active and reactive power flow measurement functions $P_{ij} \triangleq h_{P_{ij}}(\cdot)$ and $Q_{ij} \triangleq h_{Q_{ij}}(\cdot)$. 

\begin{tcolorbox}[title=Active and Reactive Power Flow Measurement Functions]
Thus, measurements $\{M_{P_{ij}},$ $M_{Q_{ij}}\} \in \mathcal{M}$, $(i,j) \in \mathcal{E}$ are associated with measurement functions:
	\begin{subequations}
   	\begin{align}
    h_{P_{ij}}(\cdot)&=
    {V}_{i}^2(g_{ij}+g_{si})-{V}_{i}{V}_{j}(g_{ij}\cos\theta_{ij}
    +b_{ij}\sin\theta_{ij})
    \label{mf_activeF}\\
    h_{Q_{ij}}(\cdot)&=
    -{V}_{i}^2(b_{ij}+b_{si})-{V}_{i}{V}_{j}(g_{ij}\sin\theta_{ij}
    -b_{ij}\cos\theta_{ij}).
    \label{mf_reactiveF}
   	\end{align}
   	\label{mf_flow}%
	\end{subequations}
\vspace{-0.7\baselineskip}	
\end{tcolorbox}

Jacobian expressions corresponding to the measurement function $h_{P_{ij}}(\cdot)$ are defined:
	\begin{subequations}
   	\begin{align}
    \cfrac{\mathrm \partial{h_{P_{ij}}(\cdot)}}
    {\mathrm \partial \theta_{i}}&=
    {V}_{i}{V}_{j}(g_{ij}\sin\theta_{ij}-b_{ij}\cos\theta_{ij}) 			
    \label{jac_ele_PijTi}
    \\
    \cfrac{\mathrm \partial{{h_{P_{ij}}}(\cdot)}}
    {\mathrm \partial \theta_{j}}&=
    -{V}_{i}{V}_{j}(g_{ij}\sin\theta_{ij}-b_{ij}\cos\theta_{ij})
    \label{jac_ele_PijTj}
    \\
    \cfrac{\mathrm \partial{{h_{P_{ij}}}(\cdot)}}
   	{\mathrm \partial V_{i}}&=
    -{V}_{j}(g_{ij}\cos\theta_{ij}+b_{ij}\sin\theta_{ij})+
    2V_{i}(g_{ij}+g_{si}) 
    \label{jac_ele_PijVi}
    \\
    \cfrac{\mathrm \partial{h_{{P_{ij}}}(\cdot)}}{\mathrm \partial V_{j}}&=
    -{V}_{i}(g_{ij}\cos\theta_{ij}+b_{ij}\sin\theta_{ij}). 
    \label{jac_ele_PijVj}    
	\end{align}
   	\label{jac_ele_Pij}%
	\end{subequations}  	
Further, Jacobian expressions corresponding to the measurement function $h_{Q_{ij}}(\cdot)$ are:
\begin{subequations}
   	\begin{align}
    \cfrac{\mathrm \partial{h_{{Q_{ij}}}(\cdot)}}
    {\mathrm \partial \theta_{i}}&=
    -{V}_{i}{V}_{j}(g_{ij}\cos\theta_{ij}+b_{ij}\sin\theta_{ij})
    \label{jac_ele_QijTi}
    \\
    \cfrac{\mathrm \partial{h_{{Q_{ij}}}(\cdot)}}
    {\mathrm \partial \theta_{j}}&=
    {V}_{i}{V}_{j}(g_{ij}\cos\theta_{ij}+b_{ij}\sin\theta_{ij})
    \label{jac_ele_QijTj}
    \\
    \cfrac{\mathrm \partial{h_{{Q_{ij}}}(\cdot)}}{\mathrm \partial V_{i}}&=
    -{V}_{j}(g_{ij}\sin\theta_{ij}-b_{ij}\cos\theta_{ij})-
    2V_{i}(b_{ij}+b_{si})
	\label{jac_ele_QijVi}    
    \\
    \cfrac{\mathrm \partial{h_{{Q_{ij}}}(\cdot)}}
    {\mathrm \partial V_{j}}&=-
    {V}_{i}(g_{ij}\sin\theta_{ij}-b_{ij}\cos\theta_{ij}).
    \label{jac_ele_QijVj}
	\end{align}
   	\label{jac_ele_Qij}%
	\end{subequations}

The line current magnitude at the branch $(i,j) \in \mathcal{E}$ that connects buses $i$ and $j$ can be obtained using \eqref{apparent_ij}:  
	\begin{equation}
   	\begin{aligned}
	I_{ij} = \cfrac{\sqrt{P_{ij}^2 + Q_{ij}^2}}{V_i}.
	\end{aligned}
   	\label{curmag_flow}
	\end{equation}
Using \eqref{mf_flow} and \eqref{curmag_flow}, the expression that defines the line current magnitude measurement function $I_{ij} \triangleq h_{I_{ij}}(\cdot)$ can be obtained.

\begin{tcolorbox}[title=Line Current Magnitude Measurement Function]	
Hence, measurement $M_{I_{ij}}$ $\in$ $\mathcal{M}$, $(i,j) \in \mathcal{E}$ is associated with measurement function:
	\begin{equation}
	h_{I_{ij}}(\cdot) = 
    [A_\mathrm{c} V_i^2 + B_\mathrm{c} V_j^2 - 2 V_iV_j
    (C_\mathrm{c} \cos \theta_{ij}-D_\mathrm{c} \sin \theta_{ij})]^{1/2},
    \label{mf_curmag}
	\end{equation}	
where coefficients are as follows: 		
	\begin{equation}
    \begin{aligned}
    A_\mathrm{c}&=(g_{ij}+g_{\mathrm{s}i})^2+(b_{ij}+b_{\mathrm{s}i})^2;&
    B_\mathrm{c}&=g_{ij}^2+b_{ij}^2\\
    C_\mathrm{c}&=g_{ij}(g_{ij}+g_{\mathrm{s}i})+b_{ij}(b_{ij}+b_{\mathrm{s}i});&
    D_\mathrm{c}&=g_{ij}b_{\mathrm{s}i}-b_{ij}g_{\mathrm{s}i}.
    \end{aligned}
    \nonumber
	\end{equation}
\end{tcolorbox}

Jacobian expressions corresponding to the line current magnitude measurement function $h_{I_{ij}}(\cdot)$ are:  
	\begin{subequations}
   	\begin{align}
	\cfrac{\mathrm \partial{h_{I_{ij}}(\cdot)}}
	{\mathrm \partial \theta_{i}}&=
    \cfrac{V_iV_j(D_\mathrm{c}\cos\theta_{ij}+
    C_\mathrm{c}\sin\theta_{ij})}{h_{I_{ij}}(\cdot)}
    \label{jac_ele_IijTi}	
    \\
    \cfrac{\mathrm \partial{h_{I_{ij}}(\cdot)}}
    {\mathrm \partial \theta_{j}}&=-
    \cfrac{V_iV_j(D_\mathrm{c}\cos\theta_{ij}+
    C_\mathrm{c}\sin\theta_{ij})}{h_{I_{ij}}(\cdot)}
    \label{jac_ele_IijTj}
    \\
    \cfrac{\mathrm \partial{h_{I_{ij}}(\cdot)}}{\mathrm \partial V_{i}}&=
    \cfrac{V_j(D_\mathrm{c}\sin\theta_{ij}-
    C_\mathrm{c}\cos\theta_{ij})+A_\mathrm{c}V_i}{h_{I_{ij}}(\cdot)}
    \label{jac_ele_IijVi}
    \\
    \cfrac{\mathrm \partial{h_{I_{ij}}(\cdot)}}{\mathrm \partial V_{j}}&=
    \cfrac{V_i(D_\mathrm{c}\sin\theta_{ij}-
    C_\mathrm{c}\cos\theta_{ij})+B_\mathrm{c}V_j}{h_{I_{ij}}(\cdot)}. 
    \label{jac_ele_IijVj}   
	\end{align}
   	\label{jac_ele_Iij}%
	\end{subequations}	

Note that, in deregulation environment current magnitude measurements can be found in significant numbers, especially in distribution grids. The use of line current magnitude measurements can lead to various problems (e.g., the ``flat start" will cause undefined Jacobian elements), which in turn may seriously deteriorate the performance of the state estimators \cite[Sec.~9.3]{abur}.

The active and reactive power injection into the bus $i \in \mathcal{H}$ can be obtained using \eqref{apparent_inj} and \eqref{apparent_inj_state}, where the real and imaginary components determine the active and reactive power injection measurement functions $P_{i} \triangleq h_{P_{i}}(\cdot)$ and $Q_{i} \triangleq h_{Q_{i}}(\cdot)$.

\begin{tcolorbox}[title=Active and Reactive Power Injection Measurement Functions]
Thus, measurements $\{M_{P_{i}},$ $M_{Q_{i}}\}$ $\in$ $\mathcal{M}$, $i \in \mathcal{H}$ are associated with measurement functions: 
	\begin{subequations}
   	\begin{align}
    h_{P_{i}}(\cdot) &={V}_{i}\sum\limits_{j \in \mathcal{H}_i} {V}_{j}
    (G_{ij}\cos\theta_{ij}+B_{ij}\sin\theta_{ij})\label{mf_injcetionA}\\
    h_{Q_{i}}(\cdot) &={V}_{i}\sum\limits_{j \in \mathcal{H}_i} {V}_{j}
    (G_{ij}\sin\theta_{ij}-B_{ij}\cos\theta_{ij})\label{mf_injcetionR}.     
	\end{align}
   	\label{mf_injcetion}%
	\end{subequations}
\vspace{-1.0\baselineskip}
\end{tcolorbox}
 
Jacobian expressions corresponding to the measurement function $h_{P_{i}}(\cdot)$ are defined:
	\begin{subequations}
   	\begin{align}
    \cfrac{\mathrm \partial{h_{P_{i}}(\cdot)}}
    {\mathrm \partial \theta_{i}}&=
    {V}_{i}\sum_{j \in \mathcal{H}_i \setminus i } {V}_{j}
    (-G_{ij}\sin\theta_{ij}+B_{ij}\cos\theta_{ij})
    \label{jac_ele_PiTi}
    \\
    \cfrac{\mathrm \partial{h_{P_{i}}(\cdot)}}{\mathrm \partial \theta_{j}}&=
    {V}_{i}{V}_{j}(G_{ij}\sin\theta_{ij}-B_{ij}\cos\theta_{ij})
    \label{jac_ele_PiTj}
    \\
    \cfrac{\mathrm \partial{h_{P_{i}}(\cdot)}}{\mathrm \partial V_{i}}&=
    \sum_{j \in \mathcal{H}_i \setminus i } 
    {V}_{j}(G_{ij}\cos\theta_{ij}+B_{ij}
    \sin\theta_{ij})+2{V}_{i}G_{ii}
    \label{jac_ele_PiVi}
    \\
    \cfrac{\mathrm \partial{h_{P_{i}}(\cdot)}}{\mathrm \partial V_{j}}&=
    {V}_{i}(G_{ij}\cos\theta_{ij}+B_{ij}\sin\theta_{ij}),
    \label{jac_ele_PiVj}
	\end{align}
   	\label{jac_ele_Pi}%
	\end{subequations}
where $\mathcal{H}_i \setminus i$ is the set of buses adjacent to the bus $i$. Furthermore, Jacobian expressions corresponding to the measurement function $h_{Q_{i}}(\cdot)$ are:
	\begin{subequations}
   	\begin{align}
   	\cfrac{\mathrm \partial{h_{Q_{i}}(\cdot)}}
   	{\mathrm \partial \theta_{i}}&=
    {V}_{i}\sum_{j \in \mathcal{H}_i \setminus i }
    {V}_{j}(G_{ij}\cos\theta_{ij}+B_{ij}\sin\theta_{ij})
    \label{jac_ele_QiTi}
    \\
    \cfrac{\mathrm \partial{h_{Q_{i}}(\cdot)}}{\mathrm \partial \theta_{j}}&=
    {V}_{i}{V}_{j}(-G_{ij}\cos\theta_{ij}-B_{ij}
    \sin\theta_{ij})
    \label{jac_ele_QiTj}
    \\
    \cfrac{\mathrm \partial{h_{Q_{i}}(\cdot)}}{\mathrm \partial V_{i}}&=
    \sum_{j \in \mathcal{H}_i \setminus i }
    {V}_{j}(G_{ij}\mbox{sni}\theta_{ij}-B_{ij}
    \cos\theta_{ij})-2{V}_{i}B_{ii}	
    \label{jac_ele_QiVi}
    \\
    \cfrac{\mathrm \partial{h_{Q_{i}}(\cdot)}}{\mathrm \partial V_{j}}&=
    {V}_{i}(G_{ij}\sin\theta_{ij}-B_{ij}\cos\theta_{ij}).  
    \label{jac_ele_QiVj}  
	\end{align}
   	\label{jac_ele_Qi}%
	\end{subequations}

The bus voltage magnitude on the bus $i \in \mathcal{H}$ simply defines corresponding measurement function $V_{i} \triangleq h_{V_{i}}(\cdot)$. 

\begin{tcolorbox}[title=Bus Voltage Magnitude Measurement Function]
Hence, measurement $M_{V_{i}}$ $\in$ $\mathcal{M}$, $i \in \mathcal{H}$ is associated with measurement function:
		\begin{equation}
        \begin{aligned}
        h_{V_{i}}(\cdot) = V_i.
        \end{aligned}
        \label{mf_voltage_leg}
		\end{equation} 
\end{tcolorbox}		 

Jacobian expressions corresponding to the measurement function $h_{V_{i}}(\cdot)$ are defined:  
	\begin{subequations}
   	\begin{align}
   	\cfrac{\mathrm \partial{{h_{V_{i}}(\cdot)}}}
   	{\mathrm \partial \theta_{i}}=0;\;\;\;\;  
    \cfrac{\mathrm \partial{{h_{V_{i}}(\cdot)}}}
    {\mathrm \partial \theta_{j}}=0   	
    \label{jac_ele_ViTiTj} \\   	
   	\cfrac{\mathrm \partial{{h_{V_{i}}(\cdot)}}}
   	{\mathrm \partial V_{i}}=1; \;\;\;\;  
    \cfrac{\mathrm \partial{{h_{V_{i}}(\cdot)}}}{\mathrm 
    \partial V_{j}}=0.	
    \label{jac_ele_ViViVj}
	\end{align}
   	\label{jac_ele_Vi}%
	\end{subequations}	

\subsection{The Conventional SE Model}
The conventional SE model implies the state vector in polar coordinates $\mathbf x \equiv[\bm \uptheta,\mathbf V]^{\mathrm{T}}$, where the vector of measurement functions $\mathbf h (\mathbf x)$ and corresponding Jacobian elements of the matrix $\mathbf J (\mathbf x)$ are expressed in the same coordinate system. If we denote with $N_{\mathrm{le}}$ the number of legacy measurements, the vector of measurement values  $\mathbf z_{\mathrm{le}} \in \mathbb {R}^{N_{\mathrm{le}}}$, the vector of measurement functions $\mathbf h_{\mathrm{le}}(\mathbf x) \in \mathbb {R}^{N_{\mathrm{le}}}$ and corresponding Jacobian matrix $\mathbf {J}_\mathrm{{le}}(\mathbf x) \in \mathbb {R}^{N_{\mathrm{le}} \times n}$ are: 
	\begin{equation}
   	\begin{gathered}
   	\mathbf z_{\mathrm{le}} =
    \begin{bmatrix}    	 
	\mathbf z_{\mathrm{P_{ij}}}\\[3pt]
	\mathbf z_{\mathrm{Q_{ij}}}\\[3pt]
	\mathbf z_{\mathrm{I_{ij}}}\\[3pt]
	\mathbf z_{\mathrm{P_{i}}}\\[3pt]
	\mathbf z_{\mathrm{Q_{i}}}\\[3pt]
	\mathbf z_{\mathrm{V_{i}}}
	\end{bmatrix};		
	\;\;\;\;   	
   	\mathbf h_{\mathrm{le}} (\mathbf x) =
    \begin{bmatrix}    	 
	\mathbf h_{\mathrm{P_{ij}}}(\mathbf x)\\[3pt]
	\mathbf h_{\mathrm{Q_{ij}}}(\mathbf x)\\[3pt]
	\mathbf h_{\mathrm{I_{ij}}}(\mathbf x)\\[3pt]
	\mathbf h_{\mathrm{P_{i}}}(\mathbf x)\\[3pt]
	\mathbf h_{\mathrm{Q_{i}}}(\mathbf x)\\[3pt]
	\mathbf h_{\mathrm{V_{i}}}(\mathbf x)
	\end{bmatrix};		
	\;\;\;\;
    \mathbf J_{\mathrm{le}}(\mathbf x)=
    \begin{bmatrix} 
    \mathbf {J}_\mathrm{{P_{ij}\uptheta}}(\mathbf x) &
	\mathbf {J}_\mathrm{{P_{ij}V}}(\mathbf x)    \\[3pt]
    \mathbf {J}_\mathrm{{Q_{ij}\uptheta}}(\mathbf x) &
	\mathbf {J}_\mathrm{{Q_{ij}V}}(\mathbf x)    \\[3pt]	
	\mathbf {J}_\mathrm{{I_{ij}\uptheta}}(\mathbf x) &
	\mathbf {J}_\mathrm{{I_{ij}V}}(\mathbf x)    \\[3pt]    
    \mathbf {J}_\mathrm{{P_{i}\uptheta}}(\mathbf x) &
	\mathbf {J}_\mathrm{{P_{i}V}}(\mathbf x)    \\[3pt]
    \mathbf {J}_\mathrm{{Q_{i}\uptheta}}(\mathbf x) &
	\mathbf {J}_\mathrm{{Q_{i}V}} (\mathbf x)   \\[3pt]
	\mathbf {J}_\mathrm{{V_{i}\uptheta}}(\mathbf x) &
	\mathbf {J}_\mathrm{{V_{i}V}}(\mathbf x) 	 
	\end{bmatrix}.
	\end{gathered}
   	\label{mv_leg}
	\end{equation}
Due to assumption of uncorrelated measurement errors (i.e., usual assumption for legacy measurements), the measurement error covariance matrix $\mathbf{R}_\mathrm{le} \in \mathbb {R}^{N_{\mathrm{le}} \times N_{\mathrm{le}}}$ has the diagonal structure:
	\begin{equation}
   	\begin{gathered}
	\mathbf R_{\mathrm{le}} = \mathrm{diag}	
    (
	\mathbf R_{\mathrm{P_{ij}}},
	\mathbf R_{\mathrm{Q_{ij}}},
	\mathbf R_{\mathrm{I_{ij}}},
	\mathbf R_{\mathrm{P_{i}}},
	\mathbf R_{\mathrm{Q_{i}}},
	\mathbf R_{\mathrm{V_{i}}}),	
    \end{gathered}
   	\label{mv_cova_mat}
	\end{equation}
and each covariance sub-matrix of $\mathbf R_{\mathrm{le}}$ is the diagonal matrix that contains measurement variances.

The solution of the described SE model can be found using Gauss-Newton method, where $\mathbf z \equiv \mathbf z_{\mathrm{le}}$, $\mathbf h (\mathbf x) \equiv \mathbf h_{\mathrm{le}} (\mathbf x)$, $\mathbf J(\mathbf x) \equiv \mathbf J_{\mathrm{le}}(\mathbf x)$ and $\mathbf R\equiv \mathbf R_{\mathrm{le}}$. In \autoref{app:A}, we provide a step-by-step illustrative example to describe the SE model where legacy measurements are involved. 

\section{Phasor Measurements with Polar State Vector} \label{sub:pmu_meas}
Integration of phasor measurements in the established model with legacy measurements can be done using different approaches. To recall, phasor measurements contain line current $\mathcal{M}_{\ph{I}_{ij}}$, $(i,j) \in \mathcal{E}$ and bus voltage $\mathcal{M}_{\ph{V}_{i}}$, $i \in \mathcal{H}$ phasors. More precisely, phasor measurement provided by PMU is formed by a magnitude, equal to the root mean square value of the signal, and phase angle \cite[Sec.~5.6]{phadkeBook}, where measurement errors are also related with magnitude and angle of the phasor. Thus, the PMU outputs phasor measurement in polar coordinates. In addition, PMU outputs can be observed in the rectangular coordinates with real and imaginary parts of the bus voltage and line current phasors, but in that case, the two measurements may be affected by correlated measurement errors. \cite[Sec.~7.3]{phadkeBook}. Note that throughout this section the vector of state variables is given in polar coordinates $\mathbf x \equiv[\bm \uptheta,\mathbf V]^{\mathrm{T}}$.

\subsection{Measurements in Polar Coordinates} \label{sub:pmu_polar}
In the majority of PMUs, the voltage and current phasors in polar coordinate system are regarded as ``direct'' measurements (i.e., output from the PMU). This representation delivers the more accurate state estimates in comparison to the rectangular measurement representation, but it requires larger computing time \cite{manousakis}. This representation is called simultaneous SE formulation, where measurements provided by PMUs are handled in the same manner as legacy measurements \cite{catalina}. Measurement errors are uncorrelated, with measurement variances that correspond to each components of the phasor measurements (i.e., magnitude and angle). 

The bus voltage phasor on the bus $i \in \mathcal{H}$ in the polar coordinate system is described:
	\begin{equation}
   	\begin{aligned}
    \ph V_i = V_{i}\mathrm{e}^{\mathrm{j}\theta_{i}},
   	\end{aligned}
   	\label{state_vol1}
	\end{equation} 
and due the fact that the state vector is given in the polar coordinate system $\mathbf x \equiv[\bm \uptheta,\mathbf V]^{\mathrm{T}}$, measurement functions are defined as $V_{i} \triangleq h_{V_{i}}(\cdot)$, $\theta_{i} \triangleq h_{\theta_{i}}(\cdot)$. 

\begin{tcolorbox}[title=Bus Voltage Phasor Measurement Functions]
Measurement $\mathcal{M}_{\ph{V}_{i}}= $ $\{M_{{V}_{i}},$ $M_{{\theta}_{i}}\}$ $\subseteq$ $\mathcal{M}$, $i \in \mathcal{H}$ is associated with measurement functions:
	\begin{subequations}
   	\begin{align}
	h_{{V}_{i}}(\cdot) = V_i\\
	h_{\theta_i}(\cdot) = \theta_i.
    \end{align}    
   	\label{vol_pha_meas_fun}
	\end{subequations} 
\vspace{-1.0\baselineskip}	
\end{tcolorbox}	
	
Jacobian expressions corresponding to the measurement function $h_{{V}_{i}}(\cdot)$ are defined:  
	\begin{subequations}
   	\begin{align}
   	\cfrac{\mathrm \partial{{h_{{V}_{i}}(\cdot)}}}
   	{\mathrm \partial \theta_{i}}=0;\;\;\;\;  
    \cfrac{\mathrm \partial{{h_{{V}_{i}}(\cdot)}}}
    {\mathrm \partial \theta_{j}}=0   	
    \label{jac_ele_ViTiTj_pmu} \\   	
   	\cfrac{\mathrm \partial{{h_{{V}_{i}}(\cdot)}}}
   	{\mathrm \partial V_{i}}=1; \;\;\;\;  
    \cfrac{\mathrm \partial{{h_{{V}_{i}}(\cdot)}}}{\mathrm 
    \partial V_{j}}=0,	
    \label{jac_ele_ViViVj_pmu}
	\end{align}
   	\label{jac_ele_Vi_pmu}%
	\end{subequations}	
while Jacobian expressions corresponding to the measurement function $h_{\theta_i}(\cdot)$ are:  
	\begin{subequations}
   	\begin{align}
   	\cfrac{\mathrm \partial{{h_{\theta_i}(\cdot)}}}
   	{\mathrm \partial \theta_{i}}=1;\;\;\;\;  
    \cfrac{\mathrm \partial{{h_{\theta_i}(\cdot)}}}
    {\mathrm \partial \theta_{j}}=0   	
    \label{jac_ele_TiTiTj_pmu} \\   	
   	\cfrac{\mathrm \partial{{h_{\theta_i}(\cdot)}}}
   	{\mathrm \partial V_{i}}=0; \;\;\;\;  
    \cfrac{\mathrm \partial{{h_{\theta_i}(\cdot)}}}{\mathrm 
    \partial V_{j}}=0.	
    \label{jac_ele_TiViVj_pmu}
	\end{align}
   	\label{jac_ele_Ti_pmu}%
	\end{subequations}

The line current phasor at the branch $(i,j) \in \mathcal{E}$ that connects buses $i$ and $j$ in polar coordinates is defined as:
	\begin{equation}
   	\begin{aligned}
	\mathscr{I}_{ij}&=I_{ij}\mathrm{e}^{\mathrm{j}\phi_{ij}},
    \end{aligned}    
   	\label{curr_pha}
	\end{equation}
where $I_{ij}$ and $\phi_{ij}$ are magnitude and angle of the line current phasor, respectively. The line current phasor measurement directly measures magnitude and angle of the phasor. It is easy to obtain magnitude and angle equations of the line current phasor using \eqref{current_ij}, where the vector of state variables is given in the polar coordinate system $\mathbf x \equiv[\bm \uptheta,\mathbf V]^{\mathrm{T}}$. Thus, the line current phasor measurement $\mathcal{M}_{\ph{I}_{ij}}= $ $\{M_{{I}_{ij}},$ $M_{{\phi}_{ij}}\}$ $\subseteq$ $\mathcal{M}$, $(i,j) \in \mathcal{E}$ is associated with magnitude $I_{ij} \triangleq h_{I_{ij}}(\cdot)$ and angle $\phi_{ij} \triangleq h_{\phi_{ij}}(\cdot)$ measurement functions.   
	
\begin{tcolorbox}[title=Magnitude of Line Current Phasor Measurement Function]	
To recall, measurement $M_{{I}_{ij}}$ $\in$ $\mathcal{M},$ $(i,j) \in \mathcal{E}$ is associated with measurement function:
	\begin{equation}
	h_{{I}_{ij}}(\cdot) = 
    [A_\mathrm{c} V_i^2 + B_\mathrm{c} V_j^2 - 2 V_iV_j
    (C_\mathrm{c} \cos \theta_{ij}-D_\mathrm{c} \sin \theta_{ij})]^{1/2},
    \label{mf_curmag_pmu}
	\end{equation}	
where coefficients are as follows: 		
	\begin{equation}
    \begin{aligned}
    A_\mathrm{c}&=(g_{ij}+g_{\mathrm{s}i})^2+(b_{ij}+b_{\mathrm{s}i})^2;&
    B_\mathrm{c}&=g_{ij}^2+b_{ij}^2\\
    C_\mathrm{c}&=g_{ij}(g_{ij}+g_{\mathrm{s}i})+b_{ij}(b_{ij}+b_{\mathrm{s}i});&
    D_\mathrm{c}&=g_{ij}b_{\mathrm{s}i}-b_{ij}g_{\mathrm{s}i}.
    \end{aligned}
    \nonumber
	\end{equation}
\end{tcolorbox}	
	
Jacobian expressions corresponding to the measurement function $h_{I_{ij}}(\cdot)$ are given in \eqref{jac_ele_Iij}.	

\begin{tcolorbox}[title=Angle of Line Current Phasor Measurement Function]			
Furthermore, measurement $M_{{\phi}_{ij}}$ $\in$ $\mathcal{M},$ $(i,j) \in \mathcal{E}$ is associated with measurement function:
	\begin{equation}
	h_{{\phi}_{ij}}(\cdot) =\mathrm{arctan}\Bigg[ 
	\cfrac{(A_\mathrm{a} \sin\theta_i
    +B_\mathrm{a} \cos\theta_i)V_i 
    - (C_\mathrm{a} \sin\theta_j + D_\mathrm{a}\cos\theta_j)V_j}
   	{(A_\mathrm{a} \cos\theta_i
    -B_\mathrm{a} \sin\theta_i)V_i 
    - (C_\mathrm{a} \cos\theta_j - D_\mathrm{a} \sin\theta_j)V_j} \Bigg],    
    \label{mf_curang}
	\end{equation}
where coefficients are as follows: 		
	\begin{equation}
    \begin{aligned}
    A_\mathrm{a}&=g_{ij}+g_{\mathrm{s}i};&
    B_\mathrm{a}&=b_{ij}+b_{\mathrm{s}i}\\
    C_\mathrm{a}&=g_{ij};&
    D_\mathrm{a}&=b_{ij}.
    \end{aligned}
    \nonumber
	\end{equation}	
\end{tcolorbox}	
	
Jacobian expressions corresponding to the measurement function $h_{{\phi}_{ij}}(\cdot)$ are:	
	\begin{subequations}
   	\begin{align}
	\cfrac{\mathrm \partial{h_{{\phi}_{ij}}(\cdot)}}
	{\mathrm \partial \theta_{i}}&=
	\frac{A_\mathrm{c} V_i^2+ (D_\mathrm{c} \sin \theta_{ij}-C_\mathrm{c} 
	\cos\theta_{ij})V_iV_j}
   	{h_{{I}_{ij}}(\cdot)}\\
   	\cfrac{\mathrm \partial{h_{{\phi}_{ij}}(\cdot)}}
	{\mathrm \partial \theta_{j}}&=
	\frac{B_\mathrm{c} V_j^2+ (D_\mathrm{c}
   	\sin \theta_{ij}-C_\mathrm{c} \cos
   	\theta_{ij})V_iV_j}{h_{{I}_{ij}}(\cdot)}	\\
   	\cfrac{\mathrm \partial{h_{{\phi}_{ij}}(\cdot)}}
	{\mathrm \partial V_{i}}&=
   	-\frac{V_j (C_\mathrm{c} \sin\theta_{ij}+D_\mathrm{c} \cos\theta_{ij})}
   	{h_{{I}_{ij}}(\cdot)}\\
   	\cfrac{\mathrm \partial{h_{{\phi}_{ij}}(\cdot)}}
	{\mathrm \partial V_{j}}&=
	\frac{V_i (C_\mathrm{c} \sin\theta_{ij}+D_\mathrm{c} \cos\theta_{ij})}
   	{h_{{I}_{ij}}(\cdot)}.		
	\end{align}
   	\label{jac_ele_Iij_pmu}%
	\end{subequations}

To summarize, presented measurement model associated with line current phasor measurements is non-linear. However, if we denote with $N_{\mathrm{ph}}$ the number of phasor measurements, the vector of measurement values  $\mathbf z_{\mathrm{ph}} \in \mathbb {R}^{2N_{\mathrm{ph}}}$, the vector of measurement functions $\mathbf h_{\mathrm{ph}}(\mathbf x) \in \mathbb {R}^{2N_{\mathrm{ph}}}$ and corresponding Jacobian matrix $\mathbf {J}_\mathrm{{ph}}(\mathbf x) \in \mathbb {R}^{(2N_{\mathrm{ph}}) \times n}$ are: 
	\begin{equation}
   	\begin{gathered}
   	\mathbf z_{\mathrm{ph}} =
    \begin{bmatrix}    	 
	\mathbf z_{\mathrm{V_{i}}}\\[3pt]
	\mathbf z_{\mathrm{\uptheta_{i}}}\\[3pt]
	\mathbf z_{\mathrm{I_{ij}}}\\[3pt]
	\mathbf z_{\mathrm{\upphi_{ij}}}
	\end{bmatrix};		
	\;\;\;\;   	
   	\mathbf h_{\mathrm{ph}} (\mathbf x) =
    \begin{bmatrix}    	 
	\mathbf h_{\mathrm{V_{i}}}(\mathbf x)\\[3pt]
	\mathbf h_{\mathrm{\uptheta_{i}}}(\mathbf x)\\[3pt]
	\mathbf h_{\mathrm{I_{ij}}}(\mathbf x)\\[3pt]
	\mathbf h_{\mathrm{\upphi_{i}}}(\mathbf x)
	\end{bmatrix};		
	\;\;\;\;
    \mathbf J_{\mathrm{ph}}(\mathbf x)=
    \begin{bmatrix} 
    \mathbf {J}_\mathrm{{V_{i}\uptheta}}(\mathbf x) &
	\mathbf {J}_\mathrm{{V_{i}V}}(\mathbf x)    \\[3pt]
    \mathbf {J}_\mathrm{{\uptheta_{i}\uptheta}}(\mathbf x) &
	\mathbf {J}_\mathrm{{\uptheta_{i}V}}(\mathbf x)    \\[3pt]	
	\mathbf {J}_\mathrm{{I_{ij}\uptheta}}(\mathbf x) &
	\mathbf {J}_\mathrm{{I_{ij}V}}(\mathbf x)    \\[3pt]    
    \mathbf {J}_\mathrm{{\upphi_{ij}\uptheta}}(\mathbf x) &
	\mathbf {J}_\mathrm{{\upphi_{ij}V}}(\mathbf x)	 
	\end{bmatrix}.
	\end{gathered}
   	\label{mv_pmu_pol}
	\end{equation}

When phasor measurements are given in polar coordinate system, measurement errors are uncorrelated and assume zero-mean Gaussian distribution whose covariance matrix $\mathbf{R}_\mathrm{ph} \in \mathbb {R}^{(2N_{\mathrm{ph}}) \times (2N_{\mathrm{ph}})}$ has the diagonal structure:
	\begin{equation}
   	\begin{gathered}
	\mathbf R_{\mathrm{ph}} = \mathrm{diag}	
    (
	\mathbf R_{\mathrm{V_{i}}},
	\mathbf R_{\mathrm{\uptheta_{i}}},
	\mathbf R_{\mathrm{I_{ij}}},
	\mathbf R_{\mathrm{\upphi_{ij}}}),	
    \end{gathered}
   	\label{mv_cova_pmu_pol}
	\end{equation}
where each covariance sub-matrix of $\mathbf R_{\mathrm{ph}}$ is the diagonal matrix that contains measurement variances.

The solution of the SE model with legacy and phasor measurements can be found using Gauss-Newton method, where:
	\begin{equation}
   	\begin{gathered}
   	\mathbf z \equiv
    \begin{bmatrix}    	 
	\mathbf z_{\mathrm{le}}\\[3pt]
	\mathbf z_{\mathrm{ph}}
	\end{bmatrix};		
	\;\;\;\;   	
   	\mathbf h (\mathbf x) \equiv
    \begin{bmatrix}    	 
	\mathbf h_{\mathrm{le}}(\mathbf x)\\[3pt]
	\mathbf h_{\mathrm{ph}}(\mathbf x)
	\end{bmatrix};		
	\;\;\;\;
    \mathbf J(\mathbf x)\equiv
    \begin{bmatrix} 
    \mathbf J_{\mathrm{le}}(\mathbf x)\\[3pt]
    \mathbf J_{\mathrm{ph}}(\mathbf x)
	\end{bmatrix}\;\;\;\;
	\mathbf R \equiv
    \begin{bmatrix}    	 
	\mathbf R_{\mathrm{le}} & \mathbf{0}\\
	\mathbf{0} &\mathbf R_{\mathrm{ph}}
	\end{bmatrix}.		
	\end{gathered}
   	\label{mv_pmu_pol_all}
	\end{equation}
In \autoref{app:A}, we provide a step-by-step illustrative example to describe the SE model with legacy and phasor measurements. 

\subsection{Measurements in Rectangular Coordinates} \label{sub:pmu_rect}
The bus voltage and line current phasors in rectangular coordinate system are regarded as ``indirect'' measurements obtained from measurements in polar coordinates \cite{manousakis}. Thus, measurements contain the real and imaginary parts of the line current phasor measurement and the real and imaginary parts of the bus voltage phasor measurement. As before, the vector of state variables is given in polar coordinates $\mathbf x \equiv[\bm \uptheta,\mathbf V]^{\mathrm{T}}$.

The bus voltage phasor on the bus $i \in \mathcal{H}$ in the rectangular coordinate system is given: 
	\begin{equation}
   	\begin{aligned}
	\ph{V}_{i}&=\Re(\ph{V}_{i}) +\mathrm{j}\Im(\ph{V}_{i}). 
    \end{aligned}    
   	\label{vol_pha_rct}
	\end{equation}
The state vector is given in polar coordinate system $\mathbf x \equiv[\bm \uptheta,\mathbf V]^{\mathrm{T}}$, hence using \eqref{vol_polar1}, one can obtain the real and imaginary components that define corresponding measurement functions $\Re(\ph{V}_{i}) \triangleq h_{\Re (\ph{V}_{i})}(\cdot)$ and $\Im(\ph{V}_{i}) \triangleq h_{\Im (\ph{V}_{i})}(\cdot)$, respectively.
	
\begin{tcolorbox}[title=Bus Voltage Phasor Measurement Functions]	
Measurement $\mathcal{M}_{\ph{V}_{i}}= $ $\{M_{{\Re (\ph{V}_{i})}},$ $M_{{\Im (\ph{V}_{i})}}\}$ $\subseteq$ $\mathcal{M}$, $i \in \mathcal{H}$ is associated with measurement functions:
	\begin{subequations}
   	\begin{align}
    h_{\Re (\ph{V}_{i})}(\cdot) &= V_i \cos\theta_i
    \label{mf_voltage_pmuRe}\\
    h_{\Im (\ph{V}_{i})}(\cdot) &= V_i \sin\theta_i. 
    \label{mf_voltage_pmuIm}   
	\end{align}
   	\label{mf_voltage_pmu_rct}%
	\end{subequations}	
\vspace{-1.5\baselineskip}		
\end{tcolorbox}

Jacobians expressions corresponding to the measurement function $h_{\Re (\ph{V}_{i})}(\cdot)$ are defined:
	\begin{subequations}
   	\begin{align}
   	\cfrac{\mathrm \partial{h_{\Re (\ph{V}_{i})}(\cdot)}}
   	{\mathrm \partial \theta_{i}}=
    -V_i \sin\theta_i;\;\;\;\;  
    \cfrac{\mathrm \partial{h_{\Re (\ph{V}_{i})}(\cdot)}}
    {\mathrm \partial \theta_{j}}=0 \\
   	\cfrac{\mathrm \partial{h_{\Re (\ph{V}_{i})}(\cdot)}}
   	{\mathrm \partial V_{i}}=
    \cos\theta_i;\;\;\;\;  
    \cfrac{\mathrm \partial{h_{\Re (\ph{V}_{i})}(\cdot)}}
    {\mathrm \partial V_{j}}=0,    
	\end{align}
   	\label{jac_ele_pmuRe_Vi}%
	\end{subequations}
while Jacobians expressions corresponding to the measurement function $h_{\Im (\ph{V}_{i})}(\cdot)$ are: 
	\begin{subequations}
   	\begin{align}
   	\cfrac{\mathrm \partial{h_{\Im (\ph{V}_{i})}(\cdot)}}
   	{\mathrm \partial \theta_{i}}=
    V_i \cos\theta_i;\;\;\;\; 
    \cfrac{\mathrm \partial{h_{\Im (\ph{V}_{i})}(\cdot)}}
    {\mathrm \partial \theta_{j}}=0\\
    \cfrac{\mathrm \partial{h_{\Im (\ph{V}_{i})}(\cdot)}}
   	{\mathrm \partial V_{i}}=
    \sin\theta_i;\;\;\;\; 
    \cfrac{\mathrm \partial{h_{\Im (\ph{I}_{ij})}(\cdot)}}
    {\mathrm \partial V_{j}}=0.
	\end{align}
   	\label{jac_ele_pmuIm_Vi}%
	\end{subequations}	

In contrast to measurements represented in the polar coordinates, measurement functions and corresponding Jacobian elements are non-linear functions, which makes the polar coordinate system preferable.  

The line current phasor at the branch $(i,j) \in \mathcal{E}$ that connects buses $i$ and $j$ in the rectangular coordinate system is given: 
	\begin{equation}
   	\begin{aligned}
	\ph{I}_{ij}&=\Re(\ph{I}_{ij}) +\mathrm{j}\Im(\ph{I}_{ij}). 
    \end{aligned}    
   	\label{vol_pha_rct1}
	\end{equation}
Using \eqref{current_ij} and \eqref{vol_polar}, where the state vector is given in polar coordinate system $\mathbf x \equiv[\bm \uptheta,\mathbf V]^{\mathrm{T}}$, the real and imaginary components of the line current phasor $\ph{I}_{ij}$ define measurement functions $\Re (\ph{I}_{ij}) \triangleq h_{\Re (\ph{I}_{ij})}(\cdot)$ and $\Im (\ph{I}_{ij}) \triangleq h_{\Im (\ph{I}_{ij})}(\cdot)$. 
	
\begin{tcolorbox}[title=Line Current Phasor Measurement Functions]	
Consequently, measurement $M_{\ph{I}_{ij}} = $ $\{M_{{\Re (\ph{I}_{ij})}},$ $M_{{\Im (\ph{I}_{ij})}}\}$ $\subseteq$ $\mathcal{M}$, $(i,j) \in \mathcal{E}$ is associated with measurement functions: 
	\begin{subequations}
   	\begin{align}
    h_{\Re (\ph{I}_{ij})}(\cdot) &= 
    V_i(A_\mathrm{a} \cos\theta_i - 
    B_\mathrm{a}\sin\theta_i) - 
    V_j(C_\mathrm{a}\cos\theta_j - 
    D_\mathrm{a}\sin\theta_j)\label{mf_current_pmuRe}\\
    h_{\Im (\ph{I}_{ij})}(\cdot) &=
    V_i(A_\mathrm{a}\sin\theta_i + 
    B_\mathrm{a}\cos\theta_i) - 
    V_j(C_\mathrm{a}\sin\theta_j + D_\mathrm{a}\cos\theta_j)\label{mf_current_pmuIm}.     
	\end{align}
   	\label{mf_current_pmu}%
	\end{subequations}
\vspace{-1.0\baselineskip}		
\end{tcolorbox}

Jacobians expressions corresponding to the measurement function $h_{\Re (\ph{I}_{ij})}(\cdot)$ are defined:	
	\begin{subequations}
   	\begin{align}
   	\cfrac{\mathrm \partial{h_{\Re (\ph{I}_{ij})}(\cdot)}}
   	{\mathrm \partial \theta_{i}}&=
    -V_i (A_\mathrm{a} \sin\theta_i + B_\mathrm{a}\cos\theta_i)\\
    \cfrac{\mathrm \partial{h_{\Re (\ph{I}_{ij})}(\cdot)}}
    {\mathrm \partial \theta_{j}}&=
    V_j (C_\mathrm{a} \sin\theta_j + D_\mathrm{a}\cos\theta_j)
    \\
    \cfrac{\mathrm \partial{h_{\Re (\ph{I}_{ij})}(\cdot)}}
   	{\mathrm \partial V_{i}}&=
    A_\mathrm{a} \cos\theta_i - B_\mathrm{a}\sin\theta_i\\
    \cfrac{\mathrm \partial{h_{\Re (\ph{I}_{ij})}(\cdot)}}
    {\mathrm \partial V_{j}}&=
    -C_\mathrm{a} \cos\theta_j + D_\mathrm{a}\sin\theta_j,
	\end{align}
   	\label{jac_ele_pmuRe_Iij}%
	\end{subequations}
while Jacobians expressions corresponding to the measurement function $h_{\Im (\ph{I}_{ij})}(\cdot)$ are:
	\begin{subequations}
   	\begin{align}
   	\cfrac{\mathrm \partial{h_{\Im (\ph{I}_{ij})}(\cdot)}}
   	{\mathrm \partial \theta_{i}}&=
    V_i (A_\mathrm{a} \cos\theta_i - B_\mathrm{a}\sin\theta_i)\\
    \cfrac{\mathrm \partial{h_{\Im (\ph{I}_{ij})}(\cdot)}}
    {\mathrm \partial \theta_{j}}&=
    -V_j (C_\mathrm{a} \cos\theta_j - D_\mathrm{a}\sin\theta_j)
    \\
   	\cfrac{\mathrm \partial{h_{\Im (\ph{I}_{ij})}(\cdot)}}
   	{\mathrm \partial V_{i}}&=
    A_\mathrm{a} \sin\theta_i + B_\mathrm{a}\cos\theta_i\\
    \cfrac{\mathrm \partial{h_{\Im (\ph{I}_{ij})}(\cdot)}}
    {\mathrm \partial V_{j}}&=
    -C_\mathrm{a} \sin\theta_j - D_\mathrm{a}\cos\theta_j.    
	\end{align}
   	\label{jac_ele_pmuIm_Iij}%
	\end{subequations}

Same as before, functions associated with line current phasor measurements are non-linear. In addition, the rectangular representation of the line current phasor resolves ill-conditioned problems that arise in polar coordinates due to small values of current magnitudes \cite{manousakis, catalina}. The main disadvantage of this approach is related to measurement errors, because measurment errors correspond to polar coordinates (i.e. magnitude and phase errors), and hence, the covariance matrix must be transformed from polar to rectangular coordinates \cite{poor_wu, phadke, nuquistate}. As a result, measurement errors of a single PMU are correlated and covariance matrix does not have diagonal form. Despite that, the measurement error covariance matrix is usually considered as diagonal matrix, which has the effect on the accuracy of the SE. Note that, combining representation of measurements in polar and rectangular is possible, for example, the bus voltage phasor in polar form and the line current phasor in rectangular form is often used \cite{catalina}. 

The vector of measurement values  $\mathbf z_{\mathrm{ph}} \in \mathbb {R}^{2N_{\mathrm{ph}}}$, the vector of measurement functions $\mathbf h_{\mathrm{ph}}(\mathbf x) \in \mathbb {R}^{2N_{\mathrm{ph}}}$ and corresponding Jacobian matrix $\mathbf {J}_\mathrm{{ph}}(\mathbf x) \in \mathbb {R}^{(2N_{\mathrm{ph}}) \times n}$ are: 
	\begin{equation}
   	\begin{gathered}
   	\mathbf z_{\mathrm{ph}} =
    \begin{bmatrix}    	 
	\mathbf z_{\Re (\mathrm{\ph{V}_{i}})}\\[3pt]
	\mathbf z_{\Im (\mathrm{\ph{V}_{i}})}\\[3pt]
	\mathbf z_{\Re (\mathrm{\ph{I}_{ij}})}\\[3pt]
	\mathbf z_{\Im (\mathrm{\ph{I}_{ij}})}	
	\end{bmatrix};		
	\;\;\;   	
   	\mathbf h_{\mathrm{ph}} (\mathbf x) =
    \begin{bmatrix} 
	\mathbf h_{\Re (\mathrm{\ph{V}_{i}})}(\mathbf x)\\[3pt]
	\mathbf h_{\Im (\mathrm{\ph{V}_{i}})}(\mathbf x)\\[3pt]
	\mathbf h_{\Re (\mathrm{\ph{I}_{ij}})}(\mathbf x)\\[3pt]
	\mathbf h_{\Im (\mathrm{\ph{I}_{ij}})}(\mathbf x)	       	 
	\end{bmatrix}		
	\\
    \mathbf J_{\mathrm{ph}}(\mathbf x)=
    \begin{bmatrix}  
	\mathbf {J}_{\Re (\mathrm{\ph{V}_{i}})\uptheta}(\mathbf x) &
	\mathbf {J}_{\Re (\mathrm{{\ph{V}_{i}})V}}(\mathbf x) \\[3pt]
	\mathbf {J}_{\Im (\mathrm{\ph{V}_{i}})\uptheta}(\mathbf x) &
	\mathbf {J}_{\Im (\mathrm{{\ph{V}_{i}})V}} (\mathbf x) \\[3pt]
	\mathbf {J}_{\Re (\mathrm{\ph{I}_{ij}})\uptheta}(\mathbf x) &
	\mathbf {J}_{\Re (\mathrm{{\ph{I}_{ij}})V}}(\mathbf x) \\[3pt]
	\mathbf {J}_{\Im (\mathrm{\ph{I}_{ij}})\uptheta}(\mathbf x) &
	\mathbf {J}_{\Im (\mathrm{{\ph{I}_{ij}})V}}(\mathbf x)	
	\end{bmatrix}.
    \end{gathered}
   	\label{mv_pmu1}
	\end{equation}

In case we neglect correlation between the measurements of a single PMU, the matrix $\mathbf{R}_\mathrm{ph} \in \mathbb {R}^{(2N_{\mathrm{ph}}) \times (2N_{\mathrm{ph}})}$ can be observed as the diagonal matrix:
	\begin{equation}
   	\begin{gathered}
	\mathbf R_{\mathrm{ph}} = \mathrm{diag}(	
	\mathbf R_{\Re (\mathrm{\ph{V}_{i}})},
	\mathbf R_{\Im (\mathrm{\ph{V}_{i}})},
	\mathbf R_{\Re (\mathrm{\ph{I}_{ij}})},
	\mathbf R_{\Im (\mathrm{\ph{I}_{ij}})}),	 
    \end{gathered}
   	\label{mv_pmu2}
	\end{equation}
where each covariance sub-matrix of $\mathbf R_{\mathrm{ph}}$ is the diagonal matrix that contains measurement variances. To recall, the solution of the SE model with legacy and phasor measurements can be found using the Gauss-Newton method, where:
	\begin{equation}
   	\begin{gathered}
   	\mathbf z \equiv
    \begin{bmatrix}    	 
	\mathbf z_{\mathrm{le}}\\[3pt]
	\mathbf z_{\mathrm{ph}}
	\end{bmatrix};		
	\;\;\;\;   	
   	\mathbf h (\mathbf x) \equiv
    \begin{bmatrix}    	 
	\mathbf h_{\mathrm{le}}(\mathbf x)\\[3pt]
	\mathbf h_{\mathrm{ph}}(\mathbf x)
	\end{bmatrix};		
	\;\;\;\;
    \mathbf J(\mathbf x)\equiv
    \begin{bmatrix} 
    \mathbf J_{\mathrm{le}}(\mathbf x)\\[3pt]
    \mathbf J_{\mathrm{ph}}(\mathbf x)
	\end{bmatrix}\;\;\;\;
	\mathbf R \equiv
    \begin{bmatrix}    	 
	\mathbf R_{\mathrm{le}} & \mathbf{0}\\
	\mathbf{0}&\mathbf R_{\mathrm{ph}}
	\end{bmatrix}.		
	\end{gathered}
   	\label{mv_pmu_pol_all1}
	\end{equation}

\section{Phasor Measurements with Rectangular State Vector} \label{sub:pmu_linear}
For the case when the vector of state variables is given in rectangular coordinates $\mathbf x \equiv[\mathbf{V}_\mathrm{re},\mathbf{V}_\mathrm{im}]^{\mathrm{T}}$, and where measurements are also represented in the same coordinates, we obtain linear measurement functions with constant Jacobian elements. Unfortunately, direct inclusion in the conventional SE model is not possible due to different coordinate systems, however, this still represents the important advantage of phasor measurements. 

The bus voltage phasor on the bus $i \in \mathcal{H}$ in the rectangular coordinates is defined as:
	\begin{equation}
   	\begin{aligned}
	\ph{V}_{i}&=\Re(\ph{V}_{i}) +\mathrm{j}\Im(\ph{V}_{i}). 
    \end{aligned}    
   	\label{vol_pha_rct_rc}
	\end{equation}
The state vector is given in the rectangular coordinate system $\mathbf x \equiv[\mathbf{V}_\mathrm{re},\mathbf{V}_\mathrm{im}]^{\mathrm{T}}$ and the real and imaginary components of \eqref{vol_pha_rct_rc} directly define measurement functions $\Re(\ph{V}_{i}) \triangleq h_{\Re (\ph{V}_{i})}(\cdot)$ and $\Im(\ph{V}_{i}) \triangleq h_{\Im (\ph{V}_{i})}(\cdot)$.

\begin{tcolorbox}[title=Bus Voltage Phasor Measurement Functions]	
Measurement $\mathcal{M}_{\ph{V}_{i}}= $ $\{M_{{\Re (\ph{V}_{i})}},$ $M_{{\Im (\ph{V}_{i})}}\}$ $\subseteq$ $\mathcal{M}$, $i \in \mathcal{H}$ is associated with measurement functions:
	\begin{subequations}
   	\begin{align}
    h_{\Re (\ph{V}_{i})}(\cdot) &= \Re(\ph{V}_{i})
    \label{mf_voltage_pmuRe_rc}\\
    h_{\Im (\ph{V}_{i})}(\cdot) &= \Im(\ph{V}_{i}). 
    \label{mf_voltage_pmuIm_rc}   
	\end{align}
   	\label{mf_voltage_pmu_rct_rc}%
	\end{subequations}	
\vspace{-1.0\baselineskip}		
\end{tcolorbox}

Jacobians expressions corresponding to the measurement function $h_{\Re (\ph{V}_{i})}(\cdot)$ are defined:
	\begin{subequations}
   	\begin{align}
   	\cfrac{\mathrm \partial{h_{\Re (\ph{V}_{i})}(\cdot)}}
   	{\mathrm \partial \Re(\ph{V}_{i})}=1;\;\;\;\;  
    \cfrac{\mathrm \partial{h_{\Re (\ph{V}_{i})}(\cdot)}}
    {\mathrm \partial \Re(\ph{V}_{j})}=0 \\
   	\cfrac{\mathrm \partial{h_{\Re (\ph{V}_{i})}(\cdot)}}
   	{\mathrm \partial \Im(\ph{V}_{i})}=0;\;\;\;\;  
    \cfrac{\mathrm \partial{h_{\Re (\ph{V}_{i})}(\cdot)}}
    {\mathrm \partial \Im(\ph{V}_{j})}=0,    
	\end{align}%
	\end{subequations}
while Jacobians expressions corresponding to the measurement function $h_{\Im (\ph{V}_{i})}(\cdot)$ are: 
	\begin{subequations}
   	\begin{align}
   	\cfrac{\mathrm \partial{h_{\Im (\ph{V}_{i})}(\cdot)}}
   	{\mathrm \partial \Re(\ph{V}_{i})}=0;\;\;\;\;  
    \cfrac{\mathrm \partial{h_{\Im (\ph{V}_{i})}(\cdot)}}
    {\mathrm \partial \Re(\ph{V}_{j})}=0 \\
   	\cfrac{\mathrm \partial{h_{\Im (\ph{V}_{i})}(\cdot)}}
   	{\mathrm \partial \Im(\ph{V}_{i})}=1;\;\;\;\;  
    \cfrac{\mathrm \partial{h_{\Im (\ph{V}_{i})}(\cdot)}}
    {\mathrm \partial \Im(\ph{V}_{j})}=0.    
	\end{align}%
	\end{subequations}	

The line current phasor at the branch $(i,j) \in \mathcal{E}$ that connects buses $i$ and $j$ in the rectangular coordinate system is given: 
	\begin{equation}
   	\begin{aligned}
	\ph{I}_{ij}&=\Re(\ph{I}_{ij}) +\mathrm{j}\Im(\ph{I}_{ij}). 
    \end{aligned}    
   	\label{vol_pha_rct2}
	\end{equation}
Using \eqref{current_ij} and \eqref{vol_polar}, where the state vector is given in the rectangular coordinate system $\mathbf x \equiv[\mathbf{V}_\mathrm{re},\mathbf{V}_\mathrm{im}]^{\mathrm{T}}$, the real and imaginary components of the line current phasor $\ph{I}_{ij}$ define measurement functions $\Re (\ph{I}_{ij}) \triangleq h_{\Re (\ph{I}_{ij})}(\cdot)$ and $\Im (\ph{I}_{ij}) \triangleq h_{\Im (\ph{I}_{ij})}(\cdot)$. 

\begin{tcolorbox}[title=Line Current Phasor Measurement Functions]	
Measurements $M_{\ph{I}_{ij}} = $ $\{M_{{\Re (\ph{I}_{ij})}},$ $M_{{\Im (\ph{I}_{ij})}}\}$ $\subseteq$ $\mathcal{M}$, $(i,j) \in \mathcal{E}$ are associated with measurement functions: 
	\begin{subequations}
   	\begin{align}
	h_{\Re (\ph{I}_{ij})}(\cdot) & = (g_{ij}+g_{\mathrm{s}i})\Re(\ph{V}_{i}) -
	(b_{ij}+b_{\mathrm{s}i})\Im(\ph{V}_{i}) - g_{ij}\Re(\ph{V}_{j}) +
	b_{ij}\Im(\ph{V}_{j})  
	\\
	h_{\Im (\ph{I}_{ij})}(\cdot) & = (b_{ij}+b_{\mathrm{s}i})\Re(\ph{V}_{i}) +
	(g_{ij}+g_{\mathrm{s}i})\Im(\ph{V}_{i}) - b_{ij}\Re(\ph{V}_{j}) -
	g_{ij}\Im(\ph{V}_{j}).  
	\end{align}%
	\end{subequations}	
\vspace{-1.0\baselineskip}		
\end{tcolorbox}

Jacobians expressions corresponding to the measurement function $h_{\Re (\ph{I}_{ij})}(\cdot)$ are defined:	
	\begin{subequations}
   	\begin{align}
   	\cfrac{\mathrm \partial{h_{\Re (\ph{I}_{ij})}(\cdot)}}
   	{\mathrm \partial \Re(\ph{V}_{i})}=
    g_{ij}+g_{\mathrm{s}i};\;\;\;\; 
    \cfrac{\mathrm \partial{h_{\Re (\ph{I}_{ij})}(\cdot)}}
    {\mathrm \partial \Re(\ph{V}_{j})}=
    -g_{ij}
    \\
    \cfrac{\mathrm \partial{h_{\Re (\ph{I}_{ij})}(\cdot)}}
   	{\mathrm \partial \Im(\ph{V}_{i})}=
    -b_{ij}-b_{\mathrm{s}i};\;\;\;\; 
    \cfrac{\mathrm \partial{h_{\Re (\ph{I}_{ij})}(\cdot)}}
    {\mathrm \partial \Im(\ph{V}_{j})}=
    b_{ij},
	\end{align}%
	\end{subequations}
while Jacobians expressions corresponding to the measurement function $h_{\Im (\ph{I}_{ij})}(\cdot)$ are: 
	\begin{subequations}
   	\begin{align}
   	\cfrac{\mathrm \partial{h_{\Im (\ph{I}_{ij})}(\cdot)}}
   	{\mathrm \partial \Re(\ph{V}_{i})}=
    b_{ij}+b_{\mathrm{s}i};\;\;\;\; 
    \cfrac{\mathrm \partial{h_{\Im (\ph{I}_{ij})}(\cdot)}}
    {\mathrm \partial \Re(\ph{V}_{j})}=
    -b_{ij}
    \\
    \cfrac{\mathrm \partial{h_{\Im (\ph{I}_{ij})}(\cdot)}}
   	{\mathrm \partial \Im(\ph{V}_{i})}=
    g_{ij}+g_{\mathrm{s}i};\;\;\;\; 
    \cfrac{\mathrm \partial{h_{\Im (\ph{I}_{ij})}(\cdot)}}
    {\mathrm \partial \Im(\ph{V}_{j})}=
    -g_{ij}.
	\end{align}%
	\end{subequations}

To summarize, presented model represents system of linear equations, where solution can be found by solving the linear WLS problem. As before, measurement errors by a single PMU are correlated and covariance matrix does not hold diagonal form. 
	
\section{The DC State Estimation} \label{sec:dc_se_model}
The DC model is obtained by linearisation of the non-linear model. In the typical operating conditions, the difference of bus voltage angles between adjacent buses $(i,j) \in \mathcal{E}$ is very small $\theta_{i}-\theta_{j} \approx 0$, which implies $\cos \theta_{ij}\approx1$ and $\sin \theta_{ij} \approx \theta_{ij}$. Further, all bus voltage magnitudes are $V_i \approx 1$, $i \in \mathcal{H}$, and all shunt elements and branch resistances can be neglected. This implies that the DC model ignores the reactive powers and transmission losses and takes into account only the active powers. Therefore, the DC SE takes only bus voltage angles $\mathbf x \equiv {\bm \uptheta}^{\mathrm{T}}$ as state variables. Consequently, the number of state variables is $n=N-1$, where one voltage angle represents the slack bus\footnote{Similar to the non-linear SE, the BP approach uses complete set of state variables.}. 

The set of DC model measurements $\mathcal{M}$ involves only active power flow $M_{P_{ij}}$, $(i,j) \in \mathcal{E}$, and active power injection $M_{P_{i}}$, $i \in \mathcal{H}$, from legacy measurments, and without loss of generality, we can include bus voltage angle $M_{\theta_{i}}$, $i \in \mathcal{H}$, from PMUs. 

\begin{tcolorbox}[title=Linear Weighted Least-Squares Method]
The DC state estimate $\hat{\mathbf x} \equiv \hat{\bm \uptheta}{}^{\mathrm{T}}$, which is a solution to the WLS problem \eqref{SE_WLS_problem}, is obtained through the non-iterative procedure by solving the system of linear equations:  
		\begin{equation}
        \begin{aligned}  
		\Big(\mathbf H^\mathrm{T} {\mathbf R}^{-1} \mathbf H \Big) 
		\hat{\mathbf x} =		
		\mathbf H^\mathrm{T} {\mathbf R}^{-1} \mathbf z, \label{DC_WLS}\\
		\end{aligned}
        \end{equation}
where $\mathbf{H}\in \mathbb {R}^{k\mathrm{x}N}$ is the Jacobian matrix of measurement functions.
\end{tcolorbox}

According to the set of measurements $\mathcal{M}$, vector and matrices are the following block structure:
	\begin{equation}
   	\begin{gathered}
   	\mathbf z =
    \begin{bmatrix}    	 
	\mathbf z_{\mathrm{P_{ij}}}\\[3pt]
	\mathbf z_{\mathrm{P_{i}}}\\[3pt]
	\mathbf z_{\uptheta_{i}}
	\end{bmatrix};		
	\;\;\;\;   	
   	\mathbf {H} =
    \begin{bmatrix} 
    \mathbf {H}_\mathrm{{P_{ij}}}\\[3pt]
    \mathbf {H}_\mathrm{{P_{i}}}\\[3pt]
	\mathbf {H}_\mathrm{{\uptheta_{i}}}
	\end{bmatrix};	
	\;\;\;\;
	\mathbf R = 	
    \begin{bmatrix} 
	\mathbf R_{\mathrm{P_{ij}}}  & \mathbf{0}& \mathbf{0} \\
	\mathbf{0} & \mathbf R_{\mathrm{P_{i}}} & \mathbf{0}  \\
	\mathbf{0} &\mathbf{0} & \mathbf {R}_\mathrm{{\uptheta_{i}}}
	\end{bmatrix}.	
    \end{gathered}
   	\label{dc_mv_leg}
	\end{equation}
Note that, each sub-matrix of $\mathbf R$ is the diagonal measurement error covariance matrix that contains measurement variances. In the following, we provide expressions for elements of $\mathbf {H}$. 

\begin{tcolorbox}[title=Active Power Flow Measurement Function (DC Model)]
The active power flow at the branch $(i,j) \in \mathcal{E}$ that connects buses $i$ and $j$ can be obtained using \eqref{mf_activeF}:
	\begin{equation}
    \begin{aligned}
    h_{P_{ij}}(\cdot) &= 
    -b_{ij}(\theta_{i}-\theta_{j}).
    \end{aligned}
    \label{DC_active_flow}
    \end{equation}
\end{tcolorbox}
    
Jacobian $\mathbf {H}_\mathrm{{P_{ij}}}$ of the function $h_{P_{ij}}(\cdot)$ associated with measurement $M_{P_{ij}}$, $(i,j) \in \mathcal{E}$ is defined as matrix with corresponding elements:
	\begin{equation}
    \begin{aligned}
    \cfrac{\mathrm \partial{h_{P_{ij}}(\cdot)}}{\mathrm 
    \partial \theta_{i}}&=
    -b_{ij};\;\;\;\;
    \cfrac{\mathrm \partial{{h_{P_{ij}}}(\cdot)}}{\mathrm 
    \partial \theta_{j}}&=
    b_{ij}.
    \end{aligned}
    \label{DC_jac_Pij}%
    \end{equation}

\begin{tcolorbox}[title=Active Power Injection Measurement Function (DC Model)]
The active power injection into bus $i \in \mathcal{H}$ can be obtained using \eqref{mf_injcetionA}:
	\begin{equation}
    \begin{aligned}
    h_{P_{i}}(\cdot) &=
    -\sum_{j \in \mathcal{H}_i \setminus i} b_{ij}(\theta_{i}-\theta_{j}),
    \end{aligned}
    \label{DC_active_injection}
    \end{equation}
where $\mathcal{H}_i\setminus i$ is the set of buses adjacent to the bus $i$. 
\end{tcolorbox}

Jacobian $\mathbf {H}_\mathrm{{P_{i}}}$ of the function $h_{P_{i}}(\cdot)$ associated with measurement $M_{P_{i}}$, $i \in \mathcal{H}$ is defined as matrix with corresponding elements:		
	\begin{equation}
    \begin{aligned}
    \cfrac{\mathrm \partial{h_{P_{i}}(\cdot)}}
    {\mathrm \partial \theta_{i}}=
    -\sum_{j \in \mathcal{H}_i \setminus i} 
    b_{ij};\;\;\;\;
    \cfrac{\mathrm \partial{{h_{P_{i}}}(\cdot)}}
    {\mathrm \partial \theta_{j}}=
    \sum_{j \in \mathcal{H}_i \setminus i} b_{ij}.
    \end{aligned}
    \label{DC_jac_Pi}%
    \end{equation}

\begin{tcolorbox}[title=Bus Voltage Angle Measurement Function (DC Model)]
The bus voltage angle on the bus $i \in \mathcal{H}$ is described with function:
	\begin{equation}
    \begin{aligned}
    h_{\theta_{i}}(\cdot) = \theta_i. 
    \end{aligned}
    \label{DC_angle}
	\end{equation} 
\end{tcolorbox}

Jacobian $\mathbf {H}_\mathrm{{\uptheta_{i}}}$ of the function $h_{\theta_{i}}(\cdot)$ associated with measurement $M_{\theta_{i}}$, $i \in \mathcal{H}$ is defined as matrix with corresponding elements:
	\begin{equation}
    \begin{aligned}
	\cfrac{\mathrm \partial{{h_{\theta_{i}}(\cdot)}}}
	{\mathrm \partial \theta_{i}}=1;\;\;\;\;
	\cfrac{\mathrm \partial{{h_{\theta_{i}}(\cdot)}}}
	{\mathrm \partial \theta_{j}}=0.
    \end{aligned}
    \label{DC_jac_Ti}%
	\end{equation}

\section{Summary}
The solution for the non-linear and DC SE model can be found by solving the optimization problem \eqref{SE_likelihood}. The solution of the non-linear SE model reduces to solving the iterative Gauss-Newton method, while the DC SE solution can be obtained through the non-iterative procedure by solving WLS problem. The DC SE provides an approximate solution, where all bus voltage magnitudes are set to one. The presented models assume uncorrelated measurement errors that define diagonal measurement error covariance matrices.

In the SE problem, each measurement function $h_i(\mathbf x)$ depends on a limited (typically small) subset of state variables $\mathbf{x}$. Hence, the likelihood function $\mathcal{L}(\mathbf{z}|\mathbf{x})$ can be factorized into factors \eqref{SE_likelihood} affecting small subsets of state variables. \textit{This fact motivates solving the SE problem scalably and efficiently using probabilistic graphical models}. The solution involves defining the factor graph corresponding to \eqref{SE_likelihood}, and deriving expressions for BP messages exchanged over the factor graph.

\chapter{Belief Propagation based DC State Estimation}	\label{ch:dc_bp}
\addcontentsline{lof}{chapter}{3 Belief Propagation based DC State Estimation}
\addcontentsline{lot}{chapter}{3 Belief Propagation based DC State Estimation}
For completeness of exposition, we present the solution of the DC SE problem using the BP algorithm; we refer to the corresponding method as the DC-BP. Furthermore, we propose a fast real-time DC state estimator and provide an in-depth convergence analysis of the DC-BP algorithm, including the additional method to improve its convergence. The material in this section sets the stage for the main contribution of this thesis - the BP-based Gauss-Newton method for the non-linear SE model. 

The DC SE model is described by the system of linear functions, where each measurement function $h_i(\mathbf x)$ involved in \eqref{SE_Gauss_mth} is defined with \eqref{DC_active_flow}, \eqref{DC_active_injection} and \eqref{DC_angle}. Due to the linearity, messages exchanged within the DC-BP algorithm can be evaluated in closed form. 

\section{The Factor Graph Construction}
For the DC model, the set of variable nodes is defined by the state variables $\mathbf x \equiv \bm \uptheta^{\mathrm{T}}$, thus  $\mathcal{V} = \{\theta_1,\dots,\theta_N\} \equiv$ $\{x_1,\dots,x_N\}$. The set of factor nodes $\mathcal{F} =\{f_1,\dots,f_k\}$ is defined by the set of measurements $\mathcal{M}$, with measurement functions \eqref{DC_active_flow}, \eqref{DC_active_injection} and \eqref{DC_angle}. Measurements define likelihood functions $\mathcal{N}(z_i|\mathbf{x},v_i)$ that are in turn equal to local functions $\psi_i({\mathcal V_i})$ associated to factor nodes. A factor node $f_i$ connects to a variable node $x_s \in \mathcal{V}$ if and only if the state variable $x_s$ is an argument of the corresponding measurement function $h_i(\mathbf x)$.

\begin{example}[Constructing factor graph] In this toy example, using a simple 3-bus model presented in \autoref{Fig_ex_bus_branch}, we demonstrate the conversion from a bus/branch model with a given measurement configuration into the corresponding factor graph for the DC model. 
	\begin{figure}[ht]
	\centering
	\begin{tabular}{@{}c@{}}
	\subfloat[]{\label{Fig_ex_bus_branch}
	\includegraphics[width=2.8cm]{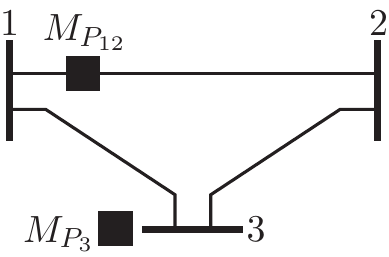}}
	\end{tabular}\quad\quad
	\begin{tabular}{@{}c@{}}
	\subfloat[]{\label{Fig_ex_DC_graph}
	\includegraphics[width=2.6cm]{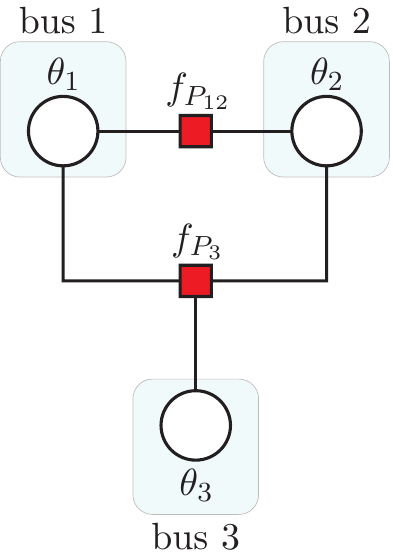}}
	\end{tabular}
	\caption{Transformation of the bus/branch model and measurement configuration 
	(subfigure a) into the corresponding factor graph for the DC model (subfigure b).}
	\label{Fig_ex_DC}
	\end{figure}
	
The variable nodes represent state variables, i.e., $\mathcal{V}=$ $\{\theta_1,$ $\theta_2,$ $\theta_3 \} \equiv$ $\{x_1,x_2,x_3\}$. Factor nodes are defined by corresponding measurements, where in our example, measurements $M_{P_{12}}$ and $M_{P_3}$ are mapped into factor nodes $\mathcal {F}=\{f_{P_{12}}$, $f_{P_3}\}$.	
\demo   
\end{example}  

\section{The Belief Propagation Algorithm}  
To recall, the BP algorithm efficiently calculates marginal distributions of state variables by passing two types of messages along the edges of the
factor graph: i) a variable node to a factor node, and ii) a factor node to a variable node messages. The marginal inference provides marginal probability distributions of each of the state variables that is used to estimate values of state variables $\mathbf{x}$. Next, we describe the DC-BP algorithm that is a version of the BP algorithm called Gaussian BP.

\subsection{Derivation of BP Messages and Marginal Inference}
\textbf{Message from a variable node to a factor node:} 
Consider a part of a factor graph shown in \autoref{Fig_v_f} with a group of factor nodes $\mathcal{F}_s=\{f_i,f_w,...,f_W\}$ $\subseteq$ $\mathcal{F}$ that are neighbours of the variable node $x_s$ $\in$ $\mathcal{V}$. 
	\begin{figure}[ht]
	\centering
	\includegraphics[width=4.3cm]{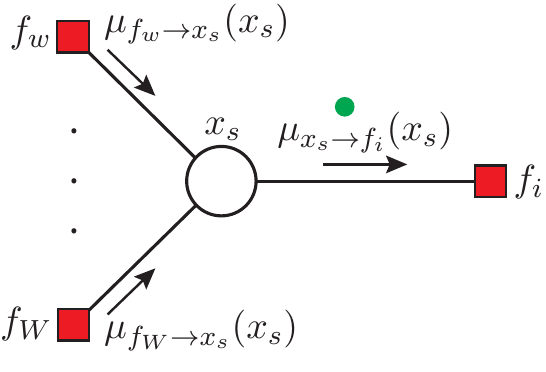}
	\caption{Message $\mu_{x_s \to f_i}(x_s)$ from variable node 
	$x_s$ to factor node $f_i$.}
	\label{Fig_v_f}
	\end{figure} \noindent
The message $\mu_{x_s \to f_i}(x_s)$ from the variable node $x_s$ to the factor node $f_i$ is equal to the product of all incoming factor node to variable node messages arriving at all the other incident edges: 
		\begin{equation}
        \begin{gathered}
        \mu_{x_s \to f_i}(x_s) =\prod_{f_a \in \mathcal{F}_s \setminus f_i} 
        \mu_{f_a \to x_s}(x_s),
        \end{gathered}
		\label{FG_v_f}
		\end{equation}
where $\mathcal{F}_s \setminus f_i$ represents the set of factor nodes incident to the variable node $x_s$, excluding the factor node $f_i$. Note that each message is a function of the variable $x_s$.

\begin{tcolorbox}[title=Message from a Variable Node to a Factor Node]
Let us assume that the incoming messages $\mu_{f_w \to x_s}(x_s)$, $\dots$, $\mu_{f_W \to x_s}(x_s)$ into the variable node $x_s$ are Gaussian and represented by their mean-variance pairs $(z_{f_w \to x_s},v_{f_w \to x_s})$, $\dots$, $(z_{f_W \to x_s},v_{f_W \to x_s})$. Note that these messages carry beliefs about the variable node $x_s$ provided by its neighbouring factor nodes $\mathcal{F}_s\setminus f_i$. According to \eqref{FG_v_f}, it can be shown that the message $\mu_{x_s \to f_i}(x_s)$ from the variable node $x_s$ to the factor node $f_i$ is proportional to:  
		\begin{equation}
        \begin{aligned}
		\mu_{x_s \to f_i}(x_s) 
		\propto 
		\mathcal{N}(x_s|z_{x_s \to f_i},
		v_{x_s \to f_i}),		
        \end{aligned}
		\label{BP_Gauss_vf} 
        \end{equation}
with mean $z_{x_s \to f_i}$ and variance $v_{x_s \to f_i}$ obtained as: 
		\begin{subequations}
        \begin{align}
        z_{x_s \to f_i} &= 
        \Bigg( \sum_{f_a \in \mathcal{F}_s\setminus f_i}
        \cfrac{z_{f_a \to x_s}}{v_{f_a \to x_s}}\Bigg)
        v_{x_s \to f_i}
        \label{BP_vf_mean}\\
		\cfrac{1}{v_{x_s \to f_i}} &= 
		\sum_{f_a \in \mathcal{F}_s\setminus f_i}
		\cfrac{1}{v_{f_a \to x_s}}.
		\label{BP_vf_var}
        \end{align}
		\label{BP_vf_mean_var}	
		\end{subequations}
\vspace{-0.5\baselineskip}			
\end{tcolorbox}	

After the variable node $x_s$ receives the messages from all of the neighbouring factor nodes from the set $\mathcal{F}_s\setminus f_i$, it evaluates the message $\mu_{x_s \to f_i}(x_s)$ according to \eqref{BP_vf_mean_var} and sends it to the factor node $f_i$. 

\textbf{Message from a factor node to a variable node:} 
Consider a part of a factor graph shown in \autoref{Fig_f_v} that consists of a group of variable nodes $\mathcal{V}_i = \{x_s, x_l,...,x_L\}$ $\subseteq$ $\mathcal V$ that are neighbours of the factor node $f_i$ $\in$ $\mathcal{F}$. 
	\begin{figure}[ht]
	\centering
	\includegraphics[width=4.5cm]{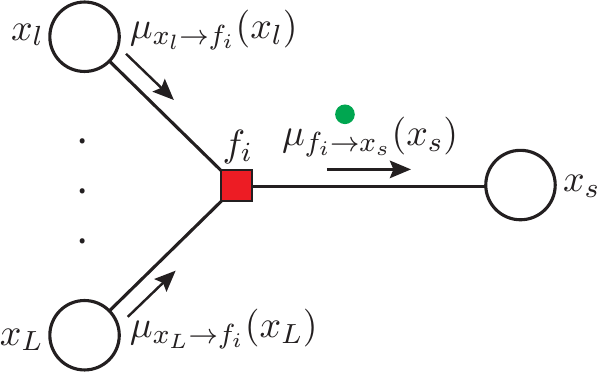}
	\caption{Message $\mu_{f_i \to x_s}(x_s)$ from factor node $f_i$ 
	to variable node $x_s$.}
	\label{Fig_f_v}
	\end{figure} \noindent
The message $\mu_{f_i \to x_s}(x_s)$ from the factor node $f_i$ to the variable node $x_s$ is defined as a product of all incoming variable node to factor node messages arriving at other incident edges, multiplied by the function $\psi_i(\mathcal{V}_i)$ associated to the factor node $f_i$, and marginalized over all of the variables associated with the incoming messages:
		\begin{equation}
		\begingroup\makeatletter\def\f@size{9}\check@mathfonts
        \begin{aligned}
        \mu_{f_i \to x_s}(x_s)= 
		\int\displaylimits_{x_l}\dots\int\displaylimits_{x_L} 
		\psi_i(\mathcal{X}_i)
		\prod_{x_b \in \mathcal{V}_i\setminus x_s} 
		\big[\mu_{x_b \to f_i}(x_b) 
		\cdot \mathrm{d}x_b\big], 
        \end{aligned}
        \endgroup
		\label{FG_f_v}
		\end{equation}		
where $\mathcal{V}_i\setminus x_s$ is the set of variable nodes incident to the factor node $f_i$, excluding the variable node $x_s$.

Due to linearity of measurement functions $h_i(\cdot)$, closed form expressions for these messages is easy to obtain and follow a Gaussian form:
		\begin{equation}
        \begin{aligned}
		\mu_{f_i \to x_s}(x_s) \propto 
		\mathcal{N}(x_s|z_{f_i \to x_s},v_{f_i \to x_s}).
        \end{aligned}
		\label{BP_Gauss_fv}
		\end{equation}
The message $\mu_{f_i \to x_s}(x_s)$ can be computed only when all other incoming messages (variable to factor node messages) are known due to synchronous scheduling. Let us assume that the messages into factor nodes are Gaussian, denoted by: 
		\begin{equation}
        \begin{aligned}
		\mu_{x_l \to f_i}(x_l) &\propto
		\mathcal{N}(x_l|z_{x_l \to f_i},
		v_{x_l \to f_i})\\
		& \vdotswithin{\propto}\\ 
		\mu_{x_L \to f_i}(x_L) &\propto
		\mathcal{N}(x_L|z_{x_L \to f_i},
		v_{x_L \to f_i}).
        \end{aligned}
		\label{BP_incoming_vf}
		\end{equation} 
The Gaussian function associated with the factor node $f_i$ is given by \eqref{SE_Gauss_mth}:
		\begin{equation}
        \begin{aligned}
		\mathcal{N}(z_i|x_s,x_l,\dots,
		x_L, v_i) 
		\propto 
        \exp\Bigg\{\cfrac{[z_i-h_i
        (x_s,x_l,\dots,x_L)]^2}
        {2v_i}\Bigg\}.
        \end{aligned}        
		\label{BP_Gauss_measurement_fun}
		\end{equation} 
The DC model contains only linear measurement functions which we represent in a general form as:
		\begin{equation}
        \begin{gathered}
		h_i(x_s,x_l,\dots,x_L) =
		C_{x_s} x_s + 
		\sum_{x_b \in \mathcal{X}_i\setminus x_s} 
		C_{x_b} x_b,
		\end{gathered}
		\label{BP_general_measurment_fun}
		\end{equation}
where $\mathcal{V}_i\setminus x_s$ is the set of variable nodes incident to the factor node $f_i$, excluding the variable node $x_s$. 

\begin{tcolorbox}[title=Message from a Factor Node to a Variable Node]
From the expression \eqref{FG_f_v}, and using \eqref{BP_incoming_vf}-\eqref{BP_general_measurment_fun}, it can be shown that the message $\mu_{f_i \to x_s}(x_s)$ from the factor node $f_i$ to the variable node $x_s$ is represented by the Gaussian function \eqref{BP_Gauss_fv}, with mean $z_{f_i \to x_s}$ and variance $v_{f_i \to x_s}$ obtained as: 
		\begin{subequations}
        \begin{align}
		z_{f_i \to x_s} &=         
        \cfrac{1}{C_{x_s}} \Bigg(z_i -  
        \sum_{x_b \in \mathcal{V}_i \setminus x_s} 
        C_{x_b} z_{x_b \to f_i}  \Bigg)
        \label{BP_fv_mean}\\
        v_{f_i \to x_s} &=         
        \cfrac{1}{C_{x_s}^2} \Bigg( v_i +  
        \sum_{x_b \in \mathcal{V}_i \setminus x_s} 
        C_{x_b}^2 v_{x_b \to f_i}  \Bigg).
		\label{BP_fv_var}
        \end{align}
		\label{BP_fv_mean_var}	
		\end{subequations}
\vspace{-0.5\baselineskip}			
\end{tcolorbox}	

To summarize, after the factor node $f_i$ receives the messages from all of the neighbouring variable nodes from the set $\mathcal{V}_i\setminus x_s$, it evaluates the message $\mu_{f_i \to x_s}(x_s)$ according to \eqref{BP_fv_mean} and \eqref{BP_fv_var}, and sends it to the variable node $x_s$.  

\textbf{Marginal inference:} 
The marginal of the variable node $x_s$, illustrated in \autoref{Fig_marginal}, is obtained as the product of all incoming messages into the variable node $x_s$:
		\begin{equation}
        \begin{gathered}
        p(x_s) =\prod_{f_c \in \mathcal{F}_s} \mu_{f_c \to x_s}(x_s),
        \end{gathered}
		\label{FG_marginal}
		\end{equation}
where $\mathcal{F}_s$ is the set of factor nodes incident to the variable node $x_s$.		
	\begin{figure}[ht]
	\centering
	\includegraphics[width=4.3cm]{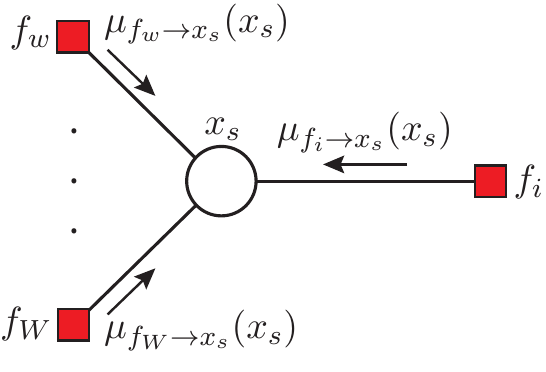}
	\caption{Marginal inference of the variable node $x_s$.}
	\label{Fig_marginal}
	\end{figure} \noindent

\begin{tcolorbox}[title=Marginal]
According to \eqref{FG_marginal}, it can be shown that the marginal of the state variable $x_s$ is represented by: 
\begin{equation}
        \begin{gathered}
        p(x_s) \propto 
        \mathcal{N}(x_s|\hat x_s,v_{x_s}),
        \end{gathered}
		\label{BP_marginal_gauss}
		\end{equation} 
with the mean value $\hat x_s$ and variance $v_{x_s}$:		
		\begin{subequations}
        \begin{align}
        \hat x_s &= 
        \Bigg( \sum_{f_c \in \mathcal{F}_s}
        \cfrac{z_{f_c \to x_s}}{v_{f_c \to x_s}}\Bigg)
        v_{x_s}
        \label{BP_marginal_mean} \\
		\cfrac{1}{v_{x_s}} &= 
		\sum_{f_c \in \mathcal{F}_s}
		\cfrac{1}{v_{f_c \to x_s}}.
		\label{BP_marginal_var}        
        \end{align}
        \label{BP_marginal_mean_var}		
		\end{subequations} 
\vspace{-0.5\baselineskip}			
\end{tcolorbox}	

Finally, the mean-value $\hat x_s$ is adopted as the estimated value of the state variable $x_s$. 

\subsection{Iterative DC-BP Algorithm}
The SE scenario is in general an instance of loopy BP since the corresponding factor graph usually contains cycles. Loopy BP is an iterative algorithm, with an iteration index $\tau = \{0,1,2, \dots \}$, and we use the synchronous scheduling, where all messages are updated in a given iteration using the output of the previous iteration as an input.

To present the algorithm precisely, we need to introduce different types of factor nodes. The \emph{indirect factor nodes} $\mathcal{F}_{\mathrm{ind}} \subset \mathcal{F}$ correspond to measurements that measure state variables indirectly. In the DC scenario, this includes active power flow and power injection measurements. The \emph{direct factor nodes} $\mathcal{F}_{\mathrm{dir}} \subset \mathcal{F}$ correspond to the measurements that measure state variables directly. For our choice of state variables for the DC scenario, an example includes measurements of bus voltage angles.

Besides direct and indirect factor nodes, we define two additional types of singly-connected factor nodes. The \emph{slack factor node} corresponds to the slack or reference bus where the voltage angle has a given value. Finally, the \emph{virtual factor node} is a singly-connected factor node used if the variable node is not directly measured, and takes the value of "flat start" with variance $v_i \to \infty $ or a priori given mean value and variance of state variables. 

We refer to direct factor nodes and two additional types of singly-connected factor nodes as local factor nodes $\mathcal{F}_{\mathrm{loc}} \subset \mathcal{F}$. We note that local factor nodes only send, but do not receive, and repeatedly transmit the same message to the corresponding variable node throughout BP iterations.

\begin{algorithm} [ht]
\caption{The DC-BP}
\label{DC}
\begin{spacing}{1.15}
\begin{algorithmic}[1] 
\Procedure {Initialization $\tau=0$}{}
  \For{each $f_s \in \mathcal{F}_{\mathrm{loc}}$}
  	  \State send $\mu_{f_s \to x_s}^{(0)}$ 
  	   to incident $x_s \in \mathcal{V}$
  \EndFor 
  \For{each $x_s \in \mathcal{V}$}
  	  \State send $\mu_{x_s \to f_i}^{(0)} = \mu_{f_s \to x_s}^{(0)}$, 
  	  to incident $f_i \in \mathcal{F}_{\mathrm{ind}}$
  \EndFor 
\EndProcedure
\myline[black](-2.4,3.22)(-2.4,0.3)

\Procedure {Iteration loop $\tau=1,2,\dots$}{}  
    \While{stopping criterion is not met} 
  \For{each $f_i \in \mathcal{F}_{\mathrm{ind}}$}
 	\State Compute $\mu_{f_i \to x_s}^{(\tau)}$ using \eqref{BP_fv_mean}*, \eqref{BP_fv_var}*
  \EndFor
   \For{each $x_s \in \mathcal{V}$}
  \State Compute $\mu_{x_s \to f_i}^{(\tau)}$ using \eqref{BP_vf_mean_var}
  \EndFor
  \EndWhile
 \EndProcedure
\myline[black](-2.4,4.18)(-2.4,0.3) 

\Procedure {Output}{}
 \For{each $x_s \in \mathcal{V}$}
    \State Compute $\hat x_s$, $v_{x_s}$ using \eqref{BP_marginal_mean_var}
  \EndFor
  \EndProcedure
\myline[black](-2.4,1.77)(-2.4,0.3)  
   \Statex *Incoming messages are obtained in previous iteration $\tau-1$  
 \end{algorithmic}
 \end{spacing}
\end{algorithm} 

\begin{example}[Different types of factor nodes] In this example, we consider the bus/branch model with three measurements illustrated in \autoref{Fig_ex_bus_branch_5meas} that we use to describe different types of factor nodes. 
	\begin{figure}[ht]
	\centering
	\begin{tabular}{@{}c@{}}
	\subfloat[]{\label{Fig_ex_bus_branch_5meas}
	\includegraphics[width=2.8cm]{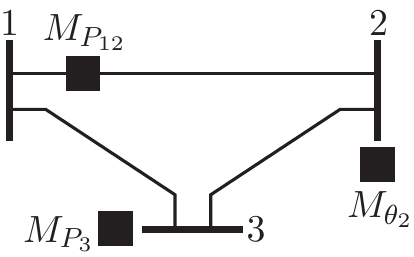}}
	\end{tabular}\quad
	\begin{tabular}{@{}c@{}}
	\subfloat[]{\label{Fig_ex_FG_5meas}
	\includegraphics[width=3.3cm]{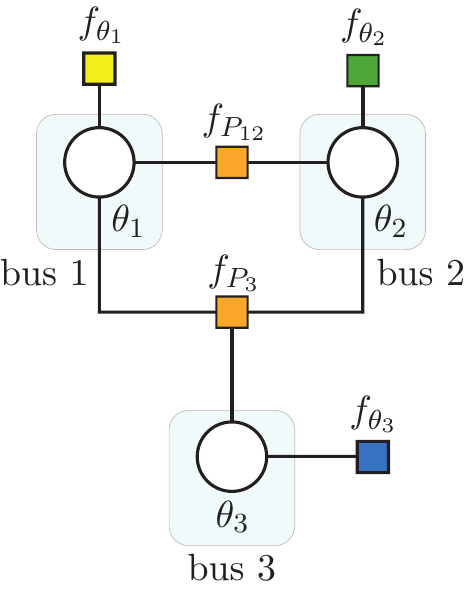}} 
	\end{tabular}
	\caption{Transformation of the bus/branch model and
	 measurement configuration (subfigure a) into the corresponding 
	 factor graph with different types of factor nodes 
	 for the DC model (subfigure b). }
	\label{Fig_ex_diff_factor}
	\end{figure} 

The indirect factor nodes (orange squares) are defined by corresponding measurements, where in our example, active power flow $M_{P_{12}}$ and active power injection $M_{P_3}$ measurements are mapped into factor nodes $\mathcal{F}_{\mathrm{ind}} = $ $\{f_{P_{12}},$ $f_{P_{3}} \}$. The set of local factor nodes $\mathcal{F}_{\mathrm{loc}}$ consists of the set of direct factor node (green square) $\mathcal{F}_{\mathrm{dir}} = $ $\{f_{\theta_{2}}\}$ defined by bus voltage angle measurement $M_{\theta_{2}}$, virtual factor node (blue square) and the slack factor node (yellow square).   \demo	  
\end{example}

The presented algorithm is an instance of a loopy Gaussian BP applied over a linear model defined by linear measurement functions $\mathbf{h}(\mathbf{x})$. It is well known that, if loopy Gaussian BP applied over a linear model converges, it will converge to a fixed point representing a solution of an equivalent WLS problem \eqref{DC_WLS} \cite{freeman}. Unlike means, the variances of Gaussian BP messages need not converge to correct values. 

The DC-BP algorithm is presented in \autoref{DC}. After the initialization (lines 1-8), the main algorithm routine starts which includes BP-based message inference (lines 9-18). Finally, the marginal inference provides the estimate of the state variables (lines 19-23). In \autoref{app:B}, we presented an illustrative numerical example of the proposed DC-BP algorithm.

\subsection{Convergence of DC-BP Algorithm} \label{sub:convergence_syn}
In this part, we present convergence analysis of DC-BP algorithm with synchronous scheduling. In the following, it will be useful to consider a subgraph of the factor graph that contains the set of variable nodes $\mathcal{V} = \{x_1, \dots, x_N \}$, the set of indirect factor nodes $\mathcal{F}_{\mathrm{ind}} = \{f_1, \dots, f_m \} \subset \mathcal{F}$, and a set of edges $\mathcal{B} \subseteq \mathcal{V} \times \mathcal{F}_{\mathrm{ind}}$ connecting them. The number of edges in this subgraph is $b = |\mathcal{B}|$. Within the subgraph, we will consider a factor node $f_i \in \mathcal{F}_{\mathrm{ind}}$ connected to its neighboring set of variable nodes $\mathcal{V}_i = \{x_q, \dots,x_Q \} \subset \mathcal{V}$ by a set of edges $\mathcal{B}_i = \{b_i^q, \dots,b_i^Q \}  \subset \mathcal{B}$, where $d_i = |\mathcal{V}_i|$ is the degree of $f_i$. Next, we provide results on convergence of both variances and means of DC-BP messages, respectively.

\textbf{Convergence of the Variances:} From equations \eqref{BP_vf_var} and \eqref{BP_fv_var}, we note that the evolution of the variances is independent of mean values of messages and measurements. Let $\mathbf {v}_{\mathrm{s}} \in \mathbb{R}^b$ denote a vector of variance values of messages from indirect factor nodes $\mathcal{F}_{\mathrm{ind}}$ to variable nodes $\mathcal{V}$. Note that this vector can be decomposed as:
		\begin{equation}
        \begin{aligned}
        \mathbf {v}_{\mathrm{s}}^{(\tau)} &=              
        [\mathbf {v}_{\mathrm{s},1}^{(\tau)}, \dots 
        \mathbf {v}_{\mathrm{s},m}^{(\tau)}]^{\mathrm{T}},
        \end{aligned}
        \label{con_2}
		\end{equation}
where the $i$-th element $\mathbf {v}_{\mathrm{s},i} \in \mathbb{R}^{d_i}$ is equal to:		
		\begin{equation}
        \begin{aligned}
        \mathbf {v}_{\mathrm{s},i}^{(\tau)} &=         
        [ v_{f_{i} \to x_q}^{(\tau)}, \dots, v_{f_{i} \to x_Q}^{(\tau)}].
        \end{aligned}
        \label{con_3}        
		\end{equation} 
Substituting \eqref{BP_vf_var} in \eqref{BP_fv_var}, the evolution of variances $\mathbf {v}_{\mathrm{s}}$ is equivalent to the following iterative equation:
		\begin{equation}
        \begin{aligned}
        \mathbf {v}_{\mathrm{s}}^{(\tau)} = 
        f \big( \mathbf {v}_{\mathrm{s}}^{(\tau-1)}\big).
        \end{aligned}
        \label{con_1}
		\end{equation}
More precisely, using simple matrix algebra, one can obtain the evolution of the variances $\mathbf {v}_{\mathrm{s}}$ in the following matrix form:		
		\begin{equation}
        \begin{aligned}
        \mathbf {v}_{\mathrm{s}}^{(\tau)} = 
        \Big[\big(\mathbf{\widetilde C}^{-1}
        \bm {\Pi} \mathbf{\widetilde C}\big)\cdot
        \big(\mathfrak{D}(\mathbf{A}) 
        \big)^{-1}
        +\bm {\Sigma}_{\mathrm{a}}\mathbf{\widetilde C}^{-1}
        \Big] \mathbf{i}, 
        \end{aligned}
        \label{con_5}        
		\end{equation}
where 
		\begin{subequations}
        \begin{align}
		\mathbf{\widetilde C} &= \mathbf{C}\mathbf{C}^{\mathrm{T}}
        \\
		\mathbf{A} &= \mathbf{\Gamma} \bm {\Sigma}_{\mathrm{s}}^{-1} 
        \mathbf{\Gamma}^\mathrm{T} + \mathbf{L}.
        \end{align}	
		\end{subequations}		
Note that in \eqref{con_5}, the dependance on $\mathbf {v}_{\mathrm{s}}^{(\tau-1)}$ is hidden in matrix $\mathbf{A}$, or more precisely, in matrix $\mathbf{{\Sigma}_{\mathrm{s}}}$. Next, we briefly describe both the matrices and matrix-operators involved in \eqref{con_5}.		
		
The operator $\mathfrak{D}(\mathbf{A}) \equiv \mathrm{diag}(A_{11}, \dots, A_{bb})$, where $A_{ii}$ is the $i$-th diagonal entry of the matrix $\mathbf{A}$. The unit vector $\mathbf{i}$ is of dimension $b$ and is equal to $\mathbf{i} = [1,\dots,1]^\mathrm{T}$. The diagonal matrix $\bm {\Sigma}_{\mathrm{s}}$ is obtained as $\bm {\Sigma}_{\mathrm{s}} = \mathrm{diag}\big(\mathbf {v}_{\mathrm{s}}^{(\tau-1)}\big) \in \mathbb{R}^{b \times b}$.

The matrix $\mathbf{C} = \mathrm{diag}\big(\mathbf{C}_{1},\dots,\mathbf{C}_{m}\big) \in \mathbb{R}^{b \times b}$ contains diagonal entries of the Jacobian non-zero elements, where $i$-th element $\mathbf {C}_i = [C_{x_q}, \dots, C_{x_Q}] \in \mathbb{R}^{d_i}$. The matrix $\bm {\Sigma}_{\mathrm{a}} = \mathrm{diag}\big(\bm {\Sigma}_{\mathrm{a,1}},  \cdots\bm {\Sigma}_{\mathrm{a},m} \big) \in \mathbb{R}^{b \times b}$ contains indirect factor node variances, with the $i$-th entry $\bm {\Sigma}_{\mathrm{a},i} = [v_i, \dots, v_i] \in \mathbb{R}^{d_i}$. 

The matrix $\mathbf{L} = \mathrm{diag}\big(\mathbf{L}_{1}, \cdots \mathbf{L}_{m} \big) \in \mathbb{R}^{b \times b}$ contains inverse variances from singly-connected factor nodes to a variable node, if such nodes exist, where the $i$-th element $\mathbf{L}_i = \big[l_{x_q}, \cdots {l}_{x_Q} \big] \in \mathbb{R}^{d_i}$. For example, $l_{x_q}$ equals:  	
		\begin{equation}
        l_{x_q}=
   		\begin{cases}
        \cfrac{1}{v_{f_{\mathrm{d},q} \to x_q}}, & \text{if}  \;\; x_q  \;\; 
        \text{is incident to $f_{\mathrm{d},q}$} \\[3pt]
        0, & \text{otherwise.}
        \end{cases}
		\end{equation}	

The matrix $\mathbf{\Pi} = \mathrm{diag}\big(\mathbf{\Pi}_1, \dots \mathbf{\Pi}_m \big) \in \mathbb{F}_2^{b \times b}$, $\mathbb{F}_2 = \{0,1\}$, is a block-diagonal matrix in which the $i$-th element is a block matrix $\mathbf{\Pi}_i = \mathbf{1}_i - \mathbf{I}_i \in \mathbb{F}_2^{d_i \times d_i}$, where the matrix $\mathbf{1}_i$ is $d_i \times d_i$ block matrix of ones, and $\mathbf{I}_i$ is $d_i \times d_i$ identity matrix. The matrix $\mathbf{\Gamma} \in \mathbb{F}_2^{b \times b}$ is of the following block structure:
		\begin{equation}
		\mathbf{\Gamma} = \left( \begin{array}{cccc}
 		\mathbf{0}_{1,1} & \mathbf{\Gamma}_{1,2} & \dots & \mathbf{\Gamma}_{1,m}  \\
 		\mathbf{\Gamma}_{2,1} & \mathbf{0}_{2,2} & \dots & \mathbf{\Gamma}_{2,m}   \\
		\vdots & \vdots  & \hfill &\vdots   \\
		\mathbf{\Gamma}_{m,1} & \mathbf{\Gamma}_{m,2}  & \dots & \mathbf{0}_{m,m} \\
		\end{array} \right),
		\label{}
		\end{equation}
where $\mathbf{0}_{i,i}$ is a block matrix $d_i \times d_i$ of zeros, and  $\mathbf{\Gamma}_{i,j} \in \mathbb{F}_2^{d_i \times d_j}$ with the $(i,j)$-th entry:
\begin{equation}
        \mathbf{\Gamma}_{i,j}(i,j)=
   		\begin{cases}
        1, & \text{if}  \;\;  
        \text{both $b_i^q$ and $b_j^q$ are incident to $x_q$} \\[3pt]
        0, & \text{otherwise.}
        \end{cases}
		\label{}		
		\end{equation}	
Note that the following holds: $\mathbf{\Gamma}_{j,i} = \mathbf{\Gamma}_{i,j}^\mathrm{T}$.

\begin{theorem}  \label{th_var}
The variances $\mathbf {v}_{\mathrm{s}}$ from indirect factor nodes to variable nodes always converge to a unique fixed point $\lim_{\tau \to \infty} \mathbf {v}_{\mathrm{s}}^{(\tau)} =\hat{\mathbf {v}}_{\mathrm{s}}$ for any initial point $\mathbf {v}_{\mathrm{s}}^{(\tau=0)} > 0$.
\end{theorem}
\begin{proof}
The theorem can be proved by showing that $f \big( \mathbf {v}_{\mathrm{s}}\big)$ satisfies the conditions of the so-called standard function \cite{hanly}, following similar steps as in the proof of Lemma 1 in \cite{zhang}.
\end{proof}

\textbf{Convergence of the Means:}
Equations \eqref{BP_vf_mean} and \eqref{BP_fv_mean} show that the evolution of the mean values depends on the variance values. Due to Theorem 3.2.2, it is possible to simplify evaluation of mean values $\mathbf {z}_{\mathrm{s}}$ from indirect factor nodes $\mathcal{F}_{\mathrm{ind}}$ to variable nodes $\mathcal{V}$ by using the fixed-point values of $\hat{\mathbf {v}}_{\mathrm{s}}$. The evolution of means $\mathbf {z}_{\mathrm{s}}$ becomes a set of linear equations: 
		\begin{equation}
        \begin{aligned}
        \mathbf {z}_{\mathrm{s}}^{(\tau)} = \mathbf{\widetilde z}
        -  \bm \Omega 
        \mathbf {z}_{\mathrm{s}}^{(\tau-1)},
        \end{aligned}
        \label{rand_mean}
		\end{equation}
where
		\begin{subequations}
        \begin{align}
		\mathbf{\widetilde z} &=  \mathbf{C}^{-1} \mathbf{z}_{\mathrm{a}} - 
        \mathbf{D} \cdot \big(\mathfrak{D} (\hat {\mathbf{A}})\big)^{-1} 
        \cdot \mathbf{L}\mathbf{z}_\mathrm{b}
        \\
		\bm \Omega &=\mathbf{D}\cdot\big(\mathfrak{D} 
        (\hat {\mathbf{A}})\big)^{-1} \cdot 
        \mathbf{\Gamma}\hat{\bm \Sigma}_{\mathrm{s}}^{-1}
        \\
        \hat {\mathbf{A}} &= \mathbf{\Gamma} \hat{\bm \Sigma}_{\mathrm{s}}^{-1}
        \mathbf{\Gamma}^\mathrm{T} + \mathbf{L}
        \\
        \mathbf{D} &= \mathbf{C}^{-1}\mathbf{\Pi} \mathbf{C}.
        \end{align}	
		\end{subequations}				
Note that the vector of means $\mathbf {z}_{\mathrm{s}} \in \mathbb{R}^b$ can be decomposed as:
		\begin{equation}
        \begin{aligned}
        \mathbf {z}_{\mathrm{s}}^{(\tau)} &=              
        [\mathbf {z}_{\mathrm{s},1}^{(\tau)}, \dots, \mathbf {z}_{\mathrm{s},m}
        ^{(\tau)}]^{\mathrm{T}},
        \end{aligned}
		\end{equation}
where the $i$-th element $\mathbf {z}_{\mathrm{s},i} \in \mathbb{R}^{d_i}$ is equal to:		
		\begin{equation}
        \begin{aligned}
        \mathbf {z}_{\mathrm{s},i}^{(\tau)} &=         
        [ z_{f_{i} \to x_k}^{(\tau)}, \dots, z_{f_{i} \to x_K}^{(\tau)}].
        \end{aligned}
		\end{equation}		

The vector $\mathbf{z}_{\mathrm{a}} = \big[\mathbf{z}_{\mathrm{a,1}}, \cdots\mathbf{z}_{\mathrm{a},m} \big]^{\mathrm{T}} \in \mathbb{R}^{b}$ contains means of indirect factor nodes, where $\mathbf{z}_{\mathrm{a},i} = [z_i, \dots, z_i] \in \mathbb{R}^{d_i}$. The diagonal matrix $\hat{\bm \Sigma}_{\mathrm{s}} \in \mathbb{R}^{b \times b}$ is obtained as $\hat{\bm \Sigma}_{\mathrm{s}} =$ $\lim_{\tau \to \infty} {\bm \Sigma}_{\mathrm{s}}^{(\tau)}$. The vector $\mathbf{z}_{\mathrm{b}} = \big[\mathbf{z}_{\mathrm{b},1}, \cdots, \mathbf{z}_{\mathrm{b},m} \big]\in \mathbb{R}^{b}$ contains means from direct and virtual factor nodes to a variable node, if such nodes exist, where the $i$-th element $\mathbf{z}_{\mathrm{b},i} = \big[z_{x_k}, \cdots {z}_{x_K} \big] \in \mathbb{R}^{d_i}$. For example, the element $z_{x_k}$ of $\mathbf{z}_{\mathrm{b},i}$ is equal to:  	
		\begin{equation}
        z_{x_k}=
   		\begin{cases}
        {z_{f_{\mathrm{d},k} \to x_k}}, & \text{if}  \;\; x_k  \;\; 
        \text{is incident to $f_{\mathrm{d},k}$} \\[3pt]
        0, & \text{otherwise.}
        \end{cases}
		\end{equation}
		
\begin{theorem} \label{th_mean_sy}
The means $\mathbf {z}_{\mathrm{s}}$ from indirect factor nodes to variable nodes converge to a unique fixed point $\lim_{\tau \to \infty} \mathbf {z}_{\mathrm{s}}^{(\tau)} =\hat {\mathbf {z}}_{\mathrm{s}}$:
		\begin{equation}
        \begin{aligned}
        \hat {\mathbf {z}}_{\mathrm{s}} =\big(\mathbf{I}+\bm \Omega \big)^{-1} 
        \mathbf{\widetilde z},
        \label{fixed}
        \end{aligned}
		\end{equation}
for any initial point $\mathbf {z}_{\mathrm{s}}^{(\tau=0)}$ if and only if the spectral radius $\rho(\bm \Omega)<1$.
\end{theorem}
\begin{proof}
The proof steps follow the proof of Theorem 5.2, \cite{hanly}.
\end{proof}

\begin{tcolorbox}[title=Convergence of the DC-BP Algorithm with Synchronous Scheduling]
To summarize, the convergence of the DC-BP algorithm depends on the spectral radius of the matrix:
		\begin{equation}
        \begin{aligned}
        \bm \Omega = \big(\mathbf{C}^{-1}\mathbf{\Pi} \mathbf{C} \big) 
        \cdot\big[\mathfrak{D} (\mathbf{\Gamma} \hat{\bm \Sigma}_{\mathrm{s}}^{-1}
        \mathbf{\Gamma}^\mathrm{T} + \mathbf{L})\big]^{-1} \cdot \big(\mathbf{\Gamma}
        \hat{\bm \Sigma}_{\mathrm{s}}^{-1}\big).
        \end{aligned}
        \label{syn_omega}
		\end{equation}
If the spectral radius $\rho(\bm \Omega)<1$, the DC-BP algorithm will converge and the resulting vector of mean values will be equal to the solution of the MAP estimator.
\end{tcolorbox}	

\subsection{Convergence of DC-BP with Randomized Damping} \label{sub:convergence_rnd}
In this section, we propose an improved DC-BP algorithm that applies synchronous scheduling with randomized damping. Several previous works reported that damping the BP messages improves the convergence of BP\cite{zhang, pretti}. Here, we propose a different randomized damping approach, where each mean value message from indirect factor node to a variable node is damped independently with probability $p$, otherwise, the message is calculated as in the standard DC-BP algorithm. The damped message is evaluated as a linear combination of the message from the previous and the current iteration, with weights $\alpha_1$ and $1-\alpha_1$, respectively. In numerical section, we demonstrate that the DC-BP with randomized damping dramatically improves convergence as compared to the standard DC-BP.    

In the proposed damping, the equation \eqref{rand_mean} is redefined as:
		\begin{equation}
        \begin{aligned}
        \mathbf {z}_{\mathrm{d}}^{(\tau)} = \mathbf {z}_{\mathrm{q}}^{(\tau)}
        +\alpha_1 \mathbf {z}_{\mathrm{w}}^{(\tau-1)} + 
        \alpha_2\mathbf {z}_{\mathrm{w}}^{(\tau)},
        \label{rand_3}
        \end{aligned}
		\end{equation}
where $0<\alpha_1<1$ is the weighting coefficient, and $\alpha_2 = 1 - \alpha_1$. In the above expression, 
$\mathbf {z}_{\mathrm{q}}^{(\tau)}$ and $\mathbf {z}_{\mathrm{w}}^{(\tau)}$ are obtained as: 
		\begin{subequations}
        \begin{align}
        \mathbf {z}_{\mathrm{q}}^{(\tau)} &= 
        \mathbf {Q} \mathbf{\widetilde z}
        - \mathbf {Q}  \bm \Omega 
        \mathbf {z}_{\mathrm{s}}^{(\tau-1)}
        \label{rand_1}\\
         \mathbf {z}_{\mathrm{w}}^{(\tau)} &= 
        \mathbf {W}\mathbf{\widetilde z}
        - \mathbf {W} \bm \Omega
        \mathbf {z}_{\mathrm{s}}^{(\tau-1)},
        \label{rand_2}
        \end{align}
		\end{subequations} 
where diagonal matrices $\mathbf {Q} \in \mathbb{F}_2^{b \times b}$ and $\mathbf {W} \in \mathbb{F}_2^{b \times b}$ are defined as $\mathbf {Q} = \mathrm{diag}(1 - q_1,...,1 - q_b)$, $q_i \sim \mathrm{Ber}(p)$, and $\mathbf {W} = \mathrm{diag}(q_1,...,q_b)$, respectively, where $\mathrm{Ber}(p) \in \{0,1\}$ is a Bernoulli random variable with probability $p$ independently sampled for each mean value message. 	

Substituting \eqref{rand_1} and \eqref{rand_2} in \eqref{rand_3}, we obtain:
		\begin{equation}
        \begin{aligned}
        \mathbf {z}_{\mathrm{d}}^{(\tau)} = \big(\mathbf {Q}+  \alpha_2 \mathbf {W}\big)
        \mathbf{\widetilde z} - 
        \big(\mathbf {Q} \bm \Omega + \alpha_2\mathbf {W} 
        \bm \Omega - \alpha_1\mathbf {W} \big) 
        \mathbf {z}_{\mathrm{s}}^{(\tau-1)}. 
        \label{rand_4}
        \end{aligned}
		\end{equation}
Note that $\mathbf {z}_{\mathrm{r}}^{(\tau-1)} = \mathbf {W} \mathbf {z}_{\mathrm{s}}^{(\tau-1)}$. In a more compact form, equation \eqref{rand_4} can be written as follows: 
		\begin{equation}
        \begin{aligned}
        \mathbf {z}_{\mathrm{d}}^{(\tau)} =\mathbf{\bar z} - \bm {\bar \Omega} 
        \mathbf {z}_{\mathrm{s}}^{(\tau-1)}, 
        \label{rand_5}
        \end{aligned}
		\end{equation}
where
		\begin{subequations}
        \begin{align}
		\mathbf{\bar z} &= \big(\mathbf {Q}+  \alpha_2 \mathbf {W}\big) 
        \mathbf{\widetilde z}
        \label{rd_z}\\
		\bm {\bar \Omega} &= \mathbf {Q} \bm \Omega + \alpha_2\mathbf {W} 
        \bm \Omega - \alpha_1\mathbf {W}.
		\label{rd_omega}
        \end{align}
		\label{rd_all}%
		\end{subequations}
		
\begin{theorem} \label{th_mean_rd}
The means $\mathbf {z}_{\mathrm{d}}$ from indirect factor nodes to variable nodes converge to a unique fixed point $\hat{\mathbf {z}}_{\mathrm{d}} = \lim_{\tau \to \infty} \mathbf {z}_{\mathrm{d}}^{(\tau)}$ for any initial point $\mathbf {z}_{\mathrm{d}}^{(\tau=0)}$ if and only if the spectral radius $\rho(\bm {\bar \Omega})<1$. For the resulting fixed point, it holds that $\hat{\mathbf {z}}_{\mathrm{d}} = \hat{\mathbf {z}}_{\mathrm{s}}$. 
\end{theorem}	

\begin{proof}
To prove theorem it is sufficient to show that equation \eqref{rand_5} converges to the fixed point defined in \eqref{fixed}. We can write:
		\begin{equation}
        \begin{aligned}
        \mathbf {z_r}^{(\tau-1)} &= 
        \mathbf {W}\mathbf{\widetilde z}
        - \mathbf {W} \bm \Omega
        \mathbf {z}_{\mathrm{s}}^{(\tau-2)}.
        \label{rand_6}
        \end{aligned}
		\end{equation}
Substituting \eqref{rand_1}, \eqref{rand_2} and \eqref{rand_6} in \eqref{rand_3}:
		\begin{equation}
        \begin{aligned}
        \mathbf {z}_{\mathrm{d}}^{(\tau)} = \big(\mathbf {Q}+  \alpha_2 \mathbf {W} 
        + \alpha_1 \mathbf {W}\big)
        \mathbf{\widetilde z} 
        - \big(\mathbf {Q} \bm \Omega + \alpha_2\mathbf {W} 
        \bm \Omega\big)\mathbf {z}_{\mathrm{s}}^{(\tau-1)} - 
        \alpha_1\mathbf {W} \bm \Omega
        \mathbf {z}_{\mathrm{s}}^{(\tau-2)}.           
        \label{rand_7}
        \end{aligned}
		\end{equation}

The fixed point $\hat{\mathbf {z}}_{\mathrm{d}}=\lim_{\tau \to \infty} \mathbf {z}_{\mathrm{d}}^{(\tau)}$ is equal to:		
		\begin{equation}
        \begin{aligned}
        \hat{\mathbf {z}}_{\mathrm{d}} = 
        \big( \mathbf {I} + \mathbf {Q} \bm \Omega + \alpha_2\mathbf {W} 
        \bm \Omega + \alpha_1\mathbf {W} \bm \Omega \big)^{-1} 
        \cdot \big(\mathbf {Q}+  \alpha_2 \mathbf {W} 
        + \alpha_1 \mathbf {W}\big)
        \mathbf{\widetilde z}.     
        \label{rand_8}
        \end{aligned}
		\end{equation}
From definitions of $\mathbf {Q}$, $\mathbf {W}$ and $\alpha_2$, we have $\mathbf {Q} \bm \Omega + \alpha_2\mathbf {W} \bm \Omega + \alpha_1\mathbf {R} \bm \Omega = \bm \Omega$ and $\mathbf {Q}+  \alpha_2 \mathbf {W} + \alpha_1 \mathbf {W} = \mathbf {I}$, thus \eqref{rand_8} becomes:
		\begin{equation}
        \begin{aligned}
        \hat{\mathbf {z}}_{\mathrm{d}} = 
        \big( \mathbf {I} + \bm \Omega \big)^{-1}\mathbf{\widetilde z}.  
        \end{aligned}
		\end{equation}
This concludes the proof.				
\end{proof}

\begin{tcolorbox}[title=Convergence of the DC-BP Algorithm with Randomized Damping]
To summarize, the convergence of the DC-BP with randomized damping depends on the spectral radius of the matrix:
		\begin{equation}
        \begin{aligned}
		\bm {\bar \Omega} &= \mathbf {Q} \bm \Omega + \alpha_2\mathbf {W} 
        \bm \Omega - \alpha_1\mathbf {W}.
        \end{aligned}
		\end{equation}
If the spectral radius $\rho(\bm {\bar \Omega})<1$, the DC-BP algorithm will converge to the same fixed point obtained by the DC-BP with synchronous scheduling.
\end{tcolorbox}	

\subsection{Randomized Damping Parameters}
The proposed randomized damping scheduling updates of selected factor to variable node means in every iteration by combining them with their values from the previous iteration using convergence parameters $p$ and $\alpha_1$:
		\begin{equation}
        \begin{aligned}
		z_{f_i \to x_s}^{(\tau)}=
		(1-q_{is})\cdot z_{f_i \to x_s}^{(\tau)} + q_{is}
		\cdot \big(\alpha_1\cdot z_{f_i \to x_s}^{(\tau-1)}+\alpha_2
		\cdot z_{f_i \to x_s}^{(\tau)}\big), 
		\label{num_fix}
        \end{aligned}
		\end{equation}
where $q_{is} \sim \mathrm{Ber}(p) \in \{0,1\}$ is independently sampled with probability $p$ for the mean from factor node $f_i$ to the variable node $x_s$.

The probability $p$ defines a fraction of a factor node to variable node messages from the current iteration that are combined with the corresponding messages from the previous iteration. The weighting coefficient $\alpha_1$ defines the ratio that determines how messages from the current and the previous iteration are combined. For example, $p = 0.2$ specifies that $20 \%$ of the messages from the current iteration will be combined with their values in the previous iteration, while $80 \%$ of messages are keeping the values calculated in the current iteration. Furthermore, if $\alpha_1 =$ $0.1$, then for the $20 \%$ of messages, the new value is obtained as a linear combination of the values calculated in the current and the previous iteration with coefficients $0.1$ and $0.9$, respectively. 

The randomized damping parameter pairs lead to trade-off between the number of non-converging simulations and the rate of convergence. In general, for the selection of $p$ and $\alpha_1$ for which only a small fraction of messages are combined with their values in the previous iteration, and that is the case for $p$ close to zero or $\alpha_1$ close to one, we observe a large number of non-converging simulations. This clearly demonstrates the necessity of using \eqref{rand_3} to ``slow down'' the BP progress, thus increasing the algorithm stability and providing improved convergence.
 
We expect that, for any selected $\alpha_1$, the BP algorithm will converge faster for smaller values of $p$, as lower $p$ leads to a reduced ``slow down'' effect. However, one needs to be careful with selection of $p$ in order to avoid the combinations of $p$ and $\alpha_1$ that lead to large number of non-converging outcomes. 
     
\section{Fast Real-Time DC State Estimation}
Monitoring and control capability of the system strongly depends on the SE accuracy as well as the periodicity of evaluation of state estimates. Ideally, in the presence of both legacy and phasor measurements, SE should run at the scanning rate (seconds or sub-second). In the following, we propose a fast real-time state estimator based on the BP algorithm. As we described, using the BP, it is possible to estimate state variables in a distributed fashion. In other words, unlike the usual scenario where measurements are transmitted directly to the control center, in the BP framework, measurements are locally collected and processed by local modules that exchange BP messages with neighboring local modules. Furthermore, even in the scenario where measurements are transmitted to the centralized control entity, the BP solution is advantageous over the classical centralized solutions in that it can be easily distributed and parallelized for high performance.

Compared to the previous section that addresses classical (static) SE problem, this section is an extension to the real-time model that operates continuously and accepts asynchronous measurements from different measurement subsystems. More precisely, we assume presence of both SCADA and WAMS infrastructure. We present appropriate models for measurement arrival processes and for the process of measurement deterioration (or ``aging'') over time. Such measurements are continuously integrated into the running instances of distributed BP-based modules. For simplicity, we present the real-time DC-BP, while extension to the non-linear SE model is possible. Furthermore, the BP-based SE is robust to ill-conditioned systems in which significant difference arise between measurement variances, thus allowing state estimator that runs without observability analysis.

To recall, the main SE routines comprise the SE algorithm, network topology processor, observability analysis and bad data analysis. The core of the SE is \emph{the SE algorithm} which provides a state estimate of the system, based on the network topology and set of measurements $\mathcal{M}$. Using information about switch and circuit breaker positions \emph{the network topology processor} generates a bus/branch model of the power network and assigns real-time measurement devices (legacy and/or PMU devices) across the bus/branch model \cite[Sec.~1.3]{abur}. As a result, the graph $\mathcal{G} =$ $(\mathcal{H},\mathcal{E})$ representing the power network is defined. In addition, the set of real-time measurements $\mathcal{M}_{\mathrm{rt}} \subseteq \mathcal{M}$ is connected to the graph $\mathcal{G}$.

According to the location and the type of real-time measurements \emph{the observability analysis} determines observable and unobservable islands. Within the observable islands, it is possible to obtain unique state estimates from the available set of real-time measurements $\mathcal{M}_{\mathrm{rt}}$, which is not the case within unobservable parts of the system. Once observability analysis is done, pseudo-measurements can be added, in order for the entire system to be observable \cite[Sec.~4.6]{abur}, \cite{monticelli}. The set of pseudo-measurements $\mathcal{M}_{\mathrm{ps}} \subset \mathcal{M}$ represents certain prior knowledge (e.g., historical data) of different electrical quantities and they are  usually assigned high values of variances \cite[Sec.~1.3]{abur}. As detailed later, we assume that, at a given time, the system measurements are either real-time or pseudo-measurements, i.e., the sets $\mathcal{M}_{\mathrm{rt}}$ and $\mathcal{M}_{\mathrm{ps}}$ are disjoint $\mathcal{M}_{\mathrm{rt}} \cap \mathcal{M}_{\mathrm{ps}} = \emptyset$ and their union is the set $\mathcal{M} = \mathcal{M}_{\mathrm{rt}} \cup \mathcal{M}_{\mathrm{ps}}$.

To summarize, in this section, we propose a fast and robust BP-based SE algorithm that can update the state estimate vector $\hat{\mathbf x}$ in a time-continuous process. Hence, the algorithm can handle each new measurement $M_i \in \mathcal{M}_{\mathrm{rt}}$ as soon as it is delivered from telemetry to the computational unit. Further, using the DC-BP algorithm, it is possible to compute the state estimate vector $\hat{\mathbf x}$ without resorting to observability analysis.

\subsection{Real-Time SE Using DC-BP}
The proposed SE solution is based on the fact that the BP-based algorithm is robust in terms of handling the ill-conditioned scenarios caused by significant differences between values of variances (e.g., phasor measurements and pseudo-measurements). Ideally, pseudo-measurements should not affect the solution within observable islands (i.e., determined with real-time measurements), therefore the variance of pseudo-measurements $M_i \in \mathcal{M}_{\mathrm{ps}}$ should be set to $v_i \to \infty$. In the conventional SE this concept is a source of ill-conditioned system. Hence, the values of pseudo-measurement variances should be defined to prevent ill-conditioned situations and ensure numerical stability of the SE algorithm (e.g., $10^{10}-10^{15}$). On the other hand, inability to define $v_i \to \infty$ causes that pseudo-measurements have more or less impact on the state estimate $\hat{\mathbf x}$, and thus the number of pseudo-measurements should be minimized to produce an observable system.

The BP SE algorithm allows the inclusion of an arbitrary number of pseudo-measurements with an extremely large values of variances (e.g., $10^{60}$), hence the impact on the observable island is negligible. Consequently, observable islands will have unique solution according to the real-time measurements, while unobservable islands will be determined according to both real-time and pseudo-measurements. Therefore, we propose a model where the network topology processor generates bus/branch model and assigns all possible measurements that exist in the power system, setting their variances to suitable values. 

Without loss of generality, we demonstrate this procedure by a toy-example, using a simple bus/branch model shown in \autoref{Fig_bus_branch} where all the possible measurements are assigned. The first step is converting the bus/branch model and its measurements configuration into the corresponding factor graph illustrated in \autoref{Fig_DC_graph}. 
	\begin{figure}[ht]
	\centering
	\begin{tabular}{@{}c@{}}
	\subfloat[]{\label{Fig_bus_branch}
	\includegraphics[width=3.3cm]{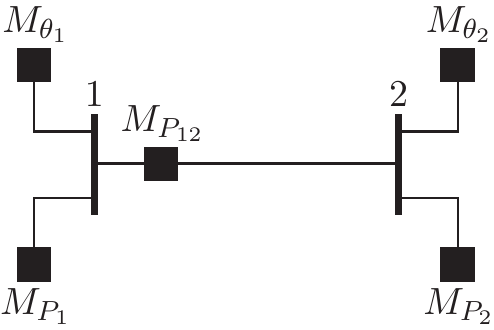}}
	\end{tabular}\quad
	\begin{tabular}{@{}c@{}}
	\subfloat[]{\label{Fig_DC_graph}
	\includegraphics[width=3.8cm]{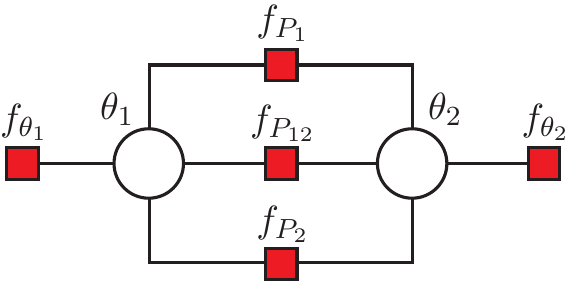}}
	\end{tabular}
	\caption{ Transformation of the bus/branch model and measurement configuration 
	(subfigure a) into the corresponding factor graph for the DC model 
	(subfigure b).}
	\label{Fig_DC}
	\end{figure}
We assume, for the time being, that all the measurements are pseudo-measurements 
$\mathcal{M} \equiv$  $\mathcal{M}_{\mathrm{ps}} =$ $\{M_{\theta_1},$ $M_{\theta_2},$ $M_{P_1},$ $M_{P_2},$ $M_{P_{12}}\}$ and $\mathcal{M}_{\mathrm{rt}} = \{\emptyset\}$, noting that the system is unobservable. Using equations \eqref{BP_vf_mean_var}, \eqref{BP_fv_mean_var} and \eqref{BP_marginal_mean_var} the BP algorithm will compute the state estimate vector $\hat{\mathbf x}$ according to the set of factor nodes $\mathcal{F}$ defined by the set of pseudo-measurements $\mathcal{M} \equiv$  $\mathcal{M}_{\mathrm{ps}}$. Hence, the system is defined according to the prior knowledge in lack of real-time measurements.

Subsequently, in an arbitrary moment, we assume that the computational unit received a real-time measurement $\mathcal{M}_{\mathrm{rt}} =$ $\{M_{\theta_1}\}$, which determines an observable island that contains bus $1$, while bus 2 remains within unobservable island. The BP algorithm in continuous process will compute the new value of state estimate $\hat{\theta}_1$ according to $M_{\theta_1}$, with insignificant impact of (high-variance) pseudo-measurements $\mathcal{M}_{\mathrm{ps}} \setminus \{M_{\theta_1}\}$, while the value of the state estimate $\hat{\theta}_2$ will be defined according to both $M_{\theta_1}$ and $\mathcal{M}_{\mathrm{ps}}\setminus \{M_{\theta_1}\}$.   

Assuming that subsequently, the computational unit receives an additional real-time measurement $M_{P_{12}}$, the system will be observable. The state estimate $\hat{\mathbf x}$ at that moment will be computed according to the real-time measurements $\mathcal{M}_{\mathrm{rt}} =$ $\{M_{\theta_1},$ $M_{P_{12}}\}$, with negligible influence of pseudo-measurements $\mathcal{M}_{\mathrm{ps}} \setminus \{M_{\theta_1}, M_{P_{12}}\}$.

Based on our extensive numerical analysis on large IEEE test cases, the proposed algorithm is able to track the state of the system in the continuous process without need for observability analysis. Note that, due the fact that the values of state variables usually fluctuate in narrow boundaries, in normal conditions, the continuous algorithm allows for fast response to new each measurement.

\section{Numerical Results} 
In this section, using numerical simulations, we analyze the convergence and evaluate the performance of the fast real-time DC-BP algorithm. In all simulated models, we start with a given IEEE test case and apply the power flow analysis to generate the exact solution. Thus, we apply the DC power flow analysis to calculate voltage angles and active powers. Further, we corrupt the exact solution by the additive white Gaussian noise of variance $v_i$ and we observe the set of measurements $\mathcal{M}$.

\subsection{Convergence Analysis}
The measurements contain active power flows and power injections, and bus voltage angles and the set of measurements is selected in such a way that the system is observable. More precisely, for each scenario, we generate 1000 random measurement configurations with the number of measurements equal either to double or triple the size of the number of state variables (i.e., we consider the redundancy to be equal 2 or 3). To evaluate the performance, we convert each of the above randomly generated IEEE test cases with a given measurement configuration into the corresponding factor graph and we run the DC-BP algorithm over the factor graph.       
	\begin{figure}[ht]
	\centering
	\captionsetup[subfigure]{oneside,margin={1.3cm,0cm}}	
	\begin{tabular}{@{}c@{}}
	\subfloat[]{
	\begin{tikzpicture}
  	\begin{axis}[width=5.5cm, height=5.0cm,
   	x tick label style={/pgf/number format/.cd,
   	set thousands separator={},fixed},
   	xlabel={Spectral Radius $\rho$},
   	ylabel={Empirical CDF $F(\rho)$},
   	label style={font=\footnotesize},
   	grid=major,
   	legend style={legend pos=north west,font=\scriptsize, column sep=0cm},
	legend columns=1,   	
   	ymin = 0, ymax = 1.1,
   	xmin = 0.55, xmax = 1.25,
   	xtick={0.6,0.7,0.8,0.9,1,1.1,1.2},
   	tick label style={font=\footnotesize},
   	ytick={0,0.1,0.2,0.3,0.4,0.5,0.6,0.7,0.8,0.9,1.0}]  
    
    \addplot[mark=otimes*,mark repeat=70, mark size=1.5pt, blue] 
   	table [x={x}, y={y}] {./chapter_03/Figs/fig3_7a/01_IEEE14_red2_synchronous.txt}; 
   	\addlegendentry{$\rho(\bm \Omega)$} 
	\addplot[mark=square*,mark repeat=70, mark size=1.5pt, red] 
   	table [x={x}, y={y}] {./chapter_03/Figs/fig3_7a/02_IEEE14_red2_randomized.txt};
   	\addlegendentry{$\rho(\bm {\bar \Omega})$}
  	\end{axis}
	\end{tikzpicture}}
	\end{tabular} \quad	
	\begin{tabular}{@{}c@{}}
	\subfloat[]{
	\begin{tikzpicture}
	\begin{axis}[width=5.5cm, height=5.0cm,
   	x tick label style={/pgf/number format/.cd,
   	set thousands separator={},fixed},
   	xlabel={Spectral Radius $\rho$},
   	ylabel={Empirical CDF $F(\rho)$},
   	label style={font=\footnotesize},
   	grid=major,
   	legend style={legend pos=north west,font=\scriptsize, column sep=0cm},
	legend columns=1,   	
   	ymin = 0, ymax = 1.1,
   	xmin = 0.75, xmax = 1.25,
   	xtick={0.8,0.9,1,1.1,1.2},
   	tick label style={font=\footnotesize},
   	ytick={0,0.1,0.2,0.3,0.4,0.5,0.6,0.7,0.8,0.9,1.0}]
	    
    \addplot[mark=otimes*,mark repeat=70, mark size=1.5pt, blue] 
   	table [x={x}, y={y}] {./chapter_03/Figs/fig3_7b/01_IEEE14_red3_synchronous.txt}; 
   	\addlegendentry{$\rho(\bm \Omega)$} 
	\addplot[mark=square*,mark repeat=70, mark size=1.5pt, red] 
   	table [x={x}, y={y}] {./chapter_03/Figs/fig3_7b/02_IEEE14_red3_randomized.txt};
   	\addlegendentry{$\rho(\bm {\bar \Omega})$} 	
  	\end{axis}
	\end{tikzpicture}}
	\end{tabular}\\
	\begin{tabular}{@{}c@{}}
	\subfloat[]{
	\centering
	\begin{tikzpicture}[spy using outlines=
	{circle, magnification=4, connect spies}]
  	\begin{axis}[width=5.5cm, height=5.0cm,
   	x tick label style={/pgf/number format/.cd,
   	set thousands separator={},fixed},
   	xlabel={Spectral Radius $\rho$},
   	ylabel={Empirical CDF $F(\rho)$},
   	label style={font=\footnotesize},
   	grid=major,
   	legend style={legend pos=south east,font=\scriptsize, column sep=0cm},
	legend columns=1,   	
   	ymin = 0, ymax = 1.1,
   	xmin = 0.95, xmax = 1.55,
   	xtick={1,1.1,1.2,1.3,1.4,1.5},
   	tick label style={font=\footnotesize},
   	ytick={0,0.1,0.2,0.3,0.4,0.5,0.6,0.7,0.8,0.9,1.0}]  
    
    \addplot[mark=otimes*,mark repeat=70, mark size=1.5pt, blue] 
   	table [x={x}, y={y}] {./chapter_03/Figs/fig3_7c/01_IEEE118_red2_synchronous.txt}; 
   	\addlegendentry{$\rho(\bm \Omega)$} 
	\addplot[mark=square*,mark repeat=70, mark size=1.5pt, red] 
   	table [x={x}, y={y}] {./chapter_03/Figs/fig3_7c/02_IEEE118_red2_randomized.txt};
   	\coordinate (spypoint) at (axis cs:1,0.9);
  	\coordinate (magnifyglass) at (axis cs:1.4,0.6);
   	\addlegendentry{$\rho(\bm {\bar \Omega})$}
  	\end{axis}
  	\spy [black, size=1.2cm] on (spypoint) in node[line width=0.15mm, fill=white] 
  	at (magnifyglass);  	
	\end{tikzpicture}}
	\end{tabular}\quad	
	\begin{tabular}{@{}c@{}}
	\subfloat[]{
	\centering
	\begin{tikzpicture}[spy using outlines=
	{circle, magnification=4, connect spies}]
  	\begin{axis}[width=5.5cm, height=5.0cm,
   	x tick label style={/pgf/number format/.cd,
   	set thousands separator={},fixed},
   	xlabel={Spectral Radius $\rho$},
   	ylabel={Empirical CDF $F(\rho)$},
   	label style={font=\footnotesize},
   	grid=major,
   	legend style={legend pos=south east,font=\scriptsize, column sep=0cm},
	legend columns=1,   	
   	ymin = 0, ymax = 1.1,
   	xmin = 0.95, xmax = 1.55,
   	xtick={1,1.1,1.2,1.3,1.4,1.5},
   	tick label style={font=\footnotesize},
   	ytick={0,0.1,0.2,0.3,0.4,0.5,0.6,0.7,0.8,0.9,1.0}]  
    
    \addplot[mark=otimes*,mark repeat=70, mark size=1.5pt, blue] 
   	table [x={x}, y={y}] {./chapter_03/Figs/fig3_7d/01_IEEE118_red3_synchronous.txt}; 
   	\addlegendentry{$\rho(\bm \Omega)$} 
	\addplot[mark=square*,mark repeat=70, mark size=1.5pt, red] 
   	table [x={x}, y={y}] {./chapter_03/Figs/fig3_7d/02_IEEE118_red3_randomized.txt};
   	\coordinate (spypoint) at (axis cs:1,1);
  	\coordinate (magnifyglass) at (axis cs:1.4,0.6);
   	\addlegendentry{$\rho(\bm {\bar \Omega})$}
  	\end{axis}
  	\spy [black, size=1.2cm] on (spypoint) in node[line width=0.15mm,fill=white] 
  	at (magnifyglass);
	\end{tikzpicture}}
	\end{tabular}
	\caption{The spectral radius of matrices $\bm \Omega$ for synchronous scheduling 
	and $\bm {\bar \Omega}$ for randomized damping for
	redundancy equal 2 for IEEE 14-bus (subfigure a) and IEEE 118-bus (subfigure c) 
	test case and for redundancy equal 3 for IEEE 14-bus (subfigure b) 
	and IEEE 118-bus (subfigure d) test case.}
	\label{fig_DC_con}
	\end{figure}
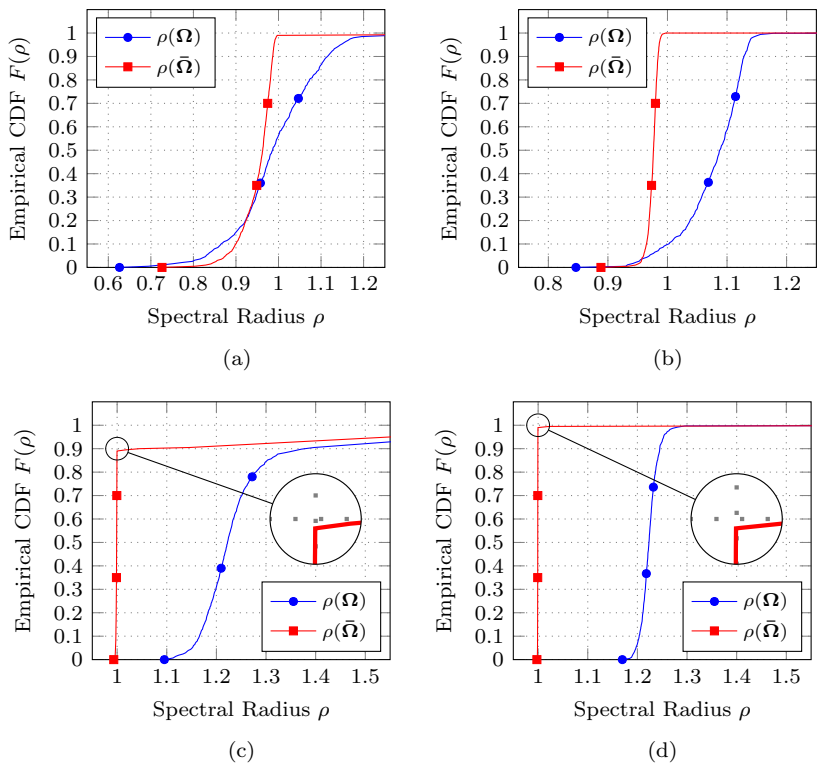     	 

As detailed in \autoref{sub:convergence_syn} and \autoref{sub:convergence_rnd}, the DC-BP with synchronous scheduling with or without randomized damping will converge if $\rho(\bm \Omega) < 1$ and $\rho(\bm {\bar \Omega}) < 1$, respectively. This condition is verified in our simulations, thus we present the convergence performance by comparing spectral radii of matrices $\bm \Omega$ and $\bm {\bar \Omega}$. 

\autoref{fig_DC_con} shows empirical cumulative density function (CDF) $F(\rho)$ of spectral radius $\rho(\bm \Omega)$ and $\rho(\bm {\bar \Omega})$ for different redundancies for IEEE 14-bus and IEEE 118-bus test case. For each scenario, the randomized damping case behaves superior in terms of the spectral radius. As an interesting and somewhat extreme case, for the IEEE 118-bus test case, the DC-BP algorithm with synchronous scheduling could not converge at all, while with randomized damping\footnote{Note that randomized damping parameters are set to $p = 0.6$ and $\alpha_1 = 0.5$.}, we recorded convergence with probability above $0.9$. As expected, the algorithm with randomized damping performs better for larger redundancy.

\subsection{Fast Real-Time DC-BP Algorithm}
We evaluate the performance of the proposed algorithm using the IEEE 14-bus test case with the measurement configuration shown in \autoref{fig_IEEE14}. The slack bus is bus 1 where the voltage angle has a given value $\theta_1 = 0$, therefore, the variance is $v_1 \to 0$ (e.g. we use $v_1 = 10^{-60}\,\mathrm{deg}$). Throughout this part, the variance of active power flow and injection pseudo-measurements are $v_{\mathrm{ps}} = 10^{60}\,\mathrm{MW}$, while voltage angle pseudo-measurements have $v_{\mathrm{ps}} = 10^{60}\,\mathrm{deg}$. Note that the base power for the IEEE 14-bus test case is $100\,\mathrm{MVA}$. 
	\begin{figure}[ht]
	\centering
	\includegraphics[width=65mm]{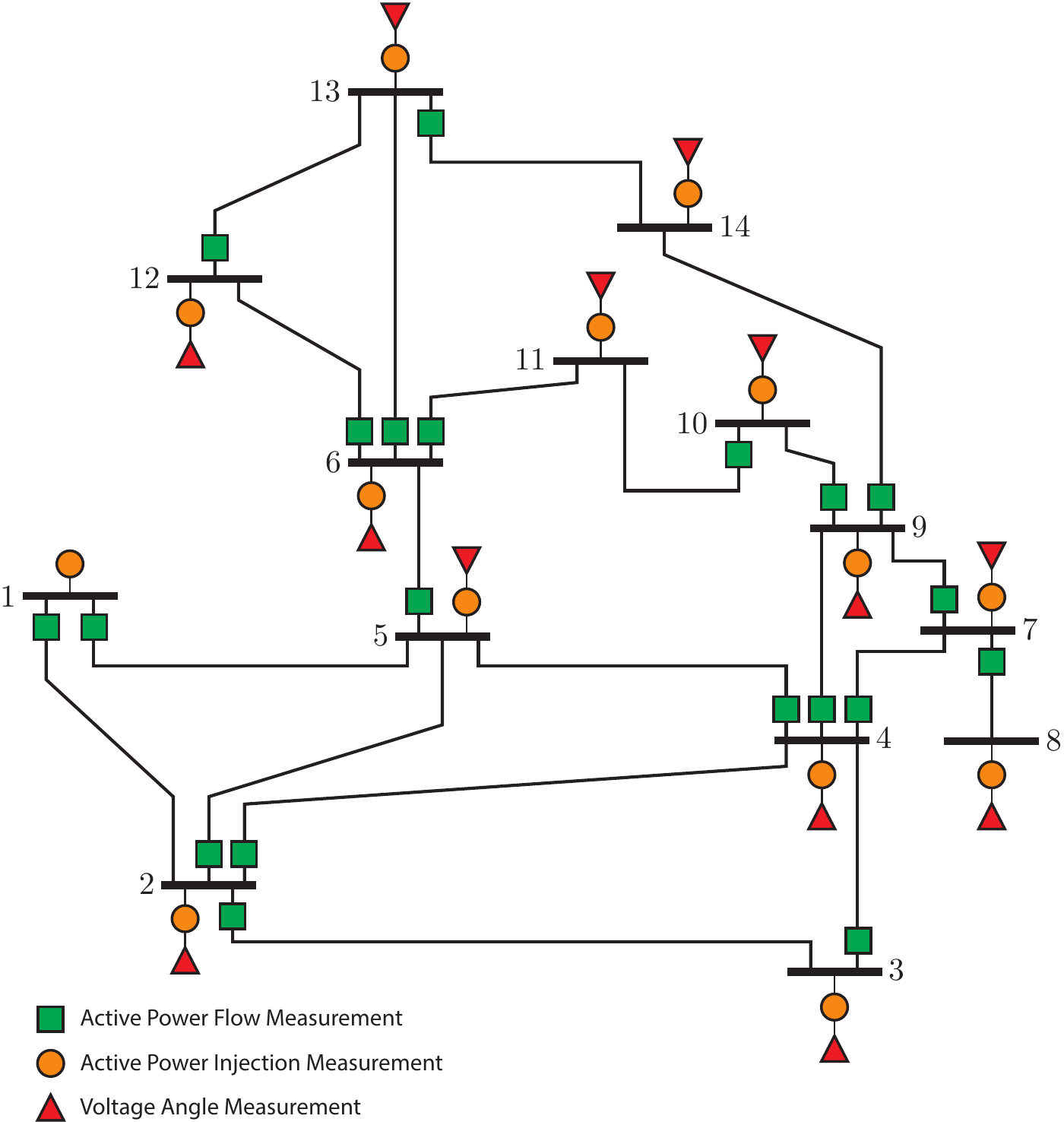}
	\caption{The IEEE 14-bus test case with measurement configuration.}
	\label{fig_IEEE14}
	\end{figure}

In each test case (described below), the algorithm starts at the time instant $t = 0$ initialized using the full set of pseudo-measurements $\mathcal{M} \equiv \mathcal{M}_{\mathrm{ps}}$ generated according to historical data.
Consider an arbitrary measurement $M_i \in \mathcal{M}$, this measurement is initialized as pseudo-measurement, i.e., at $t=0$, $M_i \in \mathcal{M}_{\mathrm{ps}}$. Let $t_{\mathrm{rt}}$ denotes the time instant when the computational unit has received the real-time measured value of $M_i$ with the predefined value of variance $v_{\mathrm{rt}}$. We model the ``aging'' of the information provided by this measurement by the linear variance increase over time up to the time instant $t_{\mathrm{ps}}$ where it becomes equal to $v_{\mathrm{ps}}$ (\autoref{fig_var}). In other words, we assume $M_i \in \mathcal{M}_{\mathrm{ps}}$ during $0 \leq t < t_{\mathrm{rt}}$ and $t \geq t_{\mathrm{ps}}$, while $M_i \in \mathcal{M}_{\mathrm{rt}}$ during $t_{\mathrm{rt}} \leq t < t_{\mathrm{ps}}$. After the transition period $t \geq t_{\mathrm{ps}}$, $M_i$ is observed as pseudo-measurement until the next real-time measurement is received.
	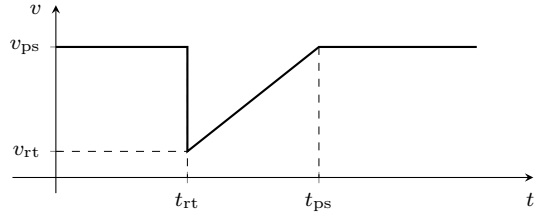
\begin{figure}[ht]
	\centering
	\begin{tikzpicture}
  	\begin{axis}[axis lines=center, axis equal image, enlargelimits=true,
  	y label style={at={(0.015,1.05)}}, x label style={at={(1.02,-0.12)}},
    xlabel={$t$},
	ylabel={$v$},
    label style={font=\footnotesize},   	
    xtick={1, 2},
    xticklabels={$t_{\mathrm{rt}}$, $t_{\mathrm{ps}}$},
    ytick={0.2,1},
    yticklabels={$v_{\mathrm{rt}}$, $v_{\mathrm{ps}}$},
    tick label style={font=\footnotesize},
    ymin = 0, ymax = 1.2,   	
   	xmin = 0, xmax = 3.3]
   	\addplot [black, no markers, thick] coordinates {(0,1) (1,1) (1,0.2) (2,1) (3.2,1)};
   	\addplot [black, no markers, dashed, very thin] coordinates {(0,0.2) (1,0.2)}; 
   	\addplot [black, no markers, dashed, very thin] coordinates {(2,1) (2,0)};
   	\addplot [black, no markers, dashed, very thin] coordinates {(1,1) (1,0)};  
  	\end{axis}
	\end{tikzpicture}
	\caption{The time-dependent function of variances for real-time measurements.}
	\label{fig_var}
	\end{figure}
	
\textbf{Test Case 1:} In the following, we analyze performance of the proposed algorithm in the scenario characterized by significant differences between variances and observe influence of the pseudo-measurements on the state estimate $\hat{\mathbf x} \equiv \hat{\bm \uptheta}{}^{\mathrm{T}}$.

In Table I, we define the (fixed) schedule and type of real-time measurements, where each real-time measurement is set to $v_{\mathrm{rt}} = 10^{-12}\,\mathrm{MW}$ at $t_{\mathrm{rt}}$ and we assume $t_{\mathrm{ps}} \to \infty$ (i.e., $v_{\mathrm{rt}}$ remains at $10^{-12}\,\mathrm{MW}$ for $t > t_{\mathrm{rt}}$ ). The example is designed in such a way that, upon reception of each real-time measurement, due to its very low variance one of the states from the estimated state vector $\hat{\bm \uptheta}{}^{\mathrm{T}}$ becomes approximately equal to the power flow solution.
	\begin{table}[ht]
	\footnotesize
  	\centering
	\begin{tabular}{c|cc||c|cc}
	\hline
	\multicolumn{1}{c|}{Time} & \multicolumn{2}{c||}{Active power flow $M_{P_{ij}}$} & 
	\multicolumn{1}{c|}{Time} & 
	\multicolumn{2}{c}{Active power flow $M_{P_{ij}}$} \rule{0pt}{1ex}\rule{0pt}{3ex}\\
	$t_{\mathrm{rt}} (\mathrm{s})$ & \multicolumn{1}{c}{from bus $i$}  & \multicolumn{1}{c||}{to bus $j$} & 
	$t_{\mathrm{rt}} (\mathrm{s})$ & \multicolumn{1}{c}{from bus $i$}  & \multicolumn{1}{c}{to bus $j$}
	\rule{0pt}{3ex}\\
	\hline
	1     & 1 & 2  & 8 & 7 & 9   
	\rule{0pt}{3ex}\\
	2     & 2 & 3  & 9 & 9 & 10    
	\rule{0pt}{2ex}\\
	3     & 3 & 4  & 10 & 10 & 11   
	\rule{0pt}{2ex}\\
	4     & 4 & 5  & 11 & 6 & 12  
	\rule{0pt}{2ex}\\
	5     & 5 & 6  & 12 & 12 & 13  
	\rule{0pt}{2ex}\\
	6     & 4 & 7  & 13 & 13 & 14  
	\rule{0pt}{2ex}\\
	7     & 7 & 8  &  &  &  
	\rule{0pt}{2ex}\\							
	\hline
	\end{tabular}
	\caption{Schedule and type of real-time measurements.}
	\label{Tab1}
	\end{table}	
	
	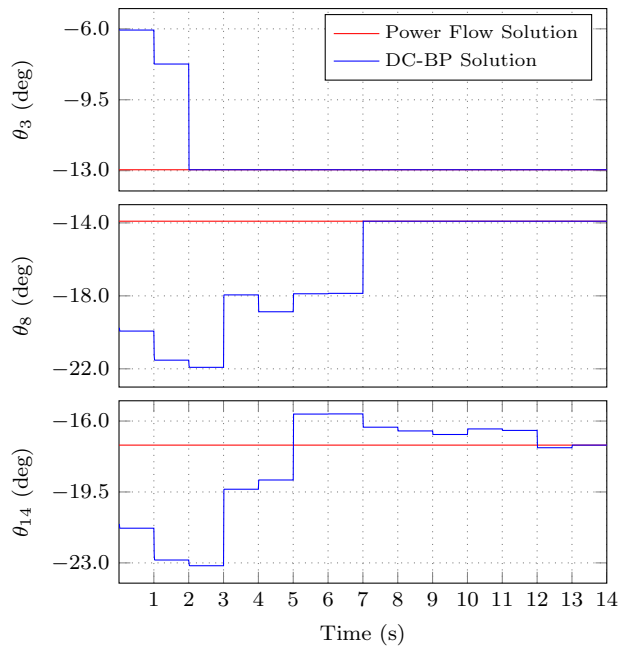
\begin{figure}[ht]
	\centering
	\begin{tikzpicture}
  	\begin{axis}[xmajorticks=false,width=8cm,height=4.0cm,at={(0cm,0cm)},
   	y tick label style={/pgf/number format/.cd,fixed,
   	fixed zerofill, precision=1, /tikz/.cd},
   	x tick label style={/pgf/number format/.cd,
   	set thousands separator={},fixed},
   	legend cell align=left,
   	legend style={legend pos=north east,font=\scriptsize},
   	legend entries={Power Flow Solution, DC-BP Solution},
   	xlabel={},
   	ylabel={$\theta_3$ (deg)},
   	label style={font=\footnotesize},   	
   	grid=major,
    xtick={1,2,3,4,5,6,7,8,9,10,11,12,13,14},
   	ytick={-6, -9.5, -13},
   	tick label style={font=\footnotesize},
   	ymin = -14.0, ymax = -5,   	
   	xmin = 0, xmax = 14]
   	\addplot [red, no markers] coordinates 
   	{(0,-12.9536631292105) (14,-12.9536631292105)}; 
   	\addplot[blue] 
   	table [x={time}, y={T3}] {./chapter_03/Figs/fig3_10/case_1.txt};
  	\end{axis}
  	\begin{axis}[xmajorticks=false,width=8cm,height=4.0cm,at={(0cm,-2.6cm)},
   	y tick label style={/pgf/number format/.cd,fixed,
   	fixed zerofill, precision=1, /tikz/.cd},
   	x tick label style={/pgf/number format/.cd,
   	set thousands separator={},fixed},
   	ylabel={$\theta_8$ (deg)},
   	label style={font=\footnotesize},   	
   	grid=major,
    xtick={1,2,3,4,5,6,7,8,9,10,11,12,13,14},
    ytick={-14, -18, -22},
   	tick label style={font=\footnotesize},
   	ymin = -23.0, ymax = -13,   	
   	xmin = 0, xmax = 14]
   	\addplot [red, no markers] coordinates 
   	{(0,-13.9070545899204) (14,-13.9070545899204)}; 
   	\addplot[blue]
   	table [x={time}, y={T8}] {./chapter_03/Figs/fig3_10/case_1.txt};
  	\end{axis} 
    \begin{axis}[width=8cm,height=4.0cm,at={(0cm,-5.2cm)},
   	y tick label style={/pgf/number format/.cd,fixed,
   	fixed zerofill, precision=1, /tikz/.cd},
   	x tick label style={/pgf/number format/.cd,
   	set thousands separator={},fixed},
   	xlabel={Time (s)},
   	ylabel={$\theta_{14}$ (deg)},
   	label style={font=\footnotesize},   	
   	grid=major,
    xtick={1,2,3,4,5,6,7,8,9,10,11,12,13,14},
  	ytick={-16, -19.5, -23},
   	tick label style={font=\footnotesize},
    ymin = -24, ymax = -15,   	
   	xmin = 0, xmax = 14]
   	\addplot [red, no markers] coordinates 
   	{(0,-17.1882875702935) (14,-17.1882875702935)}; 
   	\addplot[blue]
   	table [x={time}, y={T14}] {./chapter_03/Figs/fig3_10/case_1.txt};
  	\end{axis}
	\end{tikzpicture}
	\caption{Real-Time estimates of voltage angles $\theta_3$, 
	$\theta_8$ and $\theta_{14}$ where the computational unit received 
	active power flow real-time measurements every $t = 1\,\mathrm{s}$ with variance
	$v_{\mathrm{rt}} = 10^{-12}\,\mathrm{MW}$.}
	\label{Fig_case1}
	\end{figure}  
	
\autoref{Fig_case1} shows estimated values of voltage angles $\theta_3$, $\theta_8$ and $\theta_{14}$ for the scenario defined in \autoref{Tab1}. One can note the robustness of the proposed BP SE solution in a sense that, at any time instant, the extreme difference in variances between already received real-time measurements and remaining set of pseudo-measurements (that typically lead to ill-conditioned scenarios), are accurately solved by the BP estimator. As expected, in our pre-designed example, we clearly note a sequential refinement of the state estimate, where each new received real-time measurement $M_{P_{ij}}$ accurately defines the corresponding state variable $\theta_j$. More precisely, starting from the slack bus that has a known state value,  the real-time measurement $M_{P_{12}}$ specifies the state value of $\theta_2$ at time $t=1\,\mathrm{s}$. The chain of refinements repeats successively until $t=13\,\mathrm{s}$ when the final state variable $\theta_{14}$ is accurately estimated.   

Although somewhat trivial, the above example demonstrates that the BP-based SE algorithm provides a solution according to the real-time measurements, irrespective of the presence of (all) pseudo-measurements. In addition, \autoref{Fig_case1} shows how BP influence propagates through the network (e.g., upon reception, measurement $M_{P_{12}}$ affects the distant state variable $\theta_{14}$).

\textbf{Test Case 2:} In order to investigate how fast BP influence propagates through the network, we use the same setup given in Test Case 1, and analyse the response of the system to the received real-time measurement of different variance $v_{\mathrm{rt}} =$ $\{20^2,$ $10^2,$ $10^{-2}\}\,\mathrm{MW}$. In particular, we track the convergence of the (iterative message passing) BP algorithm over time, from the moment the real-time measurement is received, to the moment when the state estimate reaches a steady state.
	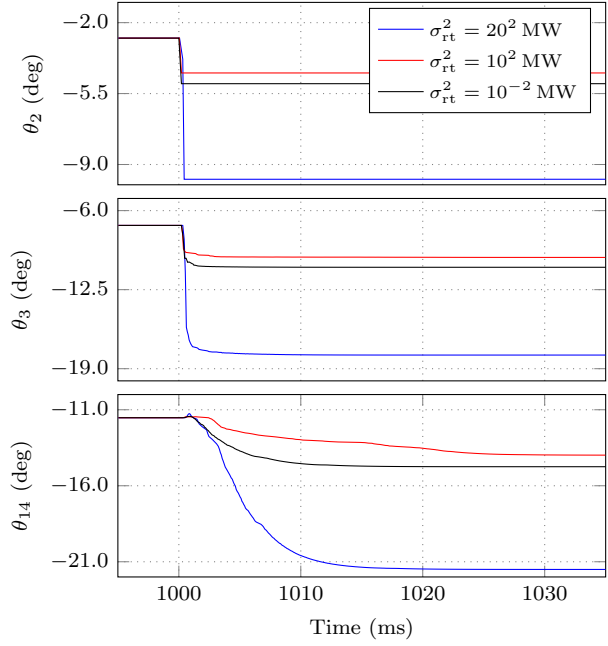
\begin{figure}[ht]
	\centering
	\begin{tikzpicture}
    \begin{axis}[xmajorticks=false,width=8cm,height=4cm,at={(0cm,0cm)},
   	y tick label style={/pgf/number format/.cd,fixed,
   	fixed zerofill, precision=1, /tikz/.cd},
   	x tick label style={/pgf/number format/.cd,
   	set thousands separator={},fixed, precision=1},
   	legend cell align=left,
   	legend style={legend pos=north east,font=\scriptsize},
   	legend entries={$\sigma_{\mathrm{rt}}^2=20^2\,\mathrm{MW}$, 
   	$\sigma_{\mathrm{rt}}^2=10^2\,\mathrm{MW}$,
   	$\sigma_{\mathrm{rt}}^2=10^{-2}\,\mathrm{MW}$},
   	ylabel={$\theta_{2}$ (deg)},
   	label style={font=\footnotesize},   	
   	grid=major,
  	ytick={-9, -5.5, -2.0},
   	tick label style={font=\footnotesize},
 	ymin = -10, ymax = -1,   	
 	xmin = 995, xmax = 1035]
   	\addplot[blue]
   	table [x={time}, y={T2}] {./chapter_03/Figs/fig3_11/case_n20.txt};
   	\addplot[red]
   	table [x={time}, y={T2}] {./chapter_03/Figs/fig3_11/case_n10.txt};
   	\addplot[black]
   	table [x={time}, y={T2}] {./chapter_03/Figs/fig3_11/case_n10-1.txt};
  	\end{axis} 	
  	\begin{axis}[xmajorticks=false,width=8cm,height=4.0cm,at={(0cm,-2.6cm)},
   	y tick label style={/pgf/number format/.cd,fixed,
   	fixed zerofill, precision=1, /tikz/.cd},
   	x tick label style={/pgf/number format/.cd,
   	set thousands separator={},fixed},
   	ylabel={$\theta_{3}$ (deg)},
   	label style={font=\footnotesize},   	
   	grid=major,
  	ytick={-19, -12.5, -6.0},
   	tick label style={font=\footnotesize},
 	ymin = -20, ymax = -5,   	
 	xmin = 995, xmax = 1035]
   	\addplot[blue]
   	table [x={time}, y={T3}] {./chapter_03/Figs/fig3_11/case_n20.txt};
   	\addplot[red]
   	table [x={time}, y={T3}] {./chapter_03/Figs/fig3_11/case_n10.txt};
   	\addplot[black]
   	table [x={time}, y={T3}] {./chapter_03/Figs/fig3_11/case_n10-1.txt};
  	\end{axis}
    \begin{axis}[width=8cm,height=4cm,at={(0cm,-5.2cm)},
   	y tick label style={/pgf/number format/.cd,fixed,
   	fixed zerofill, precision=1, /tikz/.cd},
   	x tick label style={/pgf/number format/.cd,
   	set thousands separator={},fixed, precision=1},
   	ylabel={$\theta_{14}$ (deg)},
   	xlabel={Time (ms)},   	
   	label style={font=\footnotesize},   	
   	grid=major,
  	ytick={-21, -16, -11},
   	tick label style={font=\footnotesize},
 	ymin = -22, ymax = -10,   	
 	xmin = 995, xmax = 1035]
   	\addplot[blue]
   	table [x={time}, y={T14}] {./chapter_03/Figs/fig3_11/case_n20.txt};
   	\addplot[red]
   	table [x={time}, y={T14}] {./chapter_03/Figs/fig3_11/case_n10.txt};
   	\addplot[black]
   	table [x={time}, y={T14}] {./chapter_03/Figs/fig3_11/case_n10-1.txt};
  	\end{axis} 
	\end{tikzpicture}
	\caption{Real-Time estimates of voltage angles $\theta_2$, $\theta_3$  
	and $\theta_{14}$ where the computational unit received 
	active power flow real-time measurement $M_{P_{12}}$ at the time 
	$t = 1\,\mathrm{s}$ with variances
	$v_{\mathrm{rt}} = \{20^2, 10^2, 10^{-2}\}\,\mathrm{MW}$.}
	\label{Fig_case2}
	\end{figure}

\autoref{Fig_case2} illustrates the influence of the real-time measurement $M_{P_{12}}$ received at $t_{\mathrm{rs}} = 1\,\mathrm{s}$, on the state variables  $\theta_2$, $\theta_3$ and $\theta_{14}$. As expected, the received real-time measurement has almost immediate impact on the state variable $\theta_2$, where steady state occurs within $t < 1\,\mathrm{ms}$, even for the high value of measurement variance $v_{\mathrm{rt}}=20^2\,\mathrm{MW}$. Further, this real-time measurement will influence the entire system through iterative BP message exchanges. As expected, increasing the distance between the measurement location and the bus location, more time is needed for the corresponding state variable to reach the steady state. For example, steady state of the state variable $\theta_{14}$ occurs within $t < 25\,\mathrm{ms}$. 

To summarize, the algorithm is able to provide fast response on the received real-time measurements and, for the DC SE framework, it is able to support both WAMS and SCADA technology in terms of the required computational delays. 
		  			
\textbf{Test Case 3:} In the final scenario, we consider the dynamic scenario in which the power system changes values of both generations and loads every $100\,\mathrm{s}$. In the interval between $t=0$ and $t=250\,\mathrm{s}$, only active power flow and injection real-time measurements are available with variances $v_{\mathrm{rt}} = 10^{2}\,\mathrm{MW}$ and  $t_{\mathrm{ps}}-t_{\mathrm{rt}} = 10^{3}\,\mathrm{s}$.\footnote{Although the period of $10^{3}\,\mathrm{s}$ may appear large, note that this is compensated by very high variance $v_{\mathrm{ps}} = 10^{60}\,\mathrm{MW}$ at $t_{\mathrm{ps}}$.} After $250\,\mathrm{s}$, the voltage angle real-time measurements become available with parameters $v_{\mathrm{rt}} = 10^{-6}\,\mathrm{deg}$ and $t_{\mathrm{ps}} \to \infty$. For every measurement, arrival process in each interval is modeled using Poisson process with average inter-arrival time $1/\lambda$, where for active power flow and injection real-time measurements we set $\lambda = 0.05$ and for angle real-time measurements $\lambda = 0.5$.
   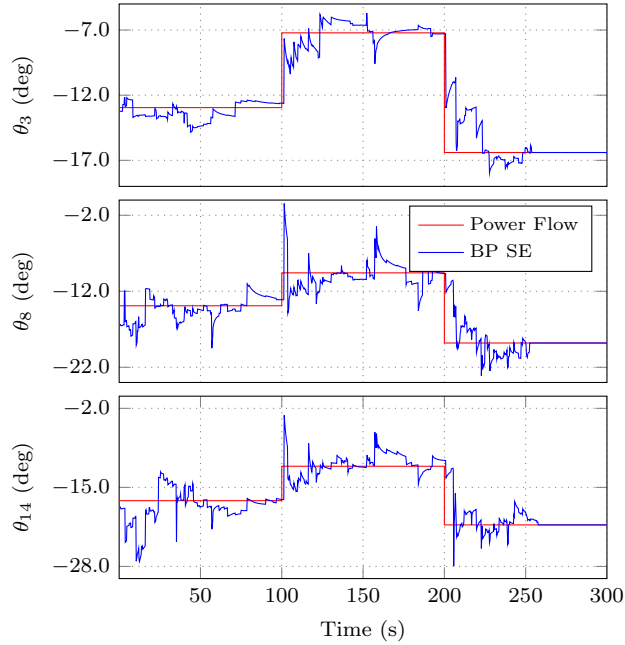
\begin{figure}[ht]
	\centering
	\begin{tikzpicture}
  	\begin{axis}[xmajorticks=false,width=8cm,height=4cm,at={(0cm,0cm)},
   	y tick label style={/pgf/number format/.cd,fixed,
   	fixed zerofill, precision=1, /tikz/.cd},
   	x tick label style={/pgf/number format/.cd,
   	set thousands separator={},fixed},
   	xlabel={},
   	ylabel={$\theta_3$ (deg)},
   	label style={font=\footnotesize},   	
   	grid=major,
    xtick={50, 100, 150, 200, 250, 300},
   	ytick={-7, -12, -17},
   	tick label style={font=\footnotesize},
   	ymin = -19.0, ymax = -5,   	
   	xmin = 0, xmax = 300]
   	\addplot [red, no markers] coordinates 
   	{(0,-12.9536631292105) (100,-12.9536631292105) 
   	(100,-7.21648794584402) (200,-7.21648794584402) 
   	(200,-16.4004583212321) (300,-16.4004583212321)};
   	\addplot[blue] 
   	table [x={time}, y={T3}] {./chapter_03/Figs/fig3_12/case_2.txt};
  	\end{axis}
  	\begin{axis}[xmajorticks=false,width=8cm,height=4cm,at={(0cm,-2.6cm)},
   	y tick label style={/pgf/number format/.cd,fixed,
   	fixed zerofill, precision=1, /tikz/.cd},
   	legend cell align=left,
   	legend style={legend pos=north east,font=\scriptsize},
   	legend entries={Power Flow, BP SE},
   	x tick label style={/pgf/number format/.cd,
   	set thousands separator={},fixed},
   	ylabel={$\theta_8$ (deg)},
   	label style={font=\footnotesize},   	
   	grid=major,
    xtick={50, 100, 150, 200, 250, 300},
   	ytick={-2, -12, -22},
   	tick label style={font=\footnotesize},
   	ymin = -24, ymax = 0,   	
   	xmin = 0, xmax = 300]
   	\addplot [red, no markers] coordinates 
   	{(0,-13.9070545899204) (100,-13.9070545899204)	
   	(100,-9.57566308601556) (200,-9.57566308601556) 
   	(200,-18.8061641931343) (300,-18.8061641931343)};
  	\addplot[blue]
   	table [x={time}, y={T8}] {./chapter_03/Figs/fig3_12/case_2.txt};
  	\end{axis} 	
    \begin{axis}[width=8cm,height=4cm,at={(0cm,-5.2cm)},
   	y tick label style={/pgf/number format/.cd,fixed,
   	fixed zerofill, precision=1, /tikz/.cd},
   	x tick label style={/pgf/number format/.cd,
   	set thousands separator={},fixed},
   	xlabel={Time (s)},
   	ylabel={$\theta_{14}$ (deg)},
   	label style={font=\footnotesize},   	
   	grid=major,
    xtick={50, 100, 150, 200, 250, 300},
  	ytick={-2, -15, -28},
   	tick label style={font=\footnotesize},
  	ymin = -30, ymax = 0,   	
   	xmin = 0, xmax = 300]
   	\addplot [red, no markers] coordinates 
   	{(0,-17.1882875702935) (100,-17.1882875702935)	
   	(100,-11.5246130227784) (200,-11.5246130227784) 
   	(200,-21.1725201082984) (300,-21.1725201082984)}; 
   	\addplot[blue]
   	table [x={time}, y={T14}] {./chapter_03/Figs/fig3_12/case_2.txt};
  	\end{axis}
	\end{tikzpicture}
	\caption{Real-time estimates of voltage angles $\theta_3$, $\theta_8$ 
	and $\theta_{14}$ where real-time measurements arrived at the 
	computational unit according to Poisson process.}
	\label{Fig_case3}
	\end{figure}	

\autoref{Fig_case3} shows state estimates of state variables $\theta_3$, $\theta_8$ and $\theta_{14}$ over the time interval of $300\,\mathrm{s}$ for the described scenario. During the first $250\,\mathrm{s}$, the BP SE provides state estimates according to incoming noisy real-time measurements and, as apparent from the figure, each new real-time measurement will affect the current state of the system. After $t=250\,\mathrm{s}$, the voltage angle real-time measurements arrived with constant and very low variance, thus providing state estimates which are considerably more accurate.

\section{Summary}
We proposed a fast real-time state estimator based on the BP algorithm. The estimator is easy to distribute and parallelize, thus alleviating computational limitations and allowing for processing measurements in real time. Convergence of the DC-BP algorithm depends of the spectral radius of the matrix that governs evolution of means from indirect factor nodes to variable nodes, and we proposed improved DC-BP algorithm using synchronous scheduling with randomized damping. 

The algorithm may run as a continuous process, with each new measurement being seamlessly processed by the distributed state estimator. In contrast to the matrix-based state estimation methods, the belief propagation approach is robust to ill-conditioned scenarios caused by significant differences between measurement variances, thus resulting in a solution that eliminates observability analysis. Using the DC model, we numerically demonstrate the performance of the state estimator in a realistic real-time system model with asynchronous measurements. We note that the extension to the non-linear state estimation is possible within the same framework.

\chapter{Native Belief Propagation based Non-Linear State Estimation}	\label{ch:native_bp}
\addcontentsline{lof}{chapter}{4 Native Belief Propagation based Non-Linear State Estimation}
The native BP-based algorithm (AC-BP) for the non-linear SE represents a logical step in the transition from a linear to a non-linear model. We use insights from the DC-BP algorithm therein to derive the AC-BP algorithm. Due to non-linearity of measurement functions, the closed-form expressions for certain classes of BP messages cannot be obtained, and using approximations, we proposed the algorithm as an approximate BP solution for the non-linear SE problem. Unfortunately, due to approximations, the AC-BP algorithm does not match the performance of the centralized non-linear SE based on Gauss-Newton method. 

Additionally, the AC-BP messages have considerably more complex form as compared to the DC-BP, and the algorithm requires prior knowledge (e.g., historical data). Despite all that, the AC-BP gives a different interpretation of the BP algorithm and establishes interesting connections between the BP algorithm and WLS equations.

Without loss of generality, in the rest of the chapter, for the AC-BP we observe only legacy measurements. To recall, the non-linear SE model is characterized by the set of state variables $\mathbf x \equiv[\bm \uptheta,\mathbf V]^{\mathrm{T}}$, while measurement functions are defined with \eqref{mf_flow}, \eqref{mf_curmag}, \eqref{mf_injcetion} and \eqref{mf_voltage_leg}.   

\section{The Factor Graph Construction}
According to \eqref{SE_likelihood}, in the non-linear scenario, the set of state variables $\mathbf x \equiv[\bm \uptheta,\mathbf V]^{\mathrm{T}}$ determines the set of variable nodes $\mathcal{V} =$ $\{(\theta_1, V_1),$ $\dots,$ $(\theta_N, V_N)\} \equiv $ $\{x_1,\dots,x_n\} $, while the set of factor nodes $\mathcal{F} =\{f_1,\dots,f_k\}$ is defined by the set of measurements $\mathcal{M}$. A factor node $f_i$ connects to a variable node $x_s \in \mathcal{V}$ if and only if the state variable $x_s$ is an argument of the corresponding measurement function $h_i(\mathbf x)$. 

\begin{example}[Constructing factor graph] In this toy example, using a simple 3-bus model presented in \autoref{Fig_ex_AC_app}, we demonstrate the conversion from a bus/branch model with a given measurement configuration into the corresponding factor graph for the AC-BP model. 
	\begin{figure}[ht]
	\centering
	\begin{tabular}{@{}c@{}}
	\subfloat[]{\label{Fig_ex_bus_branch_app}
	\includegraphics[width=2.8cm]{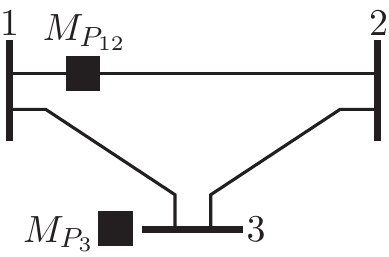}}
	\end{tabular}\quad\quad
	\begin{tabular}{@{}c@{}}
	\subfloat[]{\label{Fig_ex_AC_graph_app}
	\includegraphics[width=3.2cm]{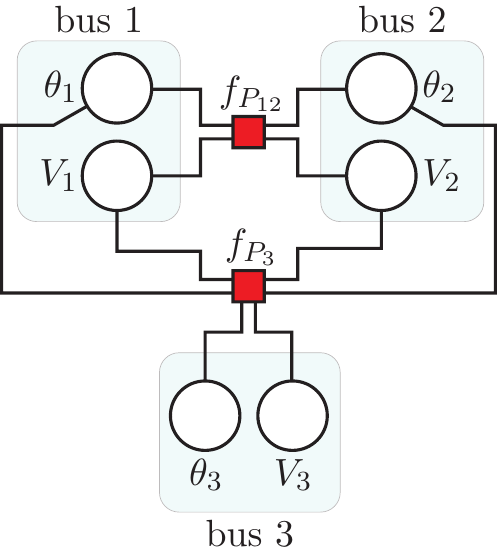}}
	\end{tabular}
	\caption{Transformation of the bus/branch model and measurement configuration 
	(subfigure a) into the corresponding factor graph for 
	the AC-BP model (subfigure b).}
	\label{Fig_ex_AC_app}
	\end{figure}
	
The variable nodes represent state variables $\mathcal{X}=$ $\{\theta_1, V_1,$ $\theta_2, V_2,$ $\theta_3, V_3\}$. Factor nodes are defined by corresponding measurements, where in our example, measurements $M_{P_{12}}$ and $M_{P_3}$ are mapped into factor nodes $\mathcal {F}=\{f_{P_{12}}$, $f_{P_3}\}$.		
\demo   
\end{example}

\section{The Belief Propagation Algorithm}
The AC-BP algorithm is based on the direct BP application over set of variable $\mathcal{V}$ and factor $\mathcal{F}$ nodes, thus insights from the DC-BP algorithm can be used.  
 
\subsection{Derivation of BP Messages and Marginal Inference} 
\textbf{Message from a variable node to a factor node:} Due to the fact that variable node output messages do not depend on measurement functions according to \eqref{FG_v_f}, relations \eqref{BP_Gauss_vf} and \eqref{BP_vf_mean_var} hold for the AC-BP.

\textbf{Message from a factor node to a variable node:} Due to non-linear measurement functions $h_i(\cdot)$, the integral in \eqref{FG_f_v} for the AC-BP cannot be evaluated in closed form. Consequently, the message from a factor node to a variable node will not be Gaussian. In the following, as \emph{an approximation}, we assume that for the AC-BP, the message $\mu_{f_i \to x_s}(x_s)$ also has the Gaussian form \eqref{BP_Gauss_fv}. According to DC-BP we provide arguments that lead us to approximations used to derive messages for the AC-BP. 

\textit{Mean value evaluation:} The expression for the mean of the DC-BP $z_{f_i \to x_s}$ is \emph{exact} and equals \eqref{BP_fv_mean}. Although the expression \eqref{BP_fv_mean} is obtained by directly evaluating \eqref{FG_f_v} for the linear DC model, we note that it has a useful interpretation via conditional expectation. For that purpose, let us define a vector $\mathbf {x}_b = \mathcal{V}_i \setminus x_s$, and let $\mathbf{z}_{\mathbf {x}_b \to f_i}$ denote a vector of mean values of messages from variable nodes $\mathcal{V}_i \setminus x_s$ to the factor node $f_i$. Then, the conditional expectation $\mathbb{E}[h_i(x_s, \mathbf {x}_b)|\mathbf {x}_b = \mathbf{z}_{\mathbf {x}_b \to f_i} ]$ can be calculated as:
	\begin{equation}
    \begin{aligned}
	\mathbb{E}[h_i(x_s, \mathbf {x}_b)|\mathbf {x}_b = 
	\mathbf{z}_{\mathbf {x}_b \to f_i} ] 
	= C_{x_s} \mathbb{E}[x_s|\mathbf {x}_b = \mathbf{z}_{\mathbf {x}_b \to f_i}] + 
	\sum_{x_b \in \mathcal{V}_i\setminus x_s} 
	C_{x_b} {z}_{x_b \to f_i} = 		
	z_i.
	\end{aligned}		
	\label{BP_cond_expectation}
	\end{equation}	

From the BP perspective, the conditional expected value $\mathbb{E}[x_s|\mathbf {x}_b = \mathbf{z}_{\mathbf {x}_b \to f_i}]$ represents the mean $z_{f_i \to x_s}$. Hence, it is possible to define the conditional expectation of non-linear measurement function $h_i(\cdot)$:
		\begin{equation}
    	\begin{gathered}
		\mathbb{E}[h_i(x_s,\mathbf {x}_b)|\mathbf {x}_b = \mathbf{z}_{\mathbf {x}_b \to f_i} ]  = 		
		z_i.
		\end{gathered}
		\label{BP_fv_mean_AC}
		\end{equation}
Due different forms of non-linear measurement functions $h_i(\cdot)$, see equations \eqref{mf_flow}, \eqref{mf_curmag} and \eqref{mf_injcetion}, the equation \eqref{BP_fv_mean_AC} will produce different forms of conditional expectation $\mathbb{E}[x_s|\mathbf{z}_{\mathbf {x}_b \to f_i}] $ $\equiv$ $z_{f_i \to x_s}$:
		\begin{subequations}
        \begin{align}
		a\mathbb{E}[x_s|\mathbf{z}_{\mathbf {x}_b \to f_i} ] + b &= 0 
        \label{BP_vf_mean_Vj}\\
		a\mathbb{E}[x_s^2|\mathbf {x}_b = 
		\mathbf{z}_{\mathbf {x}_b \to f_i} ] + 
		b\mathbb{E}[x_s|\mathbf {x}_b = \mathbf{z}_{\mathbf {x}_b \to f_i} ] +c &= 0	
		\label{BP_vf_mean_Vi}\\
		a\mathbb{E}[\sin^2 {x_s}|\mathbf {x}_b = 
		\mathbf{z}_{\mathbf {x}_b \to f_i} ] + 
		b\mathbb{E}[\sin {x_s}|\mathbf {x}_b = \mathbf{z}_{\mathbf {x}_b \to f_i} ] +c &= 0,	
		\label{BP_vf_mean_sin}
        \end{align}
		\label{BP_vf_mean_all}%
		\end{subequations}
where $a$, $b$ and $c$ are coefficients derived from non-linear measurement functions (see Appendix C for details). 

Due to quadratic form of \eqref{BP_vf_mean_Vi} and \eqref{BP_vf_mean_sin}, we may obtain two possible values for the mean value $z_{f_i \to x_s}$. Thus in order to unambiguously define $z_{f_i \to x_s}$, we assume that certain a priori knowledge of state variables, denoted as $\widetilde {\mathbf x} \equiv (\widetilde {\bm \uptheta},\widetilde{\mathbf V})$, is available (e.g., historical data). Given the prior data, we evaluate the mean value as: 
		\begin{equation}
        z_{f_i \to x_s}=
   		\begin{cases}
        z_{f_i \to x_s}^{(1)}, & \text{if}\ \Delta > 0 \;\; \text{and}\ d_1 < d_2\\[3pt]
        z_{f_i \to x_s}^{(2)}, & \text{if}\ \Delta > 0 \;\; \text{and}\ d_1 > d_2\\[3pt]
        \widetilde {x}_s, 	   & \text{if}\ \Delta < 0, 
        \end{cases}
		\label{BP_mean_cond}		
		\end{equation}
where $\Delta$ is the discriminant of the quadratic polynomial, and $d_1 = |z_{f_i \to x_s}^{(1)}- \widetilde {x}_s|$, $d_2 = |z_{f_i \to x_s}^{(2)}- \widetilde {x}_s|$, (see Appendix C for details).

\textit{The variance evaluation:} The expression for the variance of the DC-BP $v_{f_i \to x_s}$ is equal \eqref{BP_fv_var}. Let us provide another interpretation of the variance $v_{f_i \to x_s}$. For this purpose, we observe the factor graph presented in \autoref{Fig_marginalAC}.
	\begin{figure}[ht]
	\centering
	\includegraphics[width=6cm]{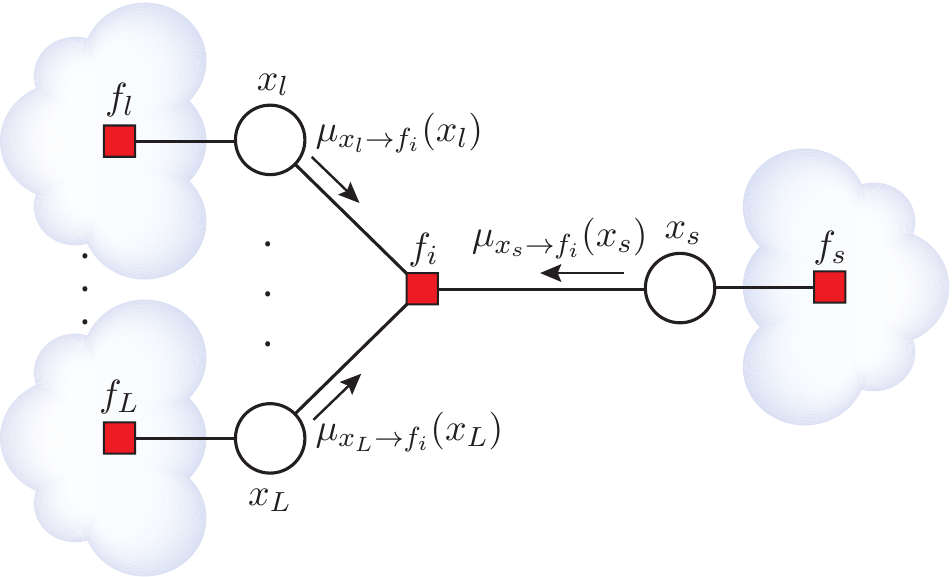}
	\caption{Factor graph which provides interpretation 
	of the variance $v_{f_i \to x_s}$.}
	\label{Fig_marginalAC}
	\end{figure} \noindent 
	
Consider the set of messages $\mu_{x_b \to f_i} = \mathcal{N}(x_b|z_{x_b \to f_i},v_{x_b \to f_i})$ arriving to the factor node $f_i$ from any variable node neighbour $x_b \in \mathcal{V}_i$. Informally, we note that this message carries a ``belief'' about itself that the variable node $x_b$ sends to the factor node $f_i$, representing collective evidence the rest of the factor graph provides about the variable node $x_b$. Let us represent this belief by an equivalent factor node attached to each variable node. Thus for a set of variable nodes $\mathcal{V}_i$, we introduce a set of factor nodes $\mathcal{F}_{\mathrm{eq}} =$ $\{f_s,$ $f_l,$ $\dots,$ $f_L \}$, where for each $x_b \in   \mathcal{V}_i$, the corresponding factor node $f_b \in \mathcal{F}_{\mathrm{eq}}$ is singly-connected to $x_b$ and by $\mathcal{N}(x_b|z_{x_b \to f_i},v_{x_b \to f_i})$. Note that, from the perspective of SE, this factor node can observed as a measurement defined by the value $z_{x_b \to f_i}$, variance $v_{x_b \to f_i}$, and measurement function $h_b(x_b)=x_b$. 

Let us now solve the system illustrated in \autoref{Fig_marginalAC} using the WLS method. It is easy to show that the corresponding Jacobian matrix\footnote{Note that the measurement function of the factor node $f_i$ is given by \eqref{BP_general_measurment_fun}, while for all other factor nodes $f_b \in \mathcal{F}_{\mathrm{eq}}$, it is equal to $h_b(x_b)=x_b$.} $\mathbf{H}$ and the measurement error covariance matrix $\mathbf{R}$ have the following form:
		\begin{equation}
		\mathbf{H} = \left( \begin{array}{cccc}
 		C_{x_s} & C_{x_l} & \dots & C_{x_L}  \\
 		1 & 0 & \dots & 0   \\
		0 & 1 & \dots & 0   \\
		\vdots & \vdots  & \hfill &\vdots   \\
		0 & 0 & \dots & 1 \\
		\end{array} \right)
		\label{BP_H}
		\end{equation}
 		\begin{equation}
    	\begin{gathered}
		\mathbf{R} = \mathrm{diag}({v_i},  v_{x_s \to f_i}, 
		v_{x_l \to f_i}, \dots, v_{x_L \to f_i}).
		\end{gathered}
		\label{BP_W}
		\end{equation}
		
A variance-covariance matrix of WLS method is defined as: 
 		\begin{equation}
    	\begin{aligned}
		\mathbb{V}(\mathbf{x}_i) &= (\mathbf H^\mathrm{T} 
		\mathbf R^{-1} \mathbf H)^{-1}
		&\setlength\arraycolsep{2pt}
		= \left( \begin{array}{cccc}
 		\mathrm{var}(x_s) & \mathrm{cov}(x_s,x_l) & \dots & 
 		\mathrm{cov}(x_s,x_L)  \\
		\mathrm{cov}(x_l,x_s) & \mathrm{var}(x_l) & \dots & 
		\mathrm{cov}(x_l,x_L)   \\
		\vdots & \vdots  & \hfill &\vdots   \\
		\mathrm{cov}(x_L,x_s) & \mathrm{cov}(x_L,x_l) & \dots & 
		\mathrm{var}(x_L) \\
		\end{array} \right).
		\end{aligned}
		\label{BP_cov_var}
		\end{equation}	
According to \eqref{BP_cov_var}, and using \eqref{BP_H} and \eqref{BP_W}, the variance $\mathrm{var}(x_s)$ is:
 		\begin{equation}
    	\begin{gathered}
		\cfrac{1}{\mathrm{var}(x_s)} = 
		\cfrac{1}{v_{x_s \to f_i}} +
		\Bigg [\cfrac{1}{C_{x_s}^2} \Big( v_i + 
		\sum_{x_b \in \mathcal{V}_i \setminus x_s}  
        C_{x_b}^2 v_{x_b \to f_i} \Big) \Bigg]^{-1}.
		\end{gathered}
		\label{BP_discuss_var}
		\end{equation}
Consider the second term on the right-hand side of \eqref{BP_discuss_var}. Recall that it represents the inverse of the variance $v_{f_i \to x_s}$ of the message from the factor node $f_i$ to the variable node $x_s$, as defined by \eqref{BP_fv_var}. Therefore, we have demonstrated that by applying WLS on the factor graph in \autoref{Fig_marginalAC}, one can obtain the expression for the variance of the message from the factor node $f_i$ to the variable node $x_s$.

For the SE that deals with non-linear measurement functions, it is possible to define a linear approximation of the variance-covariance matrix at a given point $\mathbf x_i$ using the Gauss-Newton method \eqref{AC_GN_increment}: 
 		\begin{equation}
    	\begin{gathered}
		\mathbb{V}(\mathbf x_i) = 
		[\mathbf J (\mathbf x_i)^\mathrm{T} \mathbf {R}^{-1} \mathbf J (\mathbf x_i)]^{-1}.
		\end{gathered}
		\label{BP_var_cov_GN}
		\end{equation}
It can be shown, using \eqref{BP_var_cov_GN}, that the variance $v_{f_i \to x_s}$ is governed by \eqref{BP_fv_var} where the coefficients $C_{x_p}$, $x_p \in \mathcal{V}_i$ are defined by Jacobian elements (see Appendix A and C for details):  
 		\begin{equation}
    	\begin{gathered} 
		C_{x_p} = \cfrac{\mathrm \partial{h_i(\cdot)}}{\mathrm \partial x_p}
		\Biggr|_{\substack{x_s = z_{f_i \to x_s} \\ \mathbf x_b = 
		\mathbf z_{\mathbf {x}_b \to f_i}}}.
		\end{gathered}
		\label{BP_fv_variance_jac}
		\end{equation}	
Note that the coefficients above are evaluated at the point $\mathbf x_i=(x_s,\mathbf {x}_b)$, where the values in $\mathbf x_i$ represent the mean-values of the corresponding messages.

To summarize, the message evaluation for the AC-BP is governed by \eqref{BP_vf_mean_all} and \eqref{BP_fv_var}, where coefficients are obtained using \eqref{BP_fv_variance_jac}.	

\textbf{Marginal inference:} The marginal of the state variable $x_s$ is governed by \eqref{BP_marginal_mean_var}. 
 
\subsection{Iterative AC-BP Algorithm} 
Here, the \emph{indirect factor nodes} $\mathcal{F}_{\mathrm{ind}} \subset \mathcal{F}$ include measurements of power flows, power injections and current magnitudes. The \emph{direct factor nodes} $\mathcal{F}_{\mathrm{dir}} \subset \mathcal{F}$ include measurements of bus voltage magnitudes.
\begin{algorithm} [ht]
\caption{The AC-BP}
\label{AC}
\begin{spacing}{1.25}
\begin{algorithmic}[1] 
\Procedure {Initialization $\tau=0$}{}
  \For{Each $f_s \in \mathcal{F}_{\mathrm{loc}}$}
  	  \State send $\mu_{f_s \to x_s}^{(0)}$ 
  	   to incident $x_s \in \mathcal{V}$
  \EndFor 
  \For{Each $x_s \in \mathcal{V}$}
  	  \State send $\mu_{x_s \to f_i}^{(0)} = \mu_{f_s \to x_s}^{(0)}$, 
  	  to incident $f_i \in \mathcal{F}_{\mathrm{ind}}$
  \EndFor 
  \For{Each $f_i \in \mathcal{F}_{\mathrm{ind}}$}
       \State send $\mu_{f_i \to x_s}^{(0)} = \mu_{x_s \to f_i}^{(0)}$ 
        to incident $x_s \in \mathcal{V}$ 
  \EndFor
\EndProcedure
\myline[black](-2.4,5.1)(-2.4,0.3)

\Procedure {Iteration loop $\tau=1,2,\dots$}{}  
    \While{stopping criterion is not met} 
  \For{Each $f_i \in \mathcal{F}_{\mathrm{ind}}$}
 	\State Compute $\mu_{f_i \to x_s}^{(\tau)}$ using \eqref{BP_vf_mean_all}*, \eqref{BP_fv_var}*
  \EndFor
   \For{Each $x_s \in \mathcal{V}$}
  \State Compute $\mu_{x_s \to f_i}^{(\tau)}$ using \eqref{BP_vf_mean_var}
  \EndFor
  \EndWhile
 \EndProcedure
\myline[black](-2.4,4.57)(-2.4,0.3) 

\Procedure {Output}{}
 \For{Each $x_s \in \mathcal{V}$}
    \State Compute $\hat x_s$, $v_{x_s}$ using \eqref{BP_marginal_mean_var}
  \EndFor
  \EndProcedure
\myline[black](-2.4,1.95)(-2.4,0.3)    
   \Statex *Incomming messages are obtained in previous iteration $\tau-1$  
 \end{algorithmic}
 \end{spacing}
\end{algorithm} 

The AC-BP algorithms are presented in \autoref{AC}. Note that, the initialization step for the DC-BP and AC-BP is different. This is due to the fact that the variance of the message from a factor node to a variable node for the AC-BP depends not only on the mean values of incoming messages, but also on the mean value of the message whose variance is being calculated.   		

\section{Numerical Results}
In the following, we compare the accuracy of the AC-BP algorithm to that of the centralized Gauss-Newton method using the IEEE 14-bus test case. We start with a given IEEE test case and apply the AC power flow analysis to generate the exact currents, voltages and powers across the network. Further, we corrupt the exact solution by the additive white Gaussian noise of variance $v_i$ and we observe the set of measurements. 

The IEEE 14-bus test case with fixed measurement configuration containing 61 measurement devices, as shown in Fig. \ref{fig_IEEE14AC}, is used to compare the accuracy of the SE algorithms. For each value of noise variance $v_i=$ $\{v_1$, $v_2\}$ $=\{10^{-10}$ $10^{-4} \} \,\mbox{p.u.}$, using Monte Carlo approach, we generate 1000 random sets of measurement values and feed them to the SE algorithms. 
	\begin{figure}[ht]
	\centering
	\includegraphics[width=65mm]{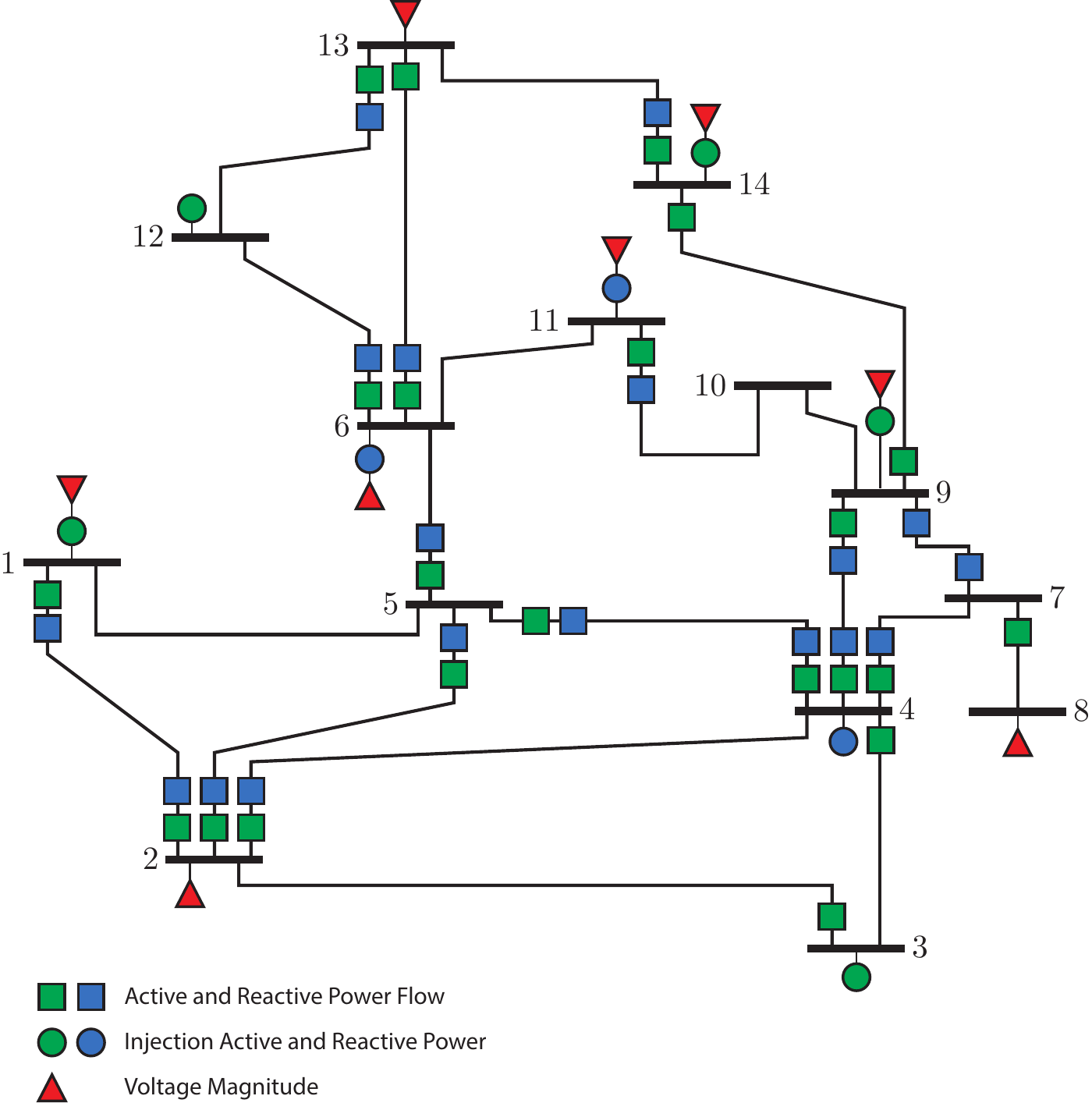}
	\caption{The IEEE 14-bus test case with given measurement configuration.}
	\label{fig_IEEE14AC}
	\end{figure} 
Note that, in order to initialize the AC-BP and the Gauss-Newton method, we use the ``flat start" assumption ($V_i=1$, $\theta_i=0$, $i=1,\dots,N$). 

To compare the accuracy of the AC-BP algorithm to that of the centralized Gauss-Newton method, we use the weighted residual sum of squares (WRSS) as a metric:
	\begin{equation}
	\begin{gathered}
	\mathrm{WRSS} = \sum_{i=1}^k 
	\cfrac{[z_i-h_i({\mathbf x})]^2}{v_i}.
	\end{gathered}
	\label{num_WRSS}
	\end{equation}
Note that WRSS is the value of the objective function of the optimization problem \eqref{SE_WLS_problem} we are solving, thus it is suitable metric for the SE accuracy. Finally, we normalize the obtained WRSS by $\mathrm{WRSS}_{\mbox{\scriptsize WLS}}$ of the centralized SE obtained using the Gauss-Newton method after 12 iterations (which we adopt as a normalization constant). This way, we compare the accuracy of BP-based algorithms to the one of the centralized SE.
	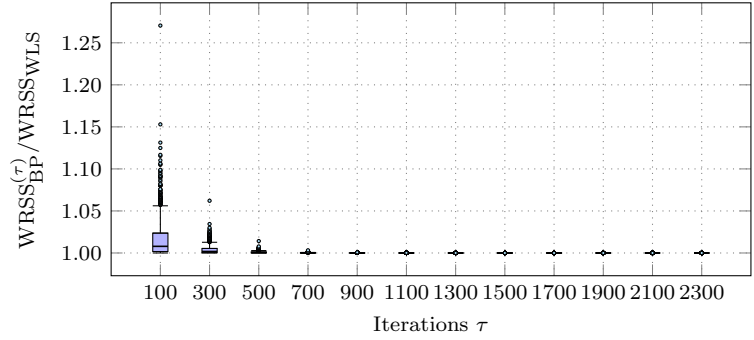
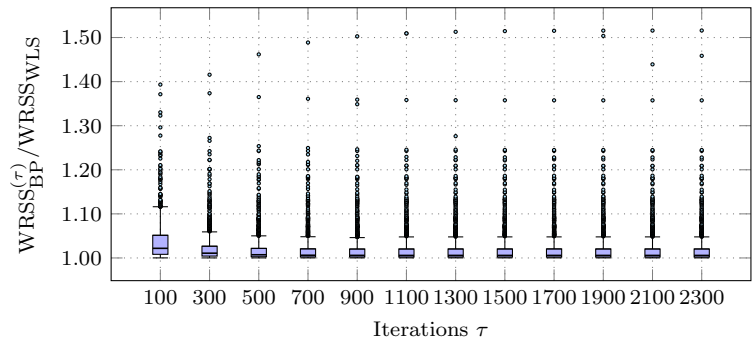
\begin{figure}[ht]
	\centering
	\captionsetup[subfigure]{oneside,margin={1.6cm,0cm}}
	\begin{tabular}{@{}c@{}}
	\subfloat[]{\label{plot3a}
	\begin{tikzpicture}
	\begin{axis} [box plot width=1.0mm,
	y tick label style={/pgf/number format/.cd,fixed,
   	fixed zerofill, precision=2, /tikz/.cd},	
	xlabel={Iterations  $\tau$},
   	ylabel={$\mathrm{WRSS}_{\mbox{\scriptsize BP}}^{(\tau)}/
   	\mathrm{WRSS}_{\mbox{\scriptsize WLS}}$},
   	grid=major,   		
   	xmin=0, xmax=13,	
   	xtick={1,2,3,4,5,6,7,8,9,10,11,12},
   	xticklabels={100, 300, 500, 700, 900, 1100, 1300, 1500, 1700, 1900, 2100, 2300},
   	ytick={1,1.05,1.1,1.15,1.2,1.25},
   	width=10cm,height=5.2cm,
   	tick label style={font=\footnotesize}, label style={font=\footnotesize}]
	\boxplot [
    forget plot, fill=blue!30,
    box plot whisker bottom index=1,
    box plot whisker top index=5,
    box plot box bottom index=2,
    box plot box top index=4,
    box plot median index=3] {./chapter_04/Figs/fig4_4a/ieee14_WRSS_v1_data.txt};   
   	
	\addplot[only marks, mark options={draw=black, fill=cyan!30},mark size=0.6pt] 
	table[x index=0, y index=1] 
	{./chapter_04/Figs/fig4_4a/ieee14_WRSS_v1_outliers.txt};   
	\end{axis}
	\end{tikzpicture}}
	\end{tabular}\\
	\begin{tabular}{@{}c@{}}
	\subfloat[]{\label{plot3b}
	\begin{tikzpicture}
	\begin{axis} [box plot width=1.0mm,
	y tick label style={/pgf/number format/.cd,fixed,
   	fixed zerofill, precision=2, /tikz/.cd},	
	xlabel={Iterations  $\tau$},
   	ylabel={$\mathrm{WRSS}_{\mbox{\scriptsize BP}}^{(\tau)}/
   	\mathrm{WRSS}_{\mbox{\scriptsize WLS}}$},
   	grid=major,   		
   	xmin=0, xmax=13,	
   	xtick={1,2,3,4,5,6,7,8,9,10,11,12},
   	xticklabels={100, 300, 500, 700, 900, 1100, 1300, 1500, 1700, 1900, 2100, 2300},
   	ytick={1,1.1,1.2,1.3,1.4,1.5},
   	width=10cm,height=5.2cm,
   	tick label style={font=\footnotesize}, label style={font=\footnotesize}]
	\boxplot [
    forget plot, fill=blue!30,
    box plot whisker bottom index=1,
    box plot whisker top index=5,
    box plot box bottom index=2,
    box plot box top index=4,
    box plot median index=3] {./chapter_04/Figs/fig4_4b/ieee14_WRSS_v2_data.txt};   
   	
	\addplot[only marks, mark options={draw=black, fill=cyan!30},mark size=0.6pt] 
	table[x index=0, y index=1] 
	{./chapter_04/Figs/fig4_4b/ieee14_WRSS_v2_outliers.txt};   
	\end{axis}
	\end{tikzpicture}}
	\end{tabular}
	\caption{The AC-BP normalized WRSS 
	(i.e., $\mathrm{WRSS}_{\mbox{\scriptsize BP}}^{(\tau)}/
   	\mathrm{WRSS}_{\mbox{\scriptsize WLS}}$) 
   	for the low noise level $v_1$(subfigure a) and the high noise 
   	level $v_2$ (subfigure b).}
	\label{fig_AC_BP_plot}
	\end{figure}

\autoref{fig_AC_BP_plot} shows the weighted residual sum of squares of the AC-BP $\mathrm{WRSS}_{\mbox{\scriptsize BP}}^{(\tau)}$ over the iterations ${\tau}$, normalized by $\mathrm{WRSS}_{\mbox{\scriptsize WLS}}$ (i.e., $\mathrm{WRSS}_{\mbox{\scriptsize BP}}^{(\tau)}/\mathrm{WRSS}_{\mbox{\scriptsize WLS}}$). We observe that the AC-BP converges for both the low and the high noise level, however, for the high noise level, the solution of the AC-BP algorithm does not correspond to the solution of the centralized SE. This is expected, since as the noise variance increases, the accuracy of Gaussian approximation of the BP messages is decreasing, which affects the accuracy of the AC-BP solution.  

\section{Summary}
The AC-BP represents an approximate BP solution for the non-linear SE problem. Despite the complexity of message forms, the AC-BP interprets the BP algorithm through conditional expectations and gives a useful insight into the relationships between the BP algorithm and WLS method. The algorithm presents the intermediate step between DC-BP and BP-based Gauss-Newton algorithm described in the next chapter.

\chapter{Distributed Gauss-Newton Method for State Estimation}	\label{ch:gn_bp}
\addcontentsline{lof}{chapter}{5 Distributed Gauss-Newton Method for State Estimation}
As the main contribution of this thesis, we adopt different methodology to derive efficient BP-based SE method. We present a novel distributed BP-based Gauss-Newton algorithm, where the BP is applied sequentially over the non-linear model, akin to what is done by the Gauss-Newton method. The resulting Gauss-Newton BP (GN-BP) algorithm represents a BP counterpart of the Gauss-Newton method. The GN-BP is the first BP-based solution for the non-linear SE model achieving exactly the same accuracy as the centralized SE via Gauss-Newton method. We note that results presented in this chapter are based on our publications \cite{cosovictsp, cosovicac}.

\section{Gauss-Newton Method as a Sequential MAP Problem}
Consider the Gauss-Newton method \eqref{AC_GN} where, at each iteration step $\nu$, the algorithm returns a new estimate of $\mathbf{x}$ denoted as $\mathbf{x}^{(\nu)}$. Note that, after a given iteration, an estimate $\mathbf{x}^{(\nu)}$ is a vector of known (constant) values. If the Jacobian matrix $\mathbf J (\mathbf x^{(\nu)})$ has a full column rank, the equation \eqref{AC_GN_increment} represents the linear WLS solution of the minimization problem \cite[Ch.~9]{hansen}:   
	\begin{equation}
    \begin{gathered}
	\min_{\Delta \mathbf x^{(\nu)}} 
	||\mathbf P^{1/2}[\mathbf r (\mathbf x^{(\nu)}) - 
	\mathbf J (\mathbf x^{(\nu)})\Delta \mathbf x^{(\nu)}]||_2^2,
    \end{gathered}
	\label{GN_WLS_increment}
	\end{equation}
where $\mathbf P = \mathbf {R}^{-1}$. Hence, at each iteration $\nu$, the Gauss-Newton method produces WLS solution of the following system of linear equations:  
	\begin{equation}
    \begin{aligned}
    \mathbf r (\mathbf x^{(\nu)})=\mathbf{g}(\Delta \mathbf x^{(\nu)})
    +\mathbf{u},
    \end{aligned}
	\label{GN_linear_SE_model}
	\end{equation}	
where $\mathbf{g}(\Delta \mathbf x^{(\nu)})= \mathbf J (\mathbf x^{(\nu)})\Delta \mathbf x^{(\nu)}$ comprises linear functions, while $\mathbf{u}$ is the vector of measurement errors. The equation \eqref{AC_GN_increment} is the weighted normal equation for the minimization problem defined in \eqref{GN_WLS_increment}, or alternatively \eqref{AC_GN_increment} is a WLS solution of \eqref{GN_linear_SE_model}.
Consequently, the probability density function associated with the \textit{i}-th measurement (i.e., the \textit{i}-th residual component $r_i$) at any iteration step $\nu$ is:
	\begin{equation}
    \begin{aligned}
	\mathcal{N}(r_i(\mathbf x^{(\nu)})|{\Delta \mathbf x^{(\nu)}},v_i)
	= 
    \cfrac{1}{\sqrt{2\pi v_i}} 
    \exp\Bigg\{\cfrac{[r_i(\mathbf x^{(\nu)}) - 
    g_i(\Delta \mathbf x^{(\nu)})]^2}{2v_i}\Bigg\}.
    \end{aligned}        
	\label{GN_m_th_residual}
	\end{equation}	

\begin{tcolorbox}[title=Gauss-Newton Method as a MAP Optimization Problem]
The MAP solution of \eqref{SE_likelihood} can be redefined as an iterative optimization problem where, instead of solving \eqref{AC_GN}, we solve:
	\begin{subequations}
    \begin{align}
    \Delta \hat {\mathbf x}^{(\nu)}&=
	\mathrm{arg} \max_{\Delta\mathbf{x}^{(\nu)}}
	\mathcal{L}\Big(\mathbf{r}(\mathbf{x}^{(\nu)})|
	\Delta\mathbf{x}^{(\nu)}\Big)	
	= \mathrm{arg} \max_{\Delta\mathbf{x}^{(\nu)}} 
	\prod_{i=1}^k \mathcal{N} \Big(r_i(\mathbf{x}^{(\nu)})|
	\Delta\mathbf{x}^{(\nu)},v_i\Big)
    \label{GN_sub_MAP}\\
	\mathbf{x}^{{(\nu+1)}} &= \mathbf{x}^{(\nu)}+ \Delta \hat {\mathbf x}^{(\nu)}.
	\label{GN_MAP_update}
    \end{align}
	\label{GN_MAP}%
	\end{subequations}
In the following, we show that the solution of the above problem \eqref{GN_MAP} can be efficiently obtained using the BP algorithm applied over the underlying factor graph.
\end{tcolorbox} 

The solution $\Delta \hat {\mathbf x}^{(\nu)}$ in each iteration $\nu = \{0,1, \dots, \nu_{\max}\}$ of the outer iteration loop, is obtained by applying the iterative BP algorithm within inner iteration loops. Every inner BP iteration loop $\tau(\nu) = \{0,1, \dots, \tau_{\max}(\nu)\}$ outputs $\Delta \hat {\mathbf x}^{(\nu,\tau_{\max}(\nu))}$ $\equiv$ $\Delta \hat {\mathbf x}^{(\nu)}$, where $\tau_{\max}(\nu)$ is the number of inner BP iterations within the outer iteration $\nu$. Note that, in general, the BP algorithm operating within inner iteration loops represents an instance of a loopy Gaussian BP over a linear model defined by linear functions $\mathbf{g}(\Delta \mathbf x^{(\nu)})$. Thus, if it converges, it provides a solution equal to the linear WLS solution $\Delta {\mathbf x}^{(\nu)}$ of \eqref{AC_GN_increment}. 

\section{The Factor Graph Construction}
From the factorization of the likelihood expression \eqref{GN_sub_MAP}, one easily obtains the factor graph corresponding to the GN-BP method as follows. The increments $\Delta \mathbf x$ of state variables $\mathbf x$ determine the set of variable nodes $\mathcal{V} = \{(\Delta \theta_1, \Delta V_1), \dots, (\Delta \theta_N, \Delta V_N)\}$ and each likelihood function $\mathcal{N} (r_i(\mathbf{x}^{(\nu)})|\Delta\mathbf{x}^{(\nu)},v_i)$ represents the local function associated with the factor node. Since the residual equals $r_i(\mathbf{x}^{(\nu)}) = z_i - h_i(\mathbf{x}^{(\nu)})$, in general, the set of factor nodes $\mathcal{F} =\{f_1,\dots,f_k\}$ is defined by the set of measurements $\mathcal{M}$. The factor node $f_i$ connects to the variable node $\Delta x_s \in \{\Delta \theta_s,\Delta V_s \}$ if and only if the increment of the state variable $\Delta x_s$ is an argument of the corresponding function ${g_i}({\Delta \mathbf x})$, i.e., if the state variable $ x_s \in \{ \theta_s, V_s \}$ is an argument of the measurement function $h_i(\mathbf x)$.

The GN-BP algorithm is applied sequentially over the non-linear model, where the main algorithm routine includes BP-based inference over MAP sub-problem \eqref{GN_sub_MAP}.  For completeness of exposition, we provide a step-by-step presentation of the GN-BP algorithm.

\subsection{Derivation of BP Messages and Marginal Inference}
\textbf{Message from a variable node to a factor node:} Consider a part of a factor graph shown in \autoref{FigGN_v_f} with a group of factor nodes $\mathcal{F}_s=\{f_i,f_w,...,f_W\}$ $\subseteq$ $\mathcal{F}$ that are neighbours of the variable node $\Delta x_s$ $\in$ $\mathcal{V}$. Let us assume that the incoming messages $\mu_{f_w \to \Delta x_s}( \Delta x_s)$, $\dots$, $\mu_{f_W \to \Delta x_s}(\Delta x_s)$ into the variable node $\Delta x_s$ are Gaussian and represented by their mean-variance pairs $(r_{f_w \to \Delta x_s},v_{f_w \to \Delta x_s})$, $\dots$, $(r_{f_W \to \Delta x_s},v_{f_W \to \Delta x_s})$. 
	\begin{figure}[ht]
	\centering
	\includegraphics[width=4.3cm]{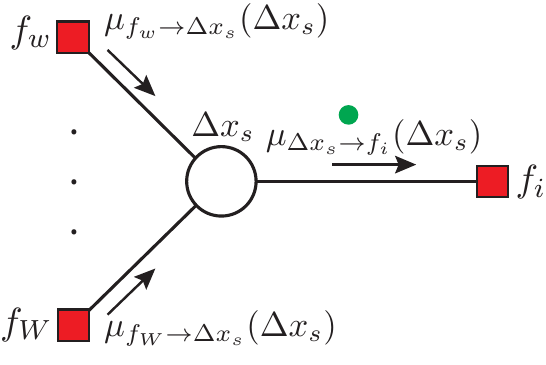}
	\caption{Message $\mu_{x_s \to f_i}(x_s)$ from variable node 
	$x_s$ to factor node $f_i$.}
	\label{FigGN_v_f}
	\end{figure} \noindent

\begin{tcolorbox}[title=Message from a Variable Node to a Factor Node]	
The message $\mu_{\Delta x_s \to f_i}(\Delta x_s)$ from the variable node $\Delta x_s$ to the factor node $f_i$ is equal to the product of all incoming factor node to variable node messages arriving at all the other incident edges \eqref{FG_v_f}. It is easy to show that the message $\mu_{\Delta x_s \to f_i}(\Delta x_s)$ is proportional to: 
		\begin{equation}
        \begin{gathered}
        \mu_{\Delta x_s \to f_i}(\Delta x_s) \propto
        \mathcal{N}(\Delta x_s|r_{\Delta x_s \to f_i}, v_{\Delta x_s \to f_i}),
        \end{gathered}
		\label{GN_Gauss_vf}
		\end{equation}	
with mean $r_{\Delta x_s \to f_i}$ and variance $v_{\Delta x_s \to f_i}$ obtained as:
		\begin{subequations}
        \begin{align}
         r_{\Delta x_s \to f_i} &= 
        \Bigg( \sum_{f_a \in \mathcal{F}_s\setminus f_i} 
        \cfrac{r_{f_{a} \to \Delta x_s}}
        {v_{f_{a} \to \Delta x_s}}\Bigg)
        v_{\Delta x_s \to f_i}
        \label{GN_vf_mean}\\
		 \cfrac{1}{v_{\Delta x_s \to f_{i}}} &= 
		\sum_{f_a \in \mathcal{F}_s\setminus f_{i}} 
		\cfrac{1}{v_{f_{a} \to \Delta x_s}},
		\label{GN_vf_var}
        \end{align}
		\label{GN_vf_mean_var}%
		\end{subequations}
where $\mathcal{F}_s \setminus f_i$ represents the set of factor nodes incident to the variable node $\Delta x_s$, excluding the factor node $f_i$. 
\end{tcolorbox}	

To conclude, after the variable node $\Delta x_s$ receives the messages from all of the neighbouring factor nodes from the set $\mathcal{F}_s\setminus f_i$, it evaluates the message $\mu_{\Delta x_s \to f_i}(\Delta x_s)$ and sends it to the factor node $f_i$. 

\textbf{Message from a factor node to a variable node:} Consider a part of a factor graph shown in \autoref{FigGN_f_v} that consists of a group of variable nodes $\mathcal{V}_i =$ $\{\Delta x_s,$ $\Delta x_l,$ $...,$ $\Delta x_L\}$ $\subseteq$ $\mathcal V$ that are neighbours of the factor node $f_i$ $\in$ $\mathcal{F}$. Let us assume that the messages $\mu_{\Delta x_l \to f_i}(\Delta x_l)$, $\dots$, $\mu_{\Delta x_L \to f_i}(\Delta x_L)$  into factor nodes are Gaussian, represented by their mean-variance pairs $(r_{\Delta x_l \to f_i}, v_{\Delta x_l \to f_i})$, $\dots$, $(r_{\Delta x_L \to f_i}, v_{\Delta x_L \to f_i})$. 
	\begin{figure}[ht]
	\centering
	\includegraphics[width=4.5cm]{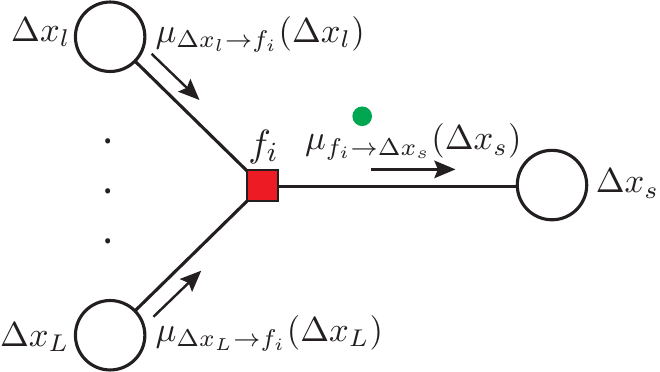}
	\caption{Message $\mu_{f_i \to \Delta x_s}(\Delta x_s)$ from factor 
	node $f_i$ to variable node $\Delta x_s$.}
	\label{FigGN_f_v}
	\end{figure} \noindent
The Gaussian function associated to the factor node $f_i$ is:
		\begin{equation}
        \begin{gathered}
		\mathcal{N}(r_i|\Delta x_s,\Delta x_l,\dots,
		\Delta x_L, v_i) 
		\propto 
        \exp\Bigg\{\cfrac{[r_i-g_i
        (\Delta x_s,\Delta x_l,\dots,\Delta x_L)]^2}
        {2v_i}\Bigg\},
		\end{gathered}        
		\label{GN_BP_Gauss_measurement_fun}
		\end{equation}
where the model contains only linear functions which we represent in a general form as:
		\begin{equation}
        \begin{gathered}
		g_i(\cdot) =
		C_{\Delta x_s} \Delta x_s + 
		\sum_{\Delta x_b \in \mathcal{V}_i\setminus \Delta x_s} 
		C_{\Delta x_b} \Delta x_b,
		\end{gathered}
		\label{GN_BP_general_measurment_fun}
		\end{equation}
where $\mathcal{V}_i\setminus \Delta x_s$ is the set of variable nodes incident to the factor node $f_i$, excluding the variable node $\Delta x_s$. 

\begin{tcolorbox}[title=Message from a Factor Node to a Variable Node]
The message $\mu_{f_i \to \Delta x_s}(\Delta x_s)$ from the factor node $f_i$ to the variable node $\Delta x_s$ is defined as a product of all incoming variable node to factor node messages arriving at other incident edges, multiplied by the function associated to the factor node $f_i$, and marginalized over all of the variables associated with the incoming messages \eqref{FG_f_v}. It can be shown that the message $\mu_{f_i \to \Delta x_s}(\Delta x_s)$ from the factor node $f_i$ to the variable node $\Delta x_s$ is represented by the Gaussian function:
		\begin{equation}
        \begin{gathered}
        \mu_{f_i \to \Delta x_s}(\Delta x_s) \propto
        \mathcal{N}(\Delta x_s|r_{f_i \to \Delta x_s}, v_{f_i \to \Delta x_s}),
        \end{gathered}
		\label{GN_Gauss_fv}
		\end{equation}
with mean $r_{f_i \to \Delta x_s}$ and variance $v_{f_i \to \Delta x_s}$ obtained as:	
		\begin{subequations}
		\begin{align}
		r_{f_i \to \Delta x_s}&=
		\cfrac{1}{C_{\Delta x_s}}
		\Bigg( r_i - \sum_{\Delta x_b \in \mathcal{V}_i \setminus \Delta x_s} 
		C_{\Delta x_b} \cdot r_{\Delta x_b \to f_i}  
		 \Bigg)
        \label{GN_fv_mean}	\\
		v_{f_i \to \Delta x_s} &= 
		\cfrac{1}{C_{\Delta x_s}^2}
		\Bigg( v_i + \sum_{\Delta x_b \in \mathcal{V}_i \setminus \Delta x_s} 
		C_{\Delta x_b}^2 \cdot v_{\Delta x_b \to f_i} 
		 \Bigg).
		\label{GN_fv_var}		
        \end{align}
        \label{GN_fv_mean_var}%
		\end{subequations}	
The coefficients $C_{\Delta x_p},\; \Delta x_p \in \mathcal{V}_i$, are Jacobian elements of the measurement function associated with the factor node $f_i$: 
		\begin{equation}
        \begin{gathered}
		C_{\Delta x_p}=\cfrac{\partial h_i(x_s,x_l,\dots, x_L)}{\partial x_p}.	
        \end{gathered}
		\label{GN_fv_coeff}
		\end{equation}	
\end{tcolorbox}
			
To summarize, after the factor node $f_i$ receives the messages from all of the neighbouring variable nodes from the set $\mathcal{V}_i\setminus \Delta x_s$, it evaluates the message $\mu_{f_i \to \Delta x_s}(\Delta x_s)$, and sends it to the variable node $\Delta x_s$.  

\textbf{Marginal Inference:} The marginal of the variable node $\Delta x_s$, illustrated in \autoref{FigGN_marginal}, is obtained as the product of all incoming messages into the variable node $\Delta x_s$ \autoref{FG_marginal}. 
	\begin{figure}[ht]
	\centering
	\includegraphics[width=4.3cm]{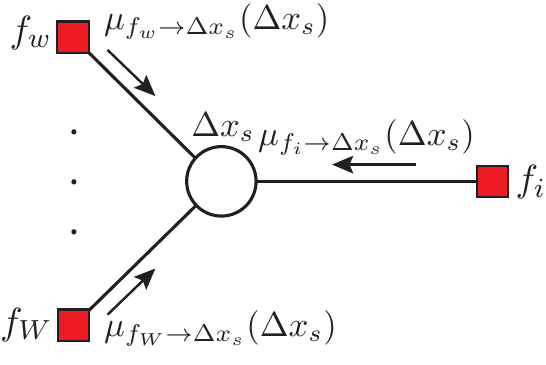}
	\caption{Marginal inference of the variable node $\Delta x_s$.}
	\label{FigGN_marginal}
	\end{figure} \noindent

\begin{tcolorbox}[title=Marginal]
It can be shown that the marginal of the state variable $\Delta x_s$ is represented by the Gaussian function: 
\begin{equation}
        \begin{gathered}
        p(\Delta x_s) \propto 
        \mathcal{N}(\Delta x_s|\Delta \hat x_s,v_{\Delta x_s}), 
        \end{gathered}
		\label{GN_Gauss_marg}
		\end{equation}
with mean $\Delta \hat x_s$ which represents the estimated value of the state variable increment $\Delta x_s$ and variance $v_{\Delta x_s}$:
		\begin{subequations}
        \begin{align}
        \Delta \hat x_s &= 
        \Bigg( \sum_{f_c \in \mathcal{F}_s} \cfrac{r_{f_{c} \to \Delta x_s}}
        {v_{f_{c} \to \Delta x_s}}\Bigg)
        v_{\Delta x_s}
        \label{GN_marginal_mean} \\
        \cfrac{1}{v_{\Delta x_s}} &= 
		\sum_{f_c \in \mathcal{F}_s} \cfrac{1}{v_{f_{c} \to \Delta x_s}}, 
		\label{GN_marginal_var}        
        \end{align}
        \label{GN_marginal_mean_var}%
		\end{subequations}
where $\mathcal{F}_s$ is the set of factor nodes incident to the variable node $\Delta x_s$.
\end{tcolorbox}		

Note that due to the fact that variable node and factor node processing preserves ``Gaussianity" of the messages, each message exchanged in BP is completely represented using only two values: the mean and the variance \cite{ping}.		

\subsection{Iterative GN-BP Algorithm}
\begin{algorithm} [ht]
\caption{The GN-BP}
\label{GN}
\begin{spacing}{1.15}
\begin{algorithmic}[1]
\Procedure {Initialization $\nu=0$}{}
  \For{Each $x_s \in \mathcal{X}$}
  	\State initialize value of $x_s^{(0)}$
  \EndFor 
\EndProcedure
\myline[black](-2.4,1.76)(-2.4,0.3)
\Procedure {Outer iteration loop $\nu=0,1,2,\dots$; $\tau = 0$}{}
\While{stopping criterion for the outer loop is not met} 
\For{Each $f_s \in \mathcal{F}_{\mathrm{dir}}$}
  	\State compute $r_{s}^{(\nu)} = z_s - x_s^{(\nu)}$
  	  \EndFor 
  	  \For{Each $f_s \in \mathcal{F}_{\mathrm{loc}}$}
  	\State send $\mu_{f_s \to \Delta x_s}^{(\nu)}$, $x_s^{(\nu)}$ 
  	to incident $\Delta x_s \in \mathcal{V}$
  \EndFor 
    
  \For{Each $\Delta x_s \in \mathcal{V}$}
  	\State send $\mu_{\Delta x_s \to f_i}^{(\nu){(\tau = 0)}} = \mu_{f_s \to \Delta x_s}^{(\nu)}$,  $x_s^{(\nu)}$ to incident $f_i \in \mathcal{F}_{\mathrm{ind}}$
  \EndFor
  \For{Each $f_i \in \mathcal{F}_{\mathrm{ind}}$}
  \State compute $r_{i}^{(\nu)} = z_i - h_i(\mathbf{x}^{(\nu)})$ and $C_{i,\Delta x_p}^{(\nu)}$; $\Delta x_p \in \mathcal{V}_i$
  \EndFor

\Procedure {Inner Iteration loop $\tau=1,2,\dots$}{} 
\While{stopping criterion for the inner loop is not met} 
  \For{Each $f_i \in \mathcal{F}_{\mathrm{ind}}$}
 	\State compute $\mu_{f_i \to \Delta x_s}^{(\tau)}$ using \eqref{GN_fv_mean_var}
  \EndFor
   \For{Each $\Delta x_s \in \mathcal{V}$}
  \State compute $\mu_{\Delta x_s \to f_i}^{(\tau)}$ using \eqref{GN_vf_mean_var}
  \EndFor
  \EndWhile 
\EndProcedure
\myline[black](-2.4,4.2)(-2.4,0.3)
  \For{Each $\Delta x_s \in \mathcal{V}$}
  	\State compute $\Delta \hat x_s^{(\nu)}$ using \eqref{GN_marginal_mean_var} 		       and $x_s^{(\nu+1)} =  x_s^{(\nu)}+\Delta \hat x_s^{(\nu)}$
  \EndFor
 \EndWhile 
\EndProcedure
\myline[black](-2.4,13.41)(-2.4,0.3) 
 \end{algorithmic}
 \end{spacing}
\end{algorithm} \noindent

The \emph{indirect factor nodes} $\mathcal{F}_{\mathrm{ind}} \subset \mathcal{F}$ correspond to measurements that measure state variables indirectly (e.g., power flows and injections). The \emph{direct factor nodes} $\mathcal{F}_{\mathrm{dir}} \subset \mathcal{F}$ correspond to the measurements that measure state variables directly (e.g., voltage magnitudes). Besides direct and indirect factor nodes, we define two additional types of singly-connected factor nodes. The \emph{slack factor node} corresponds to the slack or reference bus where the voltage angle has a given value, therefore, the residual of the corresponding state variable is equal to zero, and its variance tends to zero. Finally, the \emph{virtual factor node} is a singly-connected factor node used if the variable node is not directly measured. Residuals of virtual factor nodes approach zero, while their variances tend to infinity. 

We refer to direct factor nodes and two additional types of singly-connected factor nodes as local factor nodes $\mathcal{F}_{\mathrm{loc}} \subset \mathcal{F}$. Local factor nodes repeatedly send the same message to incident variable nodes. It is important to note that local factor nodes send messages represented by a triplet: mean (of the residual), variance and the state variable value.

The GN-BP algorithm is presented in \autoref{GN}, where the set of state variables is defined as $\mathcal{X}=\{x_1,...,x_n\}$. After the initialization (lines 1-5), the outer loop starts by computing residuals for direct and indirect factor nodes, as well as the Jacobian elements, and passes them to the inner iteration loop (lines 8-19). The inner iteration loop (lines 20-29) represents the main algorithm routine which includes BP-based message inference described in the previous subsection. We use synchronous scheduling, where all messages in a given inner iteration are updated  using the output of the previous iteration as an input \cite{elidan}. The output of the inner iteration loop is the estimate of the state variable increments. Finally, the outer loop updates the set of state variables (lines 30-32). The outer loop iterations are repeated until the stopping criteria is met.

\begin{example}[Constructing a factor graph] In this toy example, using a simple 3-bus model presented in Fig. \ref{GN_Fig_ex_bus_branch_5meas}, we demonstrate the conversion from a bus/branch model with a given measurement configuration into the corresponding factor graph. 
	\begin{figure}[ht]
	\centering
	\begin{tabular}{@{}c@{}}
	\subfloat[]{\label{GN_Fig_ex_bus_branch_5meas}
	\includegraphics[width=2.5cm]{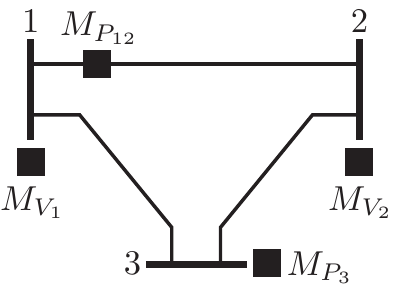}}
	\end{tabular}\quad
	\begin{tabular}{@{}c@{}}
	\subfloat[]{\label{GN_Fig_ex_FG_5meas}
	\includegraphics[width=5.1cm]{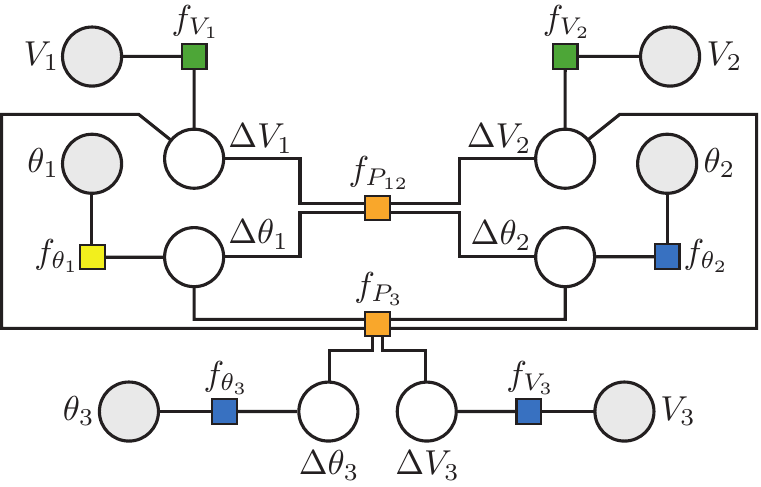}} 
	\end{tabular}
	\caption{Transformation of the bus/branch model and
	 measurement configuration (subfigure a) into the corresponding 
	 factor graph with different types of factor nodes (subfigure b).}
	\label{GN_Fig_ex_diff_factor}
	\end{figure} 
	
The corresponding factor graph is given in Fig. \ref{GN_Fig_ex_FG_5meas}, where the set of state variables is $\mathcal{X} =$ $\{(\theta_1,V_1),$  $(\theta_2,V_2),$ $(\theta_3,V_3)\}$ and the set of variable nodes is $\mathcal{V} =$ $\{(\Delta \theta_1,\Delta V_1),$ $(\Delta \theta_2, \Delta V_2),$ $(\Delta \theta_3,\Delta V_3)\}$. The indirect factor nodes (orange squares) are defined by corresponding measurements, where in our example, active power flow $M_{P_{12}}$ and active power injection $M_{P_3}$ measurements  are mapped into factor nodes $\mathcal{F}_{\mathrm{ind}} = $ $\{f_{P_{12}},$ $f_{P_{3}} \}$. The set of local factor nodes $\mathcal{F}_{\mathrm{loc}}$ consists of the set of direct factor nodes (green squares) $\mathcal{F}_{\mathrm{dir}} = $ $\{f_{V_{1}},$ $f_{V_{2}} \}$ defined by bus voltage magnitude measurements $M_{V_{1}}$ and $M_{V_2}$, virtual factor nodes (blue squares) and the slack factor node (yellow square). \demo 	  
\end{example}

\subsection{Discussion}
The presented GN-BP algorithm can be easily adapted to the multi-area SE model. Therein, each area runs the GN-BP algorithm in a fully parallelized way, exchanging messages asynchronously with neighboring areas. The algorithm may run as a continuous process, with each new measurement being seamlessly processed by the distributed state estimator. The BP approach is robust to ill-conditioned scenarios caused by significant differences between measurement variances, thus alleviating the need for observability analysis. Indeed, one can include arbitrarily large set of additional pseudo-measurements initialized using extremely high variances without affecting the BP solution within the observable part of the system \cite{fastDC}.

\subsection{Convergence of GN-BP Algorithm}
In this part, we present convergence analysis of the GN-BP algorithm with synchronous scheduling, and propose an improved GN-BP algorithm that applies synchronous scheduling \emph{with randomized damping}. We emphasize that the convergence of the GN-BP algorithm critically depends on the convergence behavior of each of the inner iteration loops. Thus, the convergence analysis presented in \autoref{sub:convergence_syn} and \autoref{sub:convergence_rnd} can be used to provide analysis for the GN-BP algorithm.    

Similar to the DC-BP analysis, it will be useful to consider a subgraph of the factor graph that contains the set of variable nodes $\mathcal{V}$, the set of indirect factor nodes $\mathcal{F}_{\mathrm{ind}} = \{f_1, \dots, f_m \} \subset \mathcal{F}$, and the set of edges $\mathcal{B} \subseteq \mathcal{V} \times \mathcal{F}_{\mathrm{ind}}$ connecting them. The number of edges in this subgraph is $b = |\mathcal{B}|$. Within the subgraph, we will consider a factor node $f_i \in \mathcal{F}_{\mathrm{ind}}$ connected to its neighboring set of variable nodes $\mathcal{V}_i = \{\Delta x_q, \dots, \Delta x_Q \} \subset \mathcal{V}$ by a set of edges $\mathcal{B}_i = \{b_i^q, \dots,b_i^Q \}  \subset \mathcal{B}$, where $d_i = |\mathcal{V}_i|$ is the degree of $f_i$. Next, we provide results on convergence of both variances and means of inner iteration loop messages, respectively.
	
\textbf{Convergence of the Variances:} As we show in \autoref{sub:convergence_syn}, the evolution of the variances $\mathbf {v}_{\mathrm{s}}$ is governed by:		
		\begin{equation}
        \begin{aligned}
        \mathbf {v}_{\mathrm{s}}^{(\tau)} = 
        \Big[\big(\mathbf{\widetilde C}^{-1}
        \bm {\Pi} \mathbf{\widetilde C}\big)\cdot
        \big(\mathfrak{D}(\mathbf{A}) 
        \big)^{-1}
        +\bm {\Sigma}_{\mathrm{a}}\mathbf{\widetilde C}^{-1}
        \Big] \mathbf{i}, 
        \end{aligned}
        \label{GN_var_evo}        
		\end{equation}
where according to \autoref{th_var} variances $\mathbf {v}_{\mathrm{s}}$ from indirect factor nodes to variable nodes always converge to a unique fixed point $\hat{\mathbf {v}}_{\mathrm{s}}$.

\textbf{Convergence of the Means:}
Using equations \eqref{GN_vf_mean} and \eqref{GN_fv_mean}, the evolution of means $\mathbf {r}_{\mathrm{s}}$ becomes a set of linear equations: 
		\begin{equation}
        \begin{aligned}
        \mathbf {r}_{\mathrm{s}}^{(\tau)} = \mathbf{\widetilde r}
        -  \bm \Omega 
        \mathbf {r}_{\mathrm{s}}^{(\tau-1)},
        \end{aligned}
        \label{GN_rand_mean}
		\end{equation}
where $\mathbf{\widetilde r} =  \mathbf{C}^{-1} \mathbf{r}_{\mathrm{a}} - \mathbf{D} \cdot \big(\mathfrak{D} (\hat {\mathbf{A}})\big)^{-1} \cdot \mathbf{L}\mathbf{r}_\mathrm{b}$, $\bm \Omega =\mathbf{D}\cdot\big(\mathfrak{D} (\hat {\mathbf{A}})\big)^{-1} \cdot \mathbf{\Gamma}\hat{\bm \Sigma}_{\mathrm{s}}^{-1}$, $ \hat {\mathbf{A}} = \mathbf{\Gamma} \hat{\bm \Sigma}_{\mathrm{s}}^{-1} \mathbf{\Gamma}^\mathrm{T} + \mathbf{L}$ and $\mathbf{D} = \mathbf{C}^{-1}\mathbf{\Pi} \mathbf{C}$ (we remind the reader that we described the vectors, matrices and matrix-operators involved in \eqref{GN_rand_mean} in \autoref{sub:convergence_syn}). According to \autoref{th_mean_sy}, the means $\mathbf {r}_{\mathrm{s}}$ from indirect factor nodes to variable nodes converge to a unique fixed point $\hat {\mathbf {r}}_{\mathrm{s}}:$
		\begin{equation}
        \begin{aligned}
        \hat {\mathbf {r}}_{\mathrm{s}} =\big(\mathbf{I}+\bm \Omega \big)^{-1} 
        \mathbf{\widetilde r},
        \label{GN_fixed}
        \end{aligned}
		\end{equation}
for any initial point $\mathbf {r}_{\mathrm{s}}^{(\tau=0)}$ if and only if the spectral radius $\rho(\bm \Omega)<1$.

Consequently, the convergence of the inner iteration loop of the GN-BP algorithm depends on the spectral radius of the matrix $\bm \Omega$. If the spectral radius $\rho(\bm \Omega)<1$, the GN-BP algorithm in the inner iteration loop $\nu$ will converge and the resulting vector of mean values will be equal to the solution of the MAP estimator. Consequently, the convergence of the GN-BP with synchronous scheduling in each outer iteration loop $\nu$ depends on the spectral radius of the matrix:
		\begin{equation}
        \begin{aligned}
		\bm \Omega(\mathbf x^{(\nu)}) = \big[\mathbf{C}(\mathbf x^{(\nu)})^{-1}
		\mathbf{\Pi} \mathbf{C}(\mathbf x^{(\nu)}) \big]
		\cdot\big[\mathfrak{D} (\mathbf{\Gamma} 
		\hat{\bm \Sigma}_{\mathrm{s}}^{-1}
        \mathbf{\Gamma}^\mathrm{T} + \mathbf{L})\big]^{-1} 
        \cdot \big(\mathbf{\Gamma}
        \hat{\bm \Sigma}_{\mathrm{s}}^{-1}\big).
        \end{aligned} 
		\label{GN_syn_omega_con}
		\end{equation}

\begin{tcolorbox}[title=Convergence of the GN-BP Algorithm with Synchronous Scheduling]		
\begin{remark}		
The GN-BP with synchronous scheduling converges to a unique fixed point if and only if $\rho_{\mathrm{syn}}<1$, where:
		\begin{equation}
        \begin{aligned}
  		\rho_{\mathrm{syn}} = 
        \max\{\rho \big(\bm {\Omega}({\mathbf x}^{(\nu)}):
        \nu = 0,1,\dots,\nu_{\max}\}. 
        \end{aligned}
        \label{GN_maxsy}
		\end{equation}
\end{remark}
\end{tcolorbox}		

\subsection{Convergence of GN-BP with Randomized Damping}
Next, we propose an improved GN-BP algorithm that applies synchronous scheduling with randomized damping. Using the proposed damping in \autoref{sub:convergence_rnd}, equation \eqref{GN_rand_mean} is redefined as:
		\begin{equation}
        \begin{aligned}
        \mathbf {r}_{\mathrm{d}}^{(\tau)} = \mathbf {r}_{\mathrm{q}}^{(\tau)}
        +\alpha_1 \mathbf {r}_{\mathrm{w}}^{(\tau-1)} + 
        \alpha_2\mathbf {r}_{\mathrm{w}}^{(\tau)},
        \label{GN_rand_3}
        \end{aligned}
		\end{equation}
where $0<\alpha_1<1$ is the weighting coefficient, and $\alpha_2 = 1 - \alpha_1$. In the above expression, $\mathbf {r}_{\mathrm{q}}^{(\tau)}$ and $\mathbf {r}_{\mathrm{w}}^{(\tau)}$ are obtained as: 
		\begin{subequations}
        \begin{align}
        \mathbf {r}_{\mathrm{q}}^{(\tau)} &= 
        \mathbf {Q} \mathbf{\widetilde r}
        - \mathbf {Q}  \bm \Omega 
        \mathbf {r}_{\mathrm{s}}^{(\tau-1)}
        \label{GN_rand_1}\\
         \mathbf {r}_{\mathrm{w}}^{(\tau)} &= 
        \mathbf {W}\mathbf{\widetilde r}
        - \mathbf {W} \bm \Omega
        \mathbf {r}_{\mathrm{s}}^{(\tau-1)},
        \label{GN_rand_2}
        \end{align}
		\end{subequations} 
where diagonal matrices $\mathbf {Q} \in \mathbb{F}_2^{b \times b}$ and $\mathbf {W} \in \mathbb{F}_2^{b \times b}$ are defined as $\mathbf {Q} = \mathrm{diag}(1 - q_1,...,1 - q_b)$, $q_i \sim \mathrm{Ber}(p)$, and $\mathbf {W} = \mathrm{diag}(q_1,...,q_b)$, respectively, and where $\mathrm{Ber}(p) \in \{0,1\}$ is a Bernoulli random variable with probability $p$ independently sampled for each mean value message. In a more compact form \eqref{GN_rand_3} can be written as follows: 
		\begin{equation}
        \begin{aligned}
        \mathbf {r}_{\mathrm{d}}^{(\tau)} =\mathbf{\bar r} - \bm {\bar \Omega} 
        \mathbf {r}_{\mathrm{s}}^{(\tau-1)}, 
        \label{GN_rand_5}
        \end{aligned}
		\end{equation}
where $\mathbf{\bar r} = \big(\mathbf {Q}+  \alpha_2 \mathbf {W}\big)  \mathbf{\widetilde r}$ and $\bm {\bar \Omega} = \mathbf {Q} \bm \Omega + \alpha_2\mathbf {W} \bm \Omega - \alpha_1\mathbf {W}$. According to \autoref{th_mean_sy}, the means $\mathbf {r}_{\mathrm{d}}$ from indirect factor nodes to variable nodes converge to a unique fixed point $\hat{\mathbf {r}}_{\mathrm{d}}$, if and only if the spectral radius $\rho(\bm {\bar \Omega})<1$, and for the resulting fixed point $\hat{\mathbf {r}}_{\mathrm{d}}$, it holds that $\hat{\mathbf {r}}_{\mathrm{d}} = \hat{\mathbf {r}}_{\mathrm{s}}$. 

To summarize, the convergence of the GN-BP with randomized damping in every outer iteration loop $\nu$ is governed by the spectral radius of the matrix:
		\begin{equation}
        \begin{gathered}
		\bm {\bar \Omega}(\mathbf x^{(\nu)}) = \mathbf {Q} 
		\bm \Omega(\mathbf x^{(\nu)}) +
		\alpha_2\mathbf {W} 
        \bm \Omega(\mathbf x^{(\nu)}) - \alpha_1\mathbf {W}.	
        \end{gathered}
		\label{GN_rand_con}
		\end{equation}

\begin{tcolorbox}[title=Convergence of the GN-BP Algorithm with Randomized Damping]
\begin{remark}			
The GN-BP with randomized damping will converge to a unique fixed point if and only if $\rho_{\mathrm{rd}}<1$, where:
		\begin{equation}
        \begin{aligned}
  		\rho_{\mathrm{rd}} = 
        \max\{\rho \big(\bm {\bar \Omega}({\mathbf x}^{(\nu)}):
        \nu = 0,1,\dots,\nu_{\max}\}, 
        \end{aligned}
        \label{maxrd}
		\end{equation}
and the resulting fixed point is equal to the fixed point obtained by the GN-BP with synchronous scheduling. 
\end{remark}
\end{tcolorbox}	

In \autoref{sec:gn_numerical}, we demonstrate that the GN-BP with randomized damping dramatically improves the GN-BP convergence.    

\section{Bad Data Analysis}
Besides the SE algorithm, one of the essential SE routines is the bad data analysis, whose main task is to detect and identify measurement errors, and eliminate them if possible. SE algorithms based on the Gauss-Newton method proceed with the bad data analysis after the estimation process is finished. This is usually done by processing the measurement residuals \cite[Ch.~5]{abur}, and typically, the largest normalized residual test (LNRT) is used to identify bad data \cite{guo}. The LNRT is performed after the Gauss-Newton algorithm converged in the repetitive process of identifying and eliminating bad data measurements one after another \cite{korres}.

Using analogies from the LNRT, we define the bad data test based on the BP messages from factor nodes to variable nodes. The presented model establishes local criteria to detect and identify bad data measurements. In \autoref{sec:gn_numerical}, we demonstrate that the BP-based bad data test (BP-BDT) significantly improves the bad data detection over the LNRT.

\textbf{The Belief Propagation Bad Data Test:} Consider a part of the factor graph shown in Fig. \ref{Fig_bad} and focus on a single measurement $M_i \in \mathcal{M}$ that defines the factor node $f_i$ $\in$ $\mathcal{F}$. Factor nodes $\{f_s,$ $f_l,$ $\dots,$ $f_L\}$ carry a collective evidence of the rest of the factor graph about the group of variable nodes $\mathcal{V}_i = \{\Delta x_s, \Delta x_l,..., \Delta x_L\}$ $\subseteq$ $\mathcal V$ incident to $f_i$. 
	\begin{figure}[ht]
	\centering
	\includegraphics[width=7.0cm]{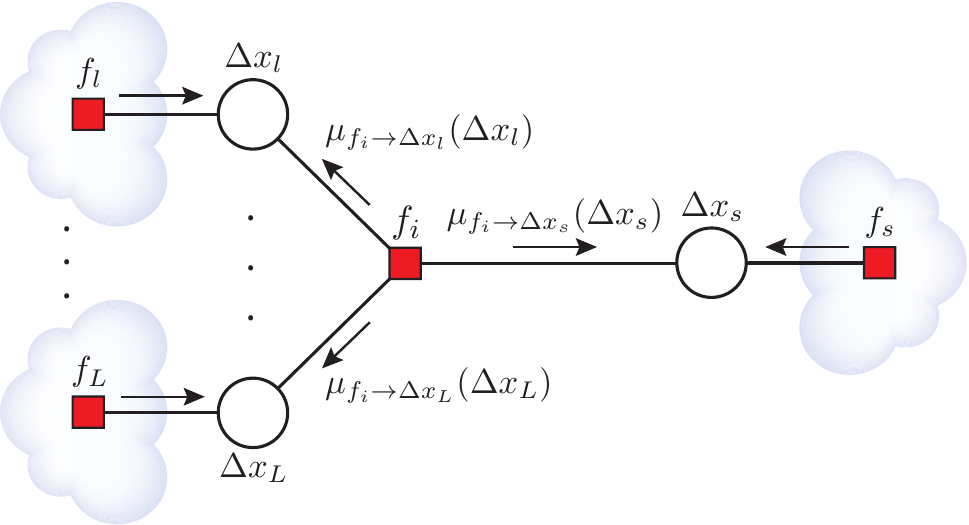}
	\caption{The part of the factor graph with messages 
	from factor node $f_i$ 
	to group of variable nodes $\mathcal{V}_i = \{\Delta x_s, \Delta x_l,..., 
	\Delta x_L\}$.}
	\label{Fig_bad}
	\end{figure} 

Assume that the estimation process is done, and the residual of the measurement $M_i$ is given as:  
		\begin{equation}
        \begin{gathered}
        r_i(\mathbf{ x}_i + \Delta {\hat{\mathbf x}}_i) = 
        z_i - h_i(\mathbf{x}_i + \Delta {\hat{\mathbf x}}_i), 
        \end{gathered}
		\label{residual}
		\end{equation}
where $\mathbf{x}_i =$ $[x_s,$ $x_l,$ $\dots,$ $x_L]^{\mathrm{T}}$ is the vector of state variables, while $\Delta {\hat{\mathbf x}}_i =$ $[\Delta \hat{x}_s,$ $\Delta \hat{x}_l,$ $\dots,$ $\Delta \hat{x}_L]^{\mathrm{T}}$ is the corresponding estimate vector of state variable increments. Let us define vectors $\mathbf{r}_{f_i} = $ $[r_{f_i \to \Delta x_s},$ $r_{f_i \to \Delta x_l},$ $\dots,$ $r_{f_i \to \Delta x_L}]^{\mathrm{T}}$ and $\mathbf{v}_{f_i} = $ $[v_{f_i \to \Delta x_s},$ $v_{f_i \to \Delta x_l},$ $\dots,$ $v_{f_i \to \Delta x_L}]^{\mathrm{T}}$ of mean and variance values of BP messages sent from the factor node $f_i$ to the variable nodes in $\mathcal{V}_i$, respectively. 

According to \eqref{GN_marginal_mean}, the vector of state variable increments $\Delta {\hat{\mathbf x}}_i$ is determined as:  
		\begin{equation}
        \begin{gathered}
        \Delta {\hat{\mathbf x}}_i = 
        [\mathrm{diag}(\mathbf{v}_{{\Delta x}_i})]\cdot
        [\mathrm{diag}(\mathbf{v}_{f_i})]^{-1} \cdot \mathbf{r}_{f_i} 
        + \mathbf{b},
        \end{gathered}
		\label{evo}
		\end{equation}
where $\mathbf{v}_{{\Delta x}_i} =$ $[v_{\Delta x_s},$ $v_{\Delta x_l},$ $\dots,$ $v_{\Delta x_L}]^{\mathrm{T}}$ is the vector of variable node variances obtained using \eqref{GN_marginal_var} and the vector $\mathbf{b}$ carries evidence of the rest of the graph about the corresponding variable nodes $\mathcal{V}_i$.

From \eqref{evo}, one can note that the BP-based SE algorithm decomposes the contribution of each factor node to state variable increments, thus providing insight in the structure of measurement residual in \eqref{residual}, where the impact of each measurement can be observed. More precisely, the expression $[\mathrm{diag}(\mathbf{v}_{f_i})]^{-1} \cdot$ $\mathbf{r}_{f_i}$ determines the influence of the measurement $M_i$ to the residual \eqref{residual}. 
To recall, the mean-value messages $\mathbf{r}_{f_i}$ contain ``beliefs'' of the factor node $f_i$ about variable nodes in $\mathcal{V}_i$, with the corresponding variances $\mathbf{v}_{f_i}$. Consequently, if the measurement $M_i$ represents bad data,  it will likely provide an inflated values of the normalized residual components $[\mathrm{diag}(\mathbf{v}_{f_i})]^{-1} \cdot$ $\mathbf{r}_{f_i}$ in \eqref{evo}. 

\begin{tcolorbox}[title=BP-based Bad Data Test Criteria]
We observe the following vector corresponding to each factor node $f_i$ to detect the bad data:   
		\begin{equation}
        \begin{gathered}
        \mathbf{r}_{\mbox{\scriptsize BP},f_i} = 
        [\mathrm{diag}(\mathbf{v}_{f_i})]^{-1}
        \cdot [\mathrm{diag}(\mathbf{r}_{f_i})] \cdot
        \mathbf{r}_{f_i}. 
        \end{gathered}
		\label{det_bad}
		\end{equation}
Note, the expression $[\mathrm{diag}(\mathbf{r}_{f_i})] \cdot$ $\mathbf{r}_{f_i} = $ $[r_{f_i \to \Delta x_s}^2,$ $r_{f_i \to \Delta x_l}^2,$ $\dots,$ $r_{f_i \to \Delta x_L}^2]^{\mathrm{T}}$ favors larger values of $\mathbf{r}_{f_i}$.
\end{tcolorbox} 

Finally, the BP-BDT is given in Algorithm \ref{BD} following similar steps as the LNRT \cite[Sec.~5.7]{abur}. Namely, after the state estimation process is done, we compute $\mathbf{r}_{\mbox{\scriptsize BP},f_i}$, $f_i$ $\in$ $\mathcal{F}$, using \eqref{det_bad}, and observe $\bar {r}_{\mbox{\scriptsize BP},f_i}$ as the largest element of $\mathbf{r}_{\mbox{\scriptsize BP},f_i}$. Comparing $\bar {r}_{\mbox{\scriptsize BP},f_i}$ values among all factor nodes, we find the largest such value ${r}_{\mbox{\scriptsize BP},f_m}$ corresponding to the $m$-th factor node. If ${r}_{\mbox{\scriptsize BP},f_m} > \kappa$, then the $m$-th measurement is suspected as bad data, where $\kappa$ is the bad data identification threshold.
\begin{algorithm} [ht]
\caption{The BP-BDT}
\label{BD}
\begin{spacing}{1.15}
\begin{algorithmic}[1]
\If {the GN-BP algorithm is converged}{}
  \For{Each $f_i \in \mathcal{F}$}
  	\State compute $\mathbf{r}_{\mbox{\tiny BP},f_i}$ using \eqref{det_bad}
  	\State find $\bar {r}_{\mbox{\tiny BP},f_i}$ as the largest 
  	 element of $\mathbf{r}_{\mbox{\tiny BP},f_i}$  
  \EndFor 
  	\State find ${r}_{\mbox{\tiny BP},f_m}$ as 
  	the largest element among all  $\bar {r}_{\mbox{\tiny BP},f_i}$
  	\If {${r}_{\mbox{\tiny BP},f_m} > \tau$}
  	\State the measurement $m$-th is suspected as bad data
  	\EndIf
\EndIf
 \end{algorithmic}
 \end{spacing}
\end{algorithm} 

\section{Numerical Results} \label{sec:gn_numerical}
In the simulated model, we start with a given IEEE test case and apply the power flow analysis to generate the exact solution. Further, we corrupt the exact solution by the additive white Gaussian noise of variance $v_i$, and we observe the set of measurements: legacy (active and reactive injections and power flows, line current magnitudes and bus voltage magnitudes) and phasor measurements (bus voltage and line current phasors). The set of measurements is selected in such a way that the system is observable. More precisely, for each scenario, we generate 300 random measurement configurations in order to obtain average performances.

In all models, we use measurement variance equal to $v_i = 10^{-10}\,\mbox{p.u.}$ for PMUs, and $v_i = 10^{-4}\,\mbox{p.u.}$ for legacy devices. To initialize the GN-BP and Gauss-Newton method, we run algorithms using ``flat start" with a small random perturbation \cite[Sec.~9.3]{abur} or ``warm start" where we use the same initial point as the one applied in AC power flow. Finally, randomized damping parameters are set to $p = 0.8$ and $\alpha_1 = 0.4$ (obtained by exhaustive search). To evaluate the performance of the GN-BP algorithm, we convert each of the above randomly generated IEEE test cases with a given measurement configuration into the corresponding factor graph, and we run the GN-BP algorithm.  

\textbf{Convergence and Accuracy:} We consider IEEE 30-bus test case with 5 PMUs and the set of legacy measurements with redundancy $\gamma$ $\in$ $\{2,3,4,5\}$. We first set the number of inner iterations to a high value of $\tau_{\max}(\nu)=5000$ iterations for each outer iteration $\nu$, where $\nu_{\max} = 11$, with the goal of investigating convergence and accuracy of GN-BP.

Fig. \ref{plot1} shows empirical cumulative density function (CDF) $F(\rho)$ of spectral radius $\rho_{\mathrm{syn}}$ and $\rho_{\mathrm{rd}}$ for different redundancies for ``flat start" and ``warm start". For each scenario, the randomized damping case is superior in terms of the spectral radius. For example, for redundancy $\gamma = 5$ and ``flat start", we record convergence with probability $0.98$ for randomized damping and $0.25$ for synchronous scheduling. When operated in ``warm start" via, e.g., large-scale historical data, the GN-BP can be integrated into continuous real-time SE framework following similar steps as in \cite{fastDC}.
	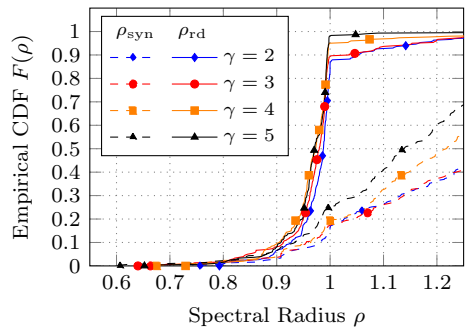
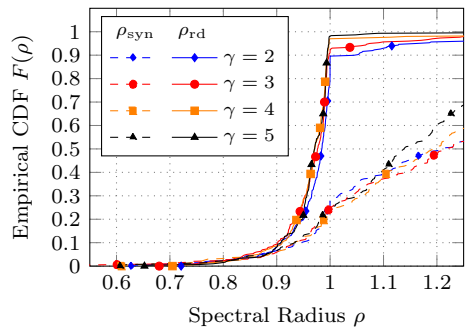
\begin{figure}[ht]
	\centering
	\captionsetup[subfigure]{oneside,margin={1.3cm,0cm}}
	\begin{tabular}{@{}c@{}}
	\subfloat[]{\label{plot1a}
	\centering
	\begin{tikzpicture}
  	\begin{axis}[width=6.5cm, height=5.0cm,
   	x tick label style={/pgf/number format/.cd,
   	set thousands separator={},fixed},
   	xlabel={Spectral Radius $\rho$},
   	ylabel={Empirical CDF $F(\rho)$},
   	label style={font=\footnotesize},
   	grid=major,
   	legend style={legend pos=north west,font=\scriptsize, column sep=0cm},
	legend columns=2,   	
   	ymin = 0, ymax = 1.1,
   	xmin = 0.55, xmax = 1.25,
   	xtick={0.6,0.7,0.8,0.9,1,1.1,1.2},
   	tick label style={font=\footnotesize},
   	ytick={0,0.1,0.2,0.3,0.4,0.5,0.6,0.7,0.8,0.9,1.0}]
	\addlegendimage{legend image with text=$\rho_{\mathrm{syn}}$}
    \addlegendentry{}
    \addlegendimage{legend image with text=$\rho_{\mathrm{rd}}$}
    \addlegendentry{}   
    
    \addplot[mark=diamond*,mark repeat=70, mark size=1.5pt, blue, dashed] 
   	table [x={x}, y={y}] {./chapter_05/Figs/fig5_5a/01_std10m2red2_sy.txt}; 
   	\addlegendentry{}  
	\addplot[mark=diamond*,mark repeat=70, mark size=1.5pt, blue] 
   	table [x={x}, y={y}] {./chapter_05/Figs/fig5_5a/02_std10m2red2_rd.txt};
   	\addlegendentry{$\gamma = 2$}
   	
   	\addplot[mark=otimes*, mark repeat=68, mark size=1.5pt, red, dashed]
	table [x={x}, y={y}] {./chapter_05/Figs/fig5_5a/03_std10m2red3_sy.txt};
   	\addlegendentry{}
	\addplot[mark=otimes*, mark repeat=68, mark size=1.5pt, red]
	table [x={x}, y={y}] {./chapter_05/Figs/fig5_5a/04_std10m2red3_rd.txt};   	
   	\addlegendentry{$\gamma = 3$}
   	
	\addplot[mark=square*,mark repeat=58, mark size=1.4pt, orange, dashed]
   	table [x={x}, y={y}] {./chapter_05/Figs/fig5_5a/05_std10m2red4_sy.txt};
   	\addlegendentry{}
   	\addplot[mark=square*,mark repeat=58, mark size=1.4pt, orange]
   	table [x={x}, y={y}] {./chapter_05/Figs/fig5_5a/06_std10m2red4_rd.txt};
   	\addlegendentry{$\gamma = 4$}
   	
    \addplot[mark=triangle*,mark repeat=74, mark size=1.5pt, black, dashed]
   	table [x={x}, y={y}] {./chapter_05/Figs/fig5_5a/07_std10m2red5_sy.txt};
   	\addlegendentry{}
   	\addplot[mark=triangle*,mark repeat=74, mark size=1.5pt, black]
   	table [x={x}, y={y}] {./chapter_05/Figs/fig5_5a/08_std10m2red5_rd.txt};
   	\addlegendentry{$\gamma = 5$}
  	\end{axis}
	\end{tikzpicture}}
	\end{tabular}\\
	\begin{tabular}{@{}c@{}}
	\subfloat[]{\label{plot1b}
	\begin{tikzpicture}
	\begin{axis}[width=6.5cm, height=5.0cm,
   	x tick label style={/pgf/number format/.cd,
   	set thousands separator={},fixed},
   	xlabel={Spectral Radius $\rho$},
   	ylabel={Empirical CDF $F(\rho)$},
   	label style={font=\footnotesize},
   	grid=major,
   	legend style={legend pos=north west,font=\scriptsize, column sep=0cm},
	legend columns=2,   	
   	ymin = 0, ymax = 1.1,
   	xmin = 0.55, xmax = 1.25,
   	xtick={0.6,0.7,0.8,0.9,1,1.1,1.2},
   	tick label style={font=\footnotesize},
   	ytick={0,0.1,0.2,0.3,0.4,0.5,0.6,0.7,0.8,0.9,1.0}]
	\addlegendimage{legend image with text=$\rho_{\mathrm{syn}}$}
    \addlegendentry{}
    \addlegendimage{legend image with text=$\rho_{\mathrm{rd}}$}
    \addlegendentry{}   
    
    \addplot[mark=diamond*,mark repeat=70, mark size=1.5pt, blue, dashed] 
   	table [x={x}, y={y}] {./chapter_05/Figs/fig5_5b/01_std10m2red2_sy.txt}; 
   	\addlegendentry{}  
	\addplot[mark=diamond*,mark repeat=70, mark size=1.5pt, blue] 
   	table [x={x}, y={y}] {./chapter_05/Figs/fig5_5b/02_std10m2red2_rd.txt};
   	\addlegendentry{$\gamma = 2$}
   	
   	\addplot[mark=otimes*, mark repeat=70, mark size=1.5pt, red, dashed]
	table [x={x}, y={y}] {./chapter_05/Figs/fig5_5b/03_std10m2red3_sy.txt};
   	\addlegendentry{}
	\addplot[mark=otimes*, mark repeat=70, mark size=1.5pt, red]
	table [x={x}, y={y}] {./chapter_05/Figs/fig5_5b/04_std10m2red3_rd.txt};   	
   	\addlegendentry{$\gamma = 3$}
   	
	\addplot[mark=square*,mark repeat=59, mark size=1.4pt, orange, dashed]
   	table [x={x}, y={y}] {./chapter_05/Figs/fig5_5b/05_std10m2red4_sy.txt};
   	\addlegendentry{}
   	\addplot[mark=square*,mark repeat=59, mark size=1.4pt, orange]
   	table [x={x}, y={y}] {./chapter_05/Figs/fig5_5b/06_std10m2red4_rd.txt};
   	\addlegendentry{$\gamma = 4$}
   	
    \addplot[mark=triangle*,mark repeat=65, mark size=1.5pt, black, dashed]
   	table [x={x}, y={y}] {./chapter_05/Figs/fig5_5b/07_std10m2red5_sy.txt};
   	\addlegendentry{}
   	\addplot[mark=triangle*,mark repeat=65, mark size=1.5pt, black]
   	table [x={x}, y={y}] {./chapter_05/Figs/fig5_5b/08_std10m2red5_rd.txt};
   	\addlegendentry{$\gamma = 5$}   	
  	\end{axis}
	\end{tikzpicture}}
	\end{tabular}
	\caption{The maximum spectral radii $\rho_{\mathrm{syn}}$ with synchronous and $\rho_{\mathrm{rd}}$ with randomized damping 
	scheduling over outer iterations $\nu = \{0,1,2,\dots,12\}$ for legacy redundancy $\gamma$ $\in$ $\{2,3,4,5\}$ and variance $v = 10^{-4}$ for IEEE 30-bus test case using ``flat start" (subfigure a) and ``warm start" (subfigure b).}
	\label{plot1}
	\end{figure} 

In the following, we compare the accuracy of the GN-BP algorithm to that of the Gauss-Newton method. We use the weighted residual sum of squares (WRSS) as a metric:
	\begin{equation}
	\begin{gathered}
	\mathrm{WRSS} = \sum_{i=1}^k 
	\cfrac{[z_i-h_i({\mathbf x})]^2}{v_i}.
	\end{gathered}
	\label{GN_num_WRSS}
	\end{equation}
Finally, we normalize the obtained $\mathrm{WRSS}_{\mbox{\scriptsize BP}}^{(\nu)}$ over outer iterations $\nu$ by $\mathrm{WRSS}_{\mbox{\scriptsize WLS}}$ of the centralized SE obtained using the Gauss-Newton method after 12 iterations (which we adopt as a normalization constant). 
	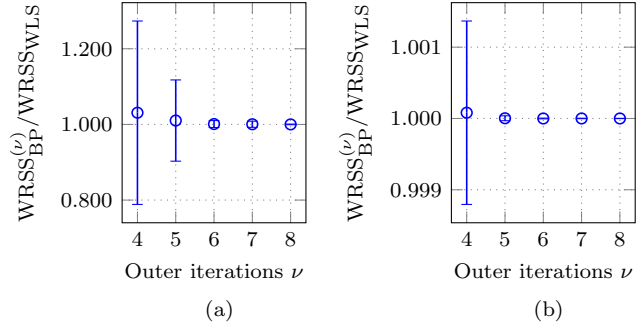
\begin{figure}[ht]
	\centering
	\captionsetup[subfigure]{oneside,margin={1.8cm,0cm}}
	\begin{tabular}{@{}c@{}}
	\subfloat[]{\label{plot2a}
	\centering	
	\begin{tikzpicture}
  	\begin{axis}[width=4cm, height=4.5cm,
   	x tick label style={/pgf/number format/.cd,
   	set thousands separator={},fixed},
   	y tick label style={/pgf/number format/.cd,fixed,
   	fixed zerofill, precision=3, /tikz/.cd},
   	xlabel={Outer iterations  $\nu$},
   	ylabel={$\mathrm{WRSS}_{\mbox{\scriptsize BP}}^{(\nu)}/
   	\mathrm{WRSS}_{\mbox{\scriptsize WLS}}$},
   	label style={font=\footnotesize},
   	grid=major,
   	xtick={1,2,3,4,5},
   	xticklabels={$4$, $5$, $6$, $7$, $8$},
   	tick label style={font=\footnotesize}]
	\addplot [color=blue, only marks, mark=o, line width=0.6pt]
 	plot [error bars, y dir = both, y explicit, 
 	error bar style={line width=0.6pt}]
 	table[x =x, y =y, y error =e]{./chapter_05/Figs/fig5_6/errorbar_red4.txt};
  	\end{axis}
	\end{tikzpicture}}
	\end{tabular}
	\begin{tabular}{@{}c@{}}
	\subfloat[]{\label{plot2b}
	\begin{tikzpicture}
  	\begin{axis}[width=4cm, height=4.5cm,
   	x tick label style={/pgf/number format/.cd,
   	set thousands separator={},fixed},
   	y tick label style={/pgf/number format/.cd,fixed,
   	fixed zerofill, precision=3, /tikz/.cd},
   	xlabel={Outer iterations  $\nu$},
   	ylabel={$\mathrm{WRSS}_{\mbox{\scriptsize BP}}^{(\nu)}/
   	\mathrm{WRSS}_{\mbox{\scriptsize WLS}}$},
   	label style={font=\footnotesize},
   	grid=major,
   	ytick={1,1.001,0.999},
   	xtick={1,2,3,4,5},
   	xticklabels={$4$, $5$, $6$, $7$, $8$},
   	tick label style={font=\footnotesize}]
	\addplot [color=blue, only marks, mark=o, line width=0.6pt]
 	plot [error bars, y dir = both, y explicit, 
 	error bar style={line width=0.6pt}]
 	table[x =x, y =y, y error =e]{./chapter_05/Figs/fig5_6/errorbar_red5.txt};
  	\end{axis}
	\end{tikzpicture}}
	\end{tabular}
	\caption{The GN-BP normalized WRSS (i.e., 
	$\mathrm{WRSS}_{\mbox{\tiny BP}}^{(\nu)}/
   	\mathrm{WRSS}_{\mbox{\tiny WLS}}$) for IEEE 30-bus test case 
   	using ``flat start" and legacy redundancy $\gamma = 4$ (subfigure a) and
   	$\gamma = 5$ (subfigure b).}
	\label{plot2}
	\end{figure} 


\textbf{Scalability and Complexity:}  Next, we use the mean absolute difference (MAD) between the state variables in two consecutive iterations as a metric:
	\begin{equation}
	\begin{gathered}
	\mathrm{MAD} =  \cfrac{1}{n} \sum_{i=1}^n 
	|\Delta x_i|.
	\end{gathered}
	\label{num_MAD}
	\end{equation}	
The MAD value represents average component-wise shift of the state estimate over the iterations, thus it may be used to quantify the rate of convergence. 

To investigate the rate of convergence as the size of the system increases, we provide MAD values for IEEE 118-bus and 300-bus test case using the ``warm start" and legacy redundancy $\gamma = 4$ with $20$ and $50$ PMUs, respectively. In the following, in order to reduce the number of inner iterations, we define an alternative inner iteration scheme. Namely, as before, we are running algorithm up to $\tau_{\max}(\nu)$, but here we allow interruption of the inner iteration loops when accuracy-based criterion is met. More precisely, the algorithm in the inner iteration loop is running until the following criterion is reached:    
		\begin{equation}
        \begin{gathered}      
		|\mathbf{r}_{f \to \Delta x}^{(\nu,\tau)}-
		\mathbf{r}_{f \to \Delta x}^{(\nu, \tau-1})| < \epsilon(\nu)
		\;\;\mathrm{or}\;\;
		\tau(\nu) = \tau_{\max}(\nu),
		\end{gathered}
		\label{GN_num_break2}
		\end{equation}
where $\mathbf{r}_{f \to \Delta x}$ represents the vector of mean-value messages from factor nodes to variable nodes, $\epsilon(\nu) = [10^{-2},$ $10^{-4},$ $10^{-6},$ $10^{-8},$ $10^{-10}]$ is the threshold at iteration $\nu$. The upper limit on inner iterations is $\tau_{\max}(\nu)=6000$ for each outer iteration $\nu$, where $\nu_{\max} = 4$.
	\begin{figure}[ht]
	\centering
	\captionsetup[subfigure]{oneside,margin={1.6cm,0cm}}
	\begin{tabular}{@{}c@{}}
	\subfloat[]{\label{GN_plot3a}
	\begin{tikzpicture}
	\begin{semilogyaxis} [box plot width=1.0mm,
	xlabel={Outer iterations  $\nu$},
   	ylabel={$\mathrm{MAD}$},
   	grid=major,   		
   	xmin=0, xmax=12, ymin=0.00000000001, ymax =0.05,	
   	xtick={1,2,3,4,5,7,8,9,10,11},
   	xticklabels={0, 1, 2, 3, 4, 0, 1, 2, 3, 4},
   	ytick={0, 0.01, 0.0001, 0.000001, 0.00000001, 0.0000000001},
   	width=9cm,height=6cm,
   	tick label style={font=\footnotesize}, label style={font=\footnotesize},
   	legend style={draw=black,fill=white,legend cell align=left,font=\tiny,
   	legend pos=south west}]
	\boxplot [
    forget plot, fill=blue!30,
    box plot whisker bottom index=1,
    box plot whisker top index=5,
    box plot box bottom index=2,
    box plot box top index=4,
    box plot median index=3] {./chapter_05/Figs/fig5_7/ieee118_mad_bp_data.txt};   
    \draw[solid, fill=blue!30, draw=black]
    (axis cs:0.685, 3) rectangle (axis cs:0.965, 4);
	\addlegendimage{area legend,fill=blue!30,draw=black}
	\addlegendentry{GN-BP};
	   
	\boxplot [
    forget plot, fill=cyan!30, 
    box plot whisker bottom index=1,
    box plot whisker top index=5,
    box plot box bottom index=2,
    box plot box top index=4,
    box plot median index=3] {./chapter_05/Figs/fig5_7/ieee118_mad_wls_data.txt};        
	\draw[solid, fill=red!30, draw=black] 
	(axis cs:0.685, 3) rectangle (axis cs:0.965, 4);
	\addlegendimage{area legend,fill=cyan!30,draw=black}
	\addlegendentry{Gauss-Newton}; 
	
    \addplot[only marks, mark options={draw=black, fill=blue!30},mark size=0.5pt] 
	table[x index=0, y index=1] 
	{./chapter_05/Figs/fig5_7/ieee118_mad_bp_outliers.txt};	
	\addplot[only marks, mark options={draw=black, fill=cyan!30},mark size=0.5pt] 
	table[x index=0, y index=1] 
	{./chapter_05/Figs/fig5_7/ieee118_mad_wls_outliers.txt};   
    
	\draw [thin] (60,\pgfkeysvalueof{/pgfplots/ymin}) -- 
	(60,\pgfkeysvalueof{/pgfplots/ymax});

	\end{semilogyaxis}
	\end{tikzpicture}}
	\end{tabular}\\
	\begin{tabular}{@{}c@{}}
	\subfloat[]{\label{GN_plot3b}
	\begin{tikzpicture}
	\begin{semilogyaxis} [box plot width=1.0mm,
	xlabel={Outer iterations  $\nu$},
   	ylabel={$\mathrm{MAD}$},
   	grid=major,   		
   	xmin=0, xmax=12, ymin=0.000000001, ymax =0.05,	
   	xtick={1,2,3,4,5,7,8,9,10,11},
   	xticklabels={0, 1, 2, 3, 4, 0, 1, 2, 3, 4},
   	ytick={0, 0.01, 0.0001, 0.000001, 0.00000001, 0.0000000001},
   	width=9cm,height=6cm,
   	tick label style={font=\footnotesize}, label style={font=\footnotesize},
   	legend style={draw=black,fill=white,legend cell align=left,font=\tiny,
   	legend pos=south west}]
	\boxplot [
    forget plot, fill=blue!30,
    box plot whisker bottom index=1,
    box plot whisker top index=5,
    box plot box bottom index=2,
    box plot box top index=4,
    box plot median index=3] {./chapter_05/Figs/fig5_7/ieee300_mad_bp_data.txt};   
    \draw[solid, fill=blue!30, draw=black]
    (axis cs:0.685, 3) rectangle (axis cs:0.965, 4);
	\addlegendimage{area legend,fill=blue!30,draw=black}
	\addlegendentry{GN-BP};
	   
	\boxplot [
    forget plot, fill=cyan!30, 
    box plot whisker bottom index=1,
    box plot whisker top index=5,
    box plot box bottom index=2,
    box plot box top index=4,
    box plot median index=3] {./chapter_05/Figs/fig5_7/ieee300_mad_wls_data.txt};        
	\draw[solid, fill=cyan!30, draw=black] 
	(axis cs:0.685, 3) rectangle (axis cs:0.965, 4);
	\addlegendimage{area legend,fill=cyan!30,draw=black}
	\addlegendentry{Gauss-Newton}; 
	
    \addplot[only marks, mark options={draw=black, fill=blue!30},mark size=0.5pt] 
	table[x index=0, y index=1] 
	{./chapter_05/Figs/fig5_7/ieee300_mad_bp_outliers.txt};	
	\addplot[only marks, mark options={draw=black, fill=cyan!30},mark size=0.5pt] 
	table[x index=0, y index=1] 
	{./chapter_05/Figs/fig5_7/ieee300_mad_wls_outliers.txt};   
    
	\draw [thin] (60,\pgfkeysvalueof{/pgfplots/ymin}) -- 
	(60,\pgfkeysvalueof{/pgfplots/ymax});

	\end{semilogyaxis}
	\end{tikzpicture}}
	\end{tabular}
	\caption{The MAD values of the GN-BP algorithm and Gauss-Newton method 
	for IEEE 118-bus (subfigure a) and IEEE 300-bus (subfigure b) 
	test case.}
	\label{GN_plot3}
	\end{figure}
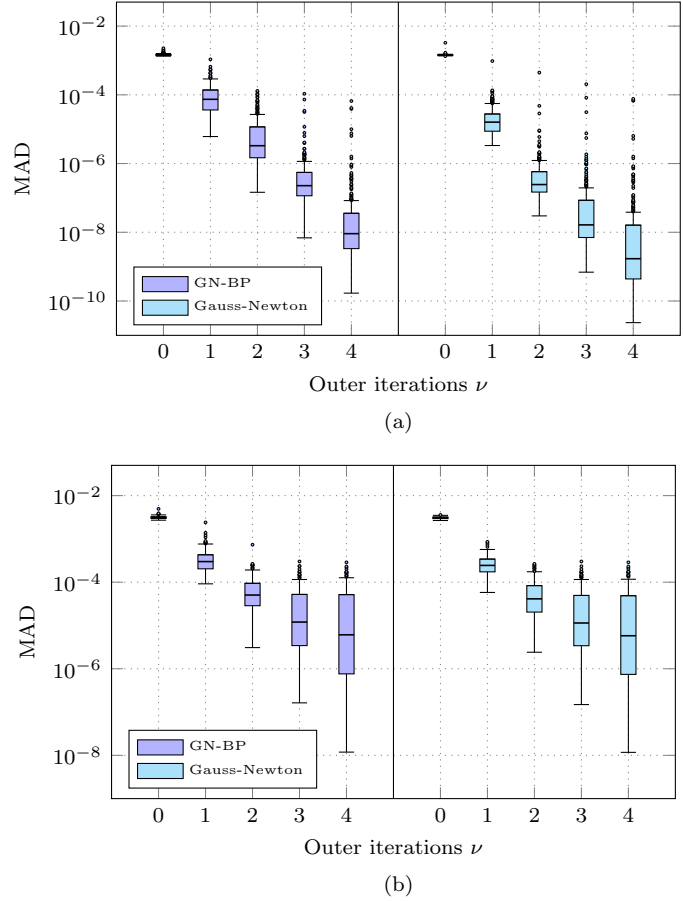

\autoref{GN_plot3} compares the MAD values of the GN-BP and Gauss-Newton method for IEEE 118-bus and 300-bus test cases within converged simulations. The GN-BP has achieved the presented performance at $\tau_{\max}(\nu) =$ $\{131,$ $488,$ $855,$ $1357,$ $2587\}$ and $\tau_{\max}(\nu) =$ $\{242,$ $1394,$ $5987,$ $6000,$ $6000\}$ (i.e., median values) for IEEE 118-bus and 300-bus test case, respectively. Note that the GN-BP exhibits very similar convergence performance to that of the centralized SE. Note also that it is difficult to directly compare the two, due to a large difference in computational loads of a single (outer) iteration. For example, the complexity of a single iteration remains constant but significant (due to matrix inversion) over iterations for the centralized SE algorithm, while it gradually increases for the GN-BP starting from an extremely low complexity at initial outer iterations. Namely, the overall complexity of the centralized SE scales as $O(n^3)$, and this can be reduced to $O(n^{2+c})$ by employing matrix inversion techniques that exploit the sparsity of involved matrices\cite{enns, alvarado}. The complexity of BP depends on the sparsity of the underlying factor graph, as the computational effort per iteration is proportional to the number of edges in the factor graph. For each of the $k$ measurements, the degree of the corresponding factor node is limited by a (typically small) constant. Indeed, for any type of measurements, the corresponding measurement function depends only on a few state variables corresponding to the buses in the local neighbourhood of the bus/branch where the measurement is taken. As $n$ and $k$ grow large, the number of edges in the factor graph scales as $O(n)$, thus the computational complexity of GN-BP scales linearly \emph{per iteration}. ased on discussion in \cite{bickson} for full matrices, the number of iterations is likely to scale with condition number of the underlying matrix, which for well-conditioned matrices may scale as low as $O(1)$. However, we leave the more detailed analysis on the scaling of the number of BP iterations as $n$ grows large for our future work.   

To summarize, BP approach builds upon the factor graph structure that directly exploits the underlying system sparsity, thus achieving minimal complexity of $\mathcal{O}(n)$ per iteration, while the scaling of the number of iterations needs further study. In contrast to the optimized centralized methods whose complexity scales as $\mathcal{O}(n^2)$, the BP method can be flexibly distributed by arbitrarily segmenting the underlying factor graph into disjoint areas. In the extreme case of the fully-distributed BP algorithm, each factor graph node operates locally and independently. Thus, the SE problem is distributed across $\mathcal{O}(n)$ nodes, and if implemented to run in parallel, can be $\mathcal{O}(n)$ times faster than the centralized solution. In addition, for fully-distributed BP, none of the nodes need to store the system-level matrices (whose storage-size typically scales as $\mathcal{O}(n^2)$), and storing only constant-size set of local parameters is sufficient.

\textbf{Bad Data Analysis:} To investigate the proposed BP-BDT, we use IEEE 14-bus and 30-bus test case, with 3 PMUs and 5 PMUs, respectively, and the set of legacy measurements of redundancy $\gamma = 3$. In each of 300 random measurement configurations, we randomly generate a bad measurement among legacy measurements, with variance set to $v_{\mathrm{b}20} = 400v_i$ or $v_{\mathrm{b}40} = 1600v_i$ (i.e., $20\sigma_i$ or $40\sigma_i$). For each simulation, we record only the largest elements ${r}_{\mbox{\scriptsize BP},f_m}$ and $r_{\mbox{\scriptsize N},m}$ obtained using BP-BDT and LNRT, respectively.   
	\begin{figure}[ht]
	\centering
	\captionsetup[subfigure]{oneside,margin={1.0cm,0cm}}
	\begin{tabular}{@{}c@{}}
	\subfloat[]{\label{plot4a}
	\begin{tikzpicture}
	\begin{axis} [box plot width=1.0mm,
	xlabel={},
   	ylabel={BP-BDT ${r}_{\mbox{\tiny BP},f_m}$},
   	grid=major,   		
    xmin=0.5, xmax=1.5, ymin=1.5, ymax =26,	
    xtick={1},
   	xticklabels={no bad data},
   	ytick={3, 10, 17, 24},
   	width=3cm,height=4.5cm,
   	tick label style={font=\footnotesize}, label style={font=\footnotesize},
   	legend style={draw=black,fill=white,legend cell align=left,font=\tiny,
   	legend pos=south west}]
	\boxplot [
    forget plot, fill=blue!30,
    box plot whisker bottom index=1,
    box plot whisker top index=5,
    box plot box bottom index=2,
    box plot box top index=4,
    box plot median index=3] {./chapter_05/Figs/fig5_8/4a_dat_ieee14_no_bad_BP.txt};  
    \addplot[only marks, mark options={draw=black, fill=blue!30},
    mark size=0.5pt] 
	table[x index=0, y index=1] 
	{./chapter_05/Figs/fig5_8/4a_out_ieee14_no_bad_BP.txt}; 
	\end{axis}
	\end{tikzpicture}}
	\end{tabular}
	\begin{tabular}{@{}c@{}}
	\subfloat[]{\label{plot4b}
	\begin{tikzpicture}
	\begin{axis} [box plot width=1.0mm,
    /pgf/number format/.cd, use comma, 1000 sep={},
	xlabel={},
   	ylabel={},
   	grid=major,   		
    xmin=0.5, xmax=1.5, ymin=-80, ymax =4000,	
    xtick={1},
   	xticklabels={$v_{\mathrm{b}20}$},
	ytick={20, 1280, 2540, 3800},
   	width=3cm,height=4.5cm,
   	tick label style={font=\footnotesize}, label style={font=\footnotesize},
   	legend style={draw=black,fill=white,legend cell align=left,font=\tiny,
   	legend pos=south west}]
	\boxplot [
    forget plot, fill=blue!30,
    box plot whisker bottom index=1,
    box plot whisker top index=5,
    box plot box bottom index=2,
    box plot box top index=4,
    box plot median index=3] {./chapter_05/Figs/fig5_8/4b_dat_ieee14_1leg_20_BP.txt}; 
    \addplot[only marks, mark options={draw=black, fill=blue!30},
    mark size=0.5pt] 
	table[x index=0, y index=1] 
	{./chapter_05/Figs/fig5_8/4b_out_ieee14_1leg_20_BP.txt};       
	\end{axis}
	\end{tikzpicture}}
	\end{tabular}
	\begin{tabular}{@{}c@{}}
	\subfloat[]{\label{plot4c}
	\begin{tikzpicture}
	\begin{axis} [box plot width=1.0mm,
    /pgf/number format/.cd, use comma, 1000 sep={},
	xlabel={},
   	ylabel={},
   	grid=major,   		
    xmin=0.5, xmax=1.5, ymin=-350, ymax =9800,	
    xtick={1},
   	xticklabels={$v_{\mathrm{b}40}$},
	ytick={80, 3180, 6280, 9380},
   	width=3cm,height=4.5cm,
   	tick label style={font=\footnotesize}, label style={font=\footnotesize},
   	legend style={draw=black,fill=white,legend cell align=left,font=\tiny,
   	legend pos=south west}]
	\boxplot [
    forget plot, fill=blue!30,
    box plot whisker bottom index=1,
    box plot whisker top index=5,
    box plot box bottom index=2,
    box plot box top index=4,
    box plot median index=3] 
    {./chapter_05/Figs/fig5_8/4c_dat_ieee14_1leg_40_BP.txt};
   \addplot[only marks, mark options={draw=black, fill=blue!30},
    mark size=0.5pt] 
	table[x index=0, y index=1] 
	{./chapter_05/Figs/fig5_8/4c_out_ieee14_1leg_40_BP.txt};    
	\end{axis}
	\end{tikzpicture}}
	\end{tabular}\\
	\captionsetup[subfigure]{oneside,margin={1.25cm,0cm}}
	\begin{tabular}{@{\hspace{-0.15cm}}c@{}}
	\subfloat[]{\label{plot4d}
	\centering
	\begin{tikzpicture}
	\begin{axis} [box plot width=1.0mm,
	xlabel={},
   	ylabel={LNRT ${r}_{\mbox{\tiny N},m}$},
   	grid=major,   		
    xmin=0.5, xmax=1.5, ymin=-10, ymax =200,	
    xtick={1},
   	xticklabels={no bad data},
   	ytick={1,64, 127,190},
   	width=3cm,height=4.5cm,
   	tick label style={font=\footnotesize}, label style={font=\footnotesize},
   	legend style={draw=black,fill=white,legend cell align=left,font=\tiny,
   	legend pos=south west}]
	\boxplot [
    forget plot, fill=cyan!30,
    box plot whisker bottom index=1,
    box plot whisker top index=5,
    box plot box bottom index=2,
    box plot box top index=4,
    box plot median index=3] 
    {./chapter_05/Figs/fig5_8/4d_dat_ieee14_no_bad_WLS.txt};   
    \addplot[only marks, mark options={draw=black, fill=cyan!30},
    mark size=0.5pt] 
	table[x index=0, y index=1] 
	{./chapter_05/Figs/fig5_8/4d_out_ieee14_no_bad_WLS.txt}; 
	\end{axis}		
	\end{tikzpicture}}
	\end{tabular}
	\captionsetup[subfigure]{oneside,margin={0.7cm,0cm}}	
	\begin{tabular}{@{\hspace{0.31cm}}c@{}}
	\subfloat[]{\label{plot4e}
	\begin{tikzpicture}
	\begin{axis} [box plot width=1.0mm,
    /pgf/number format/.cd, use comma, 1000 sep={},
	xlabel={},
   	ylabel={},
   	grid=major,   		
    xmin=0.5, xmax=1.5, ymin=1, ymax =65,	
    xtick={1},
   	xticklabels={$v_{\mathrm{b}20}$},
	ytick={4,23,42,61},
   	width=3cm,height=4.5cm,
   	tick label style={font=\footnotesize}, label style={font=\footnotesize},
   	legend style={draw=black,fill=white,legend cell align=left,font=\tiny,
   	legend pos=south west}]
	\boxplot [
    forget plot, fill=cyan!30,
    box plot whisker bottom index=1,
    box plot whisker top index=5,
    box plot box bottom index=2,
    box plot box top index=4,
    box plot median index=3] 
    {./chapter_05/Figs/fig5_8/4e_dat_ieee14_1leg_20_WLS.txt};
    \addplot[only marks, mark options={draw=black, fill=cyan!30},
    mark size=0.5pt] 
	table[x index=0, y index=1] 
	{./chapter_05/Figs/fig5_8/4e_out_ieee14_1leg_20_WLS.txt}; 
	\end{axis}       
	\end{tikzpicture}}
	\end{tabular}
	\captionsetup[subfigure]{oneside,margin={0.9cm,0cm}}		
	\begin{tabular}{@{\hspace{0.14cm}}c@{}}
	\subfloat[]{\label{plot4f}
	\begin{tikzpicture}
	\begin{axis} [box plot width=1.0mm,
    /pgf/number format/.cd, use comma, 1000 sep={},
	xlabel={},
   	ylabel={},
   	grid=major,   		
    xmin=0.5, xmax=1.5, ymin=0, ymax =130,	
    xtick={1},
   	xticklabels={$v_{\mathrm{b}40}$},
	ytick={5,45,85,	125},
   	width=3cm,height=4.5cm,
   	tick label style={font=\footnotesize}, label style={font=\footnotesize},
   	legend style={draw=black,fill=white,legend cell align=left,font=\tiny,
   	legend pos=south west}]
	\boxplot [
    forget plot, fill=cyan!30,
    box plot whisker bottom index=1,
    box plot whisker top index=5,
    box plot box bottom index=2,
    box plot box top index=4,
    box plot median index=3] 
    {./chapter_05/Figs/fig5_8/4f_dat_ieee14_1leg_40_WLS.txt};
    \addplot[only marks, mark options={draw=black, fill=cyan!30},
    mark size=0.5pt] 
	table[x index=0, y index=1] 
	{./chapter_05/Figs/fig5_8/4f_out_ieee14_1leg_40_WLS.txt};        
	\end{axis}
	\end{tikzpicture}}
	\end{tabular}
	\caption{Comparisons between BP-BDT and LNRT for bad data 
	free measurement set (subfigure a and d), a single bad data 
	in the measurement set with
	variance $v_{\mathrm{b}20}$ (subfigure b and e) 
	and $v_{\mathrm{b}40}$ (subfigure c and f)
	for IEEE 14-bus test case using ``warm start".}
	\label{plot4}
	\end{figure}
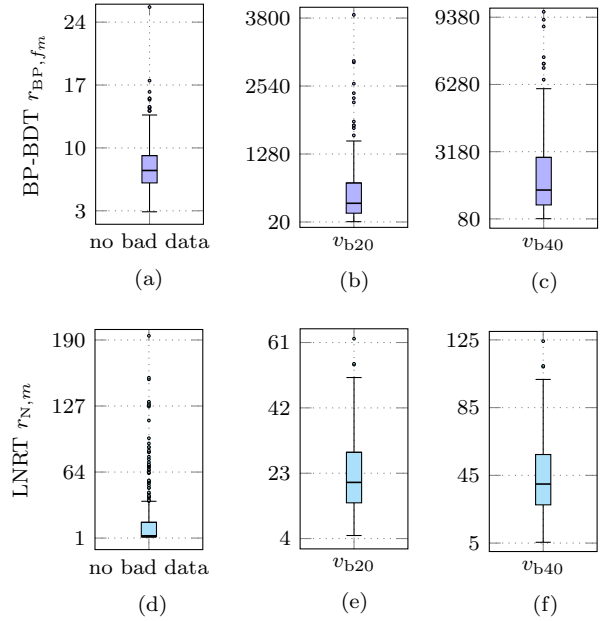

	\begin{figure}[ht]
	\centering
	\captionsetup[subfigure]{oneside,margin={1.1cm,0cm}}
	\begin{tabular}{@{}c@{}}
	\subfloat[]{\label{plot5a}
	\begin{tikzpicture}
	\begin{axis} [box plot width=1.0mm,
	xlabel={},
   	ylabel={BP-BDT ${r}_{\mbox{\tiny BP},f_m}$},
   	grid=major,   		
    xmin=0.5, xmax=1.5, ymin=1.5, ymax =80,	
    xtick={1},
   	xticklabels={no bad data},
   	ytick={3, 27, 51, 75},
   	width=3cm,height=4.5cm,
   	tick label style={font=\footnotesize}, label style={font=\footnotesize},
   	legend style={draw=black,fill=white,legend cell align=left,font=\tiny,
   	legend pos=south west}]
	\boxplot [
    forget plot, fill=blue!30,
    box plot whisker bottom index=1,
    box plot whisker top index=5,
    box plot box bottom index=2,
    box plot box top index=4,
    box plot median index=3] {./chapter_05/Figs/fig5_9/5a_dat_ieee30_no_bad_BP.txt};  
    \addplot[only marks, mark options={draw=black, fill=blue!30},
    mark size=0.5pt] 
	table[x index=0, y index=1]
	{./chapter_05/Figs/fig5_9/5a_out_ieee30_no_bad_BP.txt}; 
	\end{axis}
	\end{tikzpicture}}
	\end{tabular}
	\begin{tabular}{@{}c@{}}
	\subfloat[]{\label{plot5b}
	\begin{tikzpicture}
	\begin{axis} [box plot width=1.0mm,
    /pgf/number format/.cd, use comma, 1000 sep={},
	xlabel={},
   	ylabel={},
   	grid=major,   		
    xmin=0.5, xmax=1.5, ymin=-80, ymax =5100,	
    xtick={1},
   	xticklabels={$v_{\mathrm{b}20}$},
	ytick={20, 1620, 3220, 4820},
   	width=3cm,height=4.5cm,
   	tick label style={font=\footnotesize}, label style={font=\footnotesize},
   	legend style={draw=black,fill=white,legend cell align=left,font=\tiny,
   	legend pos=south west}]
	\boxplot [
    forget plot, fill=blue!30,
    box plot whisker bottom index=1,
    box plot whisker top index=5,
    box plot box bottom index=2,
    box plot box top index=4,
    box plot median index=3] 
    {./chapter_05/Figs/fig5_9/5b_dat_ieee30_2leg_20_BP.txt}; 
    \addplot[only marks, mark options={draw=black, fill=blue!30},
    mark size=0.5pt] 
	table[x index=0, y index=1] 
	{./chapter_05/Figs/fig5_9/5b_out_ieee30_2leg_20_BP.txt};       
	\end{axis}
	\end{tikzpicture}}
	\end{tabular}
	\begin{tabular}{@{}c@{}}
	\subfloat[]{\label{plot5c}
	\begin{tikzpicture}
	\begin{axis} [box plot width=1.0mm,
    /pgf/number format/.cd, use comma, 1000 sep={},
	xlabel={},
   	ylabel={},
   	grid=major,   		
    xmin=0.5, xmax=1.5, ymin=-350, ymax =9800,	
    xtick={1},
   	xticklabels={$v_{\mathrm{b}40}$},
	ytick={80, 3180, 6280, 9380},
   	width=3cm,height=4.5cm,
   	tick label style={font=\footnotesize}, label style={font=\footnotesize},
   	legend style={draw=black,fill=white,legend cell align=left,font=\tiny,
   	legend pos=south west}]
	\boxplot [
    forget plot, fill=blue!30,
    box plot whisker bottom index=1,
    box plot whisker top index=5,
    box plot box bottom index=2,
    box plot box top index=4,
    box plot median index=3] 
    {./chapter_05/Figs/fig5_9/5c_dat_ieee30_2leg_40_BP.txt};
   \addplot[only marks, mark options={draw=black, fill=blue!30},
    mark size=0.5pt] 
	table[x index=0, y index=1] 
	{./chapter_05/Figs/fig5_9/5c_out_ieee30_2leg_40_BP.txt};    
	\end{axis}
	\end{tikzpicture}}
	\end{tabular}
	\caption{The BP-BDT performances for IEEE 30 bus test case 
	using ``flat start" for bad data 
	free measurement set (subfigure a), two bad data 
	in the measurement set with
	with variances $v_{\mathrm{b}20}$ (subfigure b) 
	and $v_{\mathrm{b}40}$ (subfigure c).}
	\label{plot5}
	\end{figure}
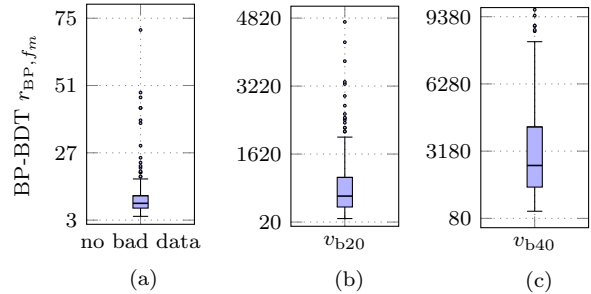

Fig. \ref{plot4} compares the BP-BDT to the LNRT for IEEE 14-bus test case using \emph{``warm start"}. The BP-BDT successfully identified the bad measurement in 291 and 294 cases, while LNRT succeeded in 220 and 240 cases, for $v_{\mathrm{b}20}$ and $v_{\mathrm{b}40}$, respectively. Figs. \ref{plot4b}, \ref{plot4c}, \ref{plot4e} and \ref{plot4f} show observed distributions of BP-BDT and LNRT metrics (${r}_{\mbox{\scriptsize BP},f_m}$ and $r_{\mbox{\scriptsize N},m}$) when tests succeeded in identifying the bad measurement. Clearly, the metric resolution between the cases without bad data (Figs. \ref{plot4a} and \ref{plot4d}) and the cases when the bad data exists in the measurement set, allows easier identification of bad data with the BP-BDT, providing for easier adjustment of the bad data identification threshold $\kappa$, in contrast to the LNRT. 
  
The BP-BDT reconfirmed the improved bad data detection for the case where two bad measurements exist in the measurement set (both with variance $v_{\mathrm{b}20}$ or $v_{\mathrm{b}40}$) for IEEE 30-bus test case initialized via \emph{``flat start"}. The BP-BDT successfully identified one of the two bad data samples after the first cycle (i.e., in the presence of another bad measurement) in 267 and 275 cases, while the LNRT identified the first bad data sample in 222 and 251 cases.	

\section{Summary}
In this chapter, we presented a novel GN-BP algorithm, which is an efficient and accurate BP-based implementation of the iterative Gauss-Newton method. The GN-BP can be highly parallelized and flexibly distributed in the context of multi-area SE.  The GN-BP is the first BP-based solution for the non-linear SE model achieving exactly the same accuracy as the centralized SE via Gauss-Newton method.
 
\chapter{Conclusions}	

In this thesis, we presented an in-depth study of the application of the BP algorithm to the SE problem in power systems. We provided detailed derivation, convergence and performance analysis of BP-based SE algorithms for both DC and non-linear model. The main contribution of our study is the GN-BP algorithm, which is shown to represent a BP-based implementation of the iterative Gauss-Newton method. GN-BP can be highly parallelized and flexibly distributed in the context of multi-area SE. In our ongoing work, we are investigating GN-BP in asynchronous, dynamic and real-time SE with online bad data detection, supported by future 5G communication infrastructure \cite{commag}. 

In the forthcoming years, 5G technology will provide ideal arena for the development of future distributed smart grid services. These services will rely on massive and reliable acquisition of timely information from the system, in combination with large-scale computing and storage capabilities, providing highly responsive, robust and scalable monitoring and control solution for future smart grids, and the proposed BP algorithms have a promising properties in such a
5G communications scenario. 

In addition, we presented the fast real-time DC SE model based on the powerful BP algorithm, which is able to provide state estimates without resorting to observability analysis. The proposed BP estimator can be distributed and parallelized which allows for flexible and low-delay centralized or distributed implementation suitable for integration in emerging WAMS. For the future work, we plan to provide extensive numerical analysis of the proposed algorithm, including the AC SE model implemented within the same framework, and extended to the generalized SE model. 

\begin{appendices}
\addtocontents{toc}{\protect\setcounter{tocdepth}{1}}
\makeatletter
\addtocontents{toc}{%
  \begingroup
  \let\protect\l@chapter\protect\l@section
  \let\protect\l@section\protect\l@subsection
}
\makeatother

\chapter{The SE in Power System: Toy Example}  \label{app:A}
\addcontentsline{lof}{chapter}{A The SE in Power System: Toy Example}
\addcontentsline{lot}{chapter}{A The SE in Power System: Toy Example}
An illustrative example presented in \autoref{appA_fig_1} will be used to
provide a step-by-step presentation of the centralized SE algorithm. The power system consists of 3 buses and 3 branches, where we observe two legacy measurements, active power flow $M_{P_{12}}$ and active power injection $M_{P_{3}}$, while bus $2$ contains one PMU that provides line current $M_{\ph I_{21}}$ and $M_{\ph I_{23}}$, and bus voltage $M_{\ph V_{2}}$ phasor measurements. Note, bus 1 is the slack, where the voltage angle has a given value.  
	\begin{figure}[ht]
	\centering
	\includegraphics[width=32mm]{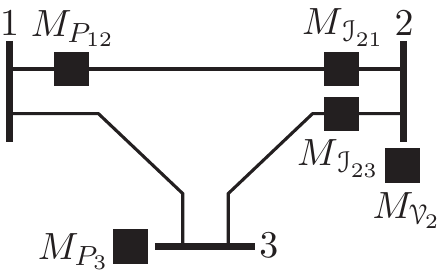}
	\caption{The 3-bus power system with given measurement configuration.}
	\label{appA_fig_1}
	\end{figure}	

We presented simultaneous SE algorithm, where state variables and phasor measurements are given in polar coordinate system. Each measurement $M_i \in \mathcal{M}$ is associated with measurement value $z_i$, variance $v_i$ and function $h_i(\mathbf{x})$ as shown in \autoref{appA_tab_1}.  system.  
	\begin{table}[H]
	\footnotesize
  	\begin{center}	
	\begin{tabular}{cccc}
	\hline															
	\rule{0pt}{2.5ex}
	Measurement & Measurement Value & Measurement Variance & Measurement Function  	\\
	$M_i$     &   $z_i$    &   $v_i$  &  							\\[1.3pt]
	\hline \rule{0pt}{3ex}%
	$M_{P_{12}}$ 		& $z_{P_{12}}$		& $v_{P_{12}}$ & $h_{P_{12}}(\cdot)$ 
	\rule{0pt}{2.2ex}\\
    $M_{P_{3}}$ 		& $z_{P_{3}}$		& $v_{P_{3}}$ & $h_{P_{3}}(\cdot)$  
    \rule{0pt}{2.2ex}\\
	$M_{V_{2}}$ 		& $z_{V_{2}}$		& $v_{V_{2}}$ & $h_{V_{2}}(\cdot)$
	 \rule{0pt}{2.2ex}\\
	$M_{\theta_{2}}$ & $z_{\theta_{2}}$	& $v_{\theta_{2}}$ & 
	$h_{\theta_{2}}(\cdot)$
	 \rule{0pt}{2.2ex}\\	
	$M_{I_{21}}$ 		& $z_{I_{21}}$		& $v_{I_{21}}$ & $h_{I_{21}}(\cdot)$
	 \rule{0pt}{2.2ex}\\
	$M_{I_{23}}$ 		& $z_{I_{23}}$		& $v_{I_{23}}$ & $h_{I_{23}}(\cdot)$
	 \rule{0pt}{2.2ex}\\
	$M_{\phi_{21}}$ & $z_{\phi_{21}}$	& $v_{\phi_{21}}$ & $h_{\phi_{21}}(\cdot)$ 
	 \rule{0pt}{2.2ex}\\	 
	$M_{\phi_{23}}$ & $z_{\phi_{23}}$	& $v_{\phi_{23}}$ & $h_{\phi_{23}}(\cdot)$
	\rule{0pt}{2.2ex}\\[2pt]
	\hline
	\end{tabular}
	\end{center}
 	\caption{Measurement data.}
 	\label{appA_tab_1}
	\end{table} 

\subsection*{Vector of Measurement Values and Covariance Matrix}
The vector of measurement values $\mathbf z \in \mathbb {R}^{N_{\mathrm{le}} + 2N_{\mathrm{ph}}}$ and the covariance matrix  $\mathbf R \in \mathbb {R}^{{(N_{\mathrm{le}} + 2N_{\mathrm{ph}})} \times {(N_{\mathrm{le}} + 2N_{\mathrm{ph}}})}$ according to the measurement configuration are:
	\begin{subequations}
   	\begin{align*}
   	\mathbf z &=
    \begin{bmatrix}    	 
	z_{P_{12}}\;
	z_{P_{3}}\;
	z_{V_{2}}\;
	z_{\theta_{2}}\;
	z_{I_{21}}\;
	z_{I_{23}}\;
	z_{\phi_{21}}\;
	z_{\phi_{23}}
	\end{bmatrix}^{\mathrm{T}}\\
	\mathbf R &= \mathrm{diag}
    (   	 
	v_{P_{12}},\; 	
	v_{P_{3}},\; 
	v_{V_{2}},\;
	v_{\theta_{2}},\;
	v_{I_{21}},\;
	v_{I_{23}},\;
	v_{\phi_{21}},\;
	v_{\phi_{23}}).
    \end{align*}
	\end{subequations}
Note that, due to uncorrelated measurement errors the covariance matrix $\mathbf{R}$ has the diagonal structure.

\subsection*{Measurement Functions}
The vector of measurement functions $\mathbf h(\mathbf x) \in \mathbb {R}^{N_{\mathrm{le}} + 2N_{\mathrm{ph}}}$ is:
	\begin{subequations}
   	\begin{align*}
   	\mathbf z &=
    \begin{bmatrix}    	 
	h_{P_{12}}(\cdot)\;
	h_{P_{3}}(\cdot)\;
	h_{V_{2}}(\cdot)\;
	h_{\theta_{2}}(\cdot)\;	
	h_{I_{21}}(\cdot)\;
	h_{I_{23}}(\cdot)\;
	h_{\phi_{21}}(\cdot)\;
	h_{\phi_{23}}(\cdot)
	\end{bmatrix}^{\mathrm{T}}.
    \end{align*}
	\end{subequations}

The measurement functions associated with legacy measurements $M_{P_{12}}$ and $M_{P_3}$ are:
	\begin{subequations}
   	\begin{align*}
	h_{P_{12}}(\cdot)&=
    {V}_{1}^2(g_{12}+g_{s1})-{V}_{1}{V}_{2}(g_{12}\cos\theta_{12}
    +b_{12}\sin\theta_{12})\\[4pt]
	h_{P_{3}}(\cdot) &={V}_{3}^2G_{33} +
	{V}_{3} \big[
	{V}_{1}(G_{31}\cos\theta_{31}+B_{31}\sin\theta_{31})+
	{V}_{2}(G_{32}\cos\theta_{32}+B_{32}\sin\theta_{32})\big].   
	\end{align*}
	\end{subequations}

The measurement functions associated with the bus phasor measurement $\mathcal{M}_{\ph{V}_{2}}= $ $\{M_{{V}_{2}},$ $M_{{\theta}_{2}}\}$ are:
	\begin{subequations}
   	\begin{align*}
	h_{{V}_{2}}(\cdot) = V_2; \;\;\;\;
	h_{\theta_2}(\cdot) = \theta_2.
	\end{align*}
	\end{subequations}

The measurement functions associated with line current phasor measurements $\mathcal{M}_{\ph{I}_{21}}= $ $\{M_{{I}_{21}},$ $M_{{\phi}_{21}}\}$ and $\mathcal{M}_{\ph{I}_{23}}= $ $\{M_{{I}_{23}},$ $M_{{\phi}_{23}}\}$ are as follows:
	\begin{subequations}
   	\begin{align*}
	h_{{I}_{21}}(\cdot) &= 
    [A_\mathrm{21c} V_2^2 + B_\mathrm{21c} V_1^2 - 2 V_2V_1
    (C_\mathrm{21c} \cos \theta_{21}-D_\mathrm{21c} \sin \theta_{21})]^{1/2}\\[4pt]
    h_{{I}_{23}}(\cdot) &= 
    [A_\mathrm{23c} V_2^2 + B_\mathrm{23c} V_3^2 - 2 V_2V_3
    (C_\mathrm{23c} \cos \theta_{23}-D_\mathrm{23c} \sin \theta_{23})]^{1/2}\\[4pt]
	h_{{\phi}_{21}}(\cdot) &=\mathrm{arctan}\Bigg[ 
	\cfrac{(A_\mathrm{21a} \sin\theta_2
    +B_\mathrm{21a} \cos\theta_2)V_2 
    - (C_\mathrm{21a} \sin\theta_1 + D_\mathrm{21a}\cos\theta_1)V_1}
   	{(A_\mathrm{21a} \cos\theta_2
    -B_\mathrm{21a} \sin\theta_2)V_2 
    - (C_\mathrm{21a} \cos\theta_1 - D_\mathrm{21a} \sin\theta_1)V_1} \Bigg]\\[4pt]
     h_{{\phi}_{23}}(\cdot) &=\mathrm{arctan}\Bigg[ 
	\cfrac{(A_\mathrm{23a} \sin\theta_2
    +B_\mathrm{23a} \cos\theta_2)V_2 
    - (C_\mathrm{23a} \sin\theta_3 + D_\mathrm{23a}\cos\theta_3)V_3}
   	{(A_\mathrm{23a} \cos\theta_2
    -B_\mathrm{23a} \sin\theta_2)V_2 
    - (C_\mathrm{23a} \cos\theta_3 - D_\mathrm{23a} \sin\theta_3)V_3} \Bigg],   
	\end{align*}
	\end{subequations}
where coefficients are: 		
	\begin{equation}
    \begin{aligned}
    A_\mathrm{21c}&=(g_{21}+g_{\mathrm{s}2})^2+(b_{21}+b_{\mathrm{s}2})^2;&
    B_\mathrm{21c}&=g_{21}^2+b_{21}^2\\
    C_\mathrm{21c}&=g_{21}(g_{21}+g_{\mathrm{s}2})+b_{21}(b_{21}+b_{\mathrm{s}2});&
    D_\mathrm{21c}&=g_{21}b_{\mathrm{s}2}-b_{21}g_{\mathrm{s}2}\\
    A_\mathrm{21a}&=g_{21}+g_{\mathrm{s}2};&
    B_\mathrm{21a}&=b_{21}+b_{\mathrm{s}2}\\
    C_\mathrm{21a}&=g_{21};&
    D_\mathrm{21a}&=b_{21}\\
    A_\mathrm{23c}&=(g_{23}+g_{\mathrm{s}2})^2+(b_{23}+b_{\mathrm{s}2})^2;&
    B_\mathrm{23c}&=g_{23}^2+b_{23}^2\\
    C_\mathrm{23c}&=g_{23}(g_{23}+g_{\mathrm{s}2})+b_{23}(b_{23}+b_{\mathrm{s}2});&
    D_\mathrm{23c}&=g_{23}b_{\mathrm{s}2}-b_{23}g_{\mathrm{s}2}\\
    A_\mathrm{23a}&=g_{23}+g_{\mathrm{s}2};&
    B_\mathrm{23a}&=b_{23}+b_{\mathrm{s}2}\\
    C_\mathrm{23a}&=g_{23};&
    D_\mathrm{23a}&=b_{23}.
    \end{aligned}
    \nonumber
	\end{equation}
Note that, it holds $g_{21}=g_{12}$, $b_{21}=b_{12}$.

\subsection*{Jacobian Matrix}
The Jacobian matrix $\mathbf {J}(\mathbf x) \in \mathbb {R}^{(N_{\mathrm{le}} + 2N_{\mathrm{ph}}) \times n}$ is defined: 
	\begin{equation}
   	\begin{gathered}
    \mathbf J(\mathbf x)=
    \left[
    \begin{array}{cc:ccc}
	\cfrac{\mathrm \partial{{h_{P_{12}}}(\cdot)}}
    {\mathrm \partial \theta_{2}} 
    & 0 & \cfrac{\mathrm \partial{{h_{P_{12}}}(\cdot)}}
   	{\mathrm \partial V_{1}}  
   	& \cfrac{\mathrm \partial{h_{{P_{12}}}(\cdot)}}{\mathrm \partial V_{2}} & 0 
   	\\[8pt]  
   	\cfrac{\mathrm \partial{h_{P_{3}}(\cdot)}}
    {\mathrm \partial \theta_{2}} &
    \cfrac{\mathrm \partial{h_{P_{3}}(\cdot)}}
    {\mathrm \partial \theta_{3}} &
    \cfrac{\mathrm \partial{h_{P_{3}}(\cdot)}}
    {\mathrm \partial V_{1}} &
    \cfrac{\mathrm \partial{h_{P_{3}}(\cdot)}}
    {\mathrm \partial V_{2}} &
    \cfrac{\mathrm \partial{h_{P_{3}}(\cdot)}}
    {\mathrm \partial V_{3}} 
    \\[8pt] 
	\cfrac{\mathrm \partial{{h_{\theta_2}(\cdot)}}}
   	{\mathrm \partial \theta_{2}} & 0 & 0& 0 & 0
   	\\[8pt]   
    0 & 0 & 0 &
    \cfrac{\mathrm \partial{{h_{{V}_{2}}(\cdot)}}}
   	{\mathrm \partial V_{2}} & 0
   	\\[8pt] 
   	\cfrac{\mathrm \partial{h_{I_{21}}(\cdot)}}
	{\mathrm \partial \theta_{2}} & 0 &
	\cfrac{\mathrm \partial{h_{I_{21}}(\cdot)}}
	{\mathrm \partial V_{1}} &
	\cfrac{\mathrm \partial{h_{I_{21}}(\cdot)}}
	{\mathrm \partial V_{2}} & 0
	\\[5pt]
	\cfrac{\mathrm \partial{h_{I_{23}}(\cdot)}}
	{\mathrm \partial \theta_{2}} &
	\cfrac{\mathrm \partial{h_{I_{23}}(\cdot)}}
	{\mathrm \partial \theta_{3}} & 0 &
	\cfrac{\mathrm \partial{h_{I_{23}}(\cdot)}}
	{\mathrm \partial V_{2}} &
	\cfrac{\mathrm \partial{h_{I_{23}}(\cdot)}}
	{\mathrm \partial V_{3}}
	\\[8pt] 
	\cfrac{\mathrm \partial{h_{{\phi}_{21}}(\cdot)}}
	{\mathrm \partial \theta_{2}} & 0 &
	\cfrac{\mathrm \partial{h_{{\phi}_{21}}(\cdot)}}
	{\mathrm \partial V_{1}} &
	\cfrac{\mathrm \partial{h_{{\phi}_{21}}(\cdot)}}
	{\mathrm \partial V_{2}} & 0
	\\[5pt]
	\cfrac{\mathrm \partial{h_{{\phi}_{23}}(\cdot)}}
	{\mathrm \partial \theta_{2}} & 
	\cfrac{\mathrm \partial{h_{{\phi}_{23}}(\cdot)}}
	{\mathrm \partial \theta_{3}} & 0 &
	\cfrac{\mathrm \partial{h_{{\phi}_{23}}(\cdot)}}
	{\mathrm \partial V_{2}} & 
	\cfrac{\mathrm \partial{h_{{\phi}_{23}}(\cdot)}}
	{\mathrm \partial V_{3}}
	\end{array}\right]. \nonumber
	\end{gathered}
	\end{equation}

Jacobian expressions corresponding to the active power flow measurement function $h_{P_{12}}(\cdot)$ are:
	\begin{subequations}
   	\begin{align*}
    \cfrac{\mathrm \partial{{h_{P_{12}}}(\cdot)}}
    {\mathrm \partial \theta_{2}}&=
    -{V}_{1}{V}_{2}(g_{12}\sin\theta_{12}-b_{12}\cos\theta_{12})
    \\
    \cfrac{\mathrm \partial{{h_{P_{12}}}(\cdot)}}
   	{\mathrm \partial V_{1}}&=
    -{V}_{2}(g_{12}\cos\theta_{12}+b_{12}\sin\theta_{12})+
    2V_{1}(g_{12}+g_{s1}) 
    \\
    \cfrac{\mathrm \partial{h_{{P_{12}}}(\cdot)}}{\mathrm \partial V_{2}}&=
    -{V}_{1}(g_{12}\cos\theta_{12}+b_{12}\sin\theta_{12}).   
	\end{align*}
	\end{subequations}

Jacobian expressions corresponding to the active power injection measurement function $h_{P_{3}}(\cdot)$ are:
	\begin{subequations}
   	\begin{align*}
    \cfrac{\mathrm \partial{h_{P_{3}}(\cdot)}}{\mathrm \partial \theta_{2}}&=
    {V}_{3}{V}_{2}(G_{32}\sin\theta_{32}-B_{32}\cos\theta_{32})
    \\
    \cfrac{\mathrm \partial{h_{P_{3}}(\cdot)}}
    {\mathrm \partial \theta_{3}}&=
    {V}_{3}[{V}_{1}(-G_{31}\sin\theta_{31}+B_{31}\cos\theta_{31})+
    {V}_{2}(-G_{32}\sin\theta_{32}+B_{32}\cos\theta_{32})]
    \\
    \cfrac{\mathrm \partial{h_{P_{3}}(\cdot)}}{\mathrm \partial V_{1}}&=
    {V}_{3}(G_{31}\cos\theta_{31}+B_{31}\sin\theta_{31})
 	 \\
    \cfrac{\mathrm \partial{h_{P_{3}}(\cdot)}}{\mathrm \partial V_{2}}&=
    {V}_{3}(G_{32}\cos\theta_{32}+B_{32}\sin\theta_{32})
    \\
    \cfrac{\mathrm \partial{h_{P_{3}}(\cdot)}}{\mathrm \partial V_{3}}&=
    {V}_{1}(G_{31}\cos\theta_{31}+B_{31} + {V}_{2}(G_{32}\cos\theta_{32}+B_{32}
    \sin\theta_{32})+2{V}_{3}G_{33}.  
	\end{align*}
	\end{subequations}	

Jacobian expressions corresponding to the bus phasor measurement functions $h_{V_{2}}(\cdot)$ and $h_{\theta_{2}}(\cdot)$ are as follows:		
	\begin{subequations}
   	\begin{align*}
   	\cfrac{\mathrm \partial{{h_{{V}_{2}}(\cdot)}}}
   	{\mathrm \partial V_{2}}=1; \;\;\;\;
   	\cfrac{\mathrm \partial{{h_{\theta_2}(\cdot)}}}
   	{\mathrm \partial \theta_{2}}=1.  
	\end{align*}
	\end{subequations}	

Jacobian expressions corresponding to the line current magnitude measurement function $h_{I_{21}}(\cdot)$ are:
	\begin{subequations}
   	\begin{align*}
	\cfrac{\mathrm \partial{h_{I_{21}}(\cdot)}}
	{\mathrm \partial \theta_{2}}&=
    \cfrac{V_2V_1(D_\mathrm{21c}\cos\theta_{21}+
    C_\mathrm{21c}\sin\theta_{21})}{h_{I_{21}}(\cdot)}
    \\
    \cfrac{\mathrm \partial{h_{I_{21}}(\cdot)}}{\mathrm \partial V_{1}}&=
    \cfrac{V_2(D_\mathrm{21c}\sin\theta_{21}-
    C_\mathrm{21c}\cos\theta_{21})+B_\mathrm{21c}V_1}{h_{I_{21}}(\cdot)} 
    \\    
    \cfrac{\mathrm \partial{h_{I_{21}}(\cdot)}}{\mathrm \partial V_{2}}&=
    \cfrac{V_1(D_\mathrm{21c}\sin\theta_{21}-
    C_\mathrm{21c}\cos\theta_{21})+A_\mathrm{21c}V_1}{h_{I_{21}}(\cdot)}.   
	\end{align*}
	\end{subequations}
	
Jacobian expressions corresponding to the line current magnitude measurement function $h_{I_{23}}(\cdot)$ are:  
	\begin{subequations}
   	\begin{align*}
	\cfrac{\mathrm \partial{h_{I_{23}}(\cdot)}}
	{\mathrm \partial \theta_{2}}&=
    \cfrac{V_2V_3(D_\mathrm{23c}\cos\theta_{23}+
    C_\mathrm{23c}\sin\theta_{23})}{h_{I_{23}}(\cdot)}	
    \\
    \cfrac{\mathrm \partial{h_{I_{23}}(\cdot)}}
    {\mathrm \partial \theta_{3}}&=-
    \cfrac{V_2V_3(D_\mathrm{23c}\cos\theta_{23}+
    C_\mathrm{23c}\sin\theta_{23})}{h_{I_{23}}(\cdot)}
    \\
    \cfrac{\mathrm \partial{h_{I_{23}}(\cdot)}}{\mathrm \partial V_{2}}&=
    \cfrac{V_3(D_\mathrm{23c}\sin\theta_{23}-
    C_\mathrm{23c}\cos\theta_{23})+A_\mathrm{23c}V_2}{h_{I_{23}}(\cdot)}
    \\
    \cfrac{\mathrm \partial{h_{I_{23}}(\cdot)}}{\mathrm \partial V_{3}}&=
    \cfrac{V_2(D_\mathrm{23c}\sin\theta_{23}-
    C_\mathrm{23c}\cos\theta_{23})+B_\mathrm{23c}V_3}{h_{I_{23}}(\cdot)}.   
	\end{align*}
	\end{subequations}	

Jacobian expressions corresponding to the line current angle measurement function $h_{\phi_{21}}(\cdot)$ are:
	\begin{subequations}
   	\begin{align*}
	\cfrac{\mathrm \partial{h_{{\phi}_{21}}(\cdot)}}
	{\mathrm \partial \theta_{2}}&=
	\frac{A_\mathrm{21c} V_2^2+ (D_\mathrm{21c} \sin \theta_{21}-C_\mathrm{21c} 
	\cos\theta_{21})V_2V_1}
   	{h_{{I}_{21}}(\cdot)}\\
   	\cfrac{\mathrm \partial{h_{{\phi}_{21}}(\cdot)}}
	{\mathrm \partial V_{1}}&=
	\frac{V_2 (C_\mathrm{21c} \sin\theta_{21}+D_\mathrm{21c} \cos\theta_{21})}
   	{h_{{I}_{21}}(\cdot)}	\\
   	\cfrac{\mathrm \partial{h_{{\phi}_{21}}(\cdot)}}
	{\mathrm \partial V_{2}}&=
   	-\frac{V_1 (C_\mathrm{21c} \sin\theta_{21}+D_\mathrm{21c} \cos\theta_{21})}
   	{h_{{I}_{21}}(\cdot)}	
	\end{align*}
	\end{subequations}
	
Jacobian expressions corresponding to the line current angle measurement function $h_{\phi_{23}}(\cdot)$ are:	
	\begin{subequations}
   	\begin{align*}
	\cfrac{\mathrm \partial{h_{{\phi}_{23}}(\cdot)}}
	{\mathrm \partial \theta_{2}}&=
	\frac{A_\mathrm{23c} V_2^2+ (D_\mathrm{23c} \sin \theta_{23}-C_\mathrm{23c} 
	\cos\theta_{23})V_2V_3}
   	{h_{{I}_{23}}(\cdot)}\\
   	\cfrac{\mathrm \partial{h_{{\phi}_{23}}(\cdot)}}
	{\mathrm \partial \theta_{3}}&=
	\frac{B_\mathrm{23c} V_3^2+ (D_\mathrm{23c}
   	\sin \theta_{23}-C_\mathrm{23c} \cos
   	\theta_{23})V_2V_3}{h_{{I}_{23}}(\cdot)}	\\
   	\cfrac{\mathrm \partial{h_{{\phi}_{23}}(\cdot)}}
	{\mathrm \partial V_{2}}&=
   	-\frac{V_3 (C_\mathrm{23c} \sin\theta_{23}+D_\mathrm{23c} \cos\theta_{23})}
   	{h_{{I}_{23}}(\cdot)}\\
   	\cfrac{\mathrm \partial{h_{{\phi}_{23}}(\cdot)}}
	{\mathrm \partial V_{3}}&=
	\frac{V_2 (C_\mathrm{23c} \sin\theta_{23}+D_\mathrm{23c} \cos\theta_{23})}
   	{h_{{I}_{23}}(\cdot)}.		
	\end{align*}
	\end{subequations}	
\chapter{The DC-BP Algorithm: Numerical Example} \label{app:B}
\addcontentsline{lof}{chapter}{B The DC-BP Algorithm: Numerical Example}
\addcontentsline{lot}{chapter}{B The DC-BP Algorithm: Numerical Example}
An illustrative example presented in \autoref{appB_fig_1} will be used to
provide a step-by-step presentation of the proposed DC-BP algorithm. The power system consists of 3 buses and 3 branches, where we observe 3 measurements: active power flow $M_{P_{12}}$, active power injection $M_{P_{3}}$, and bus voltage angle $M_{\theta_{2}}$. Note, bus 1 is the slack, where the voltage angle has a given value with the corresponding variance.  
	\begin{figure}[ht]
	\centering
	\includegraphics[width=30mm]{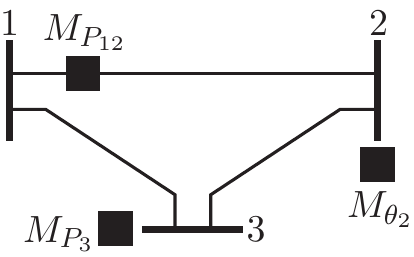}
	\caption{The 3-bus power system with given measurement configuration.}
	\label{appB_fig_1}
	\end{figure}	

\noindent
\autoref{appB_tab_1} shows the branch reactances $x_{ij}$ for the observed power system.  
	\begin{table}[ht]
	\footnotesize
  	\begin{center}	
	\begin{tabular}{cccccc}
	\hline 												\rule{0pt}{2.5ex}
	From Bus	& To Bus	& Reactance   				\\
	$i$    		&   $j$  	& $x_{ij}\;\mbox{(pu)}$  	\\[1.5pt]
	\hline \rule{0pt}{3ex}%
	1 			& 2  		& 0.040  					\rule{0pt}{2ex}\\
    1 			& 3  		& 0.020  					\rule{0pt}{2ex}\\
	2 			& 3  		& 0.025  					\rule{0pt}{2ex}\\
	\hline
	\end{tabular}
	\end{center}
 	\caption{Branch data.}
 	\label{appB_tab_1}
	\end{table} 

Each measurement $M_i \in \mathcal{M}$ is associated with measurement value $z_i$ and variance $v_i$ as shown in \autoref{appB_tab_2}. In addition, power measurements $M_{P_{12}}$ and $M_{P_{3}}$ are associated with measurement functions respectively: 
	\begin{subequations}
   	\begin{align*}
	h_{P_{12}}(\theta_1,\theta_2) &= \cfrac{\theta_1-\theta_2}{x_{12}} =
	C_{\theta_1P_{12}} \cdot \theta_1 + C_{\theta_2P_{12}} \cdot \theta_2\\
	h_{P_{3}}(\theta_1,\theta_2,\theta_3) &= 
	- \cfrac{\theta_1}{x_{13}} - 
	\cfrac{\theta_2}{x_{23}} + \cfrac{\theta_3}{x_{13}+x_{23}} =
	C_{\theta_1P_{3}} \cdot \theta_1 + C_{\theta_2P_{3}} \cdot \theta_2 +
	C_{\theta_3P_{3}} \cdot \theta_3, 
	\end{align*}
   	\label{appB_eqn_1}%
	\end{subequations}	
where coefficients are:
	\begin{equation}
    \begin{aligned}
	C_{\theta_1P_{12}} = 25\;\;\;\; 
	C_{\theta_2P_{12}} = -25\;\;\;\;
	C_{\theta_1P_{3}}  = -50\;\;\;\; 
	C_{\theta_2P_{3}}  = -40\;\;\;\;
	C_{\theta_3P_{3}}  = 90. \nonumber 
	\end{aligned} 
   	\label{appB_eqn_2}%
	\end{equation}	
	\begin{table}[ht]
	\footnotesize
  	\begin{center}	
	\begin{tabular}{cccc}
	\hline															\rule{0pt}{2.5ex}
	Measurement & Measurement Value & Measurement Variance & Unit  	\\
	$M_i$     &   $z_i$    &   $v_i$  &  							\\[1.3pt]
	\hline \rule{0pt}{3ex}%
	$M_{P_{12}}$ 		& \hspace{0.3em}1.795		& $10^{-2}$	& $\mbox{pu}$  
	\rule{0pt}{2.2ex}\\
    $M_{P_{3}}$ 		& \hspace{0.3em}1.966		& $10^{-2}$ 	& $\mbox{pu}$  
    \rule{0pt}{2.2ex}\\
	$M_{\theta_{2}}$ 	& -0.066	& $10^{-6}$	& $\mbox{rad}$ 
	\rule{0pt}{2.2ex}\\[2pt]
	\hline
	\end{tabular}
	\end{center}
 	\caption{Measurement data.}
 	\label{appB_tab_2}
	\end{table} 

\subsection*{The Factor Graph}
The first step is forming a factor graph, where set of variable nodes $\mathcal{V} = \{\theta_1, \theta_2, \theta_3\}$ is defined by state variables. The set of measurements $\mathcal{M}$ defines the set of factor nodes $\mathcal{F}$, and in addition, the set $\mathcal{F}$ is further expanded with slack and virtual factor nodes. 
	\begin{figure}[ht]
	\centering
	\includegraphics[width=65mm]{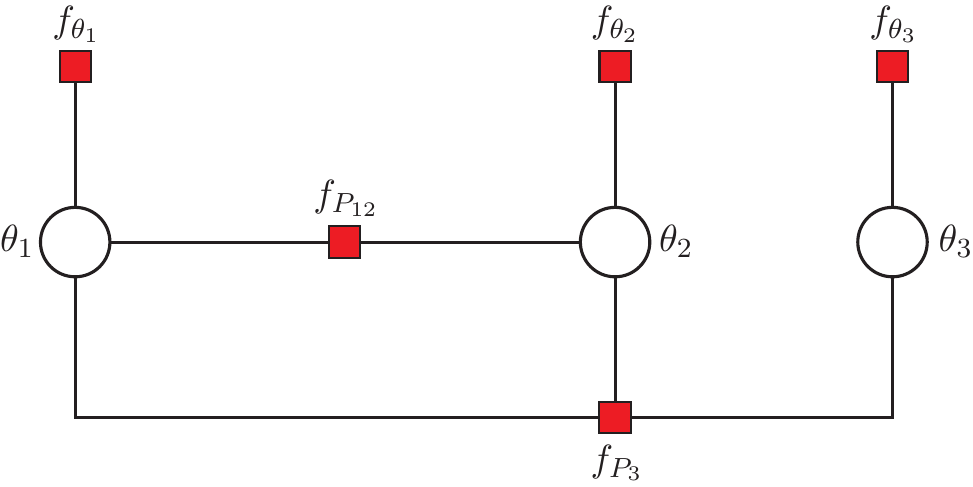}
	\caption{The factor graph.}
	\label{appB_fig_2}
	\end{figure}
More precisely, measurements $M_{P_{12}}$ and $M_{P_{3}}$ define the set of indirect factor nodes $\mathcal{F}_{\mathrm{ind}} = \{f_{P_{12}}, f_{P_{3}} \} \subset \mathcal{F}$, and measurement $M_{\theta_{2}}$ define the set of direct factor nodes $\mathcal{F}_{\mathrm{dir}} = \{f_{\theta_{2}}\} \subset \mathcal{F}$. Further, the slack bus defines the slack factor node $f_{\theta_1}$, while virtual factor node $f_{\theta_3}$ is used because variable node $\theta_3$ is not directly measured. Direct, slack and virtual factor nodes define the set of local factor nodes $\mathcal{F}_{\mathrm{loc}} \subset \mathcal{F}$. The factor graph that correspond with power system with given measurement configuration is shown in \autoref{appB_fig_1}. 

\subsection*{The DC-BP Initialization $\bm {\tau = 0}$}
\subsubsection*{Messages from local factor nodes to variable nodes}
The initialization step starts with messages from local factor nodes $\mathcal{F}_{\mathrm{loc}}$ to variable nodes $\mathcal{V}$, as shown in \autoref{appB_fig_2}. All messages are Gaussian and represent by their mean-variance pairs. 	
	\begin{figure}[ht]
	\centering
	\includegraphics[width=65mm]{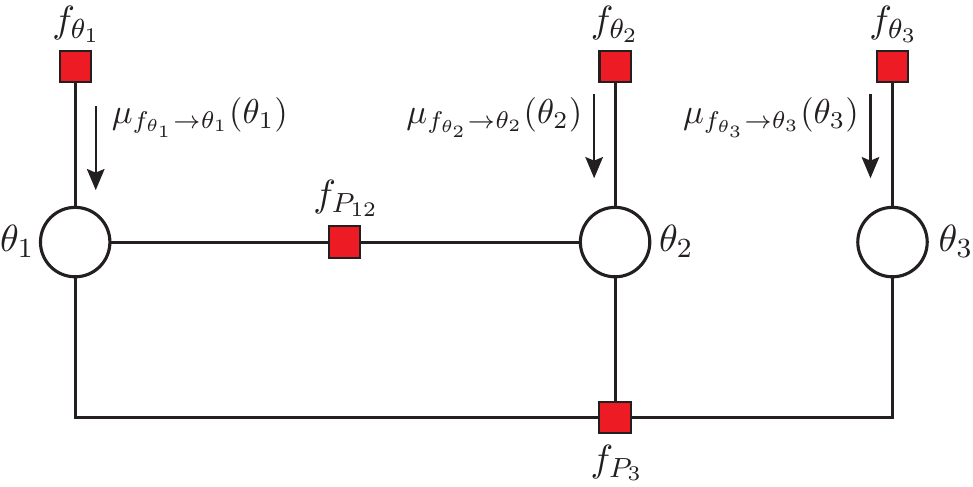}
	\caption{Messages from local factor nodes to variable nodes.}
	\label{appB_fig_3}
	\end{figure}

\noindent
According to properties of local factor nodes, messages from local factor nodes $\mathcal{F}_{\mathrm{loc}}$ to variable nodes $\mathcal{V}$ are determined:	
	\begin{subequations} 
   	\begin{align*}
	\mu_{f_{\theta_1} \to \theta_1}(\theta_1) &:= 
	(z_{f_{\theta_1} \to \theta_1}, 
	v_{f_{\theta_1} \to \theta_1}) = (0,10^{-60})\\[4pt]
	\mu_{f_{\theta_2} \to \theta_2}(\theta_2) &:= 
	(z_{f_{\theta_2} \to \theta_2}, 
	v_{f_{\theta_2} \to \theta_2}) = (-0.066,10^{-6})\\[4pt]	
	\mu_{f_{\theta_3} \to \theta_3}(\theta_3) &:= 
	(z_{f_{\theta_3} \to \theta_3}, 
	v_{f_{\theta_3} \to \theta_3}) = (0,10^{60}).
	\end{align*}
   	\label{appB_eqn_4}%
	\end{subequations}  	 
Note that we left the iteration index $\tau = 0$ as a consequence that messages from local factor nodes $\mathcal{F}_{\mathrm{loc}}$ to variable nodes $\mathcal{V}$ are constant through iterations.  

\subsubsection*{Forward incoming messages}
Then, variable nodes forward the incoming messages received from local factor nodes along remaining edges as shown in \autoref{appB_fig_3}.
	\begin{figure}[ht]
	\centering
	\includegraphics[width=65mm]{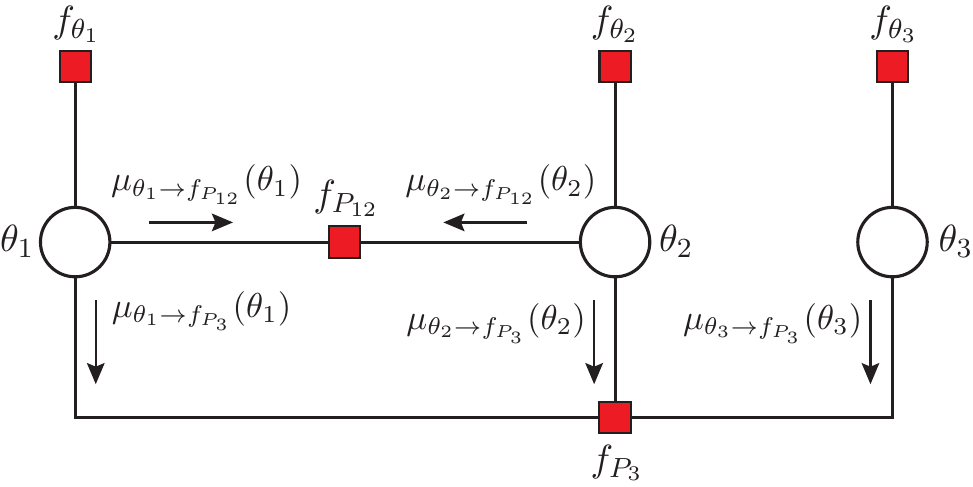}
	\caption{Variable nodes forward the incoming messages.}
	\label{appB_fig_4}
	\end{figure}	
Consequently, messages from variable nodes $\mathcal{V}$ to indirect factor nodes $\mathcal{F}_{\mathrm{ind}}$ are as follows:	
	\begin{subequations} 
   	\begin{align*}
	\mu_{\theta_1 \to f_{P_{12}}}^{(0)}(\theta_1) &:=
	(z_{\theta_1 \to f_{P_{12}}}^{(0)},v_{\theta_1 \to f_{P_{12}}}^{(0)})  
	= (0,10^{-60})\\[4pt]
	\mu_{\theta_2 \to f_{P_{12}}}^{(0)}(\theta_2) &:= 
	(z_{\theta_2 \to f_{P_{12}}}^{(0)},v_{\theta_2 \to f_{P_{12}}}^{(0)})
	= (-0.066,10^{-6})\\[4pt]
	\mu_{\theta_1 \to f_{P_{3}}}^{(0)}(\theta_1) &:= 
	(z_{\theta_1 \to f_{P_{3}}}^{(0)},v_{\theta_1 \to f_{P_{3}}}^{(0)})
	= (0,10^{-60})\\[4pt]
	\mu_{\theta_2 \to f_{P_{3}}}^{(0)}(\theta_2) &:= 
	(z_{\theta_2 \to f_{P_{3}}}^{(0)},v_{\theta_2 \to f_{P_{3}}}^{(0)})
	= (-0.066,10^{-6})\\[4pt]		
	\mu_{\theta_3 \to f_{P_{3}}}^{(0)}(\theta_3) &:= 
	(z_{\theta_3 \to f_{P_{3}}}^{(0)},v_{\theta_3 \to f_{P_{3}}}^{(0)}) 
	= (0,10^{60}).
	\end{align*}
	\end{subequations}	

\subsection*{The DC-BP Iterations $\bm {\tau = 1,2,} \dots$}
\subsubsection*{Messages from indirect factor nodes to variable nodes}
The BP iteration $\tau = 1$ starts with computing messages from indirect factor nodes $\mathcal{F}_{\mathrm{ind}}$ to variable nodes $\mathcal{V}$, as shown in \autoref{appB_fig_5}, using incoming messages from variable nodes $\mathcal{V}$ to indirect factor nodes $\mathcal{F}_{\mathrm{ind}}$ obtained in the initialization step.  
	\begin{figure}[ht]
	\centering
	\includegraphics[width=65mm]{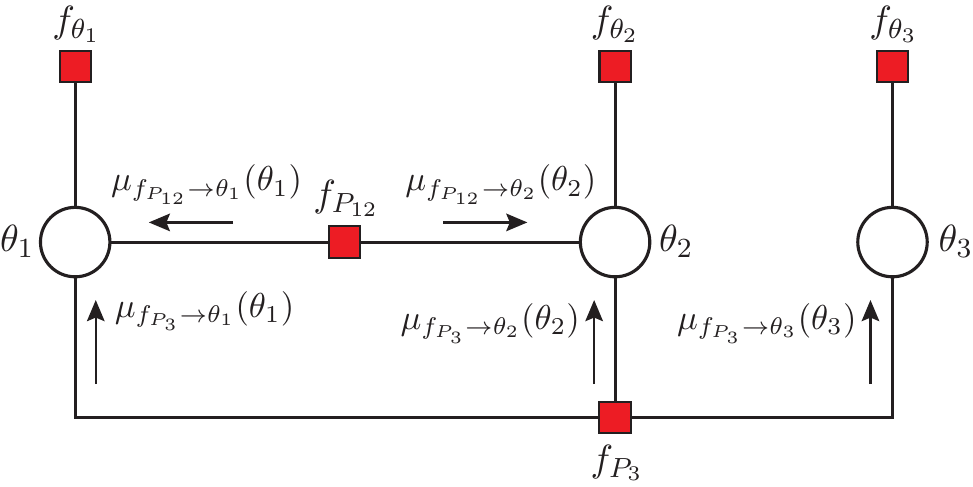}
	\caption{Messages from indirect factor nodes to variable nodes.}
	\label{appB_fig_5}
	\end{figure}

Mean and variance values of messages from factor node $f_{P{12}}$ to variable nodes $\theta_1$ and $\theta_2$ are respectively:
	\begin{subequations} 
   	\begin{align*}
	z_{f_{P_{12}} \to \theta_1}^{(1)} &= \cfrac{1}{C_{\theta_1P_{12}}}
	(z_{P_{12}} - C_{\theta_2P_{12}} \cdot z_{\theta_2 \to f_{P_{12}}}^{(0)}) 
	= 0.0058\\[3pt]
	v_{f_{P_{12}} \to \theta_1}^{(1)} &= \cfrac{1}{C_{\theta_1P_{12}}^2}
	(v_{P_{12}} + C_{\theta_2P_{12}}^2 \cdot v_{\theta_2 \to f_{P_{12}}}^{(0)}) 
	= 1.7 \cdot 10^{-5}\\[10pt]		
	z_{f_{P_{12}} \to \theta_2}^{(1)} &= \cfrac{1}{C_{\theta_2P_{12}}}
	(z_{P_{12}} - C_{\theta_1P_{12}} \cdot z_{\theta_1 \to f_{P_{12}}}^{(0)}) 
	= -0.0718\\[3pt]
	v_{f_{P_{12}} \to \theta_2}^{(1)} &= \cfrac{1}{C_{\theta_2P_{12}}^2}
	(v_{P_{12}} - C_{\theta_1P_{12}}^2 \cdot v_{\theta_1 \to f_{P_{12}}}^{(0)}) = 
	1.6 \cdot 10^{-5}.
	\end{align*}
	\end{subequations}	
\noindent
Mean and variance values of messages from factor node $f_{P{3}}$ to variable nodes $\theta_1$, $\theta_2$ and $\theta_3$ are respectively:
	\begin{subequations} 
   	\begin{align*}
	z_{f_{P_{3}} \to \theta_1}^{(1)} &= \cfrac{1}{C_{\theta_1P_{3}}}
	(z_{P_{3}} - C_{\theta_2P_{3}} \cdot z_{\theta_2 \to f_{P_{3}}}^{(0)}
	- C_{\theta_3P_{3}} \cdot z_{\theta_3 \to f_{P_{3}}}^{(0)}) = 0.0135\\[3pt]
	v_{f_{P_{3}} \to \theta_1}^{(1)} &= \cfrac{1}{C_{\theta_1P_{3}}^2}
	(v_{P_{3}} + C_{\theta_2P_{3}}^2 \cdot v_{\theta_2 \to f_{P_{3}}}^{(0)}
	+ C_{\theta_3P_{3}}^2 \cdot v_{\theta_3 \to f_{P_{3}}}^{(0)}) 
	= 3.24 \cdot 10^{60}\\[10pt]	
	z_{f_{P_{3}} \to \theta_2}^{(1)} &= \cfrac{1}{C_{\theta_2P_{3}}}
	(z_{P_{3}} - C_{\theta_1P_{3}} \cdot z_{\theta_1 \to f_{P_{3}}}^{(0)}
	- C_{\theta_3P_{3}} \cdot z_{\theta_3 \to f_{P_{3}}}^{(0)}) = -0.0491\\[3pt]
	v_{f_{P_{3}} \to \theta_2}^{(1)} &= \cfrac{1}{C_{\theta_2P_{3}}^2}
	(v_{P_{3}} + C_{\theta_1P_{3}}^2 \cdot v_{\theta_1 \to f_{P_{3}}}^{(0)}
	+ C_{\theta_3P_{3}}^2 \cdot v_{\theta_3 \to f_{P_{3}}}^{(0)}) 
	= 5.0625 \cdot 10^{60}\\[10pt]
	z_{f_{P_{3}} \to \theta_3}^{(1)} &= \cfrac{1}{C_{\theta_3P_{3}}}
	(z_{P_{3}} - C_{\theta_1P_{3}} \cdot z_{\theta_1 \to f_{P_{3}}}^{(0)}
	- C_{\theta_2P_{3}} \cdot z_{\theta_2 \to f_{P_{3}}}^{(0)}) = -0.0075\\[3pt]
	v_{f_{P_{3}} \to \theta_3}^{(1)} &= \cfrac{1}{C_{\theta_3P_{3}}^2}
	(v_{P_{3}} + C_{\theta_1P_{3}}^2 \cdot v_{\theta_1 \to f_{P_{3}}}^{(0)}
	+ C_{\theta_2P_{3}}^2 \cdot v_{\theta_2 \to f_{P_{3}}}^{(0)}) 
	= 1.4321 \cdot 10^{-6}.		
	\end{align*}
	\end{subequations}		
To summarize, corresponding messages from indirect factor nodes $\mathcal{F}_{\mathrm{ind}}$ to variable nodes $\mathcal{V}$ are:
	\begin{subequations} 
   	\begin{align*}
   	\mu_{f_{P_{12}} \to \theta_1}^{(1)}(\theta_1) &:= 
   	(z_{f_{P_{12}} \to \theta_1}^{(1)}, v_{f_{P_{12}} \to \theta_1}^{(1)}) 
   	= (0.0058, 1.7 \cdot 10^{-5})\\[4pt]
	\mu_{f_{P_{12}} \to \theta_2}^{(1)}(\theta_2)	&:= 
	(z_{f_{P_{12}} \to \theta_2}^{(1)},v_{f_{P_{12}} \to \theta_2}^{(1)})
	= (-0.0718, 1.6 \cdot 10^{-5})\\[4pt]
	\mu_{f_{P_{3}} \to \theta_1}^{(1)}(\theta_1) &:=
	(z_{f_{P_{3}} \to \theta_1}^{(1)},v_{f_{P_{3}} \to \theta_1}^{(1)})
	=(0.0135,3.24 \cdot 10^{60})\\[4pt]
	\mu_{f_{P_{3}} \to \theta_2}^{(1)}(\theta_2) &:=
	(z_{f_{P_{3}} \to \theta_2}^{(1)},v_{f_{P_{3}} \to \theta_2}^{(1)})
	= (-0.0491,5.0625 \cdot 10^{60})\\[4pt]
	\mu_{f_{P_{3}} \to \theta_3}^{(1)}(\theta_3) &:=
	(z_{f_{P_{3}} \to \theta_3}^{(1)}, v_{f_{P_{3}} \to \theta_3}^{(1)})
	=(-0.0075,1.4321 \cdot 10^{-6}).
	\end{align*}
	\end{subequations}
	
\subsubsection*{Messages from variable nodes to indirect factor nodes}
Next, the algorithm proceeds with computing messages from variable nodes $\mathcal{V}$ to indirect factor nodes $\mathcal{F}_{\mathrm{ind}}$, as shown in \autoref{appB_fig_6}, using incoming messages from factor nodes $\mathcal{F}$ to variable nodes $\mathcal{V}$.
	\begin{figure}[ht]
	\centering
	\includegraphics[width=65mm]{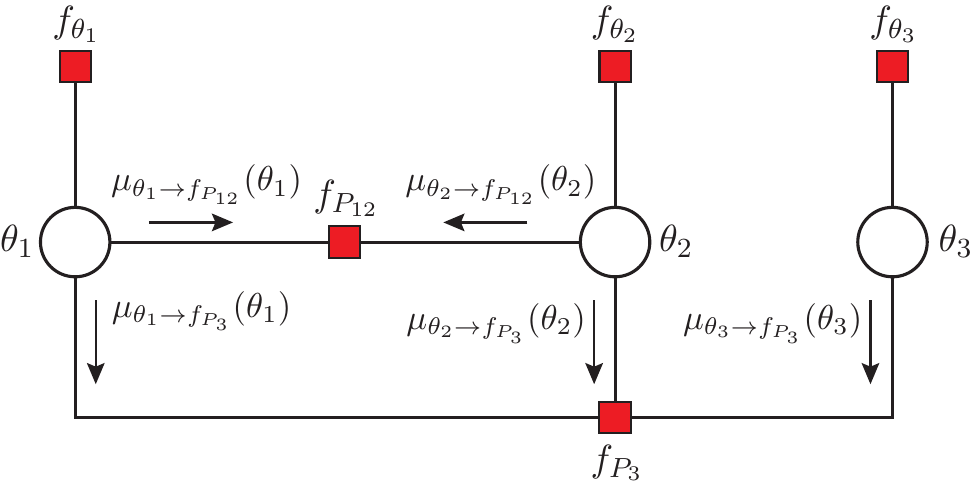}
	\caption{Messages from variable nodes to indirect factor nodes.}
	\label{appB_fig_6}
	\end{figure} 

Variance and mean values of messages from variable nodes $\theta_1$ and $\theta_2$ to factor node $f_{P{12}}$ are respectively:
	\begin{subequations} 
   	\begin{align*}
	v_{\theta_1 \to f_{P_{12}}}^{(1)} &= 
	\Bigg (\cfrac{1}{v_{f_{\theta_1} \to \theta_1}} 
	+ \cfrac{1}{v_{f_{P_3} \to \theta_1}^{(1)}}\Bigg)^{-1} = 10^{-60} \\[3pt]
	z_{\theta_1 \to f_{P_{12}}}^{(1)} &= 
	\Bigg (\cfrac{z_{f_{\theta_1} \to \theta_1}}{v_{f_{\theta_1} \to \theta_1}} 
	+ \cfrac{z_{f_{P_3} \to \theta_1}^{(1)}}{v_{f_{P_3} \to \theta_1}^{(1)}}\Bigg)
	v_{\theta_1 \to f_{P_{12}}}^{(1)} 
	 = 0 \\[10pt]
	v_{\theta_2 \to f_{P_{12}}}^{(1)} &= 
	\Bigg (\cfrac{1}{v_{f_{\theta_2} \to \theta_2}} 
	+ \cfrac{1}{v_{f_{P_3} \to \theta_2}^{(1)}}\Bigg)^{-1} = 10^{-6}\\[3pt]
	z_{\theta_2 \to f_{P_{12}}}^{(1)} &= 
	\Bigg (\cfrac{z_{f_{\theta_2} \to \theta_2}}{v_{f_{\theta_2} \to \theta_2}} 
	+ \cfrac{z_{f_{P_3} \to \theta_2}^{(1)}}{v_{f_{P_3} \to \theta_2}^{(1)}}\Bigg)
	 v_{\theta_2 \to f_{P_{12}}}^{(1)}= -0.066. 	 
	\end{align*}
	\end{subequations}		 
Variance and mean values of messages from variable nodes $\theta_1$, $\theta_2$ and $\theta_3$ to factor node $f_{P{3}}$ are respectively:
	\begin{subequations} 
   	\begin{align*}
	v_{\theta_1 \to f_{P_{3}}}^{(1)} &= 
	\Bigg (\cfrac{1}{v_{f_{\theta_1} \to \theta_1}} 
	+ \cfrac{1}{v_{f_{P_{12}} \to \theta_1}^{(1)}}\Bigg)^{-1} = 10^{-60} \\[3pt]
	z_{\theta_1 \to f_{P_{3}}}^{(1)} &= 
	\Bigg (\cfrac{z_{f_{\theta_1} \to \theta_1}}{v_{f_{\theta_1} \to \theta_1}} 
	+ \cfrac{z_{f_{P_{12}} \to \theta_1}^{(1)}}{v_{f_{P_{12}} \to \theta_1}^{(1)}}\Bigg)
	v_{\theta_1 \to f_{P_{3}}}^{(1)} = 3.4118 \cdot 10^{-58}\\[10pt]
	v_{\theta_2 \to f_{P_{3}}}^{(1)} &= 
	\Bigg (\cfrac{1}{v_{f_{\theta_2} \to \theta_2}} 
	+ \cfrac{1}{v_{f_{P_{12}} \to \theta_2}^{(1)}}\Bigg)^{-1} 
	=  9.4118 \cdot 10^{-7}\\[3pt]	
	z_{\theta_2 \to f_{P_{3}}}^{(1)} &= 
	\Bigg (\cfrac{z_{f_{\theta_2} \to \theta_2}}{v_{f_{\theta_2} \to \theta_2}} 
	+ \cfrac{z_{f_{P_{12}} \to \theta_2}^{(1)}}{v_{f_{P_{12}} \to \theta_2}^{(1)}}\Bigg)
	v_{\theta_2 \to f_{P_{3}}}^{(1)} = -0.0663\\[10pt]
	v_{\theta_3 \to f_{P_{3}}}^{(1)} &= 
	\Bigg (\cfrac{1}{v_{f_{\theta_3} \to \theta_3}} \Bigg)^{-1} = 10^{60} \\[3pt]
	z_{\theta_3 \to f_{P_{3}}}^{(1)} &= 
	\Bigg (\cfrac{z_{f_{\theta_3} \to \theta_3}}
	{v_{f_{\theta_3} \to \theta_3}} \Bigg)v_{\theta_3 \to f_{P_{3}}}^{(1)} = 0.			 
	\end{align*}
	\end{subequations}	
To summarize, corresponding messages from variable nodes $\mathcal{V}$ to indirect factor nodes $\mathcal{F}_{\mathrm{ind}}$ are:  
	\begin{subequations} 
   	\begin{align*}
   	\mu_{\theta_1 \to f_{P_{12}}}^{(1)}(\theta_1) &:=
   	(z_{\theta_1 \to f_{P_{12}}}^{(1)},v_{\theta_1 \to f_{P_{12}}}^{(1)})
   	=(0,10^{-60})\\[4pt]
	\mu_{\theta_2 \to f_{P_{12}}}^{(1)}(\theta_2)&:=
	(v_{\theta_2 \to f_{P_{12}}}^{(1)},v_{\theta_2 \to f_{P_{12}}}^{(1)}) 
	=(-0.066,10^{-6})\\[4pt]
	\mu_{\theta_1 \to f_{P_{3}}}^{(1)}(\theta_1)&:=
	(z_{\theta_1 \to f_{P_{3}}}^{(1)},v_{\theta_1 \to f_{P_{3}}}^{(1)})
	=(3.4118 \cdot 10^{-58},10^{-60})\\[4pt]
	\mu_{\theta_2 \to f_{P_{3}}}^{(1)}(\theta_2)&:=
	(v_{\theta_2 \to f_{P_{3}}}^{(1)},v_{\theta_2 \to f_{P_{3}}}^{(1)})
	=(-0.0663,9.4118 \cdot 10^{-7})	\\[4pt]
	\mu_{\theta_3 \to f_{P_{3}}}^{(1)}(\theta_3)&:=
	(z_{\theta_3 \to f_{P_{3}}}^{(1)},v_{\theta_3 \to f_{P_{3}}}^{(1)})
	=(0,10^{60}). 
	\end{align*}
	\end{subequations}	
	
Finally, the first iteration is done, and the iteration loop is repeated
until the stopping criteria is met. We define accuracy-based criterion where iteration loop is running until the following criterion is reached:    
		\begin{equation}
        \begin{gathered}      
		|\mathbf{z}_{f \to \theta}^{(\tau)}-
		\mathbf{z}_{f \to \theta}^{(\tau-1})| < \epsilon,
		\end{gathered}
		\label{num_break2}
		\end{equation}
where $\mathbf{z}_{f \to \theta}$ represents the vector of mean-value messages from factor nodes to variable nodes, and $\epsilon = 10^{-14}$ is the threshold. The algorithm converged after $\tau = 3$ iterations and final value of messages from indirect factor nodes $\mathcal{F}_{\mathrm{ind}}$ to variable nodes $\mathcal{V}$ are:
	\begin{subequations} 
   	\begin{align*}
   	\mu_{f_{P_{12}} \to \theta_1}(\theta_1) &:= 
   	(z_{f_{P_{12}} \to \theta_1}, v_{f_{P_{12}} \to \theta_1}) 
   	= (0.0058, 1.7 \cdot 10^{-5})\\[4pt]
	\mu_{f_{P_{12}} \to \theta_2}(\theta_2)	&:= 
	(z_{f_{P_{12}} \to \theta_2},v_{f_{P_{12}} \to \theta_2})
	= (-0.0718, 1.6 \cdot 10^{-5})\\[4pt]
	\mu_{f_{P_{3}} \to \theta_1}(\theta_1) &:=
	(z_{f_{P_{3}} \to \theta_1},v_{f_{P_{3}} \to \theta_1})
	=(0.0138,3.24 \cdot 10^{60})\\[4pt]
	\mu_{f_{P_{3}} \to \theta_2}(\theta_2) &:=
	(z_{f_{P_{3}} \to \theta_2},v_{f_{P_{3}} \to \theta_2})
	= (-0.0491,5.0625 \cdot 10^{60})\\[4pt]
	\mu_{f_{P_{3}} \to \theta_3}(\theta_3) &:=
	(z_{f_{P_{3}} \to \theta_3}, v_{f_{P_{3}} \to \theta_3})
	=(-0.0076,1.4205 \cdot 10^{-6}).
	\end{align*}
	\end{subequations}
	
\subsection*{The DC-BP Marginal Inference}	
The marginal of variable nodes $\mathcal{V}$ can be obtained using messages from factor nodes $\mathcal{F}$ to variable nodes $\mathcal{V}$, as shown in \autoref{appB_fig_7}. Note that the mean-value of marginal is adopted as the estimated value of the state variable.
	\begin{figure}[ht]
	\centering
	\includegraphics[width=65mm]{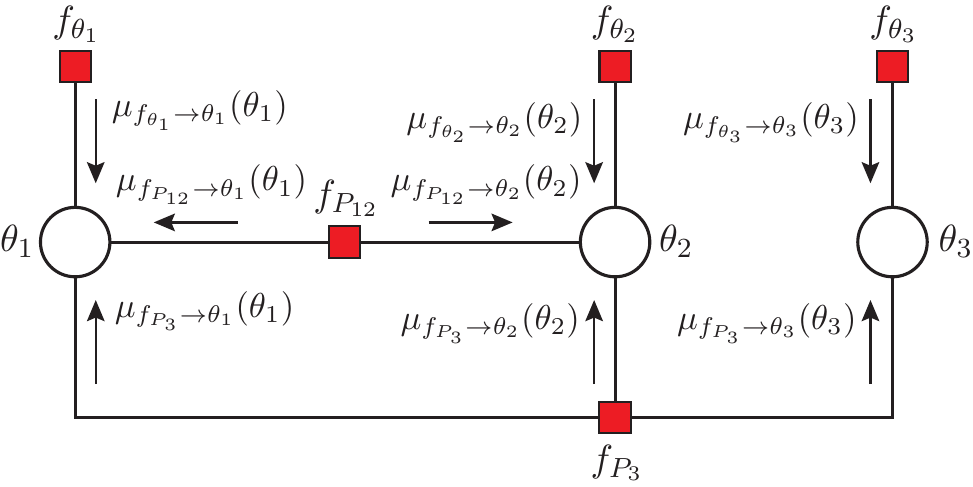}
	\caption{Messages into variable nodes.}
	\label{appB_fig_7}
	\end{figure} 

Consequently, estimated values of state variables can be obtained:
	\begin{subequations} 
   	\begin{align*}
	v_{\theta_1} &= 
	\Bigg (\cfrac{1}{v_{f_{\theta_1} \to \theta_1}} 
	+ \cfrac{1}{v_{f_{P_{12}} \to \theta_1}} 
	+ \cfrac{1}{v_{f_{P_{3}} \to \theta_1}}
	\Bigg)^{-1} 	= 10^{-60}\\[3pt]
	\hat{\theta}_1 &= 
	\Bigg (\cfrac{z_{f_{\theta_1} \to \theta_1}}{v_{f_{\theta_1} \to \theta_1}} 
	+ \cfrac{z_{f_{P_{12}} \to \theta_1}}{v_{f_{P_{12}} \to \theta_1}} 
	+ \cfrac{z_{f_{P_{3}} \to \theta_1}}{v_{f_{P_{3}} \to \theta_1}}
	\Bigg)v_{\theta_1} 	= 0 \\[10pt]
	v_{\theta_2} &= 
	\Bigg (\cfrac{1}{v_{f_{\theta_2} \to \theta_2}} 
	+ \cfrac{1}{v_{f_{P_{12}} \to \theta_2}} 
	+ \cfrac{1}{v_{f_{P_{3}} \to \theta_2}}
	\Bigg)^{-1} 	=9.4118 \cdot 10^{-7}\\[3pt]
	\hat{\theta}_2 &=
	\Bigg (\cfrac{z_{f_{\theta_2} \to \theta_2}}{v_{f_{\theta_2} \to \theta_2}} 
	+ \cfrac{z_{f_{\theta_2} \to \theta_2}}{v_{f_{P_{12}} \to \theta_2}} 
	+ \cfrac{z_{f_{\theta_2} \to \theta_2}}{v_{f_{P_{3}} \to \theta_2}}
	\Bigg)v_{\theta_2}	= -0.0663\\[10pt]
	v_{\theta_3} &= 
	\Bigg (\cfrac{1}{v_{f_{\theta_3} \to \theta_3}} 
	+ \cfrac{1}{v_{f_{P_{3}} \to \theta_3}}
	\Bigg)^{-1} =1.4205 \cdot 10^{-6}\\[3pt]
	\hat{\theta}_3 &=	
	\Bigg (\cfrac{z_{f_{\theta_3} \to \theta_3}}{v_{f_{\theta_3} \to \theta_3}} 
	+ \cfrac{z_{f_{P_{3}} \to \theta_3}}{v_{f_{P_{3}} \to \theta_3}}
	\Bigg)v_{\theta_3} =-0.0076.		 
	\end{align*}
	\end{subequations} 
To recall, the BP solution for means is equivalent to the WLS solution. Unlike means,
the variances need not converge to correct values. 		 
\chapter{The AC-BP Algorithm: Message Derivation} \label{app:C}
\addcontentsline{lof}{chapter}{C The AC-BP Algorithm: Message Derivation}
Here we present an example of evaluation of the message from a factor node to a variable node for the AC-BP algorithm. We consider a simple model containing buses $i$ and $j$, with the active power flow measurement $M_i \equiv$ $M_{P_{ij}}$ at the branch $(i,j)$. The mean $z_i$, variance $v_i$ and the measurement function $h_i(\theta_i, V_i, \theta_j, V_j)$ defined as \eqref{mf_activeF} is associated with the active power flow measurement $M_i$. The corresponding factor graph is shown in \autoref{Fig_appB_mess}. 

Further, all incoming messages from variable nodes to the factor node $f_i$ have Gaussian form. Therefore, these messages, denoted as $\mu_{\theta_i \to f_i}(\theta_i)$, $\mu_{V_i \to f_i}(V_i)$, $\mu_{\theta_j \to f_i}(\theta_j)$ and $\mu_{V_j \to f_i}(V_j)$, are represented by their mean-variance pair $(z_{\theta_i \to f_i},$ $v_{\theta_i \to f_i})$, $(z_{V_i \to f_i},$ $v_{V_i \to f_i})$, $(z_{\theta_j \to f_i},$ $v_{\theta_j \to f_i})$ and $(z_{V_j \to f_i},$ $v_{V_j \to f_i})$, respectively (\autoref{Fig_appB_V_i} - \autoref{Fig_appB_theta_j}). 

	\begin{figure}[h]
	\centering
	\begin{tabular}{@{}c@{}}
	\subfloat[]{\label{Fig_appB_V_i}
	\includegraphics[width=4.3cm]{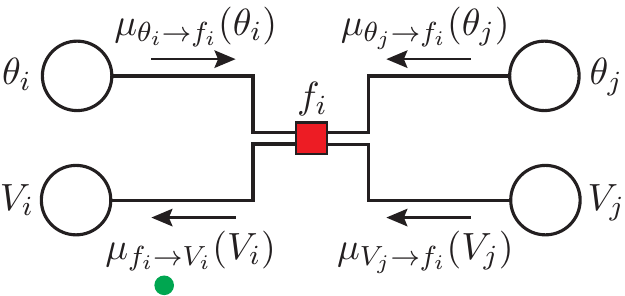}}
	\end{tabular}\quad
	\begin{tabular}{@{}c@{}}
	\subfloat[]{\label{Fig_appB_V_j}
	\includegraphics[width=4.3cm]{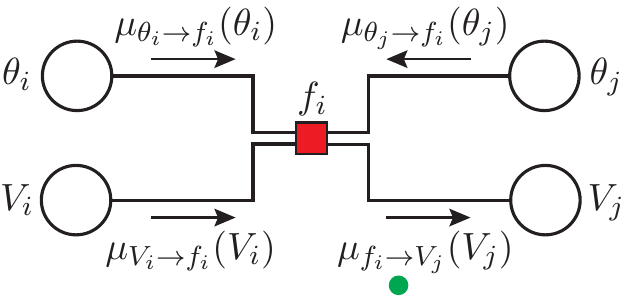}} 
	\end{tabular}\\
	\begin{tabular}{@{}c@{}}
	\subfloat[]{\label{Fig_appB_theta_i}
	\includegraphics[width=4.3cm]{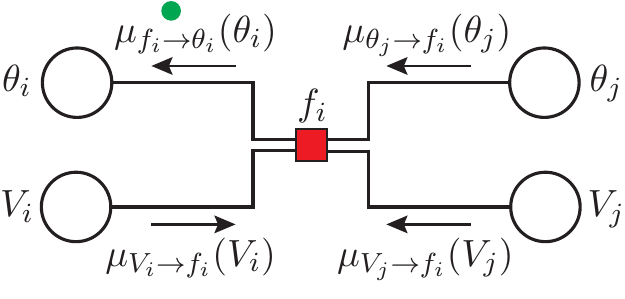}}
	\end{tabular}\quad
	\begin{tabular}{@{}c@{}}
	\subfloat[]{\label{Fig_appB_theta_j}
	\includegraphics[width=4.3cm]{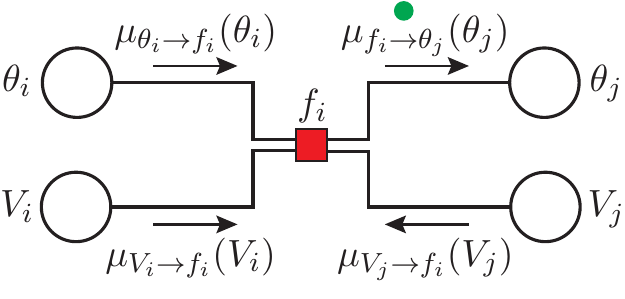}} 
	\end{tabular}
	\caption{Messages from from factor node $f_i$ to variable nodes: 
	$V_i$ (subfigure a), $V_j$ (subfigure b), $\theta_i$ (subfigure c) 
	and $\theta_j$ (subfigure d).}
	\label{Fig_appB_mess}
	\end{figure} 
	
According to assumption (see \autoref{ch:native_bp}), the messages from the factor node $f_i$ to variable nodes have Gaussian form: $\mu_{f_i \to V_i}(V_i)$, $\mu_{f_i \to V_j}(V_j)$, $\mu_{f_i \to \theta_i}(\theta_i)$ and $\mu_{f_i \to \theta_j}(\theta_j)$ with their mean-variance pair $(z_{f_i \to V_i},v_{f_i \to V_i})$, $(z_{f_i \to V_j},v_{f_i \to V_j})$, $(z_{f_i \to \theta_i}, v_{f_i \to \theta_i})$ and $(z_{f_i \to \theta_j},v_{f_i \to \theta_j})$. In the following, we consider calculation of each of these messages.

\noindent
{\scriptsize $\bullet$} The message $\mu_{f_i \to V_i}$ (\autoref{Fig_appB_V_i}): Let us first consider the mean $z_{f_i \to V_i}$.  The equation \eqref{BP_fv_mean} for the active power flow measurement boils down to \eqref{BP_vf_mean_Vi}:
		\begin{equation}
        \begin{aligned}
		a  \mathbb{E}[V_i^2|\mathbf {x}_b = \mathbf{z}_{\mathbf {x}_b \to f_i}]+ 
		b  \mathbb{E}[V_i|\mathbf {x}_b = \mathbf{z}_{\mathbf {x}_b \to f_i}] + c = 0,		
        \end{aligned}
		\nonumber
		\end{equation}		
where: $\mathbf{x}_b=(\theta_i, \theta_j, V_j)$ and $\mathbf{z}_{\mathbf {x}_b \to f_i} =$ $(z_{\theta_i \to f_i},$ $z_{\theta_j \to f_i},$ $z_{V_j \to f_i})$, with coefficients: 
		\begin{equation}
        \begin{aligned}		
		a &= g_{ij}+g_{si}\\
		b &= -z_{V_j \to f_i}
		(g_{ij}\cos z_{\theta_{ij} \to f_i}+
		b_{ij}\sin z_{\theta_{ij} \to f_i})\\
		c &= -z_i,
        \end{aligned}
		\nonumber
		\end{equation}
where $z_{\theta_{ij} \to f_i}$ is determined as $z_{\theta_{i} \to f_i}$ $- z_{\theta_j \to f_i}$. Due the fact that the conditional expected value $\mathbb{E}[V_i|\mathbf {x}_b = \mathbf{z}_{\mathbf {x}_b \to f_i}]$ represents the mean $z_{f_i \to V_i}$, we can write:
		\begin{equation}
        \begin{aligned}
		a  (z_{f_i \to V_i}^2 + v_{f_i \to V_i})+ 
		b  z_{f_i \to V_i} + c = 0.		
        \end{aligned}
		\nonumber
		\end{equation}
The mean $z_{f_i \to V_i}$ follows from the quadratic equation, where we selected a solution using \eqref{BP_mean_cond}.		

The variance $v_{f_i \to V_i}$ is determined using \eqref{BP_fv_var} as: 
        \begin{equation}
		\begin{aligned}
        \sigma_{f_i \to V_i}^2 =         
        \cfrac{1}{C_{V_i}^2} ( v_i +  
        C_{\theta_i}^2 v_{\theta_i \to f_i} +               
        C_{\theta_j}^2 v_{\theta_j \to f_i} +
        C_{V_j}^2 v_{V_j \to f_i}),
        \end{aligned}
        \nonumber
		\end{equation}
where coefficients are defined according to Jacobian elements of the measurement function $h_i(\cdot)$:
		\begin{equation}
        \begin{aligned}
		C_{\theta_i}=
		\cfrac{\mathrm \partial{h_i(V_i,\mathbf x_b)}}{\mathrm \partial \theta_i}
		\Biggr|_{\substack{V_i = z_{f_i \to V_i} \\ 
		\mathbf x_b = \mathbf z_{\mathbf {x}_b \to f_i}}} 
		\;
     	C_{\theta_j}=
		\cfrac{\mathrm \partial{h_i(V_i,\mathbf x_b)}}{\mathrm \partial \theta_j}
		\Biggr|_{\substack{V_i = z_{f_i \to V_i} \\ 
		\mathbf x_b = \mathbf z_{\mathbf {x}_b \to f_i}}}
		\\
		C_{V_i}=
		\cfrac{\mathrm \partial{h_i(V_i,\mathbf x_b)}}{\mathrm \partial V_i} 
		\Biggr|_{\substack{V_i = z_{f_i \to V_i} \\ 
		\mathbf x_b = \mathbf z_{\mathbf {x}_b \to f_i}}}
		\;
     	C_{V_j}=
		\cfrac{\mathrm \partial{h_i(V_i,\mathbf x_b)}}{\mathrm \partial V_j} 
		\Biggr|_{\substack{V_i = z_{f_i \to V_i} \\ 
		\mathbf x_b = \mathbf z_{\mathbf {x}_b \to f_i}}}	
        \end{aligned}
		\nonumber	
		\end{equation} 

\noindent
{\scriptsize $\bullet$} The message $\mu_{f_i \to V_j}$ (\autoref{Fig_appB_V_j}): The mean $z_{f_i \to V_j}$ is defined according to \eqref{BP_vf_mean_Vj} as: 
		\begin{equation}
        \begin{aligned}
		a\mathbb{E}[V_j|\mathbf {x}_b = \mathbf{z}_{\mathbf {x}_b \to f_i}]+b= 0,		
        \end{aligned}
		\nonumber
		\end{equation}
where: $\mathbf{x}_b=(\theta_i, V_i, \theta_j)$ and $\mathbf{z}_{\mathbf {x}_b \to f_i}$ $= (z_{\theta_i \to f_i},$ $z_{V_i \to f_i},$ $z_{\theta_j \to f_i} )$, with coefficients:		
		\begin{equation}
        \begin{aligned}		
		a &= z_i-z_{V_i \to f_i}^2(g_{ij}+g_{si})\\
		b &= z_{V_i \to f_i}
		(g_{ij}\cos z_{\theta_{ij} \to f_i}+
		b_{ij}\sin z_{\theta_{ij} \to f_i}).
        \end{aligned}
		\nonumber
		\end{equation}
Due the fact that the conditional expected value $\mathbb{E}[V_j|\mathbf {x}_b = \mathbf{z}_{\mathbf {x}_b \to f_i}]$ represents the mean $z_{f_i \to V_j}$, we obtain:
		\begin{equation}
        \begin{aligned}
		az_{f_i \to V_j}+b= 0.		
        \end{aligned}
		\nonumber
		\end{equation}
The variance $v_{f_i \to V_j}$ is determined using \eqref{BP_fv_var} as: 
        \begin{equation}
		\begin{aligned}
        v_{f_i \to V_j} =         
        \cfrac{1}{C_{V_j}^2} ( v_i +  
        C_{\theta_i}^2 v_{\theta_i \to f_i} +
		C_{V_i}^2 v_{V_i \to f_i} +        
        C_{\theta_j}^2 v_{\theta_j \to f_i}),
        \end{aligned}
        \nonumber
		\end{equation}
where coefficient are defined according to Jacobian elements of the measurement function $h_i(\cdot)$.

\noindent
{\scriptsize $\bullet$} The messages $\mu_{f_i \to \theta_i}$ and $\mu_{f_i \to \theta_j}$ (\autoref{Fig_appB_theta_i} and \autoref{Fig_appB_theta_j}): Means $z_{f_i \to \theta_i}$ and $z_{f_i \to \theta_j}$ are defined according to \eqref{BP_vf_mean_sin}:
		\begin{equation}
        \begin{aligned}
		a\mathbb{E}[\sin^2 x_s|\mathbf {x}_b = \mathbf{z}_{\mathbf {x}_b \to f_i}]+ 
		b \mathbb{E}[\sin x_s|\mathbf {x}_b = \mathbf{z}_{\mathbf {x}_b \to f_i}]+c= 0,	
        \end{aligned}
		\nonumber
		\end{equation} 
where: $\mathbf{x}_b=(V_i, \theta_j, V_j)$ and $\mathbf{z}_{\mathbf {x}_b \to f_i}$ $= (z_{V_i \to f_i},$ $z_{\theta_j \to f_i},$ $z_{V_j \to f_i} )$  for the message $\mu_{f_i \to \theta_i}$, $\mathbf{x}_b=(\theta_i, V_i, V_j)$ and $\mathbf{z}_{\mathbf {x}_b \to f_i}$ $= (z_{\theta_i \to f_i},$ $z_{V_i \to f_i},$ $z_{V_j \to f_i} )$ for the message $\mu_{f_i \to \theta_j}$, and $x_s \in \{\theta_i, \theta_j \}$. Due the fact that the all variables and messages preserve Gaussian distribution, the conditional expectations of sine functions are equal to $\mathbb{E}[\sin^2 x_s|\mathbf {x}_b =$ $\mathbf{z}_{\mathbf {x}_b \to f_i}] =$ $\sin^2 z_{f_i \to x_s}$ and $\mathbb{E}[\sin x_s|\mathbf {x}_b$ $= \mathbf{z}_{\mathbf {x}_b \to f_i}] =$ $\sin z_{f_i \to x_s}$, which allows us to compute the mean:
		\begin{equation}
        \begin{aligned}
		a\sin^2 z_{f_i \to x_s} + 
		b \sin z_{f_i \to x_s}+c= 0.		
        \end{aligned}
		\nonumber
		\end{equation}
To simplify expressions, we introduce coefficients $a = A^2 + B^2$, $b = -2BC$ and $c = -A^2 + C^2$:
		\begin{subequations}
        \begin{align*}     		
		A &= g_{ij} \cos z_{\theta_j \to f_i} - b_{ij} \sin z_{\theta_j \to f_i}, 
		&x_s& \equiv \theta_i\\
		A &= g_{ij} \cos z_{\theta_i \to f_i} + b_{ij} \sin z_{\theta_i \to f_i}, 
		&x_s& \equiv \theta_j\\
		B &= g_{ij} \sin z_{\theta_j \to f_i} + b_{ij} \cos z_{\theta_j \to f_i}, 
		&x_s& \equiv \theta_i\\
		B &= g_{ij} \sin z_{\theta_i \to f_i} - b_{ij} \cos z_{\theta_i \to f_i},
		&x_s& \equiv \theta_i\\
		C &= \cfrac{z_{V_i \to f_i}^2(g_{ij}+g_{si})-z_i}{z_{V_i \to f_i}z_{V_j \to f_i}}, 
		&x_s& \in \{\theta_i, \theta_j \}		
		\end{align*}   
		\end{subequations}		

The variance $v_{f_i \to \theta_i}$ is determined using \eqref{BP_fv_var} as: 
        \begin{equation}
		\begin{aligned}
        v_{f_i \to \theta_i} =         
        \cfrac{1}{C_{\theta_i}^2} ( v_i +  
        C_{V_i}^2 v_{V_i \to f_i} +
        C_{\theta_j}^2 v_{\theta_j \to f_i} +
        C_{V_j}^2 v_{V_j \to f_i}),
        \end{aligned}
        \nonumber
		\end{equation}
where coefficients are defined, as above, by calculating Jacobian elements of the measurement function $h_i(\cdot)$.

The variance $v_{f_i \to \theta_j}$ is determined according to \eqref{BP_fv_var} as: 
        \begin{equation}
		\begin{aligned}
        v_{f_i \to \theta_j} =         
        \cfrac{1}{C_{\theta_j}^2} ( v_i +  
        C_{\theta_i}^2 v_{\theta_i \to f_i} +
        C_{V_i}^2 v_{V_i \to f_i} +
        C_{V_j}^2 v_{V_j \to f_i}),
        \end{aligned}
        \nonumber
		\end{equation}
where coefficient follow Jacobian elements of the measurement function $h_{P_{i}}(\cdot)$.

Using the same methodology, it is possible to define corresponding equations for means and variances for every type of measurement functions.	
 
\chapter{The GN-BP Algorithm: Toy Example} \label{app:D}
\addcontentsline{lof}{chapter}{D The GN-BP Algorithm: Toy Example}
An illustrative example presented in \autoref{D_bus} will be used to provide a step-by-step presentation of the proposed algorithm.
	\begin{figure}[H]
	\centering
	\includegraphics[width=3cm]{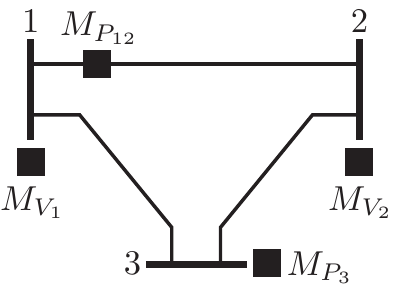}
	\caption{The 3-bus power system with given measurement configuration.}
	\label{D_bus}
	\end{figure} 	
Input data for SE from measurement devices are Gaussian-type functions represented by means and variances: $\{z_{V_1}, z_{\theta_2},z_{\theta_3}, z_{P_{12}}, z_{P_3} \}$ and $\{v_{V_1}, v_{\theta_2},v_{\theta_3}, v_{P_{12}}, v_{P_3} \}$.

\subsection*{The Factor Graph}
The corresponding factor graph is given in \autoref{D_graph}, where the set of state variables is $\mathcal{X} =$ $\{(\theta_1,V_1),$  $(\theta_2,V_2),$ $(\theta_3,V_3)\}$ and the set of variable nodes is $\mathcal{V} =$ $\{(\Delta \theta_1,\Delta V_1),$ $(\Delta \theta_2, \Delta V_2),$ $(\Delta \theta_3,\Delta V_3)\}$. The indirect factor nodes (orange squares) are defined by corresponding measurements, where in our example, active power flow $M_{P_{12}}$ and active power injection $M_{P_3}$ measurements  are mapped into factor nodes $\mathcal{F}_{\mathrm{ind}} = $ $\{f_{P_{12}},$ $f_{P_{3}} \}$. The set of local factor nodes $\mathcal{F}_{\mathrm{loc}}$ consists of the set of direct factor nodes (green squares) $\mathcal{F}_{\mathrm{dir}} = $ $\{f_{V_{1}},$ $f_{V_{2}} \}$ defined by bus voltage magnitude measurements $M_{V_{1}}$ and $M_{V_2}$, virtual factor nodes (blue squares) and the slack factor node (yellow square). 	  
	\begin{figure}[H]
	\centering
	\includegraphics[width=5.5cm]{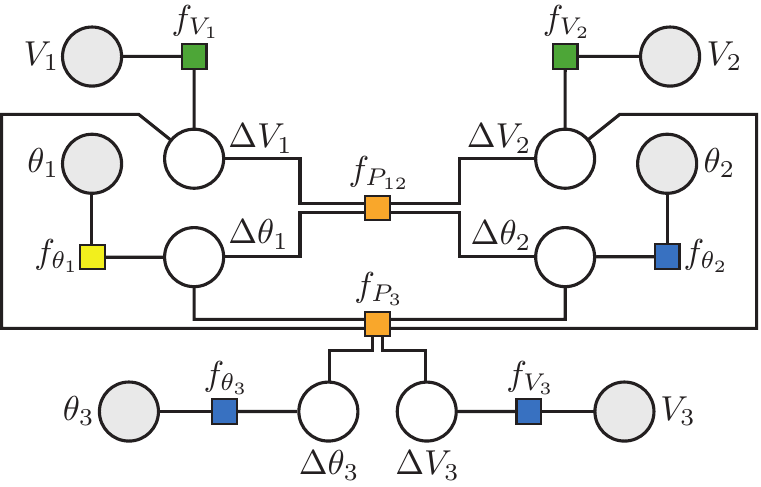}
	\caption{The factor graph.}
	\label{D_graph}
	\end{figure} 
Local factor nodes only send, but do not receive, the messages to the incident variable nodes. 

\subsection*{Algorithm Initialization $\bm {\tau = 0}$}
\begin{enumerate}[leftmargin=*]
\item The non-linear SE in electric power systems assumes ``flat start" or a priori given values of state variables: 
	 	\begin{equation}
        \begin{gathered}
        \mathbf x^{(\nu=0) }=[\theta_1\; \theta_2\; \theta_3\; V_1\; 
        V_2\; V_3]^{(\nu=0)}.
        \end{gathered}
		\nonumber
		\end{equation}
\item The residual of the slack factor node is set to $r_{\theta_1}=0$ with variance $v_{\theta_1} \to 0$. 
\item The value of virtual factor nodes are set to $r_{\theta_2} \to 0$, $r_{V_3} \to 0$ and $r_{\theta_3} \to 0$, with variances $v_{\theta_2} \to \infty$, $v_{V_3} \to \infty$ and $v_{\theta_3} \to \infty$.
\end{enumerate}

\noindent
\textbf{Iterate - Outer Loop:} $\bm\nu \mathbf{=0,1,2,\dots}; \bm\tau \mathbf{=0 }$
\begin{enumerate}[leftmargin=*,resume]			
\item Each direct factor node from the set $\mathcal{F}_{\mathrm{dir}}$ computes residual:
	 	\begin{equation}
        \begin{aligned}
		r_{V_1}^{(\nu)} &= z_{V_1}-V_1^{(\nu)}\\[4pt]
		r_{V_2}^{(\nu)} &= z_{V_2}-V_2^{(\nu)}       
		\end{aligned} 
		\nonumber
		\end{equation}	

\item Local factor nodes $\mathcal{F}_{\mathrm{loc}}$ send messages represented by a triplet (residual, variance, state variable) to incident variable nodes $\mathcal{V}$:
	 	\begin{equation}
        \begin{aligned}
		\mu_{f_{{\theta_1}} \to \Delta \theta_1}^{(\nu)} &:=
		\big(r_{\theta_1},v_{\theta_1},\theta_1^{(\nu)}\big)\\[4pt]
		\mu_{f_{{V_1}} \to \Delta V_1}^{(\nu)} &:= 
		\big(r_{V_1}^{(\nu)},v_{V_1},V_1^{(\nu)}\big).        
		\end{aligned}
		\nonumber
		\end{equation}
\item Variable nodes $\mathcal{V}$ forward the incoming messages received from local factor nodes $\mathcal{F}_{\mathrm{loc}}$ along remaining edges, e.g.: 
	 	\begin{equation}
        \begin{aligned}
		\mu_{\Delta \theta_1 \to f_{{P_{12}}}}^{(\nu,\tau)} &
		:=\big(r_{\Delta \theta_1 \to f_{{P_{12}}}}^{(\nu, \tau)},
		v_{\Delta \theta_1 \to f_{{P_{12}}}}^{(\nu, \tau)},
		\theta_1^{(\nu)}\big):= 
		\big(r_{\theta_1}^{(\nu)},v_{\theta_1},\theta_1^{(\nu)}\big) \\[4pt]
		\mu_{\Delta \theta_1 \to f_{{P_{3}}}}^{(\nu, \tau)} &
		:=\big(r_{\Delta \theta_1 \to f_{{P_{3}}}}^{(\nu, \tau)},
		v_{\Delta \theta_1 \to f_{{P_{3}}}}^{(\nu, \tau)},
		\theta_1^{(\nu)}\big):= 
		\big(r_{\theta_1}^{(\nu)},v_{\theta_1},\theta_1^{(\nu) }\big). 	
		\end{aligned}
		 \nonumber
		\end{equation}	
\item Indirect factor nodes compute residuals, e.g.:
	 	\begin{equation}
        \begin{gathered}
		r_{P_{12}}^{(\nu)}=z_{P_{12}}-
		h_{P_{12}}(\theta_1^{(\nu)}, \theta_2^{(\nu)}, V_1^{(\nu)}, V_2^{(\nu)}).			
		\end{gathered}
		\nonumber
		\end{equation}
\item Indirect factor nodes compute appropriate Jacobian elements associated with state variables, e.g.: 
		\begin{equation}
        \begin{aligned}
     	C_{P_{12},\Delta \theta_1}^{(\nu)}=
     	\cfrac{\mathrm \partial{h_{P_{12}}(\cdot)}}{\mathrm \partial \theta_{1}}&=
     	{V}_{1}^{(\nu)}{V}_{2}^{(\nu)}
     	(g_{12}\mbox{sin}\theta_{12}^{(\nu)}-b_{12}\mbox{cos}\theta_{12}^{(\nu)})\\
     	C_{P_{12},\Delta V_2}^{(\nu)}=
     	\cfrac{\mathrm \partial{h_{{P_{12}}}(\cdot)}}{\mathrm \partial V_{2}}&=
     	-{V}_{1}^{(\nu)}(g_{12}\mbox{cos}\theta_{12}^{(\nu)}+b_{12}\mbox{sin}\theta_{12}^{(\nu)}).
        \end{aligned}
		\nonumber
		\end{equation}	
\end{enumerate}

\noindent
\textbf{Iterate - Inner Loop:} $\bm\tau \mathbf{=1,2,\dots, \bm\eta(\bm\nu)}$
\begin{enumerate}[leftmargin=*,resume]	
\item Indirect factor nodes send messages as pairs along incident edges according to \eqref{GN_fv_mean_var}, e.g.:
  	 	\begin{equation}
        \begin{aligned}
		\mu_{f_{{P_{12}}} \to \Delta \theta_2}^{(\tau)} := 
		\big(r_{f_{{P_{12}}} \to \Delta \theta_2}^{(\tau)} ,
		v_{f_{{P_{12}}} \to \Delta \theta_2}^{(\tau)}\big)
		\end{aligned}
		\nonumber
		\end{equation}\\[-3ex]
		\begin{equation}
        \begin{aligned}
		r_{f_{r_{P_{12}}} \to \Delta \theta_2}^{(\tau)} =
		\cfrac{1}{C_{P_{12},\Delta \theta_2}^{(\nu)}}\Big[
		r_{P_{12}}^{(\nu)}-
		C_{P_{12},\Delta \theta_1}^{(\nu)}\cdot
		r_{\Delta \theta_1 \to f_{r_{P_{12}}}}^{(\nu, \tau-1)}\\
		-C_{P_{12},\Delta V_1}^{(\nu)}\cdot
		r_{\Delta V_1 \to f_{r_{P_{12}}}}^{(\nu, \tau-1)}
		-C_{P_{12},\Delta V_2}^{(\nu)}\cdot
		r_{\Delta V_2 \to f_{r_{P_{12}}}}^{(\nu, \tau-1)}
		\Big] 
		\end{aligned}
		\nonumber
		\end{equation}\\[-2ex]
		\begin{equation}
        \begin{aligned}
		v_{f_{r_{P_{12}}} \to \Delta \theta_2}^{(\tau)} =
		\cfrac{1}{(C_{P_{12},\Delta \theta_2}^{(\nu)})^2}\Big[
		v_{P_{12}}+
		(C_{P_{12},\Delta \theta_1}^{(\nu)})^2\cdot
		v_{\Delta \theta_1 \to f_{r_{P_{12}}}}^{(\nu, \tau-1)}\\
		+(C_{P_{12},\Delta V_1}^{(\nu)})^2\cdot
		v_{\Delta V_1 \to f_{r_{P_{12}}}}^{(\nu, \tau-1)}+
		({C_{P_{12},\Delta V_2}^{(\nu)}})^2\cdot
		v_{\Delta V_2 \to f_{r_{P_{12}}}}^{(\nu, \tau-1)}
		\Big].	
		\end{aligned}
		\nonumber
		\end{equation}
				
\item Variable nodes send messages as pairs along incident edges to indirect factor nodes according to \eqref{BP_vf_mean_var}, e.g.:
  	 	\begin{equation}
        \begin{gathered}
		\mu_{\Delta \theta_2 \to f_{r_{P_{12}}}}^{(\nu, \tau)} :=
		\big(r_{\Delta \theta_2 \to f_{r_{P_{12}}}}^{(\nu, \tau)} ,
		v_{\Delta \theta_2 \to f_{r_{P_{12}}}}^{(\nu, \tau)}\big)
		\end{gathered}
		\nonumber
		\end{equation}
		\begin{equation}
        \begin{aligned} 
		\cfrac{1}{v_{\Delta \theta_2 \to f_{r_{P_{12}}}}^{(\nu, \tau)}}&=
		\cfrac{1}{v_{\theta_2}}+
		\cfrac{1}{v_{f_{r_{P_{3}}} \to \Delta \theta_2}^{(\tau)}}\\
		r_{\Delta \theta_2 \to f_{r_{P_{12}}}}^{(\nu, \tau)}&=\Bigg(
		\cfrac{r_{\theta_2}^{(\nu)}}{v_{\theta_2}}+
		\cfrac{r_{f_{r_{P_{3}}} \to \Delta \theta_2}^{(\tau)}}
		{v_{f_{r_{P_{3}}} \to \Delta \theta_2}^{(\tau)}} \Bigg)
		v_{\Delta \theta_2 \to f_{r_{P_{12}}}}^{(\tau)}.
		\end{aligned}
		\nonumber
		\end{equation}	
\end{enumerate}

\noindent
\textbf{Iterate - Outer Loop:} $\bm\nu \mathbf{=0,1,2,\dots}; \bm\tau =\bm{\eta(\bm\nu) }$
\begin{enumerate}[leftmargin=6mm,resume]	
\item Variable nodes compute marginals according to \eqref{BP_marginal_mean_var}, e.g.:
  	 	\begin{equation}
        \begin{gathered}
        p(\Delta \theta_2) \propto 
        \mathcal{N}(\Delta \hat \theta_2^{(\nu)}|\Delta \theta_2, 
        \hat v_{\theta_2}^{(\nu)})
		\end{gathered}
		\nonumber
		\end{equation}	
  	 	\begin{equation}
        \begin{aligned}
        \cfrac{1}{\hat v_{\Delta \theta_2}^{(\nu)}}&=
		\cfrac{1}{v_{\theta_2}}+
		\cfrac{1}{v_{f_{r_{P_{12}}} \to \Delta \theta_2}^{(\tau)}}+
		\cfrac{1}{v_{f_{r_{P_{3}}} \to \Delta \theta_2}^{(\tau)}}\\
		\Delta \hat \theta_2^{(\nu)}&=\Bigg(
		\cfrac{r_{\theta_2}^{(\nu)}}{v_{\theta_2}}+
		\cfrac{r_{f_{r_{P_{12}}} \to \Delta \theta_2}^{(\tau)}}
		{v_{f_{r_{P_{12}}} \to \Delta \theta_2}^{(\tau)}}+
		\cfrac{r_{f_{r_{P_{3}}} \to \Delta \theta_2}^{(\tau)}}
		{v_{f_{r_{P_{3}}} \to \Delta \theta_2}^{(\tau)}} \Bigg)
		\hat v_{\Delta \theta_2}^{(\nu)}.
		\end{aligned}
		\nonumber
		\end{equation}	 		
\item Variable nodes update the state variables, e.g.:
  	 	\begin{equation}
        \begin{gathered}
        \theta_2^{(\nu+1)} = \theta_2^{(\nu)}+\Delta \hat \theta_2^{(\nu)}.
		\end{gathered}
		\nonumber
		\end{equation}
\item Repeat steps 4-13 until convergence.	
\end{enumerate}
\addtocontents{toc}{\endgroup}
	
\end{appendices}

\small
\cleardoublepage
\newpage
\phantomsection
\addcontentsline{toc}{chapter}{Bibliography}
\bibliographystyle{IEEEtran} 
\bibliography{thesis_bibliography}

\end{document}